\pdfoutput=1

\documentclass[sigconf, nonacm]{acmart}

\newcommand{\nop}[1]{}

\usepackage[linesnumbered,ruled,vlined]{algorithm2e}
\usepackage{algpseudocode}
\usepackage{color}
\usepackage{amsmath}
\usepackage{multirow}
\usepackage{subfigure}
\usepackage{soul}
\usepackage{balance} 
\usepackage[normalem]{ulem}
\usepackage{setspace}
\usepackage{flushend}
\usepackage{bm}
\usepackage{mathbbol}

\newtheorem{example}{Example}
\newtheorem{definition}{Definition}

\definecolor{Xiang}{rgb}{1,0,0}

\newcommand\vldbdoi{XX.XX/XXX.XX}
\newcommand\vldbpages{XXX-XXX}
\newcommand\vldbvolume{14}
\newcommand\vldbissue{1}
\newcommand\vldbyear{2020}
\newcommand\vldbauthors{\authors}
\newcommand\vldbtitle{\shorttitle} 
\newcommand\vldbavailabilityurl{URL_TO_YOUR_ARTIFACTS}
\newcommand\vldbpagestyle{plain} 

\begin{document}
\title{Efficient Exact Subgraph Matching via GNN-based Path Dominance Embedding (Technical Report)}

\author{Yutong Ye}
\affiliation{%
  \institution{East China Normal University}
  \city{Shanghai}
  \country{China}
}
\email{52205902007@stu.ecnu.edu.cn}

\author{Xiang Lian}
\affiliation{%
  \institution{Kent State University}
  \city{Kent}
  \state{Ohio}
  \country{USA}
}
\email{xlian@kent.edu}

\author{Mingsong Chen}
\affiliation{%
  \institution{East China Normal University}
  \city{Shanghai}
  \country{China}
}
\email{mschen@sei.ecnu.edu.cn}

\begin{abstract}
The classic problem of \textit{exact subgraph matching} returns those subgraphs in a large-scale data graph that are isomorphic to a given query graph, which has gained increasing importance in many real-world applications such as social network analysis, knowledge graph discovery in the Semantic Web, bibliographical network mining, and so on. In this paper, we propose a novel and effective \textit{graph neural network (GNN)-based path embedding framework} (GNN-PE), which allows efficient \textit{exact subgraph matching} without introducing \textit{false dismissals}. Unlike traditional GNN-based graph embeddings that only produce \textit{approximate} subgraph matching results, in this paper, we carefully devise GNN-based embeddings for paths, such that: if two paths (and 1-hop neighbors of vertices on them) have the subgraph relationship, their corresponding GNN-based embedding vectors will strictly follow the dominance relationship. With such a newly designed property of path dominance embeddings, we are able to propose effective pruning strategies based on path label/dominance embeddings and guarantee no false dismissals for subgraph matching. We build multidimensional indexes over path embedding vectors, and develop an efficient subgraph matching algorithm by traversing indexes over graph partitions in parallel and applying our pruning methods. We also propose a cost-model-based query plan that obtains query paths from the query graph with low query cost. To further optimize our GNN-PE approach, we also propose a more efficient \textit{GNN-based path group embedding} (GNN-PGE) technique, which performs subgraph matching over grouped path embedding vectors. We design effective pruning strategies (w.r.t. grouped path embeddings) that can significantly reduce the search space during the index traversal. Through extensive experiments, we confirm the efficiency and effectiveness of our proposed GNN-PE and GNN-PGE approaches for exact subgraph matching on both real and synthetic graph data.
\end{abstract}

\maketitle

\pagestyle{\vldbpagestyle}
\begingroup\small\noindent\raggedright\textbf{PVLDB Reference Format:}\\
\vldbauthors. \vldbtitle. PVLDB, \vldbvolume(\vldbissue): \vldbpages, \vldbyear. \\
\href{https://doi.org/\vldbdoi}{doi:\vldbdoi}
\endgroup
\begingroup
\renewcommand\thefootnote{}\footnote{\noindent
This work is licensed under the Creative Commons BY-NC-ND 4.0 International License. Visit \url{https://creativecommons.org/licenses/by-nc-nd/4.0/} to view a copy of this license. For any use beyond those covered by this license, obtain permission by emailing \href{mailto:info@vldb.org}{info@vldb.org}. Copyright is held by the owner/author(s). Publication rights licensed to the VLDB Endowment. \\
\raggedright Proceedings of the VLDB Endowment, Vol. \vldbvolume, No. \vldbissue\ %
ISSN 2150-8097. \\
\href{https://doi.org/\vldbdoi}{doi:\vldbdoi} \\
}\addtocounter{footnote}{-1}\endgroup

\ifdefempty{\vldbavailabilityurl}{}{
\vspace{.3cm}
\begingroup\small\noindent\raggedright\textbf{PVLDB Artifact Availability:}\\
The source code, data, and/or other artifacts have been made available at \url{https://github.com/JamesWhiteSnow/GNN-PE}.
\endgroup
}
\section{Introduction}

For the past decades, graph data management has received much attention from the database community, due to its wide spectrum of real applications such as the Semantic Web \cite{orogat2022smartbench},
social networks \cite{wasserman1994social}, biological networks (e.g., gene regulatory networks \cite{karlebach2008modelling} and protein-to-protein interaction networks \cite{szklarczyk2015string}), road networks \cite{chen2009monitoring,zhang2022relative}, and so on. In these graph-related applications, one of the most important and classic problems is the \textit{subgraph matching} query, which retrieves subgraphs $g$ from a large-scale data graph $G$ that match with a given query graph pattern $q$.

Below, we give an example of the subgraph matching in real applications of skilled team formation in collaboration networks.

\begin{figure}[t]
    \centering
    \subfigure[collaboration social network $G$]{
        \includegraphics[height=3.4cm]{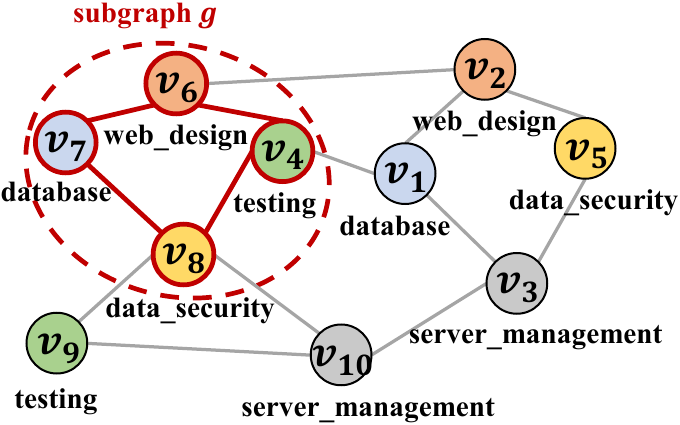}
        \label{subfig:data_graph}
    }
    \subfigure[query graph $q$]{
        \includegraphics[height=2.5cm]{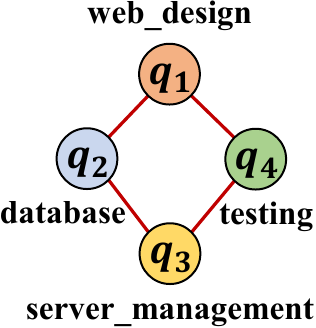}
        \label{subfig:query_graph}
    }
    \caption{An example of the subgraph matching in collaboration social networks.}
    \label{fig:matching}
\end{figure}

\begin{example}
{\bf (Skilled Team Formation in Collaboration Social Networks \cite{AnagnostopoulosBCGL12})}  In order to successfully accomplish a task, a project manager is interested in finding an experienced team that consists of members with complementary skills and having previous collaboration histories. Figure~\ref{subfig:data_graph} illustrates a collaboration social network, $G$, which contains user vertices, $v_1 \sim v_{10}$, with skill labels (e.g., $v_7$ with the skill ``database'') and edges (each connecting two user vertices, e.g., $v_7$ and $v_8$, indicating that they have collaborated in some project before). 
Figure~\ref{subfig:query_graph} shows a query graph $q$, specified by the project manager, which involves the required team members, $q_1 \sim q_4$, with specific skills and their historical collaboration requirements (e.g., the edge between nodes $q_1$ and $q_2$ that indicates the front-end and back-end collaborations). In this case, the project manager can specify this query graph $q$ and issue a subgraph matching query over the collaboration network $G$ to obtain candidate teams matching with $q$ (e.g., subgraph $g$ isomorphic to $q$, circled in Figure~\ref{subfig:data_graph}).\qquad $\blacksquare$
\end{example}

The subgraph matching has many other real applications. For example, in the Semantic Web application like the knowledge graph search  \cite{lian2011efficient}, a SPARQL query can be transformed to a query graph $q$, and thus the SPARQL query answering is equivalent to a subgraph matching query over an RDF knowledge graph, which retrieves RDF subgraphs isomorphic to the transformed query graph $q$. 

\noindent {\bf Prior Works.} The subgraph isomorphism problem is known to be NP-complete \cite{lewis1983michael,cordella2004sub,grohe2020graph}, which is thus not tractable. Prior works on the subgraph matching problem usually followed the filter-and-refine paradigm \cite{jin2021fast,yuan2021subgraph,kim2021versatile,wang2022reinforcement,zhang2022hybrid}, which first filters out subgraph false alarms with no \textit{false dismissals} and then returns actual matching subgraphs by refining the remaining candidate subgraphs.  

Due to the high computation cost of \textit{exact} subgraph matching, an alternative direction is to quickly obtain \textit{approximate} subgraph matching results, trading the accuracy for efficiency. Previous works on approximate subgraph matching \cite{du2017first,zhu2011structure,dutta2017neighbor,li2018efficient} usually searched $k$ most similar subgraphs in the data graph by using various graph similarity measures (e.g., graph edit distance \cite{zhu2011structure,li2018efficient}, chi-square statistic \cite{dutta2017neighbor}, and Sylvester equation \cite{du2017first}). 

Moreover, several recent works \cite{bai2019simgnn,xu2019cross,li2019graph} utilized deep-learning-based approaches such as \textit{Graph Neural Networks} (GNNs) to conduct approximate subgraph matching. Specifically, GNNs can be used to transform entire (small) complex data graphs into vectors in an embedding space offline. Then, we can determine the subgraph relationship between data and query graphs by comparing their embedding vectors, via either neural networks \cite{bai2019simgnn,xu2019cross} or similarity measures (e.g., Euclidean or Hamming distance \cite{li2019graph}). Although GNN-based approaches can efficiently, but approximately, assert the subgraph relationship between any two graphs, there is no theoretical guarantee about the accuracy of such an assertion, which results in approximate (but not exact) subgraph answers. Worse still, these GNN-based approaches usually work for comparing two graphs only, which are not suitable for tasks like retrieving the locations of matching subgraphs in a large-scale data graph.

\noindent {\bf Our Contributions.} In this paper, we focus on \textit{exact subgraph matching} queries over a large-scale data graph, and present a novel \textit{GNN-based path embedding} (GNN-PE) framework for exact and efficient subgraph matching. In contrast to traditional GNN-based embeddings without any evidence of accuracy guarantees, we design an effective \textit{GNN-based path dominance embedding} technique, which trains our newly devised GNN models to obtain embedding vectors of nodes (and their neighborhood structures) on paths such that: {\bf any two paths (including 1-hop neighbors of vertices on them) with the subgraph relationship will strictly yield embedding vectors with the \textit{dominance} relationship \cite{Borzsonyi01}}. This way, through GNN-based path embedding vectors, we can guarantee 100\% accuracy to effectively filter out those path false alarms. In other words, we can retrieve candidate paths (including their locations in the data graph)  matching with those in the query graph with no false dismissals. Further, we propose an effective approach to enhance the pruning power by using multiple sets of GNN-based path embeddings over randomized vertex labels in the data graph.

To deal with large-scale data graphs, we divide the data graph into multiple subgraph partitions and train GNN models for different partitions to enable parallel processing over path embeddings in a scalable manner. We also build an index over path label/dominance embedding vectors for each partition to facilitate the pruning, and develop an efficient and exact subgraph matching algorithm for (parallel) candidate path retrieval and refinement via our proposed cost-model-based query plan. 

Moreover, to further optimize the efficiency of our GNN-PE approach in the conference version \cite{ye2024efficient}, we design a novel \textit{GNN-based path group embedding} (GNN-PGE) approach to group embeddings of paths with the same starting vertices, which are treated as basic units for the index construction. Due to the grouping of path embeddings, we can significantly reduce the space cost of the index. Correspondingly, we propose effective pruning strategies with respect to such grouped path embeddings, which can rule out false alarms of path group embeddings and improve the efficiency of subgraph retrieval.  

In this paper, we make the following contributions:

\begin{enumerate}
    \item We propose a novel GNN-PE framework for designing GNN models to enable path embeddings and exact subgraph matching with no false dismissals in Section \ref{sec:problem_definition}. 

    \item  We design an effective \textit{GNN-based path dominance embedding} approach in Section \ref{sec:embedding}, which embeds data vertices/paths into vectors via GNNs, where their dominance relationships can ensure exact subgraph retrieval.

    \item We propose effective pruning methods for filtering out path false alarms, build aggregate R$^*$-tree indexes over GNN-based path embedding vectors, and develop an efficient parallel algorithm for exact subgraph matching via GNN-based path embeddings in Section \ref{sec:subgraph_matching_alg}.

    \item We devise a novel cost model for selecting the best query plan of the subgraph matching in Section \ref{sec:query_plan}.

    \item We design a \textit{GNN-based path group embedding} (GNN-PGE) technique, which groups path embeddings to further optimize the indexing and subgraph retrieval in Section \ref{sec:gnnpge}.

    \item Through extensive experiments, we evaluate our proposed GNN-PE approach for exact subgraph matching over both real and synthetic graphs and confirm the efficiency and effectiveness of our GNN-PE approach in Section \ref{sec:expr}.
\end{enumerate}

Section \ref{sec:related_work} reviews related works on exact/approximate subgraph matching and GNNs. Finally, Section \ref{sec:conclusions} concludes this paper.

\section{Problem Definition}
\label{sec:problem_definition}

In this section, we formally define the graph data model and subgraph matching queries over a graph database, and propose a GNN-based exact subgraph matching framework. Table \ref{tab:notations} depicts the commonly used symbols and their descriptions in this paper.

\begin{table}[t]\small
\begin{center}
\caption{Symbols and Descriptions}
\label{tab:notations}
\begin{tabular}{|l||p{6cm}|}
\hline
\textbf{Symbol}&\textbf{Description} \\
\hline\hline
    $G$ & a data graph\\\hline
    $q$ & a query graph\\\hline
    $g$ & a subgraph of the data graph $G$\\\hline
    $v_i$ (or $q_i$) & a vertex in graph $G$ (or $q$)\\\hline
    $e_{ij}$ (or $e_{q_iq_j}$) & an edge in graph $G$ (or $q$)\\\hline
    $V(G)$ (or $V(q)$) & a set of vertices $v_i$ (or $q_i$) \\\hline
    $E(G)$ (or $E(q)$) & a set of edges $e_{ij}$ (or $e_{q_iq_j}$) \\\hline
    $\phi(G)$ (or $\phi(q)$) &  a mapping function $V(G)\times V(G)\rightarrow E(G)$ (or $V(q)\times V(q)\rightarrow E(q)$) \\\hline
    $L(G)$ (or $L(q)$) & a labeling function of $G$ (or $q$) \\\hline
    $m$ & the number of graph partitions $G_j$ \\\hline
    $M_j$ & a GNN model for $G_j$ \\\hline
    $g_{v_i}$ (or $s_{v_i}$) & a unit star graph (or substructure) of center vertex $v_i$ \\\hline
    $o(g_{v_i})$ (or $o(v_i)$) & an embedding vector of center vertex $v_i$ from unit star graph $g_{v_i}$ \\\hline
    $o(p_z)$ (or $o(p_q)$) & a path dominance embedding vector of path $p_z$ (or $p_q$)\\\hline    
    $o_0(p_z)$ (or $o_0(p_q)$) & a path label embedding vector of path $p_z$ (or $p_q$) \\\hline
    
\end{tabular}
\end{center}
\end{table}

\subsection{Graph Data Model}
We first give the model for an undirected, labeled graph, $G$, below.

\begin{definition} \textbf{(Graph, $G$)}
A \textit{graph}, $G$, is represented by a quadruple $(V(G), E(G), \phi(G), L(G))$, where $V(G)$ is a set of vertices $v_i$, $E(G)$ is a set of edges $e_{ij}$ ($=(v_i, v_j)$) between vertices $v_i$ and $v_j$, $\phi(G)$ is a mapping function $V(G)$$\times$$V(G)$$\rightarrow$$E(G)$, and $L(G)$ is a labeling function that associates each vertex $v_i\in V(G)$ with a label $L(v_i)$.
\label{def:graph}
\end{definition}

Examples of graphs (as given in Definition~\ref{def:graph}) include social networks \cite{wasserman1994social}, road networks \cite{chen2009monitoring,zhang2022relative}, biological networks (e.g., gene regulatory networks \cite{karlebach2008modelling} and protein-to-protein interaction networks \cite{szklarczyk2015string}),  bibliographical networks \cite{tang2008arnetminer}, and so on.

\subsection{Graph Isomorphism}
In this subsection, we give the definition of the classic graph isomorphism problem between undirected, labeled graphs.

\begin{definition} \textbf{(Graph Isomorphism \cite{babai2018group,grohe2020graph}})
Given two graphs $G_A=(V_A,E_A,$ $\phi_A,L_A)$ and $G_B=(V_B, E_B, \phi_B,$ $L_B)$, we say that $G_A$ is \textit{isomorphic} to $G_B$ (denoted as $G_A \equiv G_B$), if there exists an edge-preserving bijective function $f: V_A\to V_B$, such that: i) $\forall v_i\in V_A$, $L_A(v_i)=L_B(f(v_i))$, and ii) $\forall v_i$, $v_j \in V_A$, if $(v_i, v_j)\in E_A$ holds, we have $(f(v_i),f(v_j))\in E_B$.
\label{def:subiso}
\end{definition}


In Definition \ref{def:subiso}, the graph isomorphism problem checks whether or not two graphs $G_A$ and $G_B$ exactly match each other. 

Moreover, we say that $G_A$ is \textit{subgraph isomorphic} to $G_B$ (denoted as $G_A \subseteq G_B$), if $G_A$ is isomorhic to an induced subgraph, $g_B$, of graph $G_B$. Note that, the subgraph isomorphism problem has been proven to be NP-complete \cite{cordella2004sub,lewis1983michael}.

\subsection{Subgraph Matching Queries}

We now define a \textit{subgraph matching query} over a large-scale data graph $G$, which obtains subgraphs, $g$, that match with a given query graph $q$.

\begin{definition} \textbf{(Subgraph Matching Query)}
Given a data graph $G$ and a query graph $q$, a \textit{subgraph matching query} retrieves all the subgraphs $g$ of the data graph $G$  that are isomorphic to the query graph $q$ (i.e., $g \equiv q$).
\label{def:subquery}
\end{definition}

Exact subgraph matching query (as given in Definition \ref{def:subquery}) has many real-world applications such as social network analysis, small molecule detection in computational chemistry, and pattern matching over biological networks \cite{qiao2017subgraph,sahu2017ubiquity}. For example, in real applications of fraud detection, subgraph matching is useful to detect fraudulent activities (represented by some graph patterns) in activity networks \cite{qiu2018real}. Moreover, in biological networks \cite{alon2007network}, we can also use subgraph matching to retrieve molecule structures (subgraphs) that follow some given structural pattern (i.e., query graph) with certain functions.

\subsection{Challenges}
Due to the NP-completeness of the subgraph matching problem \cite{lewis1983michael}, it is intractable to conduct exact subgraph matching over a large-scale data graph. Therefore, it is rather challenging to design effective optimization techniques to improve the query efficiency of the subgraph matching. In this paper, we will carefully devise effective graph partitioning, pruning, indexing, and refinement mechanisms to facilitate efficient and scalable subgraph matching. 

Moreover, we consider graph node/path embedding by training/using GNN models. Note that, GNNs usually provide approximate solutions (e.g., for the prediction or classification) without an accuracy guarantee. It is therefore not trivial how to utilize GNNs to guarantee the accuracy of exact subgraph matching (i.e., AI for DB without introducing false dismissals). In this work, we analyze the reason for the inaccuracy of using GNNs. Prior works usually trained and used GNNs on distinct training and testing graph data sets. In contrast, our work explores basic units of a finer resolution in the data graph (i.e., star subgraphs), such that the trained GNNs are over the same training/testing graph data set, which can ensure 100\% accuracy for the exact subgraph matching process. 

\begin{algorithm}[t!]
\caption{{\bf The GNN-Based Path Embedding (GNN-PE) Framework for Exact Subgraph Matching}}
\label{alg1}
\KwIn{
    a data graph $G$ and 
    a query graph $q$\\
}
\KwOut{
    subgraphs $g$ ($\subseteq G$) that are isomorphic to $q$
}

\tcp{\bf Offline Pre-Computation Phase}

divide graph $G$ into $m$ disjoint subgraphs $G_1$, $G_2$, ..., and $G_m$

\For{each subgraph partition $G_j$ $(1\leq j\leq m)$}{
    \tcp{train GNN models for graph node/edge embeddings}
    train a GNN model $M_j$ with node dominance embedding over vertices in $G_j$\\
    generate embedding vectors $o(p_z)$ for paths $p_z$ of lengths $l$ in $G_j$ via $M_j$\\
    \tcp{build an index over subgraph $G_j$}
    build aggregate R$^*$-tree indexes, $\mathcal{I}_j$, over embedding vectors for paths of length $l$ in $G_j$
}

\tcp{\bf Online Subgraph Matching Phase}
\For{each query graph $q$}{
    \tcp{retrieve candidate paths}
    compute a cost-model-based query plan $\varphi$ of multiple query paths $p_q$\\
    obtain a query embedding vector $o(q_i)$ of each vertex $q_i$ in $q$ from GNNs $M_j$, and embeddings $o(p_q)$ of query paths $p_q$\\
    find candidate path sets, $p_q.cand\_list$, that match with query paths, $p_q$, by traversing indexes $\mathcal{I}_j$\\
    \tcp{obtain and refine candidate subgraphs}
    assemble candidate subgraphs $g$ from candidate paths in $p_q.cand\_list$ and refine subgraphs $g$ via multi-way hash join\\
    \textbf{return} subgraphs $g$ ($\equiv q$)
}

\end{algorithm}

\subsection{GNN-Based Exact Subgraph Matching Framework}
\label{subsec:framework}

Algorithm~\ref{alg1} presents a novel \textit{GNN-based path embedding} (GNN-PE) framework for efficiently answering subgraph matching queries via path embeddings, which consists of two phases, offline pre-computation and online subgraph matching phases. That is, we first pre-process the data graph $G$ offline by building indexes over path embedding vectors via GNNs (lines 1-5), and then answer online subgraph matching queries over indexes (lines 6-11). 

Specifically, in the offline pre-computation phase, we first use the METIS \cite{karypis1998fast} to divide the large-scale data graph $G$ into $m$ disjoint subgraph partitions $G_1$, $G_2$, ..., and $G_m$, which aims to minimize the number of edge cuts (line 1). 
Then, for each subgraph partition $G_j$ ($1\leq j\leq m$), we train a GNN model $M_j$ by using our well-designed dominance embedding loss function, where the training data set contains all possible vertices in $G_j$ and their neighborhoods in $G$ (i.e., star subgraph structures; lines 2-3).
Here, the GNN model training (until the loss equals 0) is conducted offline over multiple subgraph partitions $G_j$ in parallel, which can achieve low training costs.
After that, for each vertex $v_i \in V(G_j)$, we can generate an embedding vector $o(v_i)$ via the trained GNN $M_j$.
For any paths $p_z$ of length $l$, we concatenate embedding vectors of adjacent vertices in order on paths $p_z$ to obtain path embedding vectors $o(p_z)$ (line 4). Next, for each partition $G_j$, we build an aR$^*$-tree index \cite{Lazaridis01,beckmann1990r}, $\mathcal{I}_j$, over embedding vectors of paths $p_z$ with length $l$ (starting from vertices in $G_j$) to facilitate exact subgraph matching (line 5).

In the online subgraph matching phase, given any query graph $q$, we first compute a cost-model-based query plan $\varphi$ that divides the query graph $q$ into multiple query paths $p_q$ (lines 6-7). 
Then, we obtain node embedding vectors $o(q_i)$ of query vertices $q_i \in V(q)$ in $q$ via GNNs in $M_j$ (for $1\leq j\leq m$), and embeddings of query paths $p_q$ by concatenating node embedding vectors on them (line 8). Next, we traverse indexes $\mathcal{I}_j$ (in parallel) to retrieve candidate path sets, $p_q.cand\_list$, of query paths $p_q$ (line 9). Finally, we combine candidate paths in $p_q.cand\_list$ to assemble candidate subgraphs $g$ and refine them via the multi-way hash join to return actual matching subgraphs $g$ (lines 10-11).

\section{GNN-based Dominance Embedding}
\label{sec:embedding}

In this section, we discuss how to calculate GNN-based dominance embeddings for vertices/paths (lines 3-4 of Algorithm \ref{alg1}), which can enable subgraph relationships to be preserved in the embedding space and support efficient and accurate path candidate retrieval. Section \ref{subsec:GNN} designs a GNN model for computing node embeddings in the data graph. Section \ref{subsec:node_embedding} presents a loss function used in the GNN training, which can guarantee the dominance relationships of node embeddings (when the loss equals zero) for exact subgraph matching with no false dismissals. Finally, Section \ref{subsec:path_embedding} combines embeddings of nodes on paths to obtain path dominance embeddings, which are used for online exact subgraph matching.

\subsection{GNN Model for the Node Embedding}
\label{subsec:GNN}

In this work, we use a GNN model (e.g., \textit{Graph Attention Network} (GAT) \cite{velivckovic2018graph}) to enable the node embedding in the data graph $G$. Specifically, the GNN takes a \textit{unit star graph} $g_{v_i}$ (i.e., a star subgraph containing a center vertex $v_i \in V(G)$ and its 1-hop neighbors) as input and an embedding vector, $o(v_i)$, of vertex $v_i$ as output. Figure~\ref{fig:gnnmodel} illustrates an example of this GNN model (with unit star graph $g_{v_1}$ as input), which consists of input, hidden, and output layers.

\begin{figure}[t]
    \centering
    \includegraphics[scale=0.33]{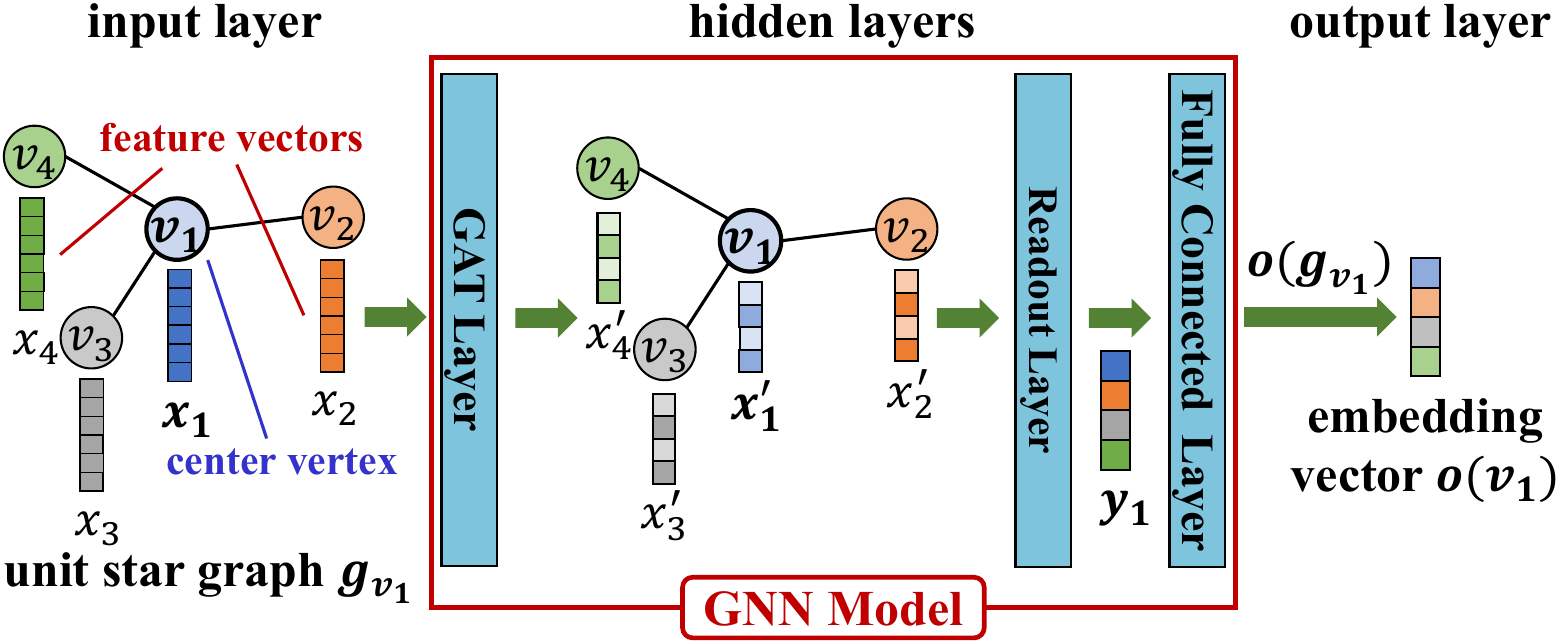}
    \caption{Illustration of our GNN model.}
    \label{fig:gnnmodel}
\end{figure}

\noindent{\bf Input Layer.}  As mentioned earlier, the input of the GNN model is a unit star graph $g_{v_i}$ (or its star substructure/subgraph, denoted as $s_{v_i}$). 
Each vertex $v_j$ in $g_{v_i}$ (or $s_{v_i}$) is associated with an initial feature vector $x_j$ of size $F$, which is obtained via either vertex label encoding or one-hot encoding \cite{bisong2019introduction}.

Figure~\ref{fig:graphdata} shows an example of unit star graph $g_{v_1}$ in data graph $G$ and one of its star substructures $s_{v_1}$ ($\subseteq g_{v_1}$), which are centered at vertex $v_1$. In Figure \ref{subfig:unit_star_graph}, vertices $v_1 \sim v_4$ have their initial feature vectors $x_1 \sim x_4$, respectively. The case of the star substructure $s_{v_1}$ in Figure \ref{subfig:substructure} is similar.

\noindent{\bf Hidden Layers.} As shown in Figure \ref{fig:gnnmodel}, we use three hidden layers in our GNN model, including GAT \cite{velivckovic2018graph}, readout, and fully connected layers.

\underline{\it GAT layer:} In the first GAT layer, the feature vector of each vertex will go through a linear transformation parameterized by a weight matrix $\textbf{W}\in \mathbb{R}^{F'\times F}$. 
Specifically, we compute an \textit{attention coefficient}, $ac_{v_iv_j}$, between any vertices $v_i$ and $v_j$ as follows:
\begin{equation}
    ac_{v_iv_j} = a\big(\textbf{W} x_i, \textbf{W} x_j\big),
\end{equation}
which indicates the importance of vertex $v_i$ to vertex $v_j$, 
where the shared attentional mechanism $a(\cdot, \cdot)$ is a function (e.g., a single-layer neural network with learnable parameters): $\mathbb{R}^{F'}\times \mathbb{R}^{F'}\rightarrow \mathbb{R}$ that outputs the correlation between two feature vectors.

Denote $\mathcal{N}(v_i)$ as the neighborhood of a vertex $v_i$. For each vertex $v_i$, we aggregate feature vectors of its 1-hop neighbors $v_j \in \mathcal{N}(v_i)$. That is, we first use a softmax function to normalize attention coefficients $ac_{v_iv_j}$ as follows:
\begin{equation}
    \alpha_{v_iv_j}=softmax(ac_{v_iv_j})=\frac{exp(ac_{v_iv_j})}{\sum_{v_k\in \mathcal{N}(v_i)}exp(ac_{v_iv_k})}.
\end{equation}

Then, the output of the GAT layer is computed by a linear combination of feature vectors:
\begin{equation}
    x'_i=\sigma \left(\sum_{v_j\in \mathcal{N}_{v_i}}\alpha_{v_iv_j}\textbf{W}x_j\right),
    \label{eq:singlehead}
\end{equation}
where $\sigma(Z)$ is a nonlinear activation function (e.g., rectified linear unit \cite{nair2010rectified}, Sigmoid \cite{han1995influence}, etc.) with input and output vectors, $Z$ and $\sigma(Z)$, of length $F'$, respectively. 

To stabilize the learning process of self-attention, GAT uses a multi-head attention mechanism similar to \cite{vaswani2017attention}, where each head is an independent attention function and $K$ heads can execute the transformation of Eq.~(\ref{eq:singlehead}) in parallel. Thus, an alternative GAT output $x'_i$ can be a concatenation of feature vectors generated by $K$ heads:
\begin{equation}
    x'_i=\big\Vert^{K}_{k=1}\sigma\left(\sum_{v_j\in \mathcal{N}_{v_i}}\alpha^{(k)}_{v_iv_j}\textbf{W}^{(k)}x_{j}\right),
\end{equation}
where $\Vert^{K}_{k=1}$ is the concatenation operator of $K$ vectors, $\alpha^{(k)}_{v_1v_j}$ is the normalized attention coefficient computed by the $k$-th attention head, and $\textbf{W}^{(k)}$ is the $k$-th parameterized weight matrix.

\begin{figure}[t!]
    \centering
    \subfigure[][{\small unit star graph $g_{v_1}\!\subseteq\!G$}]{
        \includegraphics[height=3.4cm]{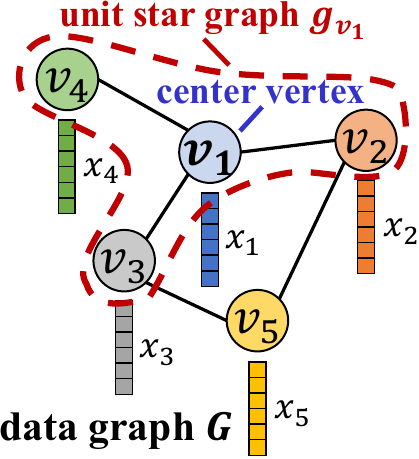}
        \label{subfig:unit_star_graph}
    }\qquad\qquad
    \subfigure[][{\small star substructure $s_{v_1}\!\subseteq\!g_{v_1}$}]{
        \includegraphics[height=3.4cm]{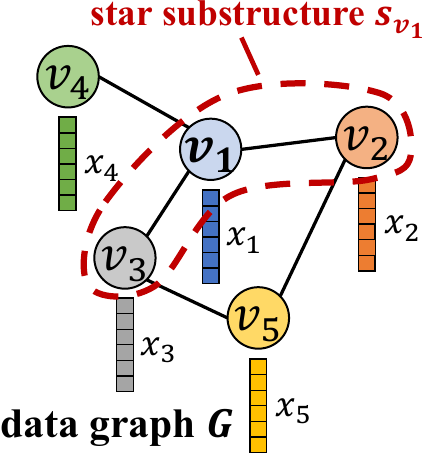}
        \label{subfig:substructure}
    }
    \caption{Illustration of the input for the GNN model.}
    \label{fig:graphdata}
\end{figure}

\underline{\it Readout layer:} A readout layer \cite{ying2018hierarchical,zhang2018end} generates an embedding vector, $y_i$, for the entire unit star graph $g_{v_i}$, by summing up feature vectors $x_j'$ of all vertices $v_j$ in $g_{v_i}$, which is permutation invariant. That is, we obtain:

\begin{equation}
    y_i =\sum_{\forall v_j \in V(g_{v_i})}x_j'.\label{eq:eq5}
\end{equation}

\underline{\it Fully Connected Layer:} A fully connected layer performs a nonlinear transformation of $y_i$ (given in Eq.~(\ref{eq:eq5})) via an activation function $\sigma(\cdot)$ and obtains the embedding vector, $o(g_{v_i})$, of size $d$ for vertex $v_i$. That is, we have:
\begin{equation}
    o(g_{v_i})=\sigma(\mathbb{W}y_i),\label{eq:eq6}
\end{equation}
where $\sigma(\cdot)$ is an activation function and $\mathbb{W}$ is a $d \times (K \cdot F')$ weight matrix. 

In this paper, we used the Sigmoid activation function
$\sigma(x)=\frac{1}{1+e^{-x}}$ $(\in (0,1))$
, where $e$ is a mathematical constant. We will leave the study of using other activation functions as our future work.

\noindent{\bf Output Layer.} In this layer, we output $o(g_{v_i})$ (given in Eq.~(\ref{eq:eq6})) as the embedding vector, $o(v_i)$, of center vertex $v_i$.

\subsection{Node Dominance Embedding}
\label{subsec:node_embedding}

In this subsection, we propose an effective \textit{GNN-based node dominance embedding} approach to train the GNN model (discussed in Section \ref{subsec:GNN}), such that our node embedding via GNN can reflect the subgraph relationship between unit star graph $g_{v_i}$ and its star substructures $s_{v_i}$ in the embedding space. Such a node dominance embedding can enable exact subgraph matching.

\noindent {\bf Loss Function.} Specifically, for the GNN model training, we use a training data set $D_j$ that contains \textit{all} pairs of unit star graphs $g_{v_i}$ and their substructures $s_{v_i}$ (for all vertices $v_i$ in the subgraph partition $G_j$ of data graph $G$). 

Then, we design a loss function $\mathcal{L}(D_j)$ over a training data set $D_j$ as follows:
\begin{equation}
    \mathcal{L}(D_j)=\sum_{\forall (g_{v_i},s_{v_i})\in D_j}\big\Vert max\big\{0,o(s_{v_i})-o(g_{v_i})\big\}\big\Vert_2^2,
    \label{eq:loss}
\end{equation}
where $o(g_{v_i})$ and $o(s_{v_i})$ are the embedding vectors of a unit star graph $g_{v_i}$ and its substructures $s_{v_i}$, respectively, and $||$$\cdot$$||_2$ is the $L_2$-norm.

\noindent {\bf GNN Model Training.} To guarantee that we do not lose any candidate vertices for exact subgraph matching, we train the GNN model, until the loss function $\mathcal{L}(\cdot)$ (given in Eq.~(\ref{eq:loss})) is equal to 0. 

Intuitively, from Eq.~(\ref{eq:loss}), when the loss function $\mathcal{L}(D_j)$ $ = 0$ holds, the embedding vector $o(s_{v_i})$ of any star substructure $s_{v_i}$ is \textit{dominating} \cite{Borzsonyi01} (or equal to) that, $o(g_{v_i})$, of its corresponding unit star graph $g_{v_i}$. In the sequel, we will simply say that $o(s_{v_i})$ \textit{dominates} $o(g_{v_i})$ (denoted as $o(s_{v_i}) \preceq o(g_{v_i})$), if $o(s_{v_i})[t] \leq o(g_{v_i})[t]$ for all $0\leq t < d$ (including the case where $o(s_{v_i}) = o(g_{v_i})$).

In other words, given the subgraph relationship between $s_{v_i}$ and $g_{v_i}$ (i.e., $s_{v_i} \subseteq g_{v_i}$), our \textit{GNN-based node dominance embedding} approach can always guarantee that their embedding vectors $o(s_{v_i})$ and $o(g_{v_i})$ follow the dominance relationship (i.e.,  $o(s_{v_i}) 	\preceq o(g_{v_i})$).

Algorithm~\ref{alg2} illustrates the training process of a GNN model $M_j$ over a subgraph parititon $G_j$ ($1\leq j\leq m$). 
For each vertex $v_i \in G_j$, we obtain all (shuffled) pairs of unit star subgraphs $g_{v_i}$ and their star substructures $s_{v_i}$, which result in a training data set $D_j$ (lines 1-5).
Then, for each training iteration, we use a training epoch to update model parameters (lines 6-10), and a testing epoch to obtain the loss $L_e$ (lines 11-15). The training loop of the GNN model $M_j$ terminates until the loss $L_e$ equals 0 (line 16).

To obtain GNN-based node embeddings with high pruning power, we train multiple GNN models with $b$ sets of random initial weights (line 17), and select the one with zero loss and the highest quality of the generated node embeddings (i.e., the lowest expected query cost, $Cost_{M_j}$, as discussed in Eq.~(\ref{eq:cost_model_select}) below; line 18). Finally, we return the best trained GNN model $M_j$ (line 19).

\begin{algorithm}[htp]
\caption{\bf GNN Model Training}
\label{alg2}
\KwIn{
    i) a subgraph partition $G_j \subseteq G$;
    ii) a training data set $D_j$, and;
    iii) a learning rate $\eta$
}
\KwOut{
    a trained GNN model $M_j$
}

\SetKwRepeat{Do}{do}{while}%

\tcp{\bf generate a training data set $D_j$}
$D_j = \emptyset$\\
\For{each vertex $v_i\in V(G_j)$}{
    obtain the unit star subgraph $g_{v_i}$ and its star substructures $s_{v_i}$ \\
    add all pairs $(g_{v_i}, s_{v_i})$ to $D_j$
}
randomly shuffle pairs in $D_j$\\

\tcp{\bf train a GNN model $M_j$ until the loss equals to $0$}

\Do{$(L_e=0)$}{
    \tcp{the training epoch}
    \For{each batch $B \subseteq D_j$}{
        obtain embedding vectors of pairs in $B$ by $M_j$\\
        compute the loss function $\mathcal{L}(B)$ of $M_j$ by Eq.~(\ref{eq:loss})\\
        $M_j.update(\mathcal{L}(B),\eta)$\\
    }
    \tcp{the testing epoch}
    $L_e = 0$\\
    \For{each batch $B \subseteq D_j$}{
        obtain embedding vectors of pairs in $B$ by $M_j$\\
        compute the loss function $\mathcal{L}(B)$ of $M_j$ by Eq.~(\ref{eq:loss})\\
        $L_e\leftarrow L_e+\mathcal{L}(B)$\\
    }
}
repeat lines 6-16 to train $b$ GNN models with random initial weight parameters to avoid local optimality\\
select the best model $M_j$ (satisfying $L_e=0$) with the smallest expected query cost $Cost_{M_j}$ (given in Eq.~(\ref{eq:cost_model_select}))\\
\Return the best trained GNN model $M_j$\\
\end{algorithm}

\noindent{\bf Complexity Analysis of the GNN Training.}
In Algorithm \ref{alg2}, we extract all star substructures $s_{v_i}$ from each unit star subgraph $g_{v_i}$ (line 3), which are used for GNN training. Thus, the time complexity of the GNN training is given by $O\big(\sum_{v_i\in V(G_j)}2^{deg(v_i)}\big)$, where $deg(v_i)$ is the degree of vertex $v_i$.
On the other hand, we train $b$ GNN models (as mentioned in Section~\ref{subsec:GNN}) with different initial weights to achieve high pruning power (line 17). For each GNN model, the time complexity of the computation on the GAT layer is $O((|V(g_{v_i})|+|E(g_{v_i})|)\cdot F')$, and that for the fully connected layer is $O(F'\cdot d)$. 
In this paper, for the input unit star subgraph $g_{v_i}$ (or star substructure $s_{v_i}$), we have $|V(g_{v_i})|=deg(v_i)+1$ and $|E(g_{v_i})|=deg(v_i)$. Therefore, the total time complexity of the GNN training is given by $O\big(b \cdot \sum_{v_i\in V(G_j)}2^{deg(v_i)}\cdot (2\cdot deg(v_i)+d+1)\cdot F'\cdot \mathbb{N}\big)$, where $\mathbb{N}$ is the number of training epochs until zero loss.

Since vertex degrees in real-world graphs usually follow the power-law distribution \cite{barabasi1999emergence,newman2005power,barabasi2003scale,albert2002statistical,dorogovtsev2003evolution}, only a small fraction of vertices have high degrees. For example, in US Patents graph data \cite{sun2020memory}, the average vertex degree is around 9, which incurs about 512 ($=2^9$) star substructures per vertex. Thus, this is usually acceptable for offline GNN training on a single machine.

In practice, for vertices $v_i$ with high degrees $deg(v_i)$ (e.g., greater than a threshold $\theta$), instead of enumerating a large number of $2^{deg(v_i)}$ star substructures, we simply set their embeddings $o(v_i)$ to all-ones vectors $\mathbb{1}$. The rationale behind this is that embeddings of those high-degree vertices often have low pruning power, and it would be better to directly consider them as vertex candidates without costly star substructure enumeration/training. 
This way, our GNN training complexity is reduced to $O\big(b \cdot \sum_{v_i\in V(G_j),deg(v_i)\leq \theta}\\2^{deg(v_i)}\cdot (2\cdot deg(v_i)+d+1)\cdot F'\cdot \mathbb{N}\big)$.

\noindent {\bf Usage of the Node Dominance Embedding for Exact Subgraph Matching.} Intuitively, with the node dominance embeddings, we can convert exact subgraph matching into the dominance search problem in the embedding space. Specifically, if a vertex $q_i$ in the query graph $q$ matches with a vertex $v_i$ in some subgraph $g$ of $G$, then it must hold that $o(g_{q_i}) \preceq o(g_{v_i})$, where $o(g_{q_i})$ is an embedding vector of vertex $q_i$ (and its 1-hop neighbors) in query graph $q$ via the trained GNN. 

This way, we can always use the embedding vector $o(g_{q_i})$ of $q_i$ to retrieve candidate vertices $v_i$ in $G$ (i.e., those vertices with embedding vectors $o(g_{v_i})$ dominated by $o(g_{q_i})$ in the embedding space). Our trained GNN with zero training loss can guarantee that vertices $v_i$ dominated by $o(g_{q_i})$ will not miss any truly matching vertices (i.e., with 100\% recall ratio). This is because all possible query star structures $g_{q_i}$ have already been offline enumerated and trained during the GNN training process (i.e., $g_{q_i} \equiv g_{s_i}$).

\underline{\it The Quality of the Generated Node Embeddings:} Note that, $b$ GNN models can produce at most $b$ sets of node embeddings that can fully satisfy the dominance relationships (i.e., with zero loss). Thus, in line 18 of Algorithm \ref{alg2}, we need to select one GNN model with the best \textit{node embedding quality}, which is defined as the expected query cost $Cost_{M_j}$ (or the expected number of embedding vectors, $o(g_{v_x})$, generated from unit star subgraphs dominated by that, $o(g_{q_i})$, of query star subgraphs) below:
\begin{equation}
    Cost_{M_j}=\frac{\sum_{\forall g_{q_i}} \left|\left\{ \forall v_x\in V(G) \text{ | } o(g_{q_i}) \preceq o(g_{v_x}) \right\}\right|}{\# \text{ of query unit star subgraphs } g_{q_i}},
    \label{eq:cost_model_select}
\end{equation}
where  $g_{q_i}$ are all the possible query unit star subgraphs (i.e., all star substructures $s_{v_x}$ extracted from $g_{v_x}$), and the $\#$ of possible query unit star subgraphs $g_{q_i}$ is given by $|D_j|$.

\begin{example} 
{\it Figure~\ref{fig:nodeembedding} illustrates an example of the node dominance embedding between data graph $G$ and query graph $q$. Each vertex $v_i \in V(G)$ has a 2D embedding vector $o(v_i)$ via the GNN, whereas each query vertex $q_i\in V(q)$ is transformed to a 2D vector $o(q_i)$. For example, as shown in the tables, we have $o(v_1) = (0.78, 0.79)$  and $o(q_1) = (0.62, 0.61)$. We plot the embedding vectors of vertices in a 2D embedding space on the right side of the figure. We can see that $o(q_1)$ is dominating $o(v_1)$ and $o(v_2)$, which implies that $g_{q_1}$ is potentially a subgraph of (i.e., matching with) $g_{v_1}$ and $g_{v_2}$. On the other hand, since $o(q_1)$ is not dominating $o(v_3)$ in the 2D embedding space, query vertex $q_1$ cannot match with vertex $v_3$ in the data graph $G$. \qquad $\blacksquare$}
\end{example}

\noindent{\bf Multi-GNN Node Dominance Embedding.}
In order to further reduce the number of candidate vertices $v_i$ that match with a query vertex $q_i\in V(q)$ (or enhance the pruning power), we use multiple independent GNNs to embed vertices. Specifically, for each subgraph partition $G_j$, we convert the label of each vertex to a new randomized label (e.g., via a hash function or using the vertex label as a seed to generate a pseudo-random number). This way, we can obtain a new subgraph $G_j'$ with the same graph structure, but different vertex labels, and train/obtain a new GNN model $M_j'$ with new embeddings $o'(v_i)$ of vertices $v_i \in V(G_j')$.

\begin{figure}[!htp]
    \centering
    \includegraphics[scale=0.35]{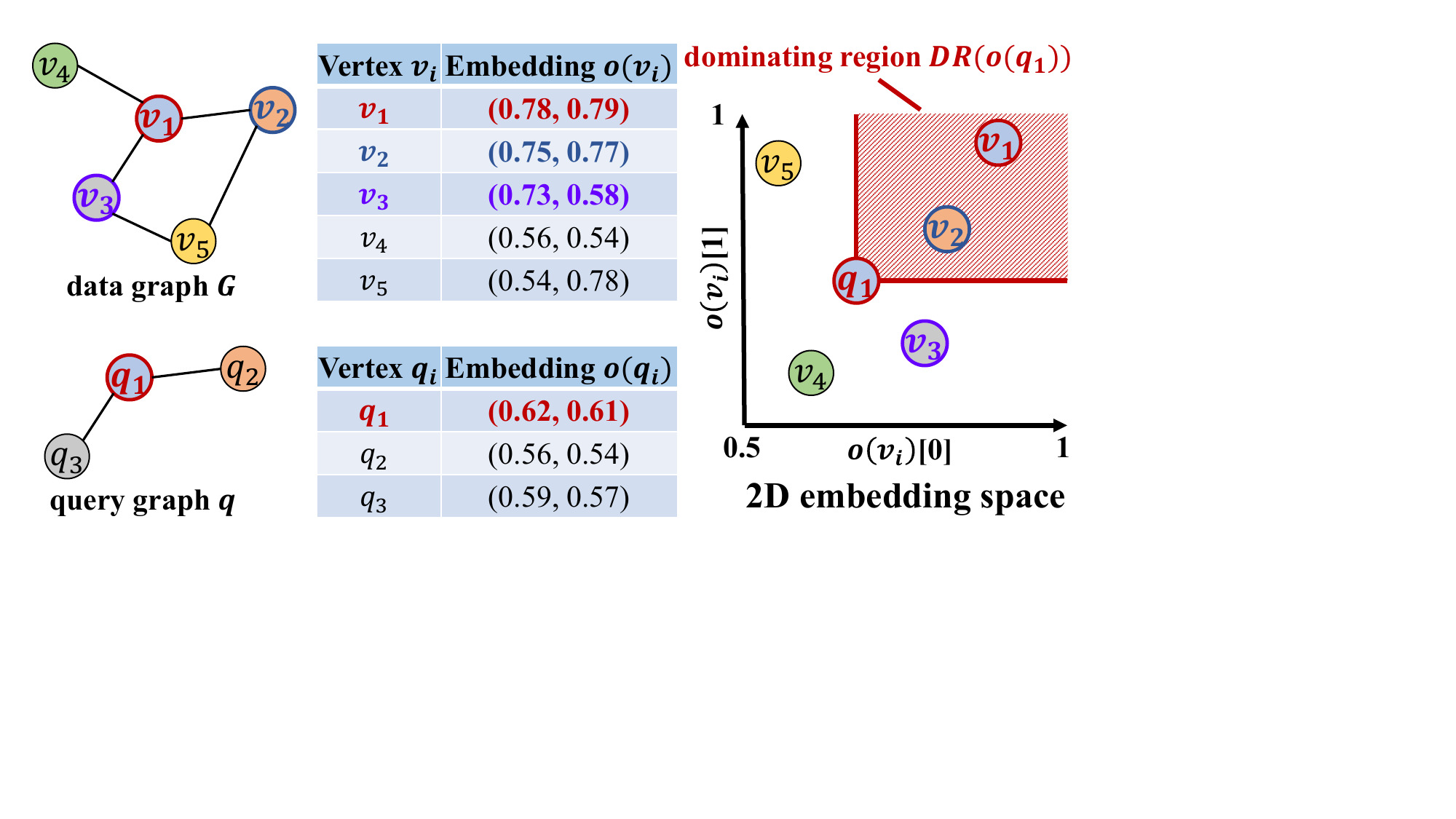}
    \caption{An example of node dominance embedding.}
    \label{fig:nodeembedding}
\end{figure}

With multiple versions of randomized vertex labels, we can obtain different vertex embeddings $o'(v_i)$, where the subgraph relationships between the unit star subgraph and its star substructures also follow the dominance relationships of their embedding vectors. 

Therefore, for each vertex $v_i$, we can compute different versions (via different GNNs) of embedding vectors (e.g., $o(v_i)$ and $o'(v_i)$), which can be used together for retrieving candidate vertices $v_i$ that match with a query vertex $q_i$ with higher pruning power.

To simplify the notations, in the following discussion, we will present the embeddings using a single GNN. Nevertheless, our proposed approach can be easily generalized to multi-GNNs.

\noindent{\bf Convergence Analysis.} To guarantee no false dismissals for online subgraph matching, we need to offline train a GNN model $M_j$, until the training loss equals zero (i.e., $\mathcal{L}(D_j) = 0$). Below, we give the convergence analysis of the GNN model training, including i) parameter settings for achieving sufficient GNN capacity, ii) the existence of GNN parameters to achieve the training goal, and iii) the target accessibility of the GNN training.

\underline{\it GNN Model Capacity for Zero-Loss Training:} As mentioned in \cite{goodfellow2016deep}, the GNN model needs enough \textit{capacity} to achieve zero loss on the training data sets. In particular, the capacity of a GNN model $M_j$ \cite{loukas2020graph} can be defined as $M_j.cap=M_j.dep\times M_j.wid$, where $M_j.dep$ and $M_j.wid$ are the depth (i.e., $\#$ of layers) and width (i.e., the maximum dimension of intermediate node embeddings in all layers) of the GNN model $M_j$, respectively. 

According to \cite{loukas2020graph}, the GNN model's capacity to accomplish graph tasks with zero-loss training needs to satisfy the following condition:
\begin{equation}
    M_j.cap \geq \tilde{\Omega}(|V(\cdot)|^\delta),
\end{equation}
where $\delta \in [1/2,2]$ is an exponent factor reflecting the complexity of the task (e.g., solving some NP-hard problems necessitates $\delta=2$), $|V(\cdot)|$ is the size of the input graph, and $\tilde{\Omega}(\cdot)$ indicates that as the graph size increases, the GNN model capacity also increases with the same rate (up to a logarithmic factor).

In our subgraph matching task, the GNN model aims to learn a partial order between a unit star subgraph $g_{v_i}$ and its star substructure $s_{v_i}$ in each pair of the training data set $D_j$. Since this partial order exists in each individual pair, the GNN capacity is only relevant to the maximum input size of unit star subgraphs, i.e., $\max_{\forall g_{v_i}\in D_j}\{|V(g_{v_i})|\}^\delta$. Therefore, by overestimating the $\delta$ value (i.e., $\delta=2$), we have the lower bound of the GNN capacity $M_j.cap$ to solve our partial-order learning problem below:
\begin{equation}
    M_j.cap\geq\max_{\forall g\in D_j}\left\{|V(g)|^{2}\right\}.
    \label{eq:cap_inequality}
\end{equation}

Note that, for our task of learning partial order, counter-intuitively, the GNN model capacity is constrained theoretically by the input graph size, instead of the size of the training data set $D_j$ \cite{loukas2020graph}.

Specifically, the GNN model $M_j$ we use in this paper (as shown in Figure \ref{fig:gnnmodel}) has 3 hidden layers (i.e., GNN depth $M_j.dep = 3$). If we set $F = 1$, $K = 3$, $F' = 32$, and $d = 2$ by default, then we have the GNN width $M_j.wid = K\cdot F' = 96$. As a result, we have the model capacity $M_j.cap=M_j.dep \times M_j.wid = 3\times 96=288$. On the other hand, since we set the degree threshold $\theta=10$ for node embeddings (as discussed in the complexity analysis above), the maximum input size of unit star subgraph does not exceed 11 (i.e., $max_{\forall g\in D_j}(|V(g)|) = 11$). Thus, we can see that Inequality~(\ref{eq:cap_inequality}) holds (i.e., $288 = M_j.cap\geq\max_{\forall g\in D_j}\left\{|V(g)|^{2}\right\} = 11^2$ holds), which implies that our GNN model $M_j$ has enough capacity to achieve zero loss on the training data set $D_j$.

\underline{\it The Existence of the GNN Model that Meets the Training Target:} \\Next, we prove that there exists at least one set of GNN model parameters that make the loss equal to zero over the training data.

In the following lemma, we give a special case of GNN model parameters, which can ensure the dominance relationship between node embedding vectors $o(g_{v_i})$ and $o(s_{v_i})$ (satisfying $o(s_{v_i})\preceq o(g_{v_i})$) of any two star subgraphs $g_{v_i}$ and $s_{v_i}$ (satisfying $s_{v_i}\subseteq g_{v_i}$).

\begin{lemma} {\bf (A Special Case of GNN Model Parameter Settings)}
    For a unit star subgraph $g_{v_i}$ and its star substructure $s_{v_i}$ ($\subseteq g_{v_i}$), their GNN-based node embedding vectors satisfy the dominance condition that: $o(s_{v_i})\preceq o(g_{v_i})$, if values of the weight matrix $\mathbb{W}$ (in Eq.~(\ref{eq:eq6})) in the fully connected layer are all zeros, i.e., $\mathbb{W} = \mathbf{0}$.
    \label{lemma:special_case}
\end{lemma}

\begin{proof}
    If all the values in the weight matrix $\mathbb{W}$ are zeros, in Eq.~(\ref{eq:eq6}), for any $y_i$, we have $\mathbb{W}y_i =\mathbf{0}$. With a Sigmoid activation function $\sigma(Z)=\frac{1}{1+e^{-Z}}$, any dimension of the final output, $o(g_{v_i})$ or $o(s_{v_i})$, in Eq.~(\ref{eq:eq6}) is always equal to $0.5$ for arbitrary input 
    Therefore, we have: $0.5=o(s_{v_i})[t]\leq o(g_{v_i})[t]=0.5$, for all dimensions $1\leq t\leq d$. In other words, when $\mathbb{W} = \mathbf{0}$ holds, any unit star subgraph $g_{v_i}$ and its star substructure $s_{v_i}$ ($\subseteq g_{v_i}$) have their node embedding vectors satisfying the dominance condition (i.e., $o(s_{v_i})\preceq o(g_{v_i})$). 
\end{proof}

In the special case of Lemma \ref{lemma:special_case}, the loss function $\mathcal{L}(D_j)$ given in Eq.~(\ref{eq:loss}) over all the training pairs in $D_j$ is always equal to 0. Note that, although there is no pruning power in this special case (i.e., all the node embedding vectors are the same, which preserves the dominance relationships), it at least indicates that there exists a set of weight parameters (in Lemma \ref{lemma:special_case}) that can achieve zero loss. 

In reality (e.g., from our experimental results), there are multiple possible sets of GNN parameters that can reach zero training loss (e.g., the special case of $o(s_{v_i})=o(g_{v_i})$ given in Lemma~\ref{lemma:special_case}). This is because we are looking for embedding vectors that preserve dominance relationships between individual pairs $(g_{v_i}, s_{v_i})$ in $D_j$ (rather than seeking for a global dominance order for all star subgraphs). 

\underline{\it The Target Accessibility of the GNN Training:}
Up to now, we have proved that we can guarantee enough GNN capacity for zero-loss training via parameter settings, and GNN parameters that can achieve zero loss exist. We now illustrate the accessibility of our GNN training that can meet the training target (i.e., the loss equals zero).

First, based on \cite{loukas2020graph}, if a GNN model $M_j$ over connected attributed graphs has enough capacity, then the GNN model can approach the optimal solution infinitely. That is, we have:
\begin{equation}
    |f_j(\cdot)-f_{opt}(\cdot)|\rightarrow 0,
    \label{eq:optimal}
\end{equation}
where $f_j(\cdot)$ is a non-linear function with input $x_i$ and output $o(v_i)$ that learned by the GNN model $M_j$, and $f_{opt}(\cdot)$ is an optimal function for GNN that achieves zero loss.

Moreover, from \cite{du2019gradient}, with randomized initial weights,  first-order methods (e.g., \textit{Stochastic Gradient Descent} (SGD) \cite{kingma2015adam} with the Adam optimizer used in our work) can achieve zero training loss, at a linear convergence rate. That is, it can find a solution with $\mathcal{L}(\cdot)\leq \epsilon$ in $O(log(1/\epsilon))$ epochs, where $\epsilon$ is the desired accuracy.

Note that, to handle some exceptional cases that zero training loss cannot be achieved within a limited number of epochs, we may remove those relevant pairs (causing the loss to be non-zero) and train a new GNN on them, which we will leave as our future work.

In summary, we can train our designed GNN model $M_j$ for node dominance embeddings, and the training process can converge to zero loss. 

\begin{figure*}[htp]
\centering
\subfigure[][{\small \# of training pairs ($|V(G_j)|$$=$$500K$)}]{                    
\scalebox{0.2}[0.2]{\includegraphics{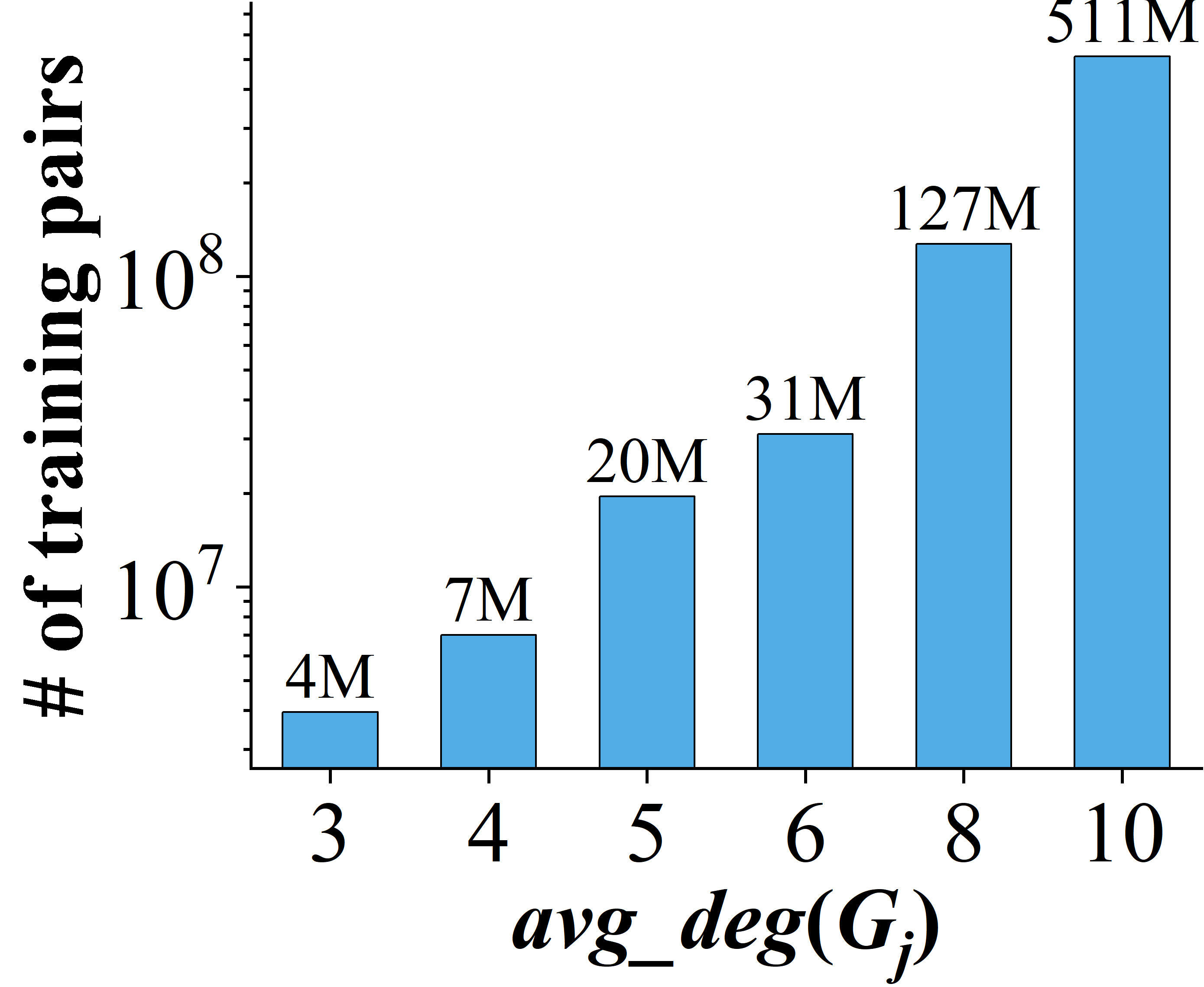}}\label{subfig:training_pair_deg}}\qquad
\subfigure[][{\small convergence  ($|V(G_j)|$$=$$500K$)}]{
\scalebox{0.2}[0.2]{\includegraphics{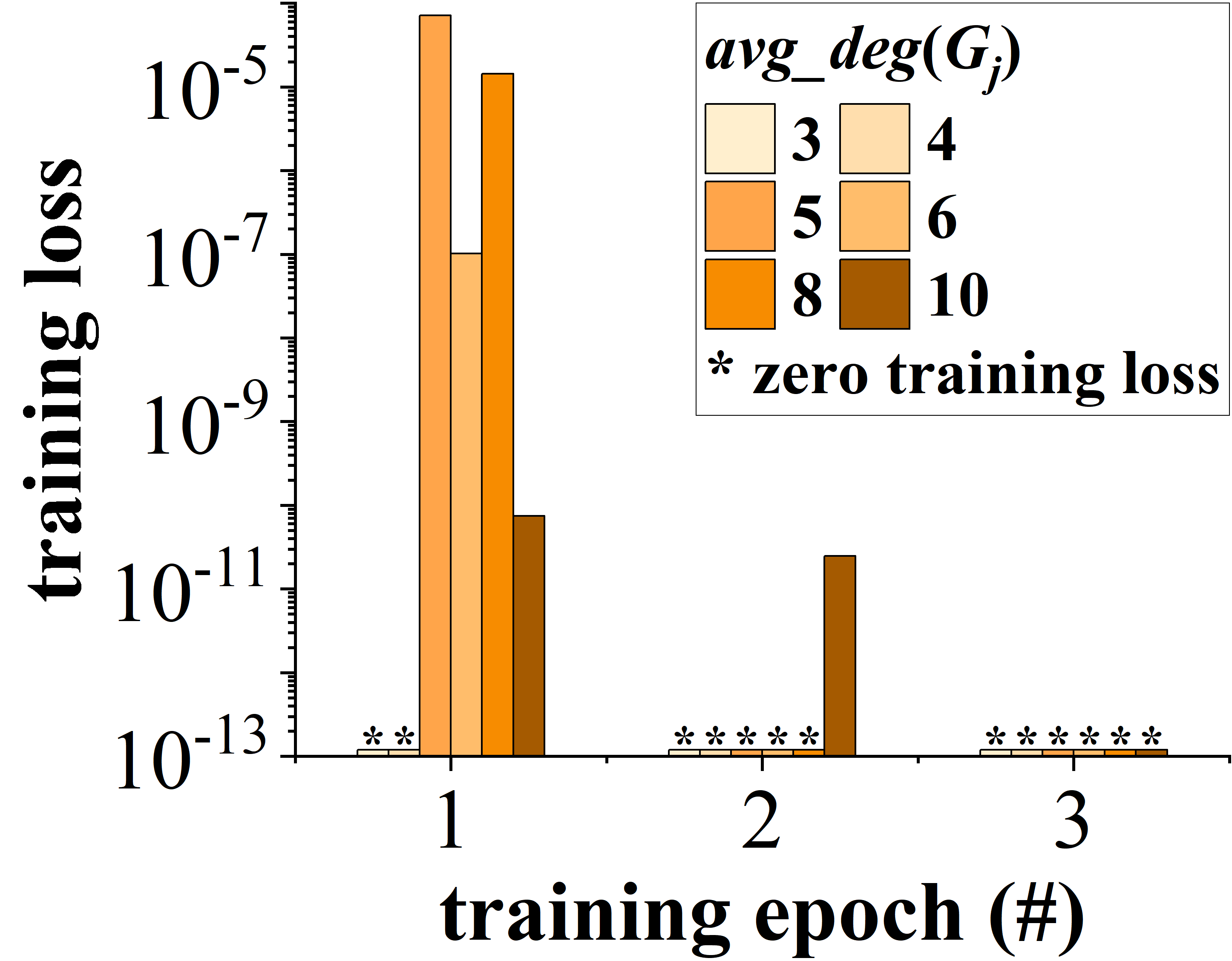}}\label{subfig:training_loss_deg}}\qquad
\subfigure[][{\small training time ($|V(G_j)|$$=$$500K$)}]{                    
\scalebox{0.2}[0.2]{\includegraphics{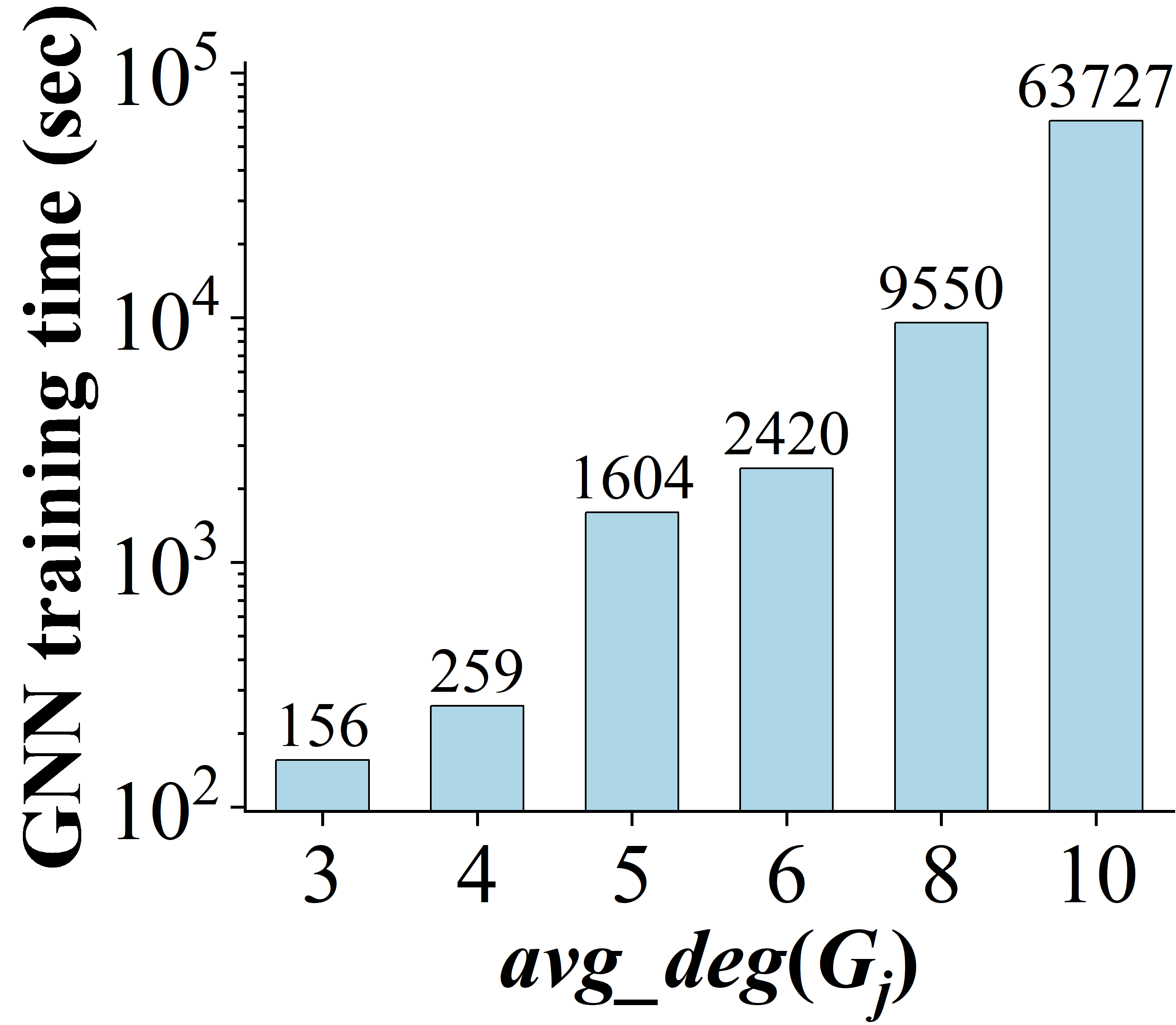}}\label{subfig:training_time_deg}}\\
\subfigure[][{\small \# of training pairs ($avg\_deg(G_j)$$=$$5$)}]{                    
\scalebox{0.2}[0.2]{\includegraphics{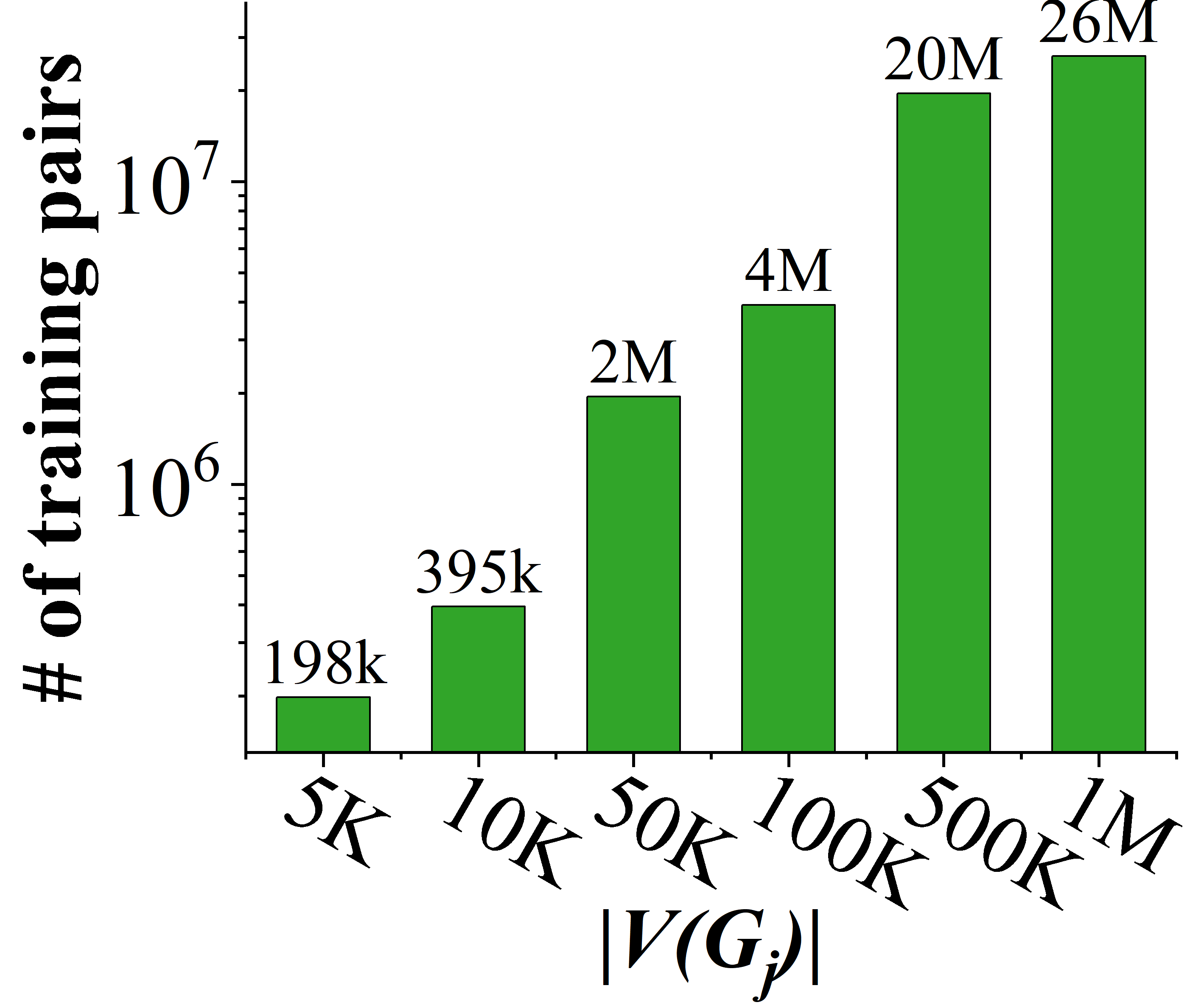}}\label{subfig:training_pair_size}}\qquad
\subfigure[][{\small convergence ($avg\_deg(G_j)$$=$$5$)}]{
\scalebox{0.2}[0.2]{\includegraphics{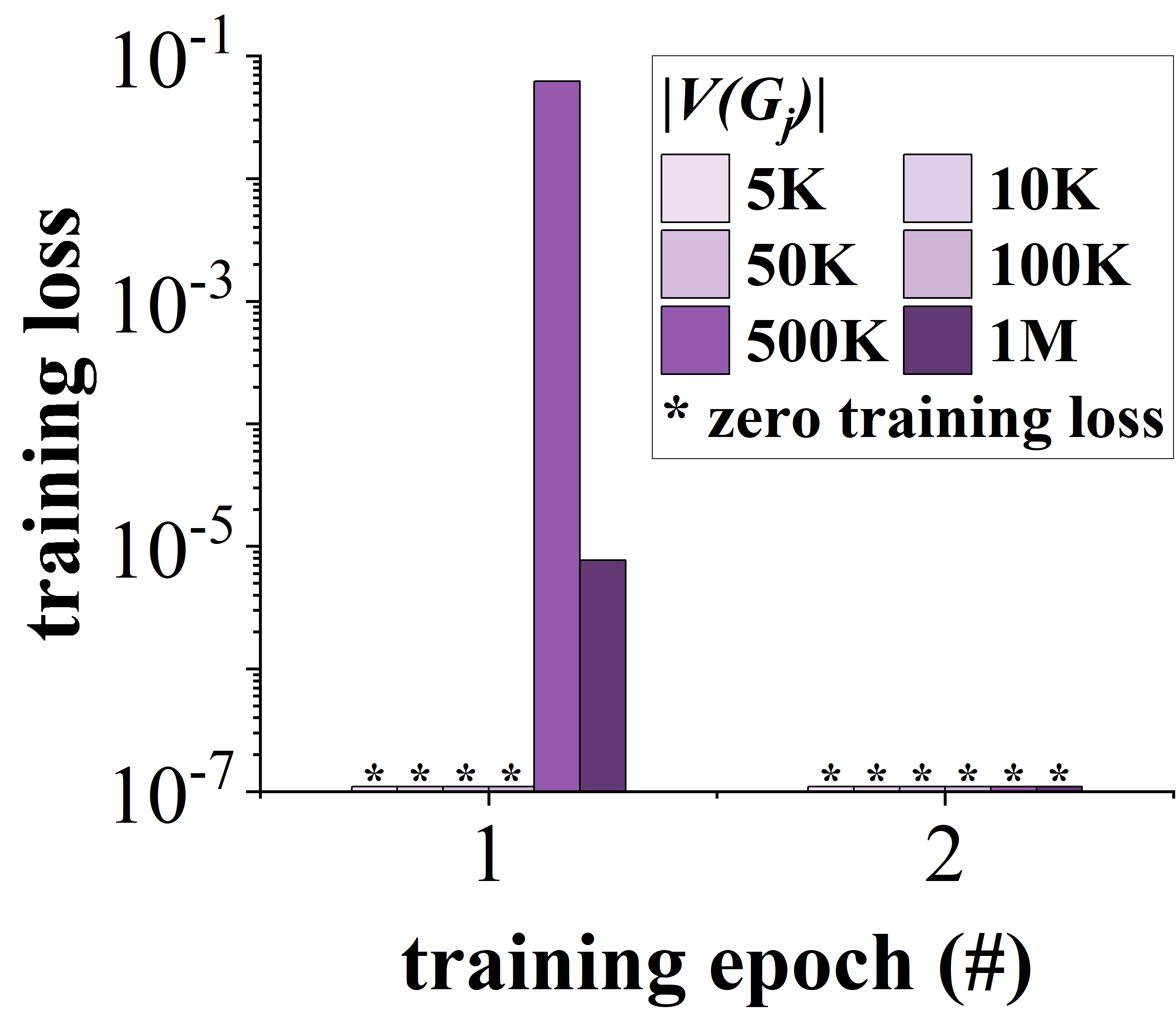}}\label{subfig:training_loss_size}}\qquad
\subfigure[][{\small training time ($avg\_deg(G_j)$$=$$5$)}]{
\scalebox{0.2}[0.2]{\includegraphics{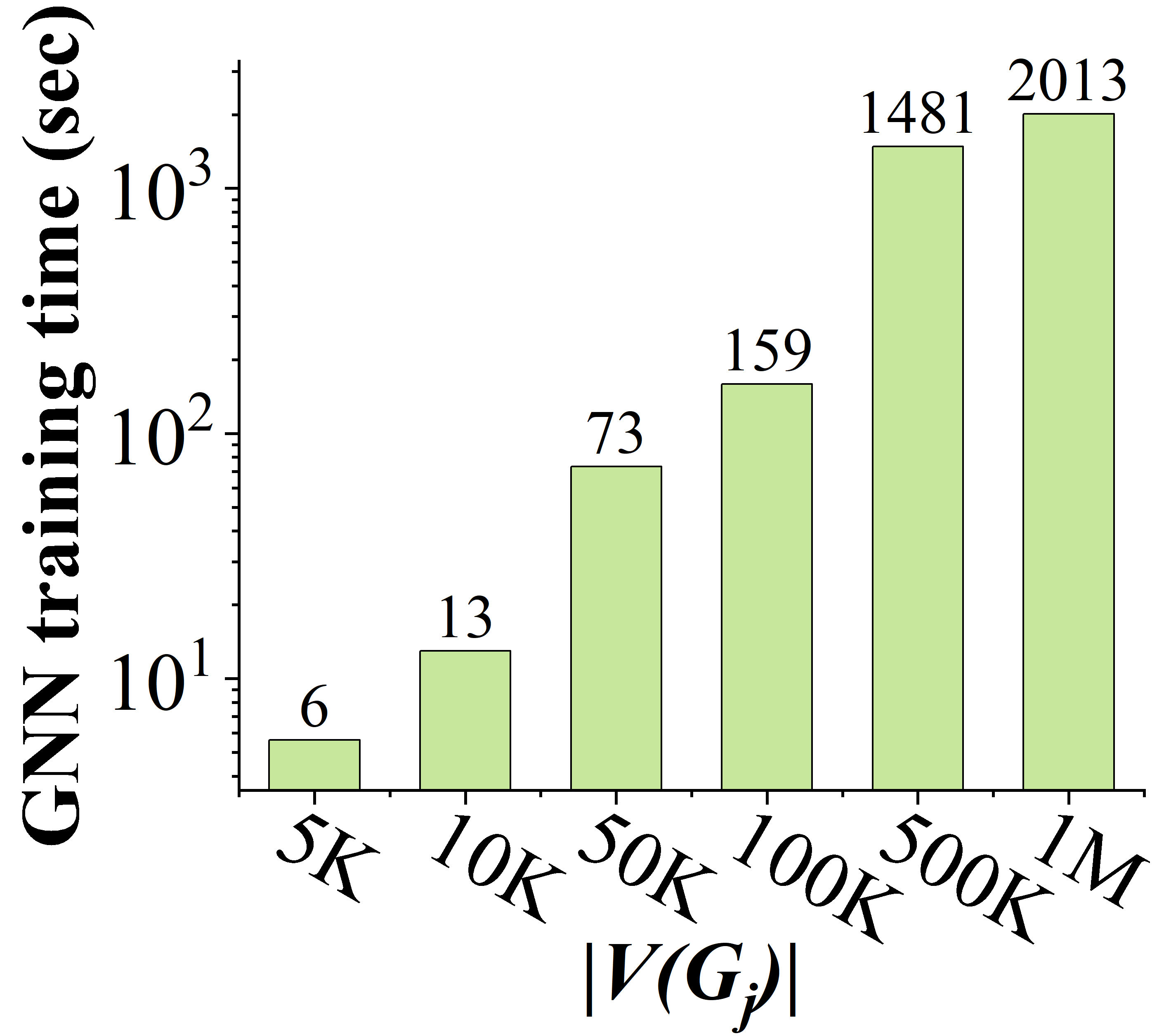}}\label{subfig:training_time_size}}
\caption{Illustration of the GNN training performance for node dominance embeddings.}
\label{fig:GNN_convergence}
\end{figure*}

\noindent{\bf The GNN Training Scalability w.r.t. Node Dominance Embedding.} We train a GNN model (with $F=1$, $K=3$, $F'=32$, and $d=2$) over large training data sets $D_j$ (containing pairs $(g_{v_i}, s_{v_i})$), where by default the vertex label domain size $|\Sigma|=500$, the default average vertex degree, $avg\_deg(G_j)=5$, $|V(G_j)|=500K$, the learning rate of the Adam optimizer $\eta=0.001$, and batch size $1K \sim 4K$.

Figure~\ref{subfig:training_pair_deg} illustrates the number of training pairs $(g_{v_i}, s_{v_i})$ that the GNN model (in Figure~\ref{fig:gnnmodel}) can support (i.e., the loss function $\mathcal{L}(D_j)$ achieves zero), where we vary the average vertex degree, $avg\_deg(G_j)$ from 3 to 10.
From our experimental results, our GNN model can learn as many as $\geq 511M$ pairs for a graph with $500K$ vertices and an average degree equal to 10. 

Figure~\ref{subfig:training_loss_deg} reports the convergence performance of our GNN-PE approach, where $avg\_deg(G_j)$ varies from 3 to 10. As described in Algorithm~\ref{alg2}, we train/update the parameters of the GNN model after each batch of a training epoch, and evaluate the loss of the GNN model after training all batches in $D_j$ at the end of each training epoch. From the figure, the GNN training needs no more than three epochs (before the loss becomes 0), which confirms that we can train our designed GNN model $M_j$ for node dominance embeddings within a small number of epochs, and the training process can converge fast to zero loss.  

Figures \ref{subfig:training_pair_size} and \ref{subfig:training_loss_size} vary the graph size $|V(G_j)|$ from $5K$ to $1M$, and similar experimental results can be obtained, in terms of the number of training pairs and the convergence performance, respectively.

Figures \ref{subfig:training_time_deg} and \ref{subfig:training_time_size} show the \textit{GNN training time} for different average degrees $avg\_deg(G_j)$ from 3 to 6 and graph sizes $|V(G_j)|$ from $5K$ to $1M$, respectively, where the GNN training needs no more than three epochs (before the loss becomes 0) and takes less than 17 hours. This indicates the efficiency and scalability of training GNNs to offline perform node dominance embeddings.

\begin{figure}[htp]
    \centering
    \includegraphics[scale=0.38]{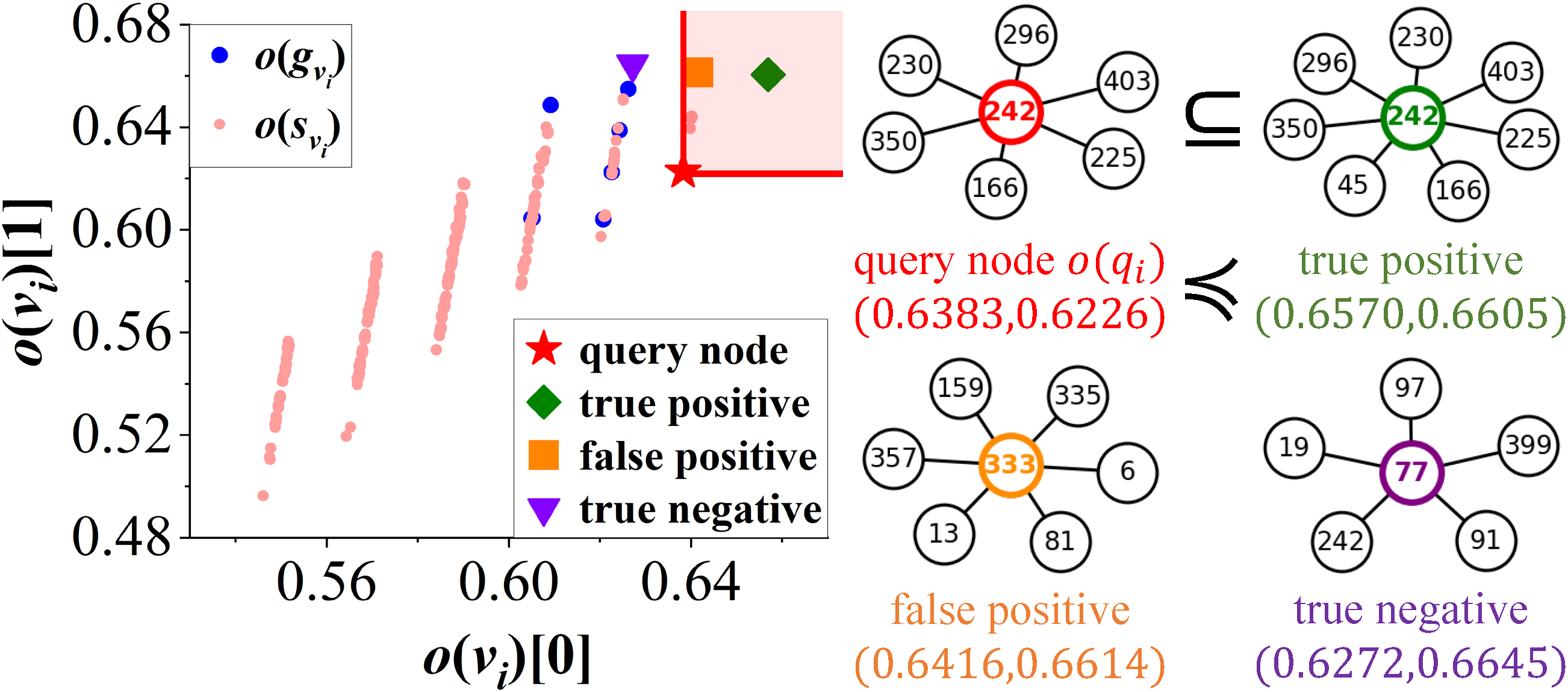}
    \caption{A visualization of node dominance embeddings.}
    \label{fig:visualization}
\end{figure}

\noindent{\bf Visualization Analysis of Node Dominance Embeddings.} We randomly sample 10 vertices $v_i$ from synthetic graph $G_j$ (used in Figure \ref{fig:GNN_convergence}). For each vertex $v_i$, we obtain its unit star subgraph $g_{v_i}$ and all star substructures $s_{v_i}$, and plot in Figure~\ref{fig:visualization} their GNN-based node dominance embeddings $o(g_{v_i})$ (blue points) and $o(s_{v_i})$ (pink points), respectively, in a 2D embedding space. 

As a case study shown in Figure~\ref{fig:visualization}, given a query node embedding $o(q_i)$ (red star point), its dominating region $DR(o(q_i))$ contains true positive (green diamond; matching vertex) and false positive (orange square; mismatching candidate vertex). The purple triangle point is not in $DR(o(q_i))$, which is true negative (i.e., not a matching vertex). From the visualization, our GNN-based embedding vectors are distributed on some piecewise curves, and their dominance relationships can be well-preserved.

\subsection{Path Dominance Embedding}
\label{subsec:path_embedding}

Next, we discuss how to obtain path dominance embedding from node embeddings (as discussed in Section \ref{subsec:node_embedding}). Specifically, given a path $p_z$ in $G$ starting from $v_i$ and with length $l$, we concatenate embedding vectors $o(v_j)$ of all consecutive vertices $v_j$ on path $p_z$ and obtain a path embedding vector $o(p_z)$ of size $((l+1)\cdot d)$, where $l$ is the length of path $p_z$ and $d$ is the dimensionality of node embedding vector $o(v_j)$. That is, we have:
\begin{equation}
    o(p_z)=\big\Vert_{\forall v_j \in p_z} o(v_j),
    \label{eq:eq8}
\end{equation}
where $\Vert$ is the concatenation operator. 

Note that, node dominance embedding can be considered as a special case of path dominance embedding, where path $p_z$ has a length equal to 0.

\noindent {\bf Property of the Path Dominance Embedding.} Given two paths $p_q$ and $p_z$, if path $p_q$ (and 1-hop neighbors of vertices on $p_q$) is a subgraph of $p_z$ (and 1-hop neighbors of vertices on $p_z$), then it must hold that $o(p_q) \preceq o(p_z)$.

\begin{example} 
{\it Figure~\ref{fig:pathembedding} shows an example of the path dominance embedding for paths with length 2. Consider a path $p_z = v_3v_1v_2$, and its 6D path embedding vector $o(p_z) = (0.73, 0.58; 0.78, 0.79; 0.75, 0.77)$, which is a concatenation of three 2D embedding vectors $o(v_3) || o(v_1)$ $|| o(v_2)$ (as given in Figure~\ref{fig:nodeembedding}). Similarly, in the query graph $q$, we can obtain a 6D embedding vector $o(p_q)$ of a path $p_q = q_3q_1q_2$. 

In the figure, we can see that $o(p_q)$ dominates $o(p_z)$, which indicates that path $p_z$ may potentially match path $p_q$ (according to the property of the path dominance embedding). 
\qquad $\blacksquare$}
\end{example}

\begin{figure}[t!]
    \centering
    \includegraphics[scale=0.38]{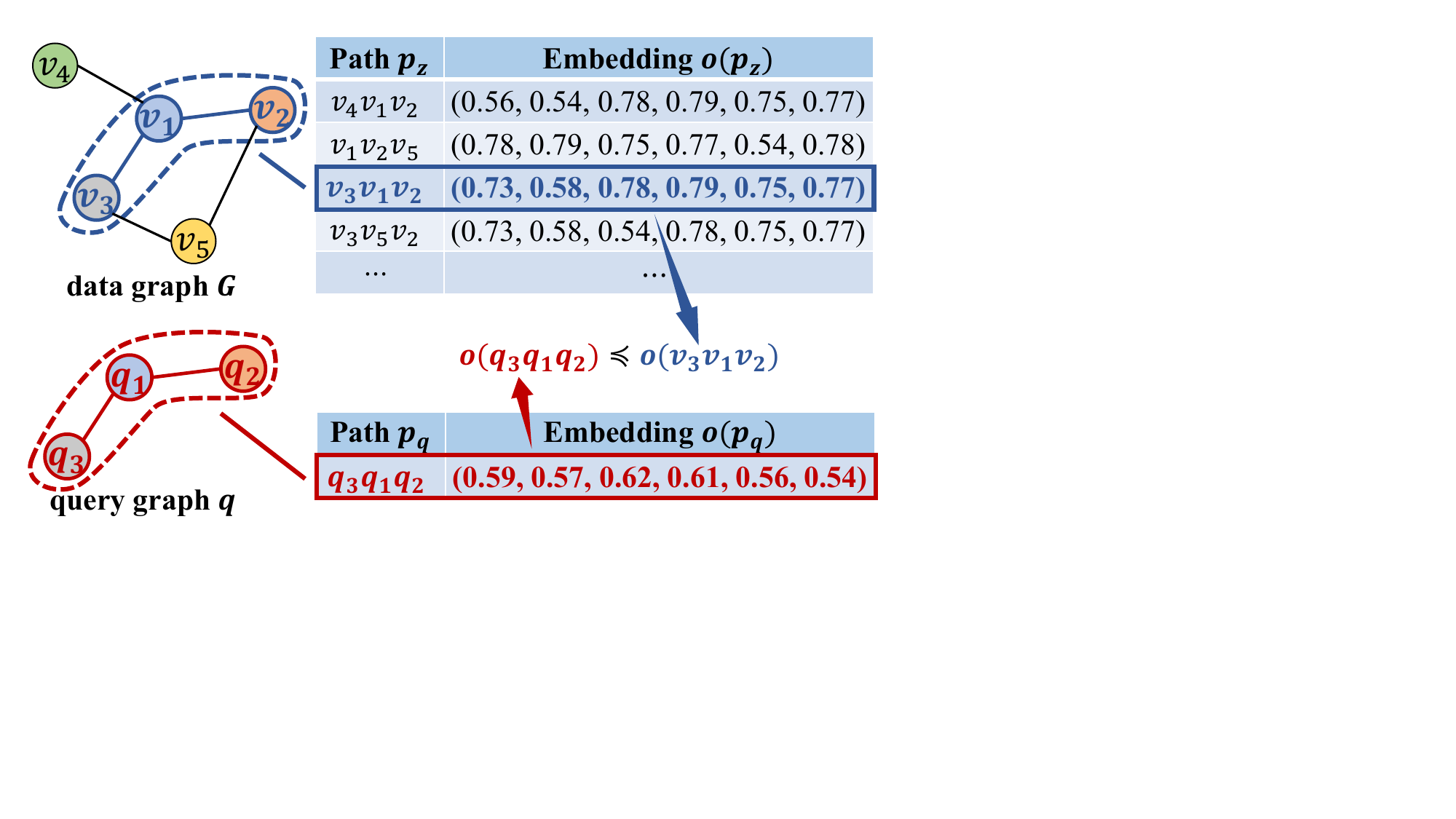}
    \caption{An example of path dominance embedding.}
    \label{fig:pathembedding}
\end{figure}

\section{Subgraph Matching with GNN-Based Path Embedding}
\label{sec:subgraph_matching_alg}

\subsection{Pruning Strategies}
\label{subsec:pruning}

In this subsection, we present effective pruning strategies, namely \textit{path label} and \textit{path dominance pruning}, to filter out false alarms of subgraphs $g$ ($\subseteq G$) that cannot match with a given query graph $q$. 


\noindent {\bf Path Label Pruning.} Let $s_0(v_i)$ (or $s_0(q_i)$) be a special star substructure containing an isolated vertex $v_i$ (or $q_i$; without any 1-hop neighbors). Assume that we can obtain an embedding vector, $o(s_0(v_i))$ (or $o(s_0(q_i))$), of the isolated vertex $v_i$ (or $q_i$) via the GNN model (discussed in Section \ref{subsec:node_embedding}). For simplicity, we denote $o(s_0(v_i))$ (or $o(s_0(q_i))$) as $o_0(v_i)$ (or $o_0(q_i)$).

Similarly, we can concatenate embedding vectors $o_0(v_i)$ (or $o_0(q_i)$) of vertices $v_i$ (or $q_i$) on path $p_z$ (or $p_q$) and obtain a \textit{path label embedding} vector $o_0(p_z) = \Vert_{\forall v_i\in p_z} o_0(v_i)$ (or $o_0(p_q) = \Vert_{\forall q_i\in p_q} o_0(q_i)$), which intuitively encodes labels of vertices on the path $p_z$ (or $p_q$).

\begin{lemma}
{\bf (Path Label Pruning)} 
Given a path $p_z$ in the subgraph $g$ of data graph $G$ and a query path $p_q$ in query graph $q$, path $p_z$ can be safely pruned, if it holds that $o_0(p_z) \ne o_0(p_q)$.
\label{lemma:path_label_pruning}
\end{lemma}

\begin{proof} 
If a path $p_z$ in a subgraph $g$ matches with a query path $p_q$ in the query graph $q$, then their corresponding vertices $v_i$ and $q_i$ on paths $p_z$ and $p_q$, respectively, must have the same labels. Through the same GNN model $M_j$ and with isolated vertices $v_i$ and $q_i$ as input, we can obtain label embedding vectors $s_0(v_i)$ and $s_0(q_i)$ as outputs, respectively, and in turn path label embeddings which are the concatenation of node label embedding vectors, that is, $o_0(p_z) = \Vert_{\forall v_i\in p_z} o_0(v_i)$ and $o_0(p_q) = \Vert_{\forall q_i\in p_q} o_0(q_i)$. If it holds that $o_0(p_z) \ne o_0(p_q)$, it indicates that at least one vertex position on paths $p_z$ and $p_q$ does not have the same label. In other words, paths $p_z$ and $p_q$ do not match with each other, and thus path $p_z$ can be safely pruned, which completes the proof of this lemma.
\end{proof}

\noindent{\bf Path Dominance Pruning.} For paths $p_z \subseteq G$ (or $p_q \subseteq q$) of length $l$, their embedding vectors, $o(p_z)$ (or $o(p_q)$), follow the property of the path dominance embedding (as discussed in Section \ref{subsec:path_embedding}). Based on this property, we have the lemma for the \textit{path dominance pruning} below.

\begin{lemma}
{\bf (Path Dominance Pruning)} 
Given a path $p_z$ in the subgraph $g$ of data graph $G$ and a query path $p_q$ in query graph $q$, path $p_z$ can be safely pruned, if $o(p_q) \preceq o(p_z)$ does not hold (denoted as $o(p_q) \npreceq o(p_z)$).
\label{lemma:path_dominance_pruning}
\end{lemma}

\begin{proof}
If a query path $p_q$ in query graph $q$ matches with a path $p_z$ in a subgraph $g$, then query path $p_q$ and its surrounding 1-hop neighbors in $q$ must be a subgraph of path $p_z$ (and 1-hop neighbors of its vertices) in $g$. Thus, based on the property of path dominance embedding in Section \ref{subsec:path_embedding}, their path dominance embeddings satisfy the condition $o(p_q) \preceq o(p_z)$. In other words (via contraposition), if this condition $o(p_q) \preceq o(p_z)$ does not hold, then $p_q$ does not match with $p_z$, and $p_z$ can be safely pruned, which completes the proof.
\end{proof}

\subsection{Indexing Mechanism}
In this subsection, we discuss how to obtain paths of length $l$ in (expanded) subgraph partitions $G_j$ and offline construct indexes, $\mathcal{I}_j$, over these paths to facilitate efficient processing of exact subgraph matching. 
Specifically, starting from each vertex $v_i \in V(G_j)$, we extract all paths $p_z$ of length $l$ (i.e., in an expanded subgraph partition that extends $G_j$ outward by $l$-hop) and compute their path dominance embedding vectors $o(p_z)$ via the GNN model $M_j$.
Then, we will build an aggregate R$^*$-tree (or aR-tree) \cite{Lazaridis01,beckmann1990r} over these path embedding vectors $o(p_z)$, by using standard \textit{insert} operator. 

In addition to the \textit{minimum bounding rectangles} (MBRs) of path embedding vectors $o(p_z)$ in index nodes, we store aggregate data such as MBRs of path label embedding $o_0(p_z)$, which can be used for path label pruning (as given by Lemma \ref{lemma:path_label_pruning} in Section \ref{subsec:pruning}).

\noindent{\bf Leaf Nodes.} Each leaf node $N \in \mathcal{I}_j$ contains multiple paths $p_z$, where each path has an embedding vector $o(p_z)$ via a GNN in $M_j$. Each path $p_z \in N$ is associated with a path label embedding vector $o_0(p_z)$ via the GNN $M_j$.

\noindent{\bf Non-Leaf Nodes.} Each non-leaf node $N \in \mathcal{I}_j$ contains multiple entries $N_c$, each of which is an MBR, $N_c.MBR$, of all path embedding vectors $o(p_z)$ for all paths $p_z$ under entry $N_c$. Each entry $N_c \in N$ is associated with an MBR, $N_c.MBR_0$, over path label embedding vector $o_0(p_z)$ for all paths $p_z$ under entry $N_c$.

\subsection{Index-Level Pruning}
In this subsection, we present effective pruning strategies on the node level of indexes $\mathcal{I}_j$, which can be used for filtering out (a group of) path false alarms in nodes.

\noindent{\bf Index-Level Path Label Pruning.} We first discuss the \textit{index-level path label pruning}, which prunes entries, $N_i$, in index nodes, containing path labels that do not match with that of the query path $p_q$.

\begin{lemma}
{\bf (Index-Level Path Label Pruning)} Given a query path $p_q$ and an entry $N_i$ of index node $N$, entry $N_i$ can be safely pruned, if it holds that $o_0(p_q)\notin N_i.MBR_0$.
\label{lemma:index_label_pruning}
\end{lemma}

\begin{proof}
MBR $N_i.MBR_0$ bounds all the label embedding vectors, $o_0(p_z)$, of paths $p_z$ under node entry $N_i$. If it holds that $o_0(p_q)\notin N_i.MBR_0$, it implies that labels of the query path $p_q$ do not match with those of any path $p_z$ under $N_i$. Therefore, node entry $N_i$ does not contain any candidate path $p_z$, and thus can be safely pruned.
\end{proof}

\noindent{\bf Index-Level Path Dominance Pruning.} Similarly, we can obtain the \textit{index-level path dominance pruning}, which rules out those index node entries $N_i$, under which all path dominance embedding vectors are not dominated by that of the query path $p_q$. 

Let $DR(o(p_q))$ be a \textit{dominating region} that is dominated by an embedding vector $o(p_q)$ in the embedding space. For example, as shown in Figure \ref{fig:nodeembedding}, the embedding vector $o(q_1)$ of vertex $q_1$ (i.e., a special case of a path with length 0) has a dominating region, $DR(o(q_1))$. Then, we have the following lemma:

\begin{lemma}
{\bf (Index-Level Path Dominance Pruning)} Given a query path $p_q$ and a node entry $N_i$, entry $N_i$ can be safely pruned, if $DR(o(p_q))\cap N_i.MBR = \emptyset$ holds.
\label{lemma:index_dominance_pruning}
\end{lemma}

\begin{proof}
For a query path $p_q$, any embedding vector $o(p_z)$ in its dominating region $DR(o(p_q))$ corresponds to a candidate path $p_z$ that may potentially be a supergraph of $p_q$. Thus, if $DR(o(p_q))$ and $N_i.MBR$ overlap with each other, then node entry $N_i$ may contain candidate paths and have to be accessed. Otherwise (i.e., $DR(o(p_q))\cap N_i.MBR = \emptyset$ holds), node entry $N_i$ can be safely pruned, since it does not contain any candidate path, which completes the proof.
\end{proof}

Intuitively, in Lemma \ref{lemma:index_dominance_pruning}, if embedding vector $o(p_q)$ does not fully or partially dominate $N_i.MBR$, then the entire index entry $N_i$ can be pruned. This is because any path $p_z$ under entry $N_i$ cannot be dominated by $p_q$ in the embedding space, and thus cannot be a candidate path that matches with query path $p_q$.

\begin{algorithm}[!h]
\caption{{\bf Exact Subgraph Matching with GNN-Based Path Dominance Embedding}}
\label{alg3}
\KwIn{
    i) a query graph $q$;
    ii) a trained GNN model $M_j$, and;
    iii) an aR-tree index $\mathcal{I}_j$ over subgraph partition $G_j$
}
\KwOut{
    a set, $\mathcal{S}$, of matching subgraphs
}
\SetKwProg{Function}{function}{}{end}

obtain a set, $Q$, of query paths with length $l$ from a cost-model-based query plan $\varphi$\\
\For{each query path $p_q \in Q$}{
    $p_q.cand\_list = \emptyset$\\
    obtain $o(p_q)$ and $o'(p_q)$ via multi-GNNs\\
    obtain $o_0(p_q)$ via $M_j$\\
}
    
\tcp{traverse index $\mathcal{I}_j$ to find candidate paths}
    initialize a \textit{maximum heap} $\mathcal{H}$ accepting entries in the form $(N, key(N))$\\
    $root\big(\mathcal{I}_j\big).list \leftarrow Q$\\
    insert $\big(root\big(\mathcal{I}_j\big), 0\big)$ into $\mathcal{H}$\\
    \While{$\mathcal{H}$ is not empty}{
        deheap an entry $(N, key(N)) = \mathcal{H}.pop()$;\\
        \If{$key(N) < \min_{\forall p_q\in Q}\{||o(p_q)||_1\}$}{
            terminate the loop;
        }
        
        \eIf{$N$ is a leaf node}{
            \For{each path $p_z\in N$}{
                \For{each query path $p_q\in N.list$}{
                
                \If{$o_0(p_q) = o_0(p_z)$ \quad \tcp{Lemma \ref{lemma:path_label_pruning}}}{
                    \If{$o(p_q)\preceq o(p_z)$}{
                       
                        $p_q.cand\_list \leftarrow p_q.cand\_list \cup \{p_z\}$ \qquad \qquad \tcp{Lemma \ref{lemma:path_dominance_pruning}}
                    }
                }
                }
            }
        }
        {
            \For{each child node $N_i \in N$}{
                \For{each query path $p_q\in N.list$}{
                    \If{$o_0(p_q)\in N_i.MBR_0$ \tcp{Lemma \ref{lemma:index_label_pruning}}}{
                        \If{$DR(o(p_q))\cap N_i.MBR \ne \emptyset$}{
                            $N_i.list \leftarrow N_i.list \cup \{p_q\}$ \qquad \tcp{Lemma \ref {lemma:index_dominance_pruning}}
                        }
                    }
                }
                \If{$N_i.list \ne \emptyset$}{
                    insert $(N_i, key(N_i))$ into heap $\mathcal{H}$\\
                }
            }
        }
    }
    concatenate all candidate paths in $p_q.cand\_list$ for $p_q\in Q$ and refine/obtain matching subgraphs $g$ in $\mathcal{S}$\\
    {\bf return} $\mathcal{S}$;\\
\end{algorithm}

\subsection{GNN-Based Subgraph Matching Algorithm}

In this subsection, we illustrate the exact subgraph matching algorithm by traversing the indexes over GNN-based path embeddings in Algorithm~\ref{alg3}. Specifically, given a query graph $q$, we first obtain all query paths of length $l$ in a set $Q$ from the query plan $\varphi$ (line 1). Then, for each query path $p_q\in Q$, we generate path embedding vectors $o(p_q)$ and $o_0(p_q)$ (via $M_j$) (lines 2-5). Next, we traverse the index $\mathcal{I}_j$ once to retrieve path candidate sets for each query path $p_q\in Q$ (lines 6-28). Finally, we refine candidate paths and join the matched paths to obtain/return subgraphs $g \in \mathcal{S}$ that are isomorphic to $q$ (lines 29-30).

\noindent {\bf Index Traversal.} To traverse the index $\mathcal{I}_j$, we initialize a \textit{maximum heap} $\mathcal{H}$, which accepts entries in the form of $(N, key(N))$ (line 6), where $N$ is an index entry and $key(N)$ is the key of index entry $N$ in the heap. Here, for an index entry $N$, we use $N.MBR_{max}$ to denote the maximal corner point of MBR $N.MBR$, which takes the maximal value of the interval for each dimension of the MBR. The $key(N)$ of entry $N$ is calculated by $key(N) = ||N.MBR_{max}||_1$, where a $d$-dimensional vector $Z$ has the $L_1$-norm $||Z||_1 = \sum_{i=1}^d|Z[i]|$.

We also initialize a list, $root(\mathcal{I}_j).list$, of the index root with the query path set $Q$, and insert an entry $(root(\mathcal{I}_j), 0)$ into heap $\mathcal{H}$ (lines 7-8). We traverse the index by accessing entries from heap $\mathcal{H}$ (lines 9-28). Specifically, each time we will pop out an index entry $(N,$ $key(N))$ with the maximum key from heap $\mathcal{H}$ (line 10). If $key(N)<\min_{\forall p_q\in Q}\{||o(p_q)||_1\}$ holds, which indicates that paths in the remaining entries of $\mathcal{H}$ cannot dominate any query paths $p_q\in Q$, then we can terminate the index traversal (lines 11-12); otherwise, we will check entries in node $N$.

When we encounter a leaf node $N$, for each path $p_z\in N$ and each query path $p_q\in N.list$, we apply the \textit{path label pruning} (Lemma~\ref{lemma:path_label_pruning}) and \textit{path dominance pruning} (Lemma~\ref{lemma:path_dominance_pruning}) (lines 13-18). If $p_z$ cannot be ruled out by these two pruning methods, we will add $p_z$ to the path candidate set $p_q.cand\_list$ (line 19).

When $N$ is a non-leaf node, we consider each child node $N_i\in N$ and each query path $p_q\in N.list$ (lines 20-22). If entry $N_i$ cannot be pruned by \textit{index-level path label pruning} (Lemma~\ref{lemma:index_label_pruning}) and \textit{index-level path dominance pruning} (Lemma~\ref{lemma:index_dominance_pruning}), then we will add $p_q$ to $N_i.list$ for further checking (lines 23-26).
If $N_i.list$ is not empty (i.e., $N_i$ is a node candidate for some query path), we will insert entry $(N_i, key(N_i))$ into heap $\mathcal{H}$ for further investigation (lines 27-28).

The index traversal terminates when either heap $\mathcal{H}$ is empty (line 9) or the remaining heap entries in $\mathcal{H}$ cannot contain path candidates (lines 11-12).

\noindent {\bf Refinement.} After finding all candidate paths in $p_q.cand\_list$ for each query path $p_q\in Q$, we will assemble these paths (with overlapping vertex IDs) into candidate subgraphs to be refined, and return the actual matching subgraph answers in $\mathcal{S}$ (lines 29-30).

Specifically, we consider the following two steps to obtain candidate subgraphs: 1) local join within each partition, and; 2) global join for partition boundaries. First, inside each partition, we perform the \textit{multi-way hash join} by joining vertex IDs of candidate paths for different query paths. Then, for those boundary candidate paths across partitions, we also use the \textit{multi-way hash join} to join them with candidate paths from all partitions globally. Finally, we can refine and return the resulting candidate subgraphs. 

\noindent {\bf Complexity Analysis.} 
In Algorithm \ref{alg3}, the time complexity of finding a query path set $Q$ (line 1) is given by $O(|P|\cdot (|Q|\cdot (l+1)))$, where $|P|$ is the number of different query path sets we evaluate for the query plan selection, $|Q|$ is the number of paths in $Q$, and $l$ is the path length.
Since the GNN computation cost is $O(|E(q)|+d^2)$ \cite{wu2020comprehensive,wang2022reinforcement}, the time complexity of obtaining embedding vectors of query paths $p_q$ in $Q$ via GNNs (lines 2-5) is given by $O(|V(q)|\cdot (|E(q)|+d^2))$, where $|V(q)|$ and $|E(q)|$ are the numbers of vertices and edges in $q$, resp., and $d$ is the size of vertex embedding vectors. 

For the index traversal (lines 6-28), assume that $h$ is the height of the aR-tree $\mathcal{I}_j$, $f$ is the average fanout of each index node $N$, and the pruning power is $PP_i$ on the $i$-th level of the tree index.
Then, on the $i$-th level, the number of index nodes (or paths) to be accessed is given by $f^{h-i+1}\cdot (1-PP_{i})$. Thus, the total index traversal cost has $O\big(\sum^{h}_{i=0}|Q|\cdot f^{h-i+1}\cdot (1-PP_{i})\big)$ complexity.

Finally, we use a multiway hash join for candidate path assembling and refinement (line 29). The time complexity is given by $O\big(\sum_{p_q\in Q}|p_q.cand\_list|\big)$, where $|p_q.cand\_list|$ is the number of candidate paths for query path $p_q\in Q$. 

Therefore, the overall time complexity of Algorithm \ref{alg3} is $O\big(|P|\cdot |Q|\cdot (l+1)+|V(q)|\cdot (|E(q)|+d^2)+\sum^{h}_{i=0}|Q|\cdot f^{h-i+1}\cdot (1-PP_{i})+\sum_{p_q\in Q}|p_q.cand\_list|\big)$.

\section{Cost-Model-Based Query Plan}
\label{sec:query_plan}

In this section, we discuss how to select a good query plan $\varphi$ from the query graph $q$, based on our proposed cost model, which returns a number of query paths $p_q$ (line 1 of Algorithm~\ref{alg3}).

\subsection{Cost Model}
\label{subsec:cost_model}

In this subsection, we provide a formal cost model to estimate the query cost of a query plan $\varphi$ which contains a set $Q$ of query paths $p_q$ from query graph $q$ (used for retrieving matching paths from the index). Intuitively, 
fewer query paths with small overlapping would result in lower query cost, and fewer candidate paths that match with query paths will also lead to lower query cost. 

Based on this observation, we define the query cost, $Cost_Q(\varphi)$, for query paths $p_q \in Q$ as follows: 

\begin{eqnarray}
    Cost_Q(\varphi)=\sum_{p_q\in Q}w(p_q),
    \label{eq:cost_model}
\end{eqnarray}
where $w(p_q)$ is the weight (or query cost) of a query path $p_q$.

Thus, our goal is to find a good query plan $\varphi$ with query paths in $Q$ that minimize the cost function $Cost_Q(\varphi)$ given in Eq.~(\ref{eq:cost_model}). 

\noindent {\bf Discussions on the Calculation of Path Weights.} We next discuss how to compute the path weight $w(p_q)$ in Eq.~(\ref{eq:cost_model}), which implies the search cost of query path $p_q$. Intuitively, when degrees of vertices in query path $p_q$ are high, the number of candidate paths that may match with $p_q$ is expected to be small, which incurs low query cost. We can thus set $w(p_q)= - \sum_{q_i\in p_q}deg(q_i)$, where $deg(q_i)$ is the degree of vertex $q_i$ on query path $p_q$. 

Alternatively, we can use other query cost metrics such as the number of candidate paths (to be retrieved and refined) dominated by $q_p$ in the embedding space. For example, we can set: $w(p_q)=|DR(o(p_q))|$, where $|DR(o(p_q))|$ is the number of candidate paths in the region, $DR(o(p_q))$, dominated by embedding vector $o(p_q)$. 

\subsection{Cost-Model-Based Query Plan Selection}
\label{subsec:queryplan}

Algorithm~\ref{alg4} illustrates how to select the query plan $\varphi$ in light of the cost model (given in Eq.~(\ref{eq:cost_model})), which returns a set, $Q$, of query paths from query graph $q$. Specifically, we first initialize an empty set $Q$ and query cost $cost_Q(\varphi)$ (line 1). Then, we select a starting vertex $q_i$ with the highest degree, whose node embedding vector expects to have high pruning power (line 2). Next, we obtain a set, $P$, of initial paths that pass through vertex $q_i$ (line 3). We start from each initial path, $p_q$, in $P$, and each time expand the local path set $local\_Q$ by including one path $p$ that minimally overlaps with $local\_Q$ and has minimum weight $w(p)$ (lines 4-9). For different initial paths in $P$, we always keep the best-so-far path set in $Q$ and the smallest query cost in $Cost_Q(\varphi)$ (lines 10-12). Finally, we return the best query path set $Q$ with the lowest query cost (line 13).

\begin{algorithm}[!h]
\caption{\bf Cost-Model-Based Query Plan Selection}
\label{alg4}
\KwIn{
    i) a query graph $q$;
    ii) path length $l$;
}
\KwOut{
    a set, $Q$, of query paths in the query plan $\varphi$
}

$Q=\emptyset$; $Cost_Q(\varphi) = +\infty$;\\
select a starting vertex $q_i$ with the highest degree\\
obtain a set, $P$, of initial paths of length $l$ containing $q_i$\\
\tcp{apply OIP, AIP, or $\boldsymbol{\varepsilon}$IP strategy in Section~\ref{subsec:queryplan}}
\For{each possible initial path $p_q\in P$}{
    $local\_Q = \{p_q\}$; $local\_cost = 0$;\\
    \While{at least one vertex of $V(q)$ is not covered}{
        select a path $p$ of length $l$ that connects with $Q$ with minimum overlapping and minimum weight $w(p)$\\
        $local\_Q \leftarrow local\_Q  \cup \{p\}$\\
        $local\_cost \leftarrow local\_cost + w(p)$
    }
    \If{ $local\_cost < Cost_Q(\varphi)$}{
        $Q \leftarrow local\_Q$\\
        $Cost_Q(\varphi) \leftarrow local\_cost$\\
    }
}
\textbf{return} $Q$\\
\end{algorithm}

\noindent {\bf Discussions on the Initial Path Selection Strategy.} In line 3 of Algorithm~\ref{alg4}, we use one of the following three strategies to select initial query path(s) in $P$:

\begin{itemize}
\item  {\bf One-Initial-Path (OIP)}: select one path $p_q$ with the minimum weight $w(p_q)$ that passes by the starting vertex $q_i$; 

\item  {\bf All-Initial-Path (AIP)}: select all paths that pass through the starting vertex $q_i$; and

\item  {\bf $\boldsymbol{\varepsilon}$-Initial-Path ($\boldsymbol{\varepsilon}$IP)}: randomly select $\varepsilon$ paths passing through the starting vertex $q_i$.
\end{itemize}

\section{GNN-PE Optimizations with Path Group Embeddings}
\label{sec:gnnpge}

Although the GNN-PE framework (as given in Algorithm \ref{alg1}) can transform the expensive subgraph search problem to an efficient dominance range search in an embedding space, it relies on the index over embeddings of all possible paths of length $l$ in the data graph $G$, which may not be space-efficient and lead to high query cost. In order to further optimize our GNN-PE approach, we propose an optimized GNN-PE approach, named \textit{GNN-based path group embedding} (GNN-PGE), which aggregates embeddings for groups of paths and queries over an index of \textit{path group embeddings} (rather than embeddings of individual paths). 

Specifically, we use \textit{minimum bounding rectangles} (MBRs) to minimally bound all path embeddings that share the same starting vertices, which can significantly reduce the space cost of the index, as well as the index search cost. We design effective pruning strategies tailored to MBRs of path group embeddings, which can facilitate efficient candidate path retrieval without false dismissals. 

\subsection{Rationale Behind the Path Grouping}

In our GNN-PE framework  (as illustrated in Algorithm \ref{alg1}), we used embeddings of paths (i.e., concatenation of vertex embeddings on the path) to enable the pruning via the dominance relationships. While we can achieve high pruning power with such path embeddings, the space cost of storing these individual path embeddings (and their index as well) could be large, which may also lead to extra index traversal cost. 

To further alleviate the space/search overheads, we propose an optimized GNN-PGE approach, by grouping similar path dominance embeddings (instead of treating each path individually). We observe that those paths with the same starting vertex usually share many common path prefixes, whose embeddings can be combined together to obtain a small MBR (i.e., a compact representation for the path group). Therefore, with this grouping strategy, the number of indexed items can be significantly reduced from $O(|V(G)|\cdot avg\_deg(G)^l)$ to $O(|V(G)|)$, especially for high-degree graph and long paths, which can greatly improve the index space cost and the efficiency of the index traversal, where $avg\_deg(G)$ is the average degree of vertices in graph $G$.

\subsection{Path Group Embeddings (PGEs)}
\label{subsec:path_group_emb}

In this subsection, we present how to obtain \textit{path group embeddings} (PGEs). In particular, for each vertex $v_i \in V(G)$, one straightforward method is to extract all paths starting from vertex $v_i$ and of length $l$, obtain their path dominance embeddings, and group them with a \textit{path group dominance embedding MBR} $v_i.MBR$. 

Observing that all paths starting from $v_i$ share many common prefixes, we do not actually need to enumerate all paths of length $l$ one by one (with redundant path prefix retrieval). Instead, we can perform a \textit{breadth-first search} (BFS) for $l$ hops, starting from vertex $v_i$. For the $x$-th hop, we extend the current paths of length $(x-1)$ by including neighbors of their ending vertices, and only need to calculate the MBR of neighbors' embeddings for the $x$-th hop. Finally, we combine/concatenate MBRs from $l$ hops to form an MBR, as the path group embedding $v_i.MBR$. This way, we can efficiently compute a compact representation of \textit{path group dominance embedding MBR}, $v_i.MBR$, for each vertex $v_i$, without redundant path traversal.

Formally, given a path length $l$, we define lower and upper bounds of MBR $v_i.MBR$ on the $(x\cdot d+t)$-th dimension (for $0 \leq x \leq l$ and $0 \leq t < d$) as follows:
\begin{eqnarray}
&&    v_i.MBR[2 (x \cdot d + t)]   = \min \{ o(v_j)[t] \text{ }|\text{ }\forall v_j \in \mathcal{N}_x(v_i)\}, \label{eq:MBR1}\\
&&    v_i.MBR[2 (x \cdot d + t)+1] = \max \{ o(v_j)[t] \text{ } |\text{ }\forall v_j \in \mathcal{N}_x(v_i)\}, \label{eq:MBR2}
\end{eqnarray}
where $d$ is the dimensionality of each vertex's embedding vector, $x$ denotes the number of hops from vertex $v_i$, $\mathcal{N}_x(v_i)$ represents a set of neighbors $x$ hops away from vertex $v_i$ along any paths, and $o(v_j)[t]$ is the $t$-th dimension of the embedding of vertex $v_j$ (for $0 \leq t < d$). 

\begin{figure}[t]
    \centering
    \includegraphics[scale=0.34]{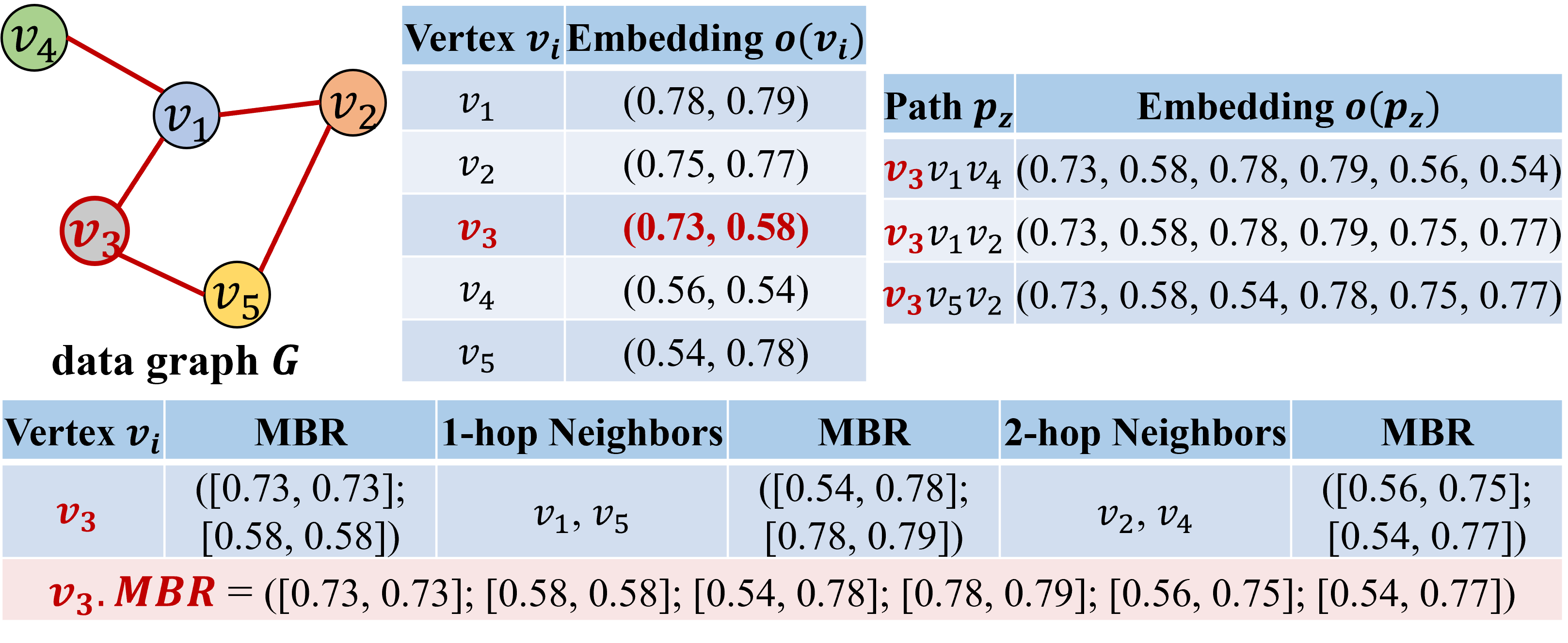}
    \caption{An example of the path group embedding ($l=2$).}
    \label{fig:path_group_example}
\end{figure}

\begin{example}
\it
    Consider the data graph $G$ in Figure~\ref{fig:path_group_example}, where each vertex $v_i$ is associated with an embedding $o(v_i)$. Following our GNN-PE framework, we have 3 paths $p_z$ of length 2 starting from vertex $v_3$, that is, $v_3v_1v_4$, $v_3v_1v_2$, and $v_3v_5v_2$, each of which has a path embedding (given by the concatenation of 2D vertex embeddings). For example, path $v_3v_1v_4$ has the path embedding $(o(v_3), o(v_1), o(v_4)) = (0.73, 0.58,$ $0.78,$ $0.79, 0.56, 0.54)$.

    To obtain the path group embedding $v_3.MBR$, we can start from vertex $v_3$ and do the BFS for 0-hop, 1-hop, and 2-hop neighbors of $v_3$ on paths, given by $\{v_3\}$, $\{v_1, v_5\}$, and $\{v_2, v_4\}$, respectively. We use an MBR to bound the set of the $x$-th hop neighbors ($0\leq x\leq 2$). For example, for 0-hop neighbors, we have MBR  $([0.73, 0.73]; [0.58, 0.58])$; for 1-hop neighbors, we obtain MBR $([0.54,0.78]; [0.78,0.79])$, and; for 2-hop neighbors, we have MBR $([0.56,0.75];$ $[0.54,$ $0.77])$. By concatenating these MBRs, we can obtain $v_3.MBR$ as shown in Figure \ref{fig:path_group_example}.     \qquad $\blacksquare$
\end{example}

\noindent{\bf Complexity Analysis of Obtaining Path Group Embedding MBRs.}
In GNN-PGE, we avoid explicit path enumeration by grouping all paths of length $l$ (sharing many prefixes) starting from the same vertex $v_i$ into a single MBR. For each vertex $v_i$, instead of enumerating all possible paths one by one, we only need to traverse vertex $v_i$'s $x$-hop neighbors ($0 \leq x \leq l$) once, and compute minimum and maximum bounds on each dimension over these $x$-hop neighbors' embeddings. Therefore, our PGE MBR construction traverses approximately $\frac{avg\_deg(G)^{l+1} -1}{avg\_deg(G) -1}$ ($=\sum_{x=0}^{l} avg\_deg(G)^x$) vertices within $l$ hops away from $v_i$, and aggregates neighbor embeddings up to $l$ hops. The time complexity reduces to $O\Big(\sum_{v_i \in V(G)} \\\frac{avg\_deg(G)^{l+1} -1}{avg\_deg(G) -1} \cdot d \Big)$. In practice, since real-world graphs typically follow a power-law degree distribution, only a small fraction of vertices have high degrees, and most vertices have low degrees, which usually leads to a low average degree $avg\_deg(G)$. Thus, $avg\_deg(G)^l$ can be treated as a constant, and in turn the time complexity $O\Big(|V(G)| \cdot \frac{avg\_deg(G)^{l+1} -1}{avg\_deg(G) -1} \cdot d \Big)$ of the PGE construction is nearly linear in the graph size.

\begin{figure}[t]
    \centering
    \includegraphics[scale=0.3]{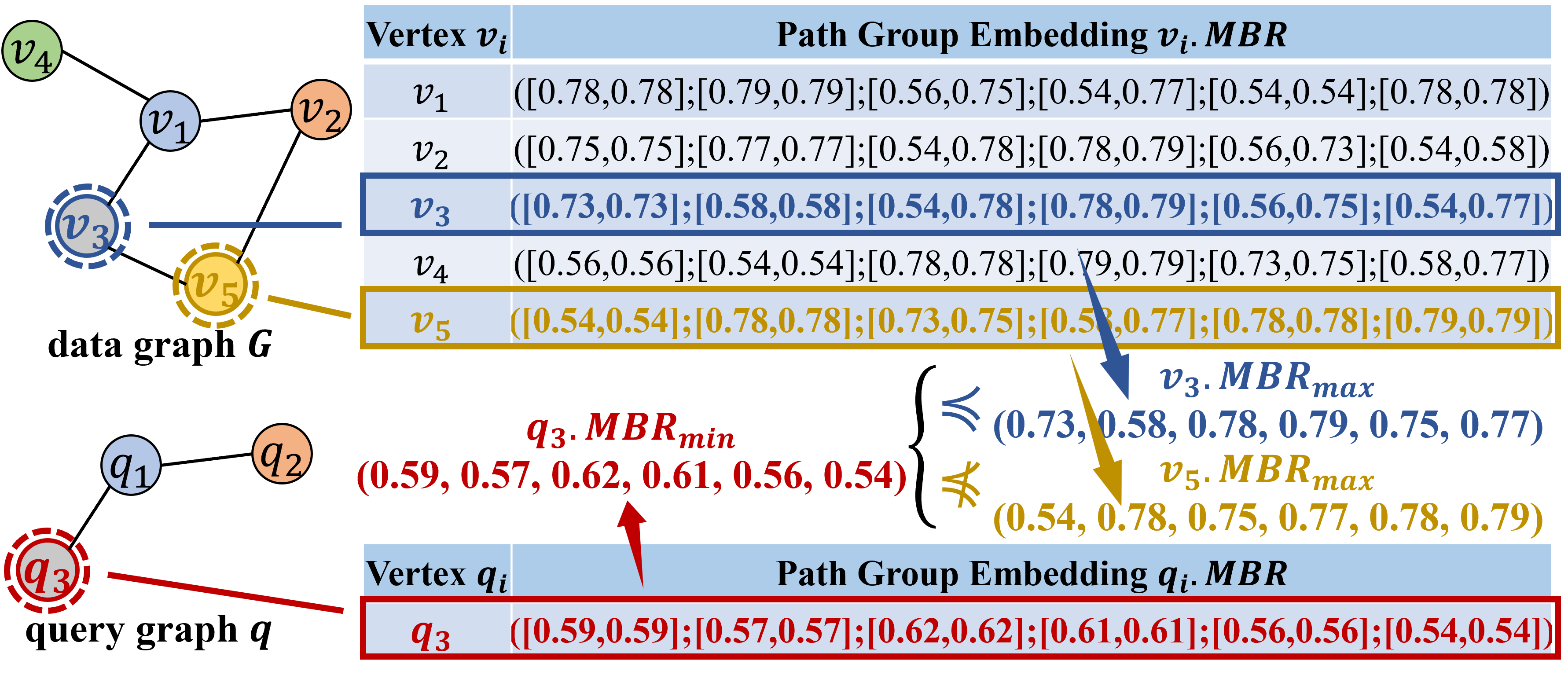}
    \caption{An example of path group embedding usage ($l=2$).}
    \label{fig:path_group_dominance}
\end{figure}

\noindent{\bf Property of the Path Group Dominance Embedding.}
Given a query vertex $q_i$ in the query graph $q$ and a vertex $v_i$ in a subgraph $g$ of the data graph $G$, we denote $q_i.MBR_{min}$ and $v_i.MBR_{max}$ as the minimum and maximum corner points of PGE MBRs $q_i.MBR$ and $v_i.MBR$, respectively, where corner points take lower and upper bounds of MBRs on each dimension, respectively. As a result, if $q_i$ matches with $v_i$, then the dominance relationship must hold that: $q_i.MBR_{min}\preceq$\\ $v_i.MBR_{max}$, that is, $q_i.MBR_{min}[x \cdot d + t]\leq v_i.MBR_{max}[x \cdot d + t]$ for all the $(l\cdot d)$ dimensions, where $0 \leq x \leq l$ and $0 \leq t < d$.

Note that, the condition $q_i.MBR_{min}\preceq v_i.MBR_{max}$ captures the potential dominance relationship between any query and data path embeddings in MBRs $q_i.MBR$ and $v_i.MBR$, respectively, in the embedding space. In other words, if the PGE MBR, $v_i.MBR$, of a data vertex $v_i$ intersects with the dominating region of minimum corner point $q_i.MBR_{min}$ of $q_i$'s PGE MBR $q_i.MBR$ (as illustrated in Section \ref{subsec:node_embedding}), then $v_i$ may match with $q_i$.

\begin{example}
\it
    Figure~\ref{fig:path_group_dominance} illustrates the dominance property of the path group embedding with path length $l=2$. 
    Based on the vertex embedding shown in Figure \ref{fig:nodeembedding}, for the query vertex $q_3$ and the data vertex $v_3$, we have $q_3.MBR_{min} = (0.59, 0.57,\\ 0.62, 0.61, 0.56, 0.54)$ and $v_3.MBR_{max} = (0.73, 0.58, 0.78, 0.79,\\ 0.75, 0.77)$. 
    Based on the PGE dominance property, we can see that $q_3.MBR_{min}$ dominates $v_3.MBR_{max}$, which indicates that vertex $v_3$ may potentially match with query vertex $q_3$.
    
    On the other hand, the maximum corner point, $v_5.MBR_{max}$ $=(0.54, 0.78, 0.75, 0.77, 0.78, 0.79)$, of PGE MBR $v_5.MBR$, is not dominated by $q_3.MBR_{min}$, i.e., $q_3.MBR_{min}\npreceq v_5.MBR_{max}$, which means vertex $v_5$ cannot match with query vertex $q_3$.    $\blacksquare$
\end{example}

\subsection{Pruning Strategies with Path Groups}
\label{subsec:pruning_group}
In this subsection, we present effective pruning strategies, namely \textit{path group label} and \textit{path group dominance pruning}, to filter out false alarms of subgraphs $g$ $(\subseteq G)$ that cannot match with a given query graph $q$.

\noindent{\bf Path Group Label Pruning.} The basic idea of the \textit{path group label pruning} is as follows. For each vertex $v_i \in V(G)$, we can find a group of paths, $p_z$, of length $l$ starting from $v_i$, and obtain a \textit{path group label embedding MBR}, denoted as $v_i.MBR_0$, minimally bounding all path group label embeddings $o_0(p_z)$ (as mentioned in Section~\ref{subsec:pruning}). Similarly, for any query vertex $q_i\in V(q)$, we can also obtain its path group label MBR $q_i.MBR_0$. Intuitively, if $q_i$ matches $v_i$, then their corresponding path label MBRs must satisfy the containment relationship (i.e., $q_i.MBR_0 \subseteq v_i.MBR_0$).

\begin{lemma}
    \textbf{(Path Group Label Pruning)}
    Given a vertex $v_i$ in the subgraph $g$ of data graph $G$ and a query vertex $q_i$ in query graph $q$, vertex $v_i$ can be safely pruned, if it holds that $q_i.MBR_0\nsubseteq v_i.MBR_0$.
    \label{lemma:path_label_pruning_group}
\end{lemma}

\begin{proof}
    If a query vertex $q_i$ in query graph $q$ matches a data vertex $v_i$ in a subgraph $g$, then the path group starting from $q_i$ must be a subset of that starting from $v_i$. As a result, the $x$-th hop neighbors of $q_i$ form a subset of the $x$-th hop neighbors of $v_i$, which leads to the containment relationship between their corresponding vertex label embedding MBRs. Therefore, if it holds that $q_i.MBR_0 \nsubseteq v_i.MBR_0$, it indicates that some neighbor(s) of $q_i$ have labels not contained in the neighbors of $v_i$. Hence, data vertex $v_i$ cannot match with query vertex $q_i$ and can thus be safely pruned, due to their neighbor label mismatch in the path group label set, which completes the proof.  \qquad $\square$
\end{proof}

\noindent{\bf Path Group Dominance Pruning.} Next, we consider the pruning via the dominance relationship for paths in path groups of query and data  vertices, $q_i$ and $v_i$, respectively. If $q_i$ matches with $v_i$, then the GNN-based embeddings of paths in $q_i.MBR$ must dominate some path group embeddings in $v_i.MBR$ (as discussed in Section \ref{subsec:path_group_emb}). In other words, if this condition does not hold, then data vertex $v_i$ can be safely pruned. Below, we have the lemma of the path group dominance pruning.

\begin{lemma}
    \textbf{(Path Group Dominance Pruning)}
    Given a vertex $v_i$ in a subgraph $g$ and a query vertex $q_i$ in query graph $q$, vertex $v_i$ cannot match with $q_i$, if $q_i.MBR_{min}\npreceq v_i.MBR_{max}$ holds.
    \label{lemma:path_dominance_pruning_group}
\end{lemma}

\begin{proof}
    Based on the dominance property of path group embedding in Section \ref{subsec:path_group_emb}, if query and data vertices, $q_i$ and $v_i$, match with each other, then their path group embeddings must satisfy the condition that: $q_i.MBR_{min}\preceq v_i.MBR_{max}$. Thus, by the contrapositive, if this condition does not hold (i.e., $q_i.MBR_{min}\npreceq v_i.MBR_{max}$), then $q_i$ does not match with $v_i$, which completes the proof.  \qquad $\square$
\end{proof}

\subsection{Indexing Mechanism Over PGE MBRs}

Similar to the GNN-PE index, for each graph partition $G_j$ ($1\leq j\leq m$), we build an aR-tree index $\mathcal{I}_j$ over path group label embeddings, $v_i.MBR_0$, and path group dominance embeddings, $v_i.MBR$, which can facilitate the pruning, as discussed in Lemmas \ref{lemma:path_label_pruning_group} and \ref{lemma:path_dominance_pruning_group}, respectively.

\noindent{\bf Leaf Nodes.}
Each leaf node $N\in \mathcal{I}_j$ contains $n$ path group embedding (PGE) MBRs, $v_i.MBR$, of vertices $v_i$. Each vertex $v_i$ (or PGE MBR $v_i.MBR$) is associated with aggregate data of path group label embeddings, $v_i.MBR_0$.

\noindent{\bf Non-Leaf Nodes.}
Each non-leaf node $N\in \mathcal{I}_j$ contains multiple entries $N_c$, each of which is represented by: (1) a path group embedding MBR, $N_c.MBR$, over PGE MBRs $v_i.MBR$ of all vertices $v_i$ under entry $N_c$, and; (2) a path group label embedding MBR, $N_c.MBR_0$, over path group label embedding MBRs $o_0(p_z)$ of all vertices $v_i$ under entry $N_c$.

\noindent{\bf Index-Level Pruning.}
Next, we discuss pruning strategies on the node level of indexes $\mathcal{I}_j$ ($1\leq j\leq m$), which can be used for effectively filtering out node entries with false alarms of vertices.

\underline{\it Index-Level Path Group Label Pruning.} For the \textit{index-level path group label pruning}, we prune those node entries $N_c$ that do not contain some path group label of the query vertex $q_i$. Specifically, by considering their path group label embeddings $N_c.MBR_0$ and $q_i.MBR_0$, we have the following lemma.

\begin{lemma}
{\bf (Index-Level Path Group Label Pruning)} Given a query vertex $q_i$ and an entry $N_c$ of index node $N$, entry $N_c$ can be safely pruned (w.r.t. $q_i$), if it holds that $q_i.MBR_0\nsubseteq N_c.MBR_0$.
\label{lemma:index_label_pruning_group}
\end{lemma}

\begin{proof}
    MBR $N_c.MBR_0$ bounds all the \textit{path group label embedding MBRs}, $v_i.MBR_0$, of vertices $v_i$ under node entry $N_c$. If it holds that $q_i.MBR_0\nsubseteq N_c.MBR_0$, it implies that some path group label of the query vertex $q_i$ does not appear in that of all vertices $v_i$ under $N_c$. Therefore, no vertex under node entry $N_c$ matches with query vertex $q_i$, and thus can be safely pruned.\qquad $\square$
\end{proof}

\underline{\it Index-Level Path Group Dominance Pruning.} The \textit{index-level path group dominance pruning} aims to rule out those node entries $N_c$, in each of which all \textit{path group dominance embedding MBRs} are not dominated by that of the query vertex $q_i$.

Let $DR(q_i.MBR_{min})$ be a \textit{dominating region} that is dominated by the minimum corner point of vertex $q_i$'s MBR $q_i.MBR$ in the embedding space. Then, we have the following lemma of index-level path group dominance pruning.

\begin{lemma}
{\bf (Index-Level Path Group Dominance Pruning)} Given a query vertex $q_i$ and a node entry $N_c$, entry $N_c$ can be safely pruned (w.r.t. $q_i$), if $DR(q_i.MBR_{min})\cap N_c.MBR = \emptyset$ (i.e., $q_i.MBR_{min} \npreceq N_c.MBR_{max}$) holds.
\label{lemma:index_dominance_pruning_group}
\end{lemma}

\begin{proof}
    If $DR(q_i.MBR_{min})$ and $N_c.MBR$ do not overlap with each other (i.e., $DR(q_i.MBR_{min})\cap N_c.MBR = \emptyset$), then node entry $N_c$ does not contain any paths whose path embeddings are dominated by that of $q_i$'s. In other words, $q_i$ does not match with any vertices under node entry $N_c$. Thus, entry $N_c$ can be safely pruned with respect to $q_i$, which completes the proof. \qquad $\square$
\end{proof}

\begin{algorithm}[t]
\caption{\bf Offline Pre-Computation of the GNN-PGE Approach}
\label{alg5}
\KwIn{
    i) a subgraph partition $G_j$ of data graph $G$;
    ii) path length $l$;
}
\KwOut{
    an index $\mathcal{I}_j$ over $G_j$; trained mult-GNN models $M_j$
}

\tcp{the model training (the same protocol as GNN-PE)}
extract unit star subgraph $g_{v_i}$ and its substructures $s_{v_i}$ to build the training data set

train GNN model $M_j$ until the loss equals 0 over this data set, following Algorithm~\ref{alg2}

generate vertex dominance embeddings $o(v_i)$ and vertex label embeddings $o_0(v_i)$ by the trained model $M_j$

\tcp{path group dominance embedding generation and index construction}

\For{each vertex $v_i\in G_j$}{

    obtain $v_i$'s $x$-hop neighbors in $\mathcal{N}_x(v_i)$ ($1 \leq x \leq l$)

    compute the path group dominance embedding MBR, $v_i.MBR$, and path group label embedding MBR, $v_i.MBR_0$, via Eqs. (\ref{eq:MBR1}) and (\ref{eq:MBR2})
    
    insert $v_i.MBR$ (associated with aggregate $v_i.MBR_0$) into the aR-tree index $\mathcal{I}_j$
}

{\bf return} $\mathcal{I}_j$ and $M_j$;\\

\end{algorithm}

\begin{algorithm}[t]
\caption{\bf Online Subgraph Matching of the GNN-PGE Approach}
\label{alg6}
\KwIn{
    i) an index $\mathcal{I}_j$ over graph partition $G_j$;
    ii) the trained multi-GNN models $M_j$, and;
    iii) a query graph $q$
}
\KwOut{
    a set, $\mathcal{S}$, of matching subgraphs
}

\tcp{generate query path group embeddings}

\For{each query vertex $q_i$}{
    generate path group dominance embedding MBR, $q_i.MBR$, and path group label embedding MBR, $q_i.MBR_0$, following lines 5-8 of Algorithm \ref{alg5}
}

\tcp{traverse index $\mathcal{I}_j$ to find candidate vertices}

initialize a \textit{maximum heap} $\mathcal{H}$ accepting entries in the form $(N, key(N))$

$root\big(\mathcal{I}_j\big).list \leftarrow V(q)$

insert $\big(root\big(\mathcal{I}_j\big), 0\big)$ into $\mathcal{H}$

\While{$\mathcal{H}$ is not empty}{
    de-heap an entry $(N, key(N)) = \mathcal{H}.pop()$
    
    \If{$key(N) < \min_{\forall q_i\in V(q)}\{||q_i.MBR_{min}||_1\}$}{
        terminate the loop;
    }
    
    \eIf{$N$ is a leaf node}{
        \For{each data vertex $v_c\in N$}{
            \For{each query vertex $q_i\in N.list$}{
                \If{$q_i.MBR_0\subseteq v_c.MBR_0$}{\tcp{Lemma \ref{lemma:path_label_pruning_group}}
                    \If{$q_i.MBR_{min}\preceq v_c.MBR_{max}$}{\tcp{Lemma \ref{lemma:path_dominance_pruning_group}}
                        $q_i.cand\_list \leftarrow q_i.cand\_list \cup \{v_c\}$
                    }
                }
            }
        }
    }
    {
        \For{each child node $N_c \in N$}{
            \For{each query vertex $q_i\in N.list$}{
                \If{$q_i.MBR_0\subseteq N_c.MBR_0$}{ \tcp{Lemma \ref{lemma:index_label_pruning_group}}
                    \If{$DR(q_i.MBR_{min})\cap N_c.MBR \ne \emptyset$}{\tcp{Lemma \ref {lemma:index_dominance_pruning_group}}
                        $N_c.list \leftarrow N_c.list \cup \{q_i\}$ 
                    }
                }
            }
            \If{$N_c.list \ne \emptyset$}{
                insert $(N_c, key(N_c))$ into heap $\mathcal{H}$\\
            }
        }
    }
}
    
\tcp{generate the matching order }

obtain an ordered list $\Phi$ based on the number of candidates for each query vertex

\tcp{refine candidate subgraphs}

invoke a back-tracking search function to obtain the actual matching subgraphs $g$ and add them to $\mathcal{S}$

{\bf return} $\mathcal{S}$;

\end{algorithm}

\subsection{GNN-PGE Algorithm}
In this subsection, we illustrate the exact subgraph matching algorithm, named GNN-PGE, by traversing indexes over GNN-based path group label/dominance embeddings, which, similar to GNN-PE, consists of \textit{offline pre-computation} and \textit{online subgraph matching phases}.

\noindent{\bf Offline Pre-Computation.} As shown in Algorithm \ref{alg5}, similar to GNN-PE, the offline pre-computation phase first trains the model $M_j$ on unit star subgraphs $g_{v_i}$ and their substructures $s_{v_i}$ until the dominance loss converges to zero, and then generates vertex dominance embeddings $o(v_i)$ and label embeddings $o_0(v_i)$ (as mentioned in Section~\ref{subsec:pruning}) for all vertices $v_i \in G_j$ (lines 1–3).
Next, we obtain an $x$-hop neighbor set, $\mathcal{N}_x(v_i)$, for each vertex $v_i$, where $1 \leq x \leq l$ (line 5). 
Next, following Eqs. (\ref{eq:MBR1}) and (\ref{eq:MBR2}), we compute both the path group dominance embedding MBR, $v_i.MBR$, and path group label embedding MBR, $v_i.MBR_0$, by concatenating $l$ MBRs of vertices’ dominance and label embeddings in $\mathcal{N}_x(v_i)$ (for $1\leq x\leq l$), respectively (line 6).
The path group dominance embedding MBR $v_i.MBR$ is then inserted into an aR-tree index $\mathcal{I}_j$, and the path group label embedding MBR $v_i.MBR_0$ is treated as the aggregate of MBR $v_i.MBR$ (line 7).
Finally, we return the constructed index $\mathcal{I}_j$ and trained GNN model $M_j$ (line 8).

\noindent{\bf Online Subgraph Matching Processing.}
Algorithm \ref{alg6} illustrates the online processing of a subgraph matching query, by traversing the indexes over GNN-based path group embedding MBRs. Specifically, given a query graph $q$, we first generate path group embeddings of query vertices $q_i$ via the trained GNN model, in the same way as introduced in Algorithm \ref{alg5} (lines 1-2).
Next, we traverse the index $\mathcal{I}_j$ once to retrieve a candidate vertex set for each query vertex $q_i\in V(q)$ (lines 3-23). Finally, we generate a matching order and 
invoke the back-tracking search to refine/return subgraphs $g\in \mathcal{S}$ that are isomorphic to $q$ (lines 24-26).

\underline{\it Index Traversal.}
To traverse the index $\mathcal{I}_j$, we initialize a \textit{maximum heap} $\mathcal{H}$, which accepts entries in the form of $(N, key(N))$ (line 3), where $N$ is an index entry and $key(N)$ is the key of index entry $N$ in the heap. 
Here, for an index entry $N$, the $key(N)$ of entry $N$ is calculated by $key(N) = ||N.MBR_{max}||_1$, where a $d$-dimensional vector $Z$ has the $L_1$-norm $||Z||_1 = \sum_{i=1}^d|Z[i]|$.

We also initialize a list, $root(\mathcal{I}_j).list$, of the index root with all query vertices in $V(q)$, and insert the first entry $(root(\mathcal{I}_j), 0)$ into heap $\mathcal{H}$ (lines 4-5). We traverse index $\mathcal{I}_j$ by accessing entries from heap $\mathcal{H}$ (lines 6-23). Specifically, each time we pop out an index entry $(N,$ $key(N))$ with the maximum key from heap $\mathcal{H}$ (line 7). If it holds that $key(N)<\min_{\forall q_i\in V(q)}\{||q_i.MBR_{min}||_1\}$, which indicates that vertices in the remaining entries of $\mathcal{H}$ cannot be dominated by any query vertex $q_i\in V(q)$, then we can terminate the index traversal (lines 8-9); otherwise, we will check entries in node $N$.

When we encounter a leaf node $N$, for each vertex $v_c\in N$ and each query vertex $q_i\in N.list$, we apply the \textit{path group label pruning} (Lemma~\ref{lemma:path_label_pruning_group}) and \textit{path group dominance pruning} (Lemma~\ref{lemma:path_dominance_pruning_group}) (lines 10-14). If $v_c$ cannot be ruled out by these two pruning methods, we will add $v_c$ to the candidate vertex set $q_i.cand\_list$ (line 15).

When $N$ is a non-leaf node, we consider each child node $N_c\in N$ and each query vertex $q_i\in N.list$ (lines 16-18). If entry $N_c$ cannot be pruned by \textit{index-level path group label pruning} (Lemma~\ref{lemma:index_label_pruning_group}) or \textit{index-level path group dominance pruning} (Lemma~\ref{lemma:index_dominance_pruning_group}), then we will add $q_i$ to $N_c.list$ for further checking (lines 19-21).
If $N_c.list$ is not empty (i.e., node $N_c$ may contain candidate vertices for some query vertex), we will insert entry $(N_c, key(N_c))$ into heap $\mathcal{H}$ for further investigation (lines 22-23).

The index traversal terminates when either heap $\mathcal{H}$ is empty (line 6) or the remaining heap entries in $\mathcal{H}$ cannot contain vertex candidates (lines 8-9).

\underline{\it Matching Order Generation.}
After finding all candidate vertices in $q_i.cand\_list$ for each query vertex $q_i\in V(q)$, we generate a matching order $\phi$ for the refinement (line 24). The selection of the matching order has the goal of reducing the size of intermediate join results. Thus, we will first choose a query vertex $q_i$ with the smallest list size $|q_i.cand\_list|$ as the first query vertex in an ordered list $\Phi$, and then iteratively append a neighbor $q_j$ of $q_i\in \Phi$ with the minimum number of candidates to $\Phi$ until all query vertices in $V(q)$ have been added to $\Phi$.

\underline{\it Refinement.}
Based on the matching order $\Phi$, we use a left-deep join based method \cite{kankanamge2017graphflow,shang2008taming} to assemble candidate vertices in $q_i.cand\_list$ into candidate subgraphs, and obtain/return the actual matching subgraph answer set $\mathcal{S}$ (lines 25-26).

\underline{\it Complexity Analysis.}
In Algorithm \ref{alg6}, since the time complexity of path group embedding construction is $O(|V(q)|\cdot l\cdot avg\_deg(q)\cdot d)$ and the GNN computation cost is $O(|E(q)|+d^2)$ \cite{wu2020comprehensive,wang2022reinforcement}, the time complexity of generating path group embeddings (lines 1-2) is given by $O(|V(q)|\cdot l\cdot avg\_deg(q)\cdot d+V|q|\cdot (|E(q)|+d^2))$. 

For the index traversal (lines 3-23), assume that $h$ is the height of the aR-tree $\mathcal{I}_j$, $f$ is the average fanout of each index node $N$, and the pruning power is $PP_i$ on the $i$-th level of the tree index.
Then, on the $i$-th level, the number of index nodes (or paths) to be accessed is given by $f^{h-i+1}\cdot (1-PP_{i})$. Thus, the total index traversal cost has $O\big(\sum^{h}_{i=0}|Q|\cdot f^{h-i+1}\cdot (1-PP_{i})\big)$ complexity.

Next, for the greedy-based matching order generation (line 24), we need to iteratively select a neighbor of vertices in the query plan $Q$, which requires $O(|V(q)|^2)$ cost.

Finally, we use a left-deep join to assemble and refine candidate vertices (line 25). The worst-case time complexity is given by $O(\prod_{i=0}^{|V(q)|-1}|q_i.cand\_list|)$.

Therefore, the overall time complexity of Algorithm \ref{alg6} is given by $O(V|q|\cdot (|E(q)|+d^2)+|V(q)|\cdot l\cdot avg\_deg(q)\cdot d+\sum^{h}_{i=0}|Q|\cdot f^{h-i+1}\cdot (1-PP_{i})+|V(q)|^2+\prod_{i=0}^{|V(q)|-1}|q_i.cand\_list|)$.

\section{Experimental Evaluation}
\label{sec:expr}

\subsection{Experimental Settings}
To evaluate the effectiveness and efficiency of our GNN-PE approach, we conduct experiments on an Ubuntu server equipped with an Intel Core i9-12900K CPU, 128GB memory, and NVIDIA GeForce RTX 4090 GPU. The GNN training of our approach is implemented by PyTorch, where embedding vectors are offline computed on the GPU. The online subgraph matching is implemented in C++ with multi-threaded support on the CPU to enable parallel search on multiple subgraph partitions. 

For the GNN model (as mentioned in Section~\ref{subsec:node_embedding}), by default, we set the dimension of initial input node feature $F=1$, attention heads $K=3$, the dimension of hidden node feature $F'=32$, and the dimension of the output node embedding $d=2$.
During the training process, we use the Adam optimizer to update parameters and set the learning rate $\eta=0.001$. Due to different data sizes, we set different batch sizes for different data sets (from 128 to 1,024). Based on the statistics of real data graphs (e.g., only $7.24\%$ vertices have a degree greater than 10 in Youtube), we set the degree threshold $\theta$ to 10 by default.

Our source code and real/synthetic graph data sets are available at URL: \textit{\url{https://github.com/JamesWhiteSnow/GNN-PE}}.

\noindent{\bf Baseline Methods.} We compare the performance of our GNN-PE approach with that of eight representative subgraph matching baseline methods as follows: 

\begin{itemize}
    \item  {\bf GraphQL (GQL)} \cite{he2008graphs} is a ``graphs-at-a-time'' method that improves the query efficiency by retrieving multiple related patterns simultaneously.

    \item {\bf QuickSI (QSI)} \cite{shang2008taming} is a direct-enumeration method that filters out the unpromising vertices during the enumeration.

    \item {\bf RI} \cite{bonnici2013subgraph} is a state space representation-based model that prioritizes constraints based on pattern graph topology to generate the query plan.

    \item {\bf CFLMatch (CFL)} \cite{bi2016efficient} is a preprocessing-enumeration method that utilizes the concept of context-free language to express substructures of a graph.

    \item {\bf VF2++ (VF)} \cite{juttner2018vf2++} is a state space representation model that employs heuristic rules, pruning strategies, and preprocessing steps to enhance the performance.

    \item {\bf DP-iso (DP)} \cite{han2019efficient} is a preprocessing-enumeration method that combines dynamic programming, adaptive matching order, and failing set to improve the matching efficiency.

    \item {\bf CECI} \cite{bhattarai2019ceci} is a compact embedding cluster index method that enables efficient matching operations by leveraging compact encoding, clustering, and neighbor mapping techniques.

    \item {\bf Hybrid} \cite{sun2020memory} is a hybrid method where the algorithms of the candidate filtering, query plan generation, and enumeration use GQL, RI, and QSI, respectively.
\end{itemize}

We used the code of baseline methods from \cite{sun2020memory}, which is implemented in C++ by enabling multi-threaded CPU support for a fair comparison.

\begin{table}[t]\small
\begin{center}
\caption{Statistics of real-world graph data sets.}
\label{tab:datasets}
\begin{tabular}{|l||c|c|c|c|}
\hline
\textbf{\text{ }\text{ }Data Sets}&\textbf{$|V(G)|$}&\textbf{$|E(G)|$}&\textbf{$|\sum|$}&\textbf{$avg\_deg(G)$} \\
\hline\hline
    Yeast (ye) & 3,112 & 12,519 & 71 & 8.0\\\hline
    Human (hu) & 4,674 & 86,282 & 44 & 36.9\\\hline
    HPRD (hp) & 9,460 & 34,998 & 307 & 7.4\\\hline
    WordNet (wn) & 76,853 & 120,399 & 5 & 3.1\\\hline
    DBLP (db) & 317,080 & 1,049,866 & 15 & 6.6\\\hline
    Youtube (yt) & 1,134,890 & 2,987,624 & 25 & 5.3\\\hline
    US Patents (up) & 3,774,768 & 16,518,947 & 20 & 8.8\\\hline
\end{tabular}
\end{center}
\end{table}

\begin{table}[t]\small
\begin{center}
\caption{Parameter settings.}
\label{tab:parameters}
\begin{tabular}{|p{5cm}||p{3.2cm}|}
\hline
\textbf{Parameters}&\textbf{Values} \\
\hline\hline
    the path length $l$ & 1, {\bf 2}, 3\\\hline
    the dimension, $d$, of the node embedding vector  & {\bf 2}, 3, 4, 5\\\hline
    the number, $n$, of multi-GNNs & 0, 1, {\bf 2}, 3, 4\\\hline 
    the number, $b$, of GNNs with randomized initial weights & {\bf 1}, 2, 3, 4, 5, 6, 7, 8, 9, 10\\\hline
    the size, $|V(q)|$, of the query graph $q$ & 5, 6, {\bf 8}, 10, 12\\\hline
    the average degree, $avg\_deg(q)$, of the query graph $q$ & 2, {\bf 3}, 4\\\hline  
    the size, $|V(G)|/m$, of subgraph partitions & 5K, 6K, {\bf 10K}, 20K, 50K\\\hline
    the number, $|\sum|$, of distinct labels & 100, 200, {\bf 500}, 800, 1K\\\hline
    the average degree, $avg\_deg(G)$, of the data graph $G$ & 3, 4, {\bf 5}, 6, 7\\\hline     
    the size, $|V(G)|$, of the data graph $G$ & 10K, 30K, {\bf 50K}, 80K, 100K, 500K, 1M\\\hline 
\end{tabular}
\end{center}
\end{table}

\noindent{\bf Real/Synthetic Graph Data Sets.} We use both real and synthetic graphs to evaluate our GNN-PE approach, compared with baselines.

\underline{\textit{Real-world graphs.}} We used seven real-world graph data used by previous works \cite{he2008graphs, shang2008taming, zhao2010graph, lee2012depth, sun2012efficient, han2013turboiso, ren2015exploiting, bi2016efficient, katsarou2017subgraph, bhattarai2019ceci, han2019efficient, sun2020memory}, which can be classified into four categories: i) biology networks (Yeast, Human, and HPRD); ii) lexical networks (WordNet); iii) bibliographical/social networks (DBLP and Youtube); and iv) citation networks (US Patents). 
Based on the graph size, we divide Yeast, Human, and HPRD into 5 partitions, WordNet into 7 subgraphs, DBLP into 30 partitions, Youtube into 346 partitions, and US Patents into 1,000 partitions. Statistics of these real graphs are summarized in Table~\ref{tab:datasets}. 

\underline{\textit{Synthetic graphs.}} We generated synthetic graphs via NetworkX \cite{hagberg2020networkx} which produces small-world graphs following the Newman-Watts-Strogatz model \cite{watts1998collective}. Parameter settings of synthetic graphs are depicted in Table~\ref{tab:parameters}. For each vertex $v_i$, we generate its label $L(v_i)$ by randomly picking up an integer in the range $[1, |\sum|]$, following the Uniform, Gaussian, or Zipf distribution. Accordingly, we obtain three types of data graphs, denoted as $Syn\text{-}Uni$, $Syn\text{-}Gau$, and $Syn\text{-}Zipf$, respectively.

\begin{figure}[t]
\centering
\subfigure[][{Syn-Uni}]{                    
\scalebox{0.13}[0.13]{\includegraphics{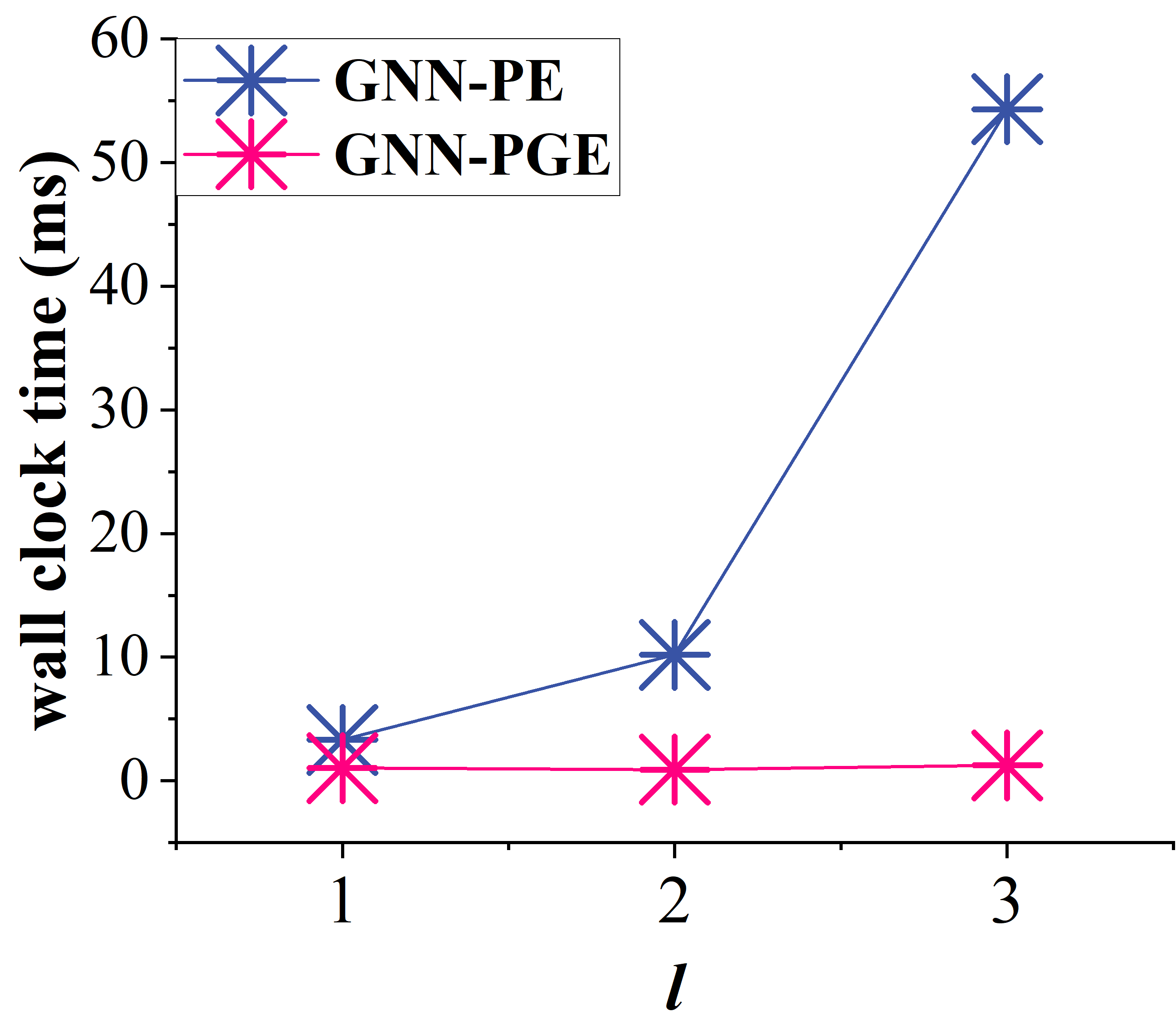}}\label{subfig:l}}
\subfigure[][{Syn-Gau}]{
\scalebox{0.13}[0.13]{\includegraphics{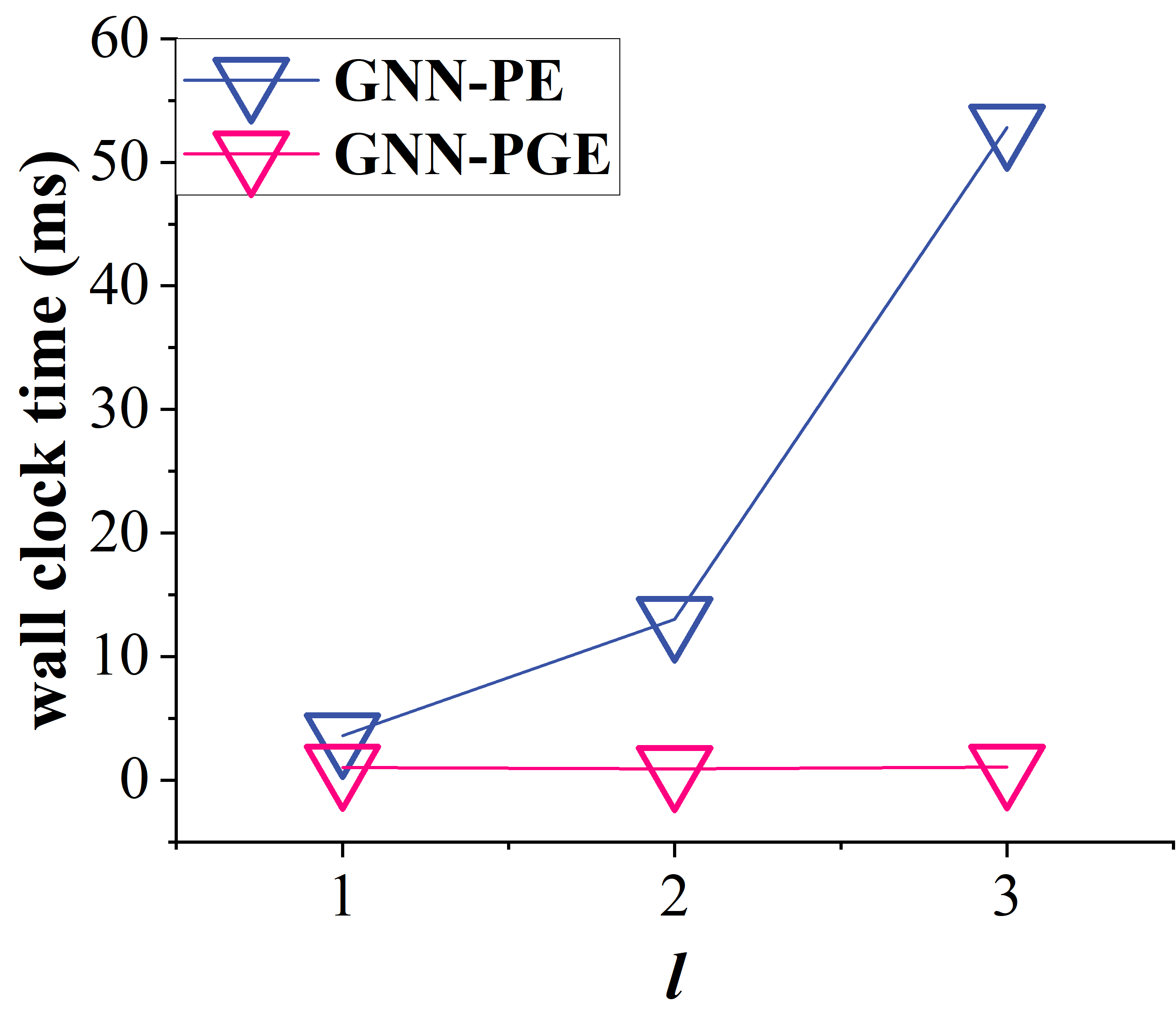}}\label{subfig:b}}
\subfigure[][{Syn-Zipf}]{
\scalebox{0.13}[0.13]{\includegraphics{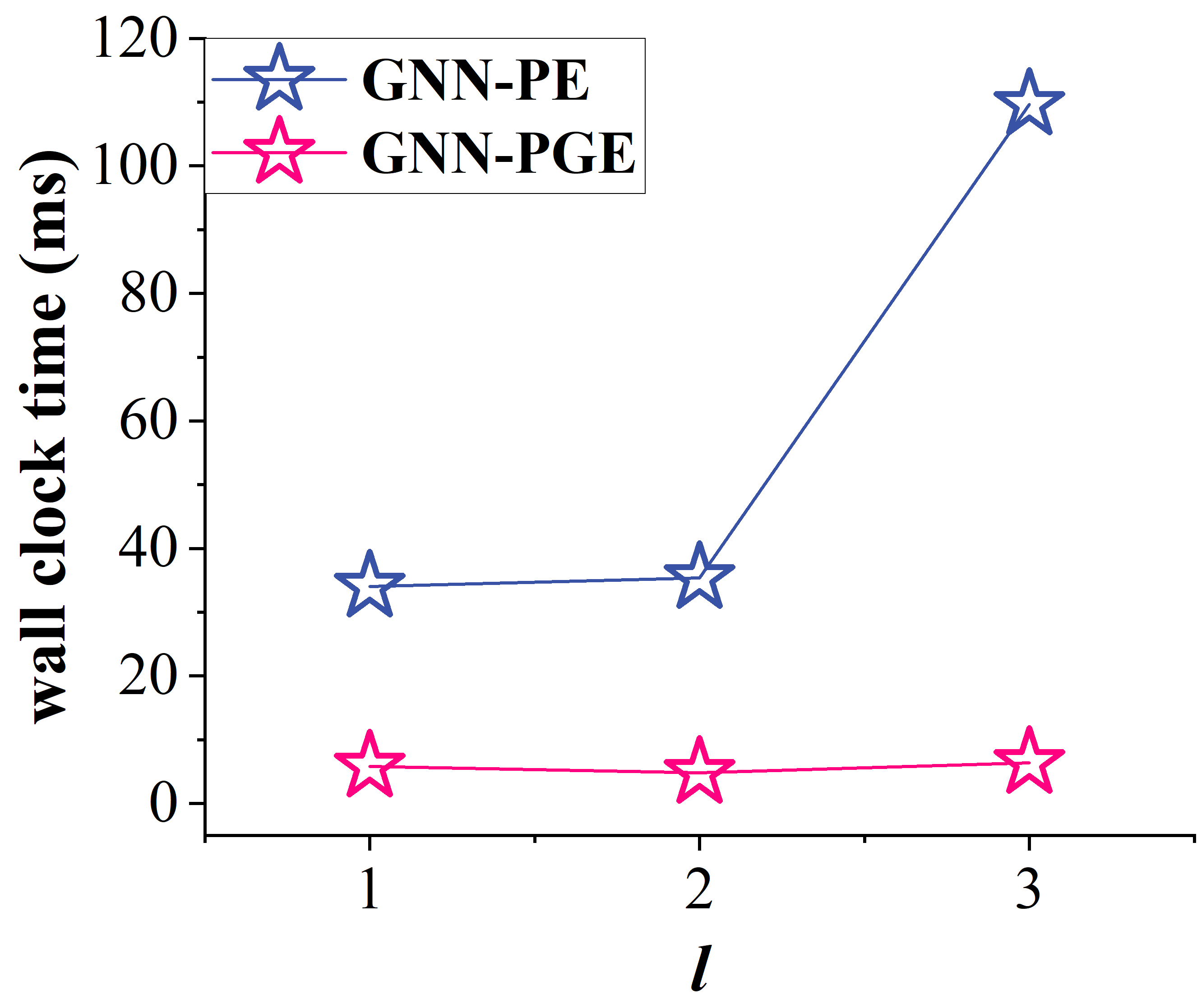}}\label{subfig:strategies}}
\caption{Efficiency evaluation w.r.t different path lengths $l$.}
\label{fig:path_length}
\end{figure}

\begin{figure}[ht]
\centering
\subfigure[][{embedding dimension $d$}]{
\scalebox{0.17}[0.17]{\includegraphics{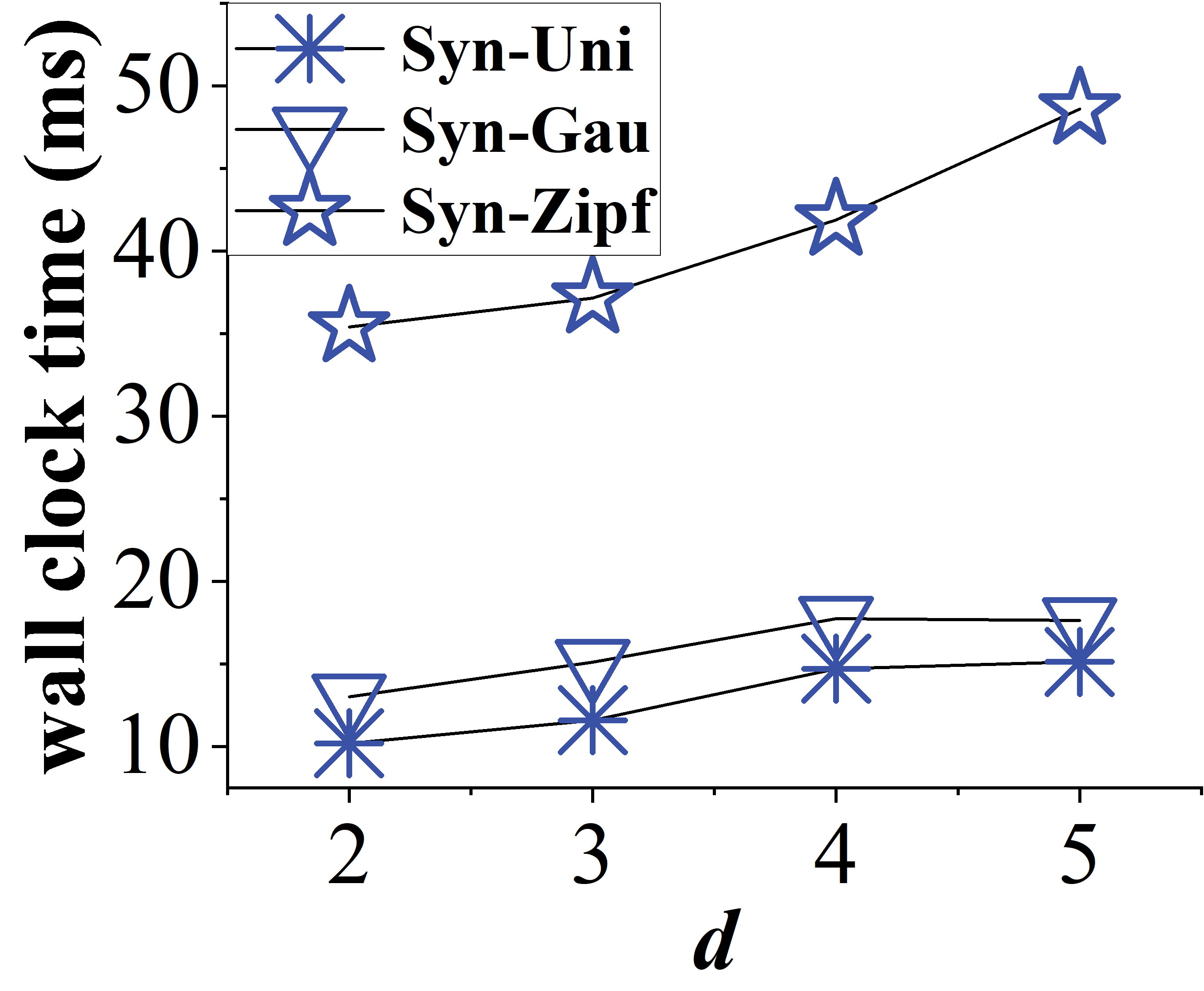}}\label{subfig:d}}
\subfigure[][{\# of multi-GNNs, $n$}]{                 
\scalebox{0.17}[0.17]{\includegraphics{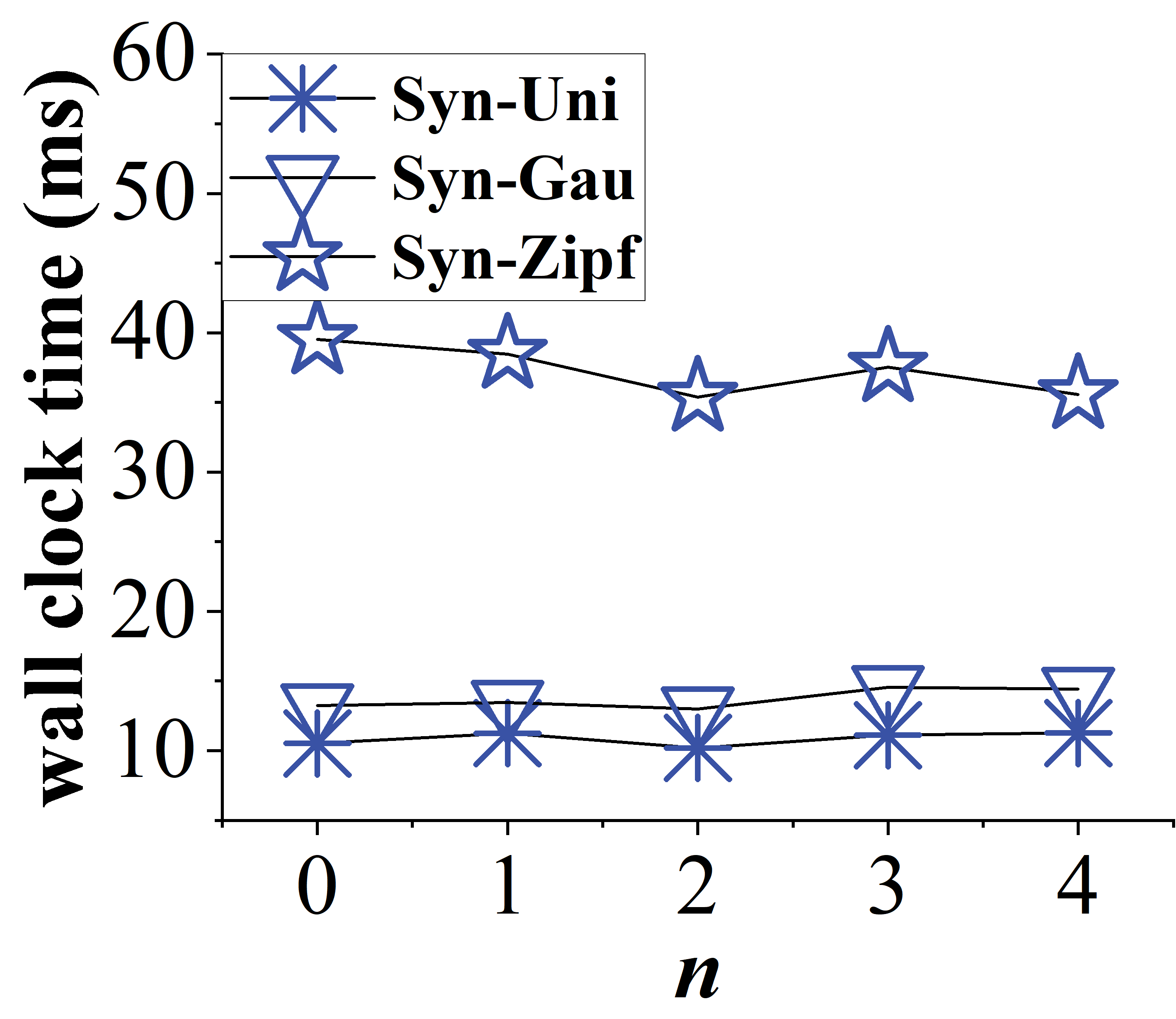}}\label{subfig:n}}
\\
\subfigure[][{\# of initial GNNs, $b$}]{
\scalebox{0.17}[0.17]{\includegraphics{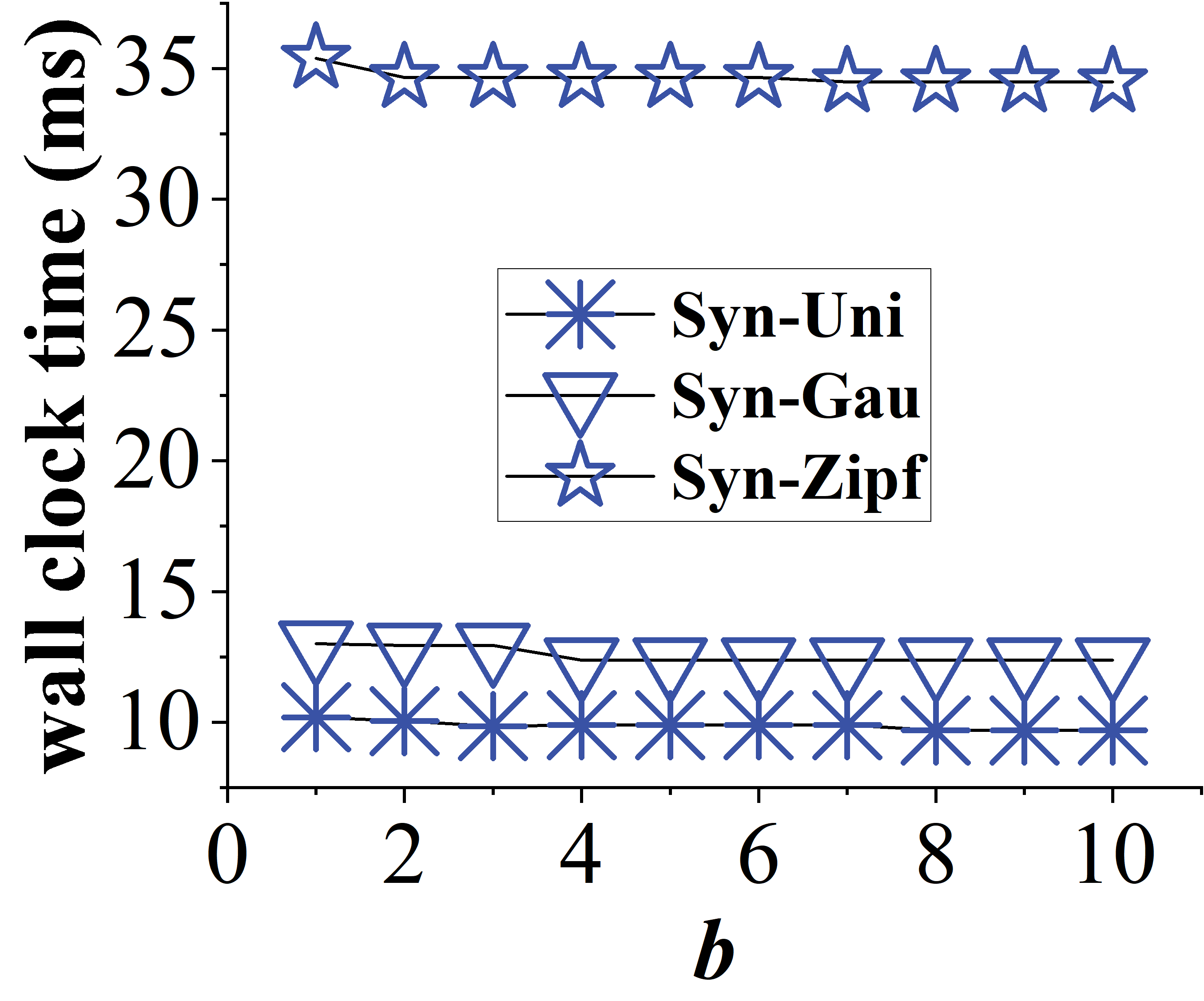}}\label{subfig:b}}
\subfigure[][{query plan selection}]{
\scalebox{0.17}[0.17]{\includegraphics{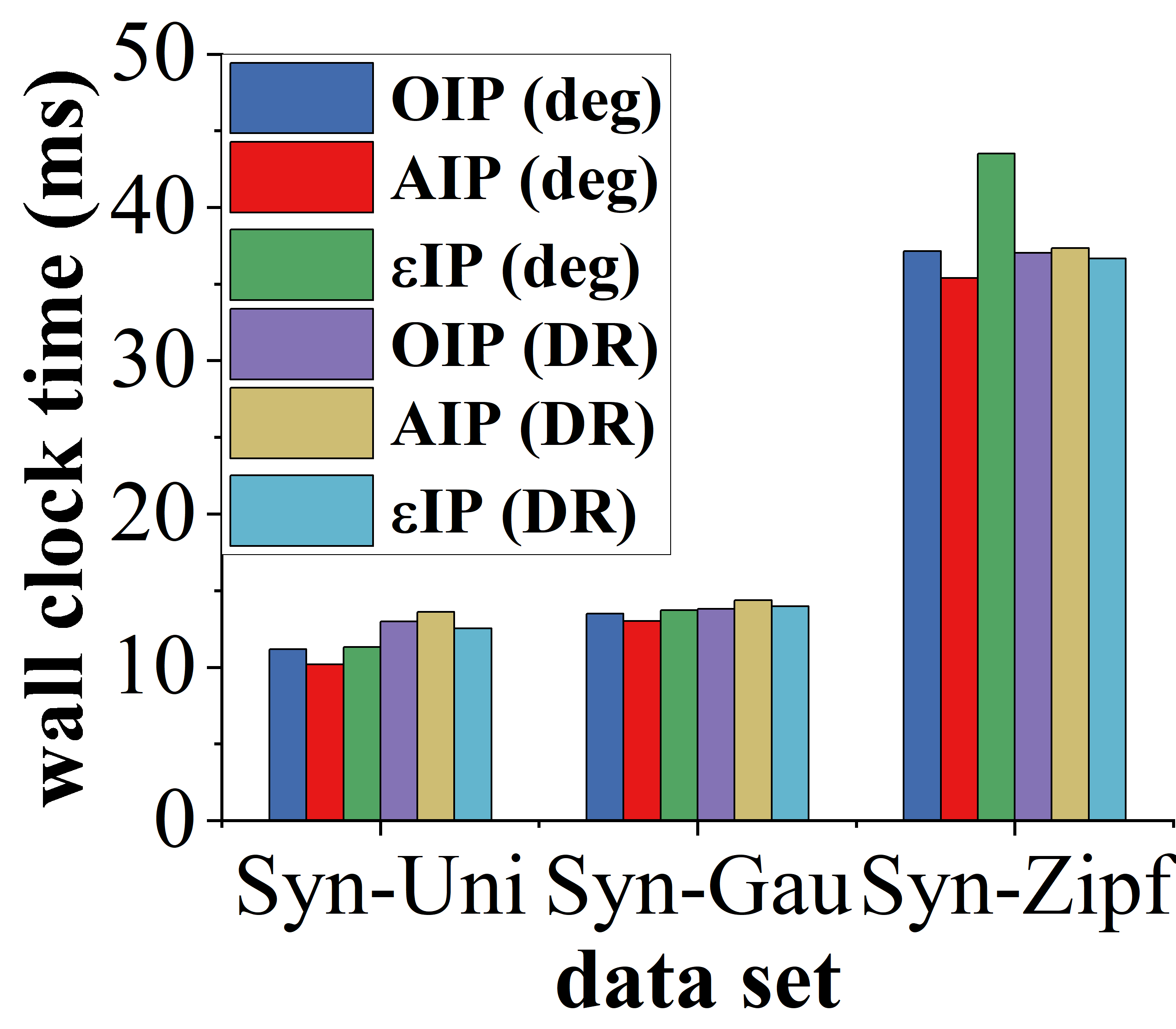}}\label{subfig:strategies}}
\caption{GNN-PE efficiency evaluation w.r.t different parameters $\bm{d}$, $\bm{n}$, $\bm{b}$, and query plan selection strategies.}
\label{fig:parameter_tuning}
\end{figure}

\noindent{\bf Query Graphs.} Similar to previous works \cite{sun2012efficient, han2013turboiso, ren2015exploiting, bi2016efficient, katsarou2017subgraph, archibald2019sequential, bhattarai2019ceci, han2019efficient}, for each graph data set $G$, we randomly extract/sample 100 connected subgraphs from $G$ as query graphs, where parameters of query graphs $q$ (e.g., $|V(q)|$ and $avg\_deg(q)$)  are depicted in Table~\ref{tab:parameters}. 
Specifically, to generate a query graph $q$, we first perform a random walk in the data graph $G$ until obtaining $|V(q)|$ vertices, and then check whether or not the average degree of the induced subgraph is larger than or equal to $avg\_deg(q)$. If yes, we randomly delete edges from the subgraph, until the average degree becomes $avg\_deg(q)$; otherwise, we start from a new vertex to perform the random walk.

\noindent{\bf Evaluation Metrics.}  
In our experiments, we report the efficiency of our GNN-PE approach and baseline methods, in terms of the \textit{wall clock time} (including both filtering and refinement time costs). We also evaluate the \textit{pruning power} of our pruning strategies (as mentioned in Sections \ref{subsec:pruning} and \ref{subsec:pruning_group}), which is the percentage of candidate vertices that can be ruled out by our pruning methods for GNN-PE and GNN-PGE, respectively. For all the experiments, we take an average of each metric over 100 runs (w.r.t. 100 query graphs, respectively). We also test offline pre-computation costs of our GNN-PE and GNN-PGE approaches, including the GNN training time, path embedding time, and index construction time. 

Table~\ref{tab:parameters} depicts parameter settings in our experiments, where default parameter values are in bold. For each set of subsequent experiments, we vary the value of one parameter while setting other parameters to their default values.

\subsection{Parameter Tuning}
In this subsection, we first tune parameters and query plan selection strategies for our GNN-PE approach over synthetic graphs.

\noindent {\bf The GNN-PE/GNN-PGE Efficiency Evaluation w.r.t. Path Length $\bm{l}$.}
Figure \ref{fig:path_length} illustrates the performance of GNN-PE and GNN-PGE, by varying the path length $l$ from 1 to 3, where other parameters are by default. When $l$ increases, more vertex labels and dominance embeddings on each path are used for pruning, which incurs higher pruning power. On the other hand, however, longer path length $l$ will result in more data paths from graph $G$ and higher dimensionality of the index, which may lead to higher costs to process more candidate paths. Thus, the GNN-PE efficiency is affected by the two factors above. From the figure, when $l=1, 2$, the wall clock times are comparable for all three synthetic graphs; when $l=3$, the time cost suddenly increases due to much more candidate paths to process and the ``dimensionality curse'' \cite{BerchtoldKK96}.
On the other hand, due to the path grouping, GNN-PGE is not very sensitive to changes in the length of a single path.
Nonetheless, for all $l$ values, the wall clock time remains low (i.e., 0.0033 $sec$ $\sim$ 0.1096 $sec$ for GNN-PE and 0.0009 $sec$ $\sim$ 0.0063 $sec$ for GNN-PGE).

\noindent {\bf The GNN-PE Efficiency Evaluation w.r.t. Embedding Dimension $\bm{d}$.}
Figure \ref{subfig:d} varies the embedding dimension $d$ via a GNN from 2 to 5,  where other parameters are set to their default values. With the increase of the embedding dimension $d$, the wall clock time also increases. This is mainly due to the ``dimensionality curse'' \cite{BerchtoldKK96} when traversing the aR-tree index. Nonetheless, for different $d$ values, the query cost remains low (i.e., less than 0.02 $sec$ for $Syn\text{-}Uni$ and $Syn\text{-}Gau$, and 0.05 $sec$ for $Syn\text{-}Zipf$).

\noindent {\bf The GNN-PE Efficiency Evaluation w.r.t. Number, $\bm{n}$, of Multi-GNNs.}
Figure \ref{subfig:n} shows the performance of our GNN-PE approach, by varying the number, $n$, of multi-GNNs from 0 to 4, where other parameters are set by default. When $n=0$ (i.e., no multi-GNNs), $Syn\text{-}Zipf$ takes a higher time cost than that for $n>0$ (due to the lower pruning power without multi-GNNs). Since more GNNs for embeddings may lead to higher pruning power, but higher processing cost as well, the total GNN-PE time cost is affected by these two factors and relatively stable over $Syn\text{-}Uni$ and $Syn\text{-}Gau$ for $n>0$. Overall, for different $n$ values, the time cost remains low (i.e., 0.01 $sec$ $\sim$ 0.04 $sec$).

\noindent{\bf The GNN-PE Efficiency Evaluation w.r.t. Number, $\bm{b}$, of GNNs with Random Initial Weights.}
Figure~\ref{subfig:b} illustrates the performance of our GNN-PE approach, by varying the number, $b$, of the trained GNN models with random initial weights from 1 to 10, where other parameters are set by default. The experimental results show that the GNN-PE performance is not very sensitive to $b$. Since more GNNs trained will lead to higher training costs, in this paper, we set $b=1$ by default. Nonetheless, for different $b$ values, the query cost remains low (i.e., 0.01 $sec$ $\sim$ 0.04 $sec$).

\noindent {\bf The GNN-PE Efficiency Evaluation w.r.t. Query Plan Selection Strategies.}
Figure \ref{subfig:strategies} reports the performance of our GNN-PE approach with different query plan selection strategies, OIP, AIP, and $\boldsymbol{\varepsilon}$IP (as mentioned in Section \ref{sec:query_plan}), where the path weight $w(p_q)$ is estimated by vertex degrees or counts in dominating regions (denoted as (deg) and (DR), respectively), and default values are used for all parameters. From the figure, we can see that different strategies result in slightly different performances, and AIP (deg) consistently achieves the best performance. The experimental results on real-world graphs are similar and thus omitted here.


In subsequent experiments, we will set parameters $l=2$, $d=2$, $n=2$, and $b=1$, and use AIP (deg) as our default query plan selection strategy.

\begin{figure}[t]
\centering
\subfigure[][{real-world graphs}]{                    
\scalebox{0.13}[0.13]{\includegraphics{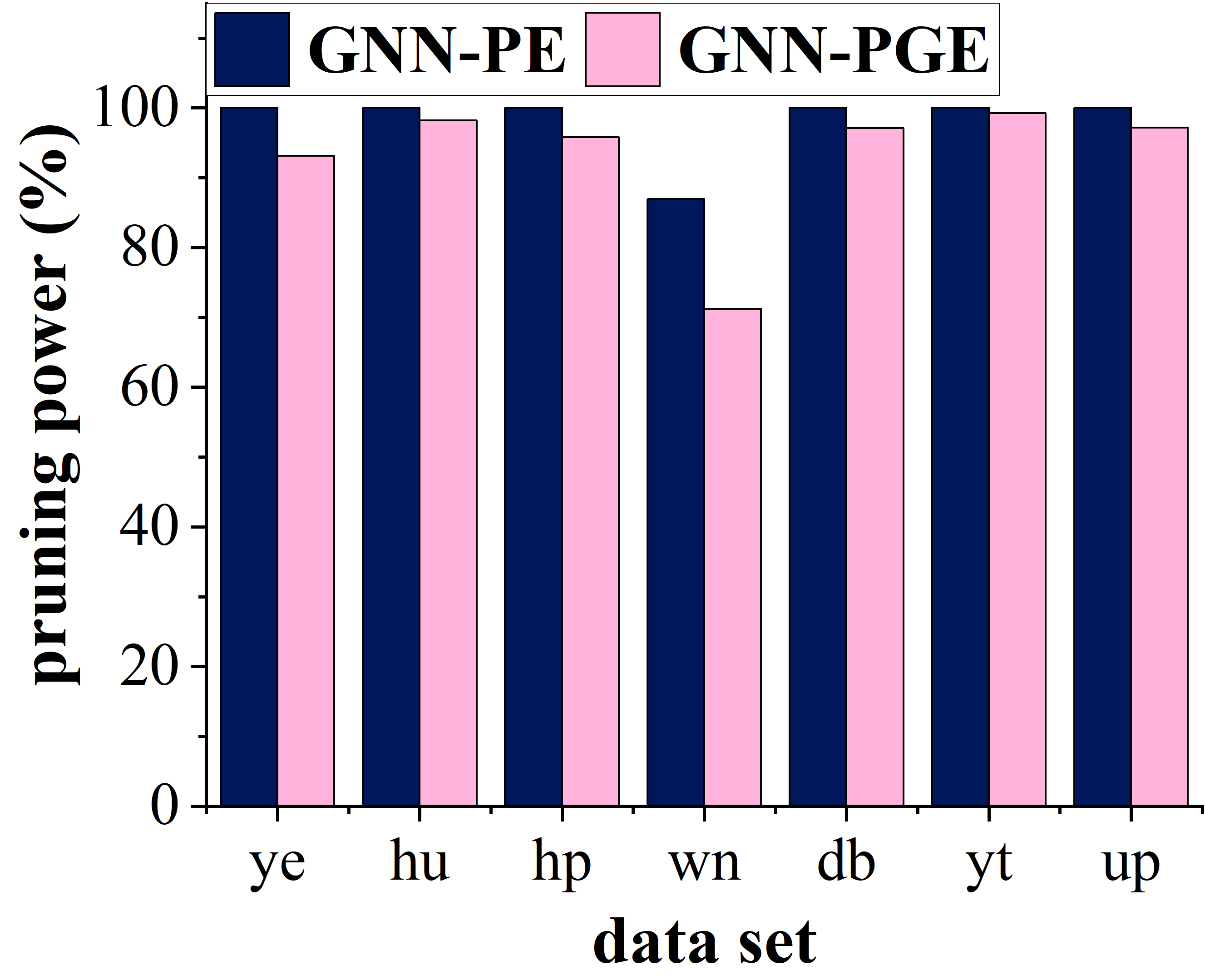}}\label{subfig:real_pp}}
\qquad\qquad
\subfigure[][{synthetic graphs}]{
\scalebox{0.13}[0.13]{\includegraphics{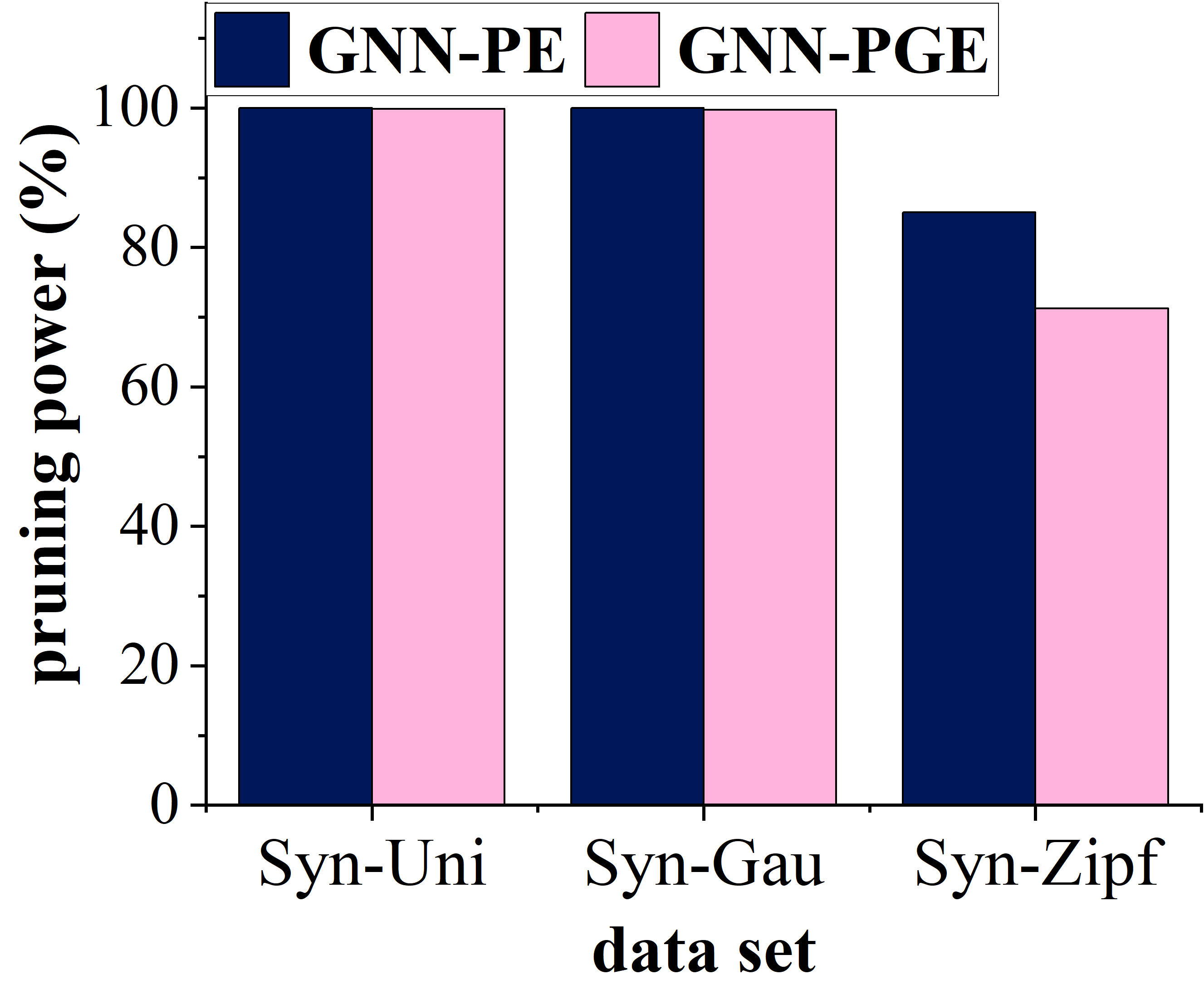}}\label{subfig:syn_pp}}
\caption{GNN-PE pruning power on real/synthetic graphs.} 
\label{fig:pruning_power}
\end{figure}

\subsection{Evaluation of the GNN-PE/GNN-PGE Effectiveness}
In this subsection, we report the pruning power of our pruning strategies for GNN-PE (Section \ref{subsec:pruning}) over real/synthetic graphs.

\noindent {\bf The GNN-PE/GNN-PGE Pruning Power on Real/Synthetic Graphs.}
Figure~\ref{fig:pruning_power} shows the \textit{pruning power} of our proposed pruning strategies (i.e., path label/dominance pruning in Section~\ref{subsec:pruning} and path group label/dominance pruning in Section \ref{subsec:pruning_group}) over both real and synthetic graphs, where the default values are used for all parameters. In subfigures, we can see that for all real/synthetic graphs, the strongest pruning power can reach as high as 99.58\% for GNN-PE and 99.23\% for GNN-PGE (i.e., filtering out 99.58\% or 99.23\% of candidate vertices), which confirms the effectiveness of our pruning strategies and the efficiency of our proposed GNN-PE and GNN-PGE approaches.

\begin{figure}[t]
\subfigcapskip=-0.cm
\centering
\subfigure[][{real-world graphs}]{
\scalebox{0.16}[0.16]{\includegraphics{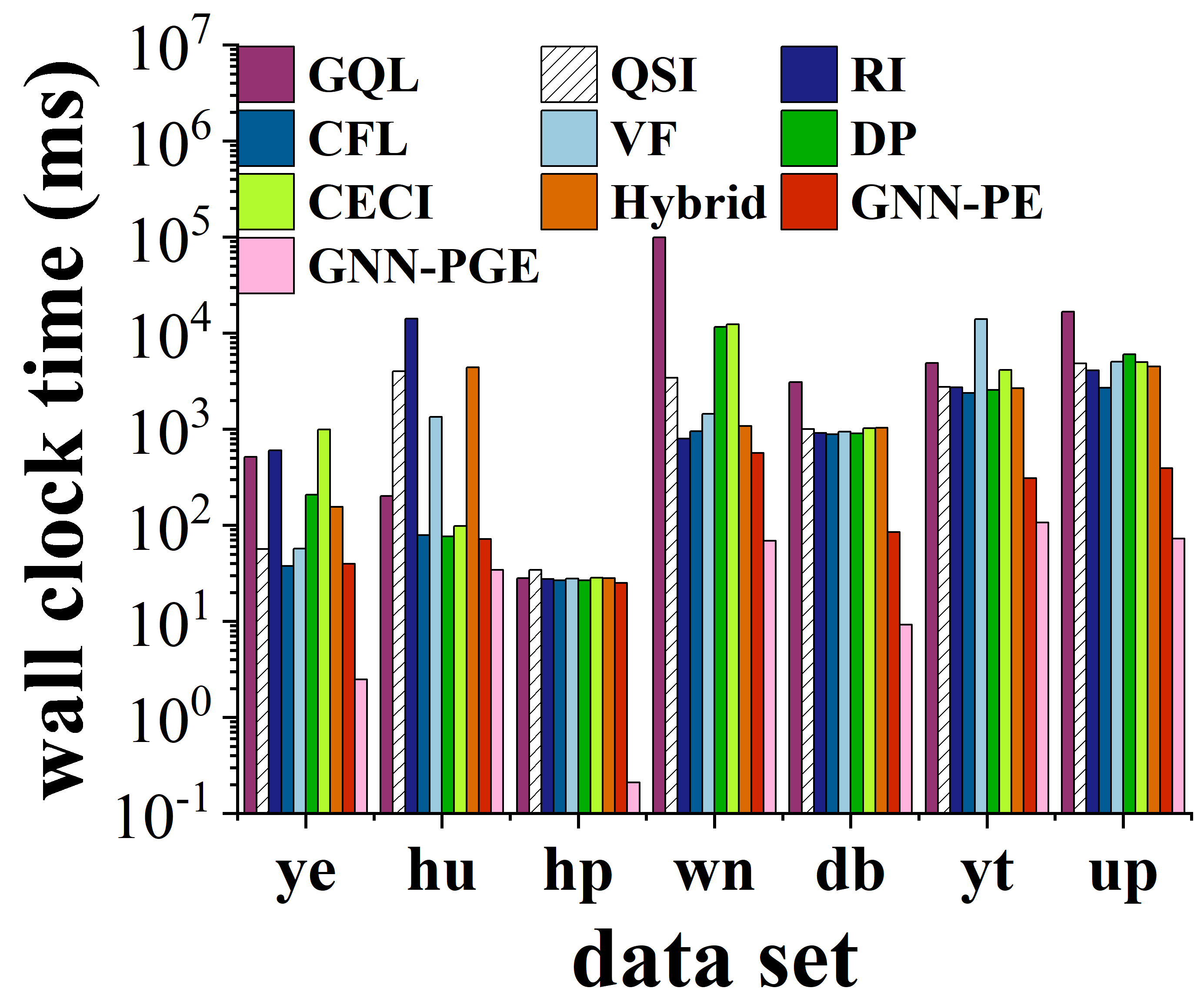}}\label{subfig:real_data}}
\quad
\subfigure[][{synthetic graphs}]{
\scalebox{0.16}[0.16]{\includegraphics{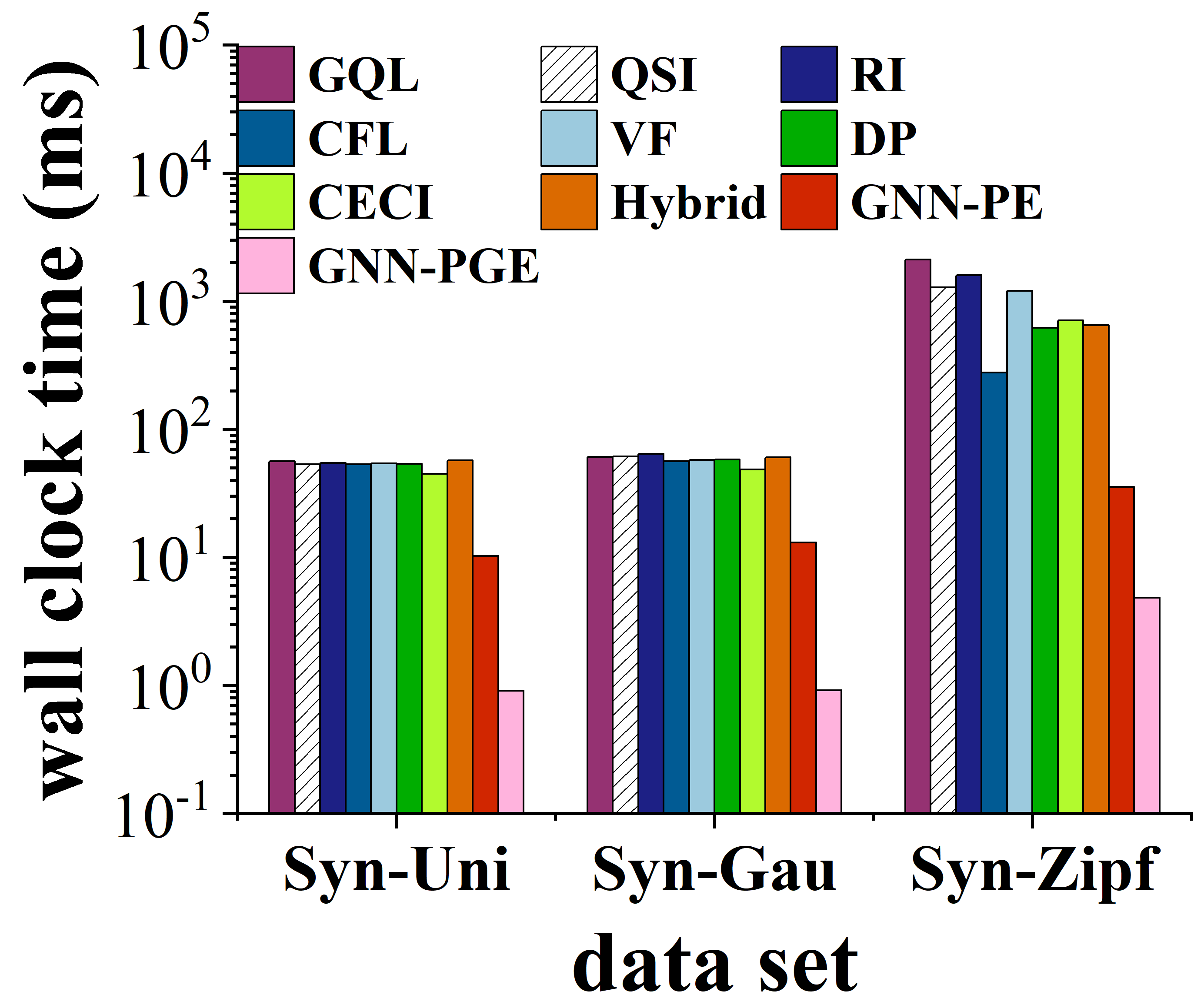}}\label{subfig:syn_data}}
\caption{GNN-PE efficiency on real/synthetic graphs, compared with baseline methods.}
\label{fig:efficiency_datasets}
\end{figure}

\begin{figure*}[htp]
\centering
\subfigure[][{chain}]{                    
\scalebox{0.2}[0.2]{\includegraphics{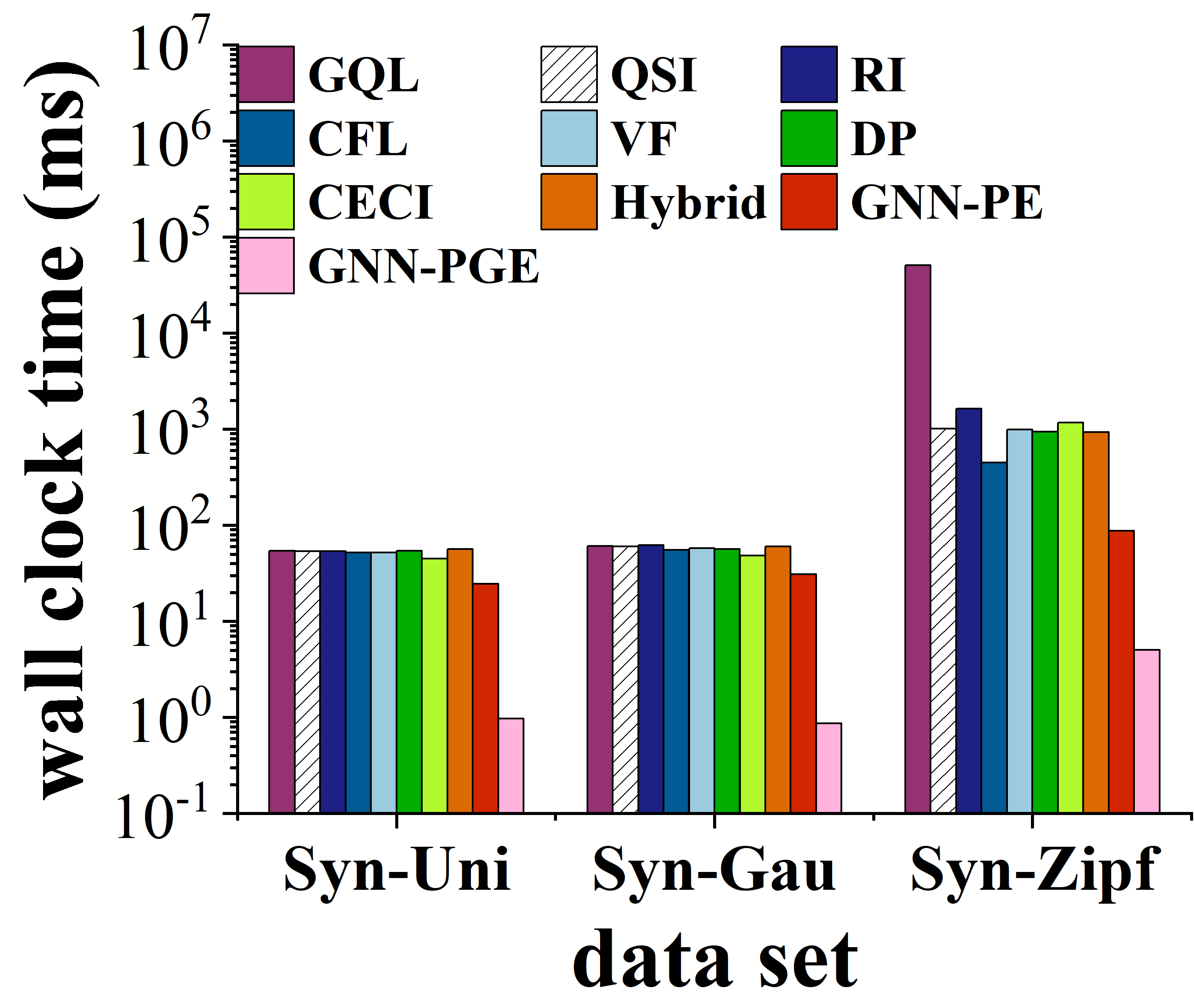}}\label{subfig:chain_syn}}\qquad\qquad
\subfigure[][{star}]{
\scalebox{0.2}[0.2]{\includegraphics{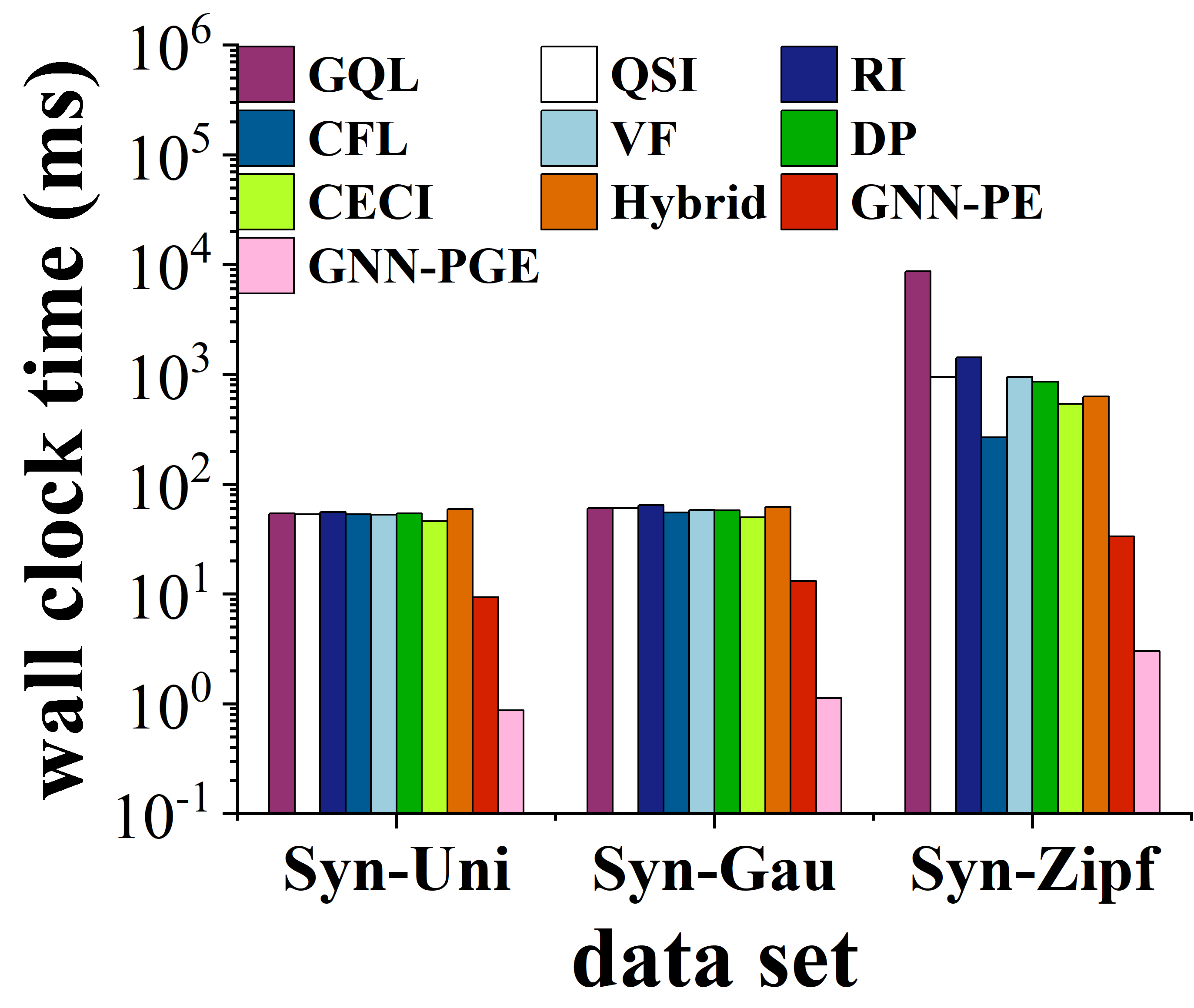}}\label{subfig:star_syn}}\qquad\qquad
\subfigure[][{tree}]{
\scalebox{0.2}[0.2]{\includegraphics{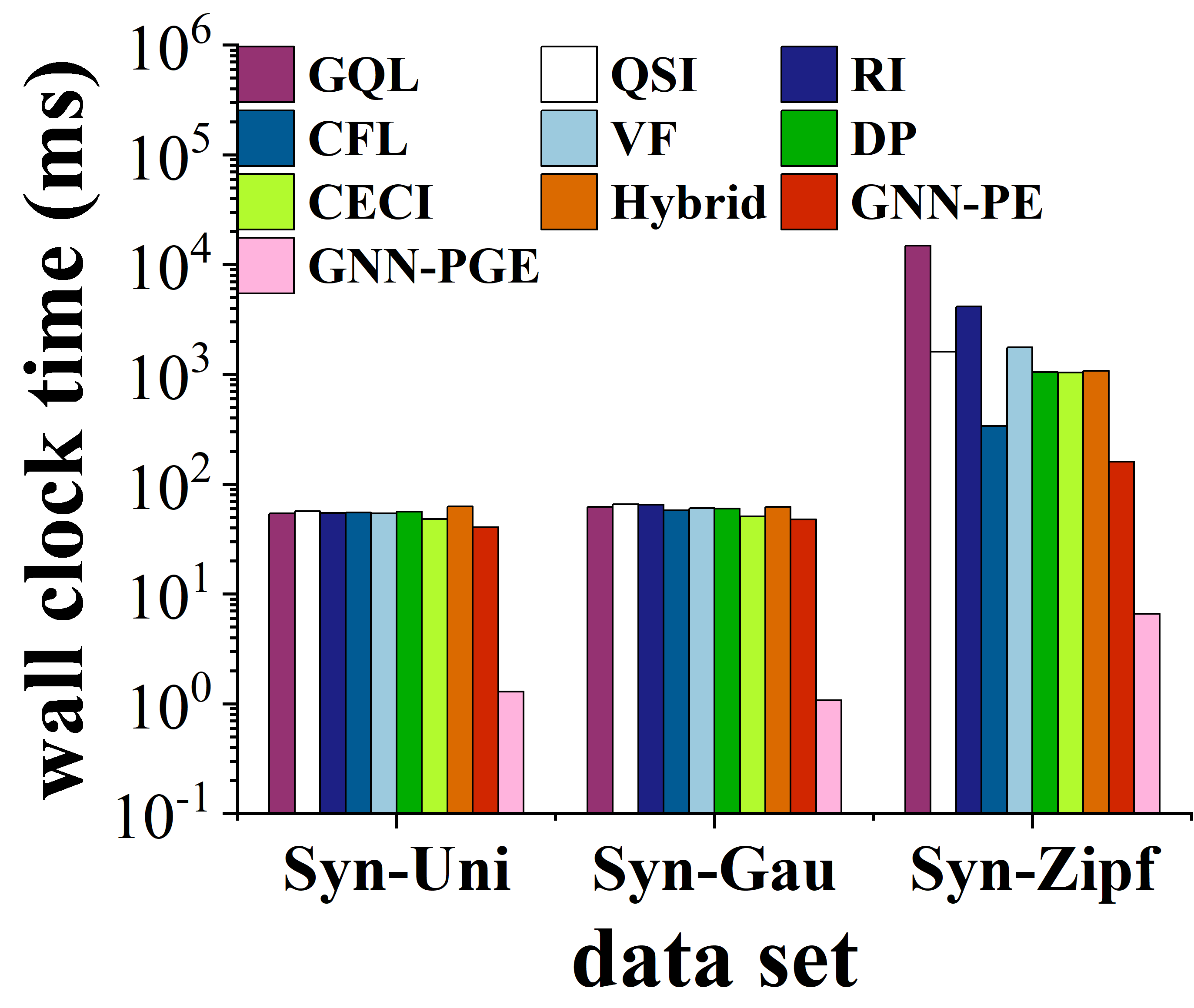}}\label{subfig:tree_syn}}
\caption{The GNN-PE efficiency w.r.t query graph patterns.}
\label{fig:query_patterns}
\end{figure*}

\subsection{Evaluation of the GNN-PE/GNN-PGE Efficiency}
\noindent {\bf The GNN-PE/GNN-PGE Efficiency on Real/Synthetic Graph Data Sets.} 
Figure \ref{fig:efficiency_datasets} compares the efficiency of our GNN-PE and GNN-PGE approaches with 8 baseline methods over both real-world and synthetic graphs, where all parameters are set to default values. From the subfigures, we can see that our GNN-PE and GNN-PGE approaches always outperform baseline methods. Especially, for large-scale real (e.g., $yt$ and $up$) and synthetic graphs ($Syn\text{-}Uni$, $Syn\text{-}Gau$, and $Syn\text{-}Zipf$), our methods can achieve better performance than baselines by 1-2 orders of magnitude. For all real/synthetic graphs (even for $up$ with 3.77M vertices), the time cost of our approaches remains low (i.e., $<$0.56 $sec$ for GNN-PE and $<$0.11 $sec$ for GNN-PGE).

\begin{figure}[t]
\centering
\subfigure[][{\small real-world graphs}]{                    
\scalebox{0.17}[0.17]{\includegraphics{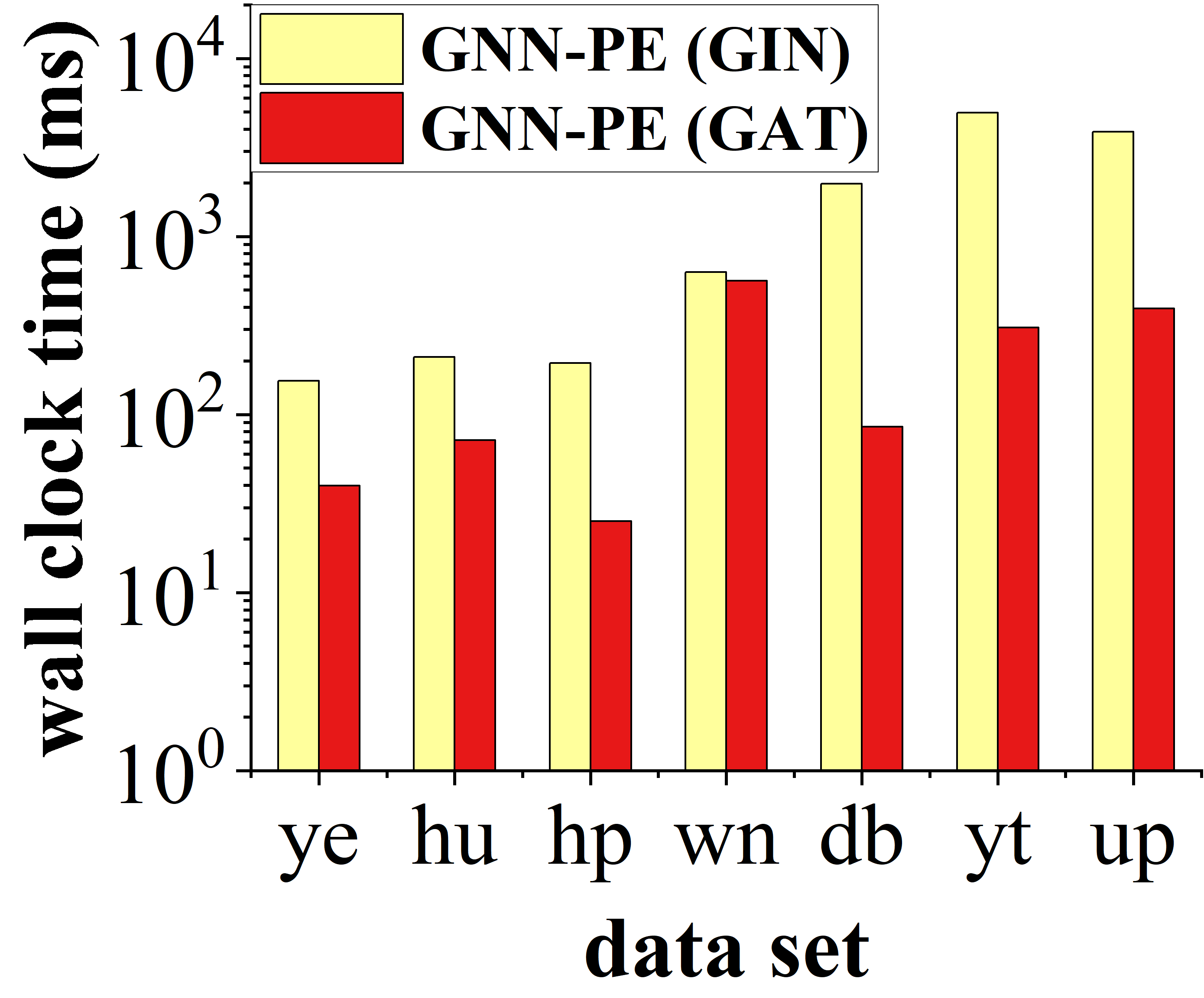}}\label{subfig:gin_real}}
\subfigure[][{\small synthetic graphs}]{
\scalebox{0.17}[0.17]{\includegraphics{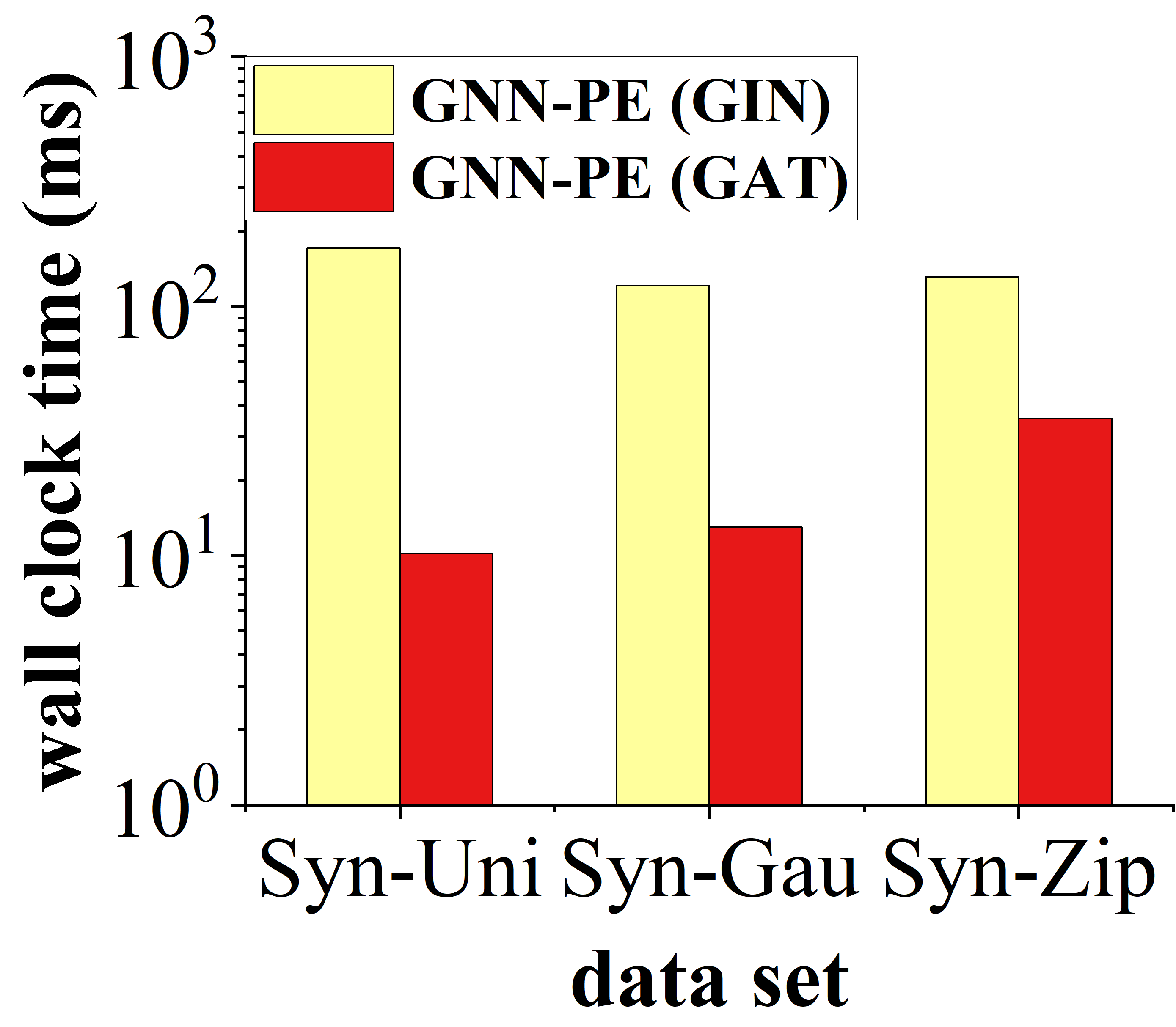}}\label{subfig:gin_syn}}
\caption{GNN-PE efficiency evaluation w.r.t. GNN layer types.}
\label{fig:gin_layer}
\end{figure}

To evaluate our GNN-PE/GNN-PGE query efficiency, in subsequent experiments, we vary the values of different parameters on synthetic graphs (e.g., query graph patterns, types of GNN layers, $|V(q)|$, $avg\_deg(q)$, $|V(G)|/m$, $|\sum|$, $avg\_deg(G)$, $|V(G)|$, average degree of partitions, and \# of edge cuts between partitions). To better illustrate the trends of curves, we omit the results of baseline methods below.

\noindent{\bf The GNN-PE/GNN-PGE Efficiency w.r.t Query Graph Patterns.} Figure~\ref{fig:query_patterns} compares the performance of our GNN-PE and GNN-PGE approaches with that of 8 baselines, for different query graph patterns, including chain, star, and tree, where the length of chain is 6, the degree of star is 6, the depth and fanout of tree are 3, respectively, and default values are used for other parameters.
From the figure, we can see that, our methods can achieve better performance than baselines by 1-2 orders of magnitude for all query graph patterns.
For all synthetic graphs, the time cost of our approaches remains low (i.e., $<$0.16 $sec$ for GNN-PE and $<$0.007 $sec$ for GNN-PGE).

\noindent{\bf The GNN-PE Efficiency  w.r.t. Types of GNN Layers.} Figure~\ref{fig:gin_layer} examines the performance of our GNN-PE approach with different GNN model types, by comparing  GAT \cite{velivckovic2018graph} with GIN \cite{xu2019powerful} used in the first layer of our GNN model, where parameters are set to their default values. 
In the GIN layer, we use a three-layer multilayer perceptron with the ReLU activation function and set the dimension of the hidden features to 32 for a fair comparison.
From the figure, we can see that GNN-PE with GIN has poorer performance than GNN-PE with GAT. 
This is because the sum aggregation operation of the GIN layer tends to generate node embeddings based on vertex degrees, which incurs low pruning power for path embeddings. We also conduct experiments with other types of GNN layers, e.g., GCN \cite{kipf2017semi} and GraphSAGE \cite{hamilton2017inductive}. We do not report similar experimental results here.

\noindent {\bf The GNN-PE/GNN-PGE Efficiency w.r.t. Query Graph Size $\bm{|V(q)|}$.}
Figure \ref{fig:query_graph_size} illustrates the performance of our GNN-PE and GNN-PGE approaches by varying the query graph size, $|V(q)|$, from 5 to 12, where default values are used for other parameters. When the number, $|V(q)|$, of vertices in query graph $q$ increases, more query vertices and paths from $q$ are expected, which results in higher query costs for index traversal and refinement. Thus, larger $|V(q)|$ incurs higher wall clock time. For different query graph sizes $|V(q)|$, our methods can achieve low time costs (i.e., 0.01 $sec$ $\sim$ 0.73 $sec$ for GNN-PE and 0.0001$sec$ $\sim$ 0.35 $sec$ for GNN-PGE).

\begin{figure*}[t]
\centering
\subfigure[][{\small Yeast}]{                    
\scalebox{0.14}[0.15]{\includegraphics{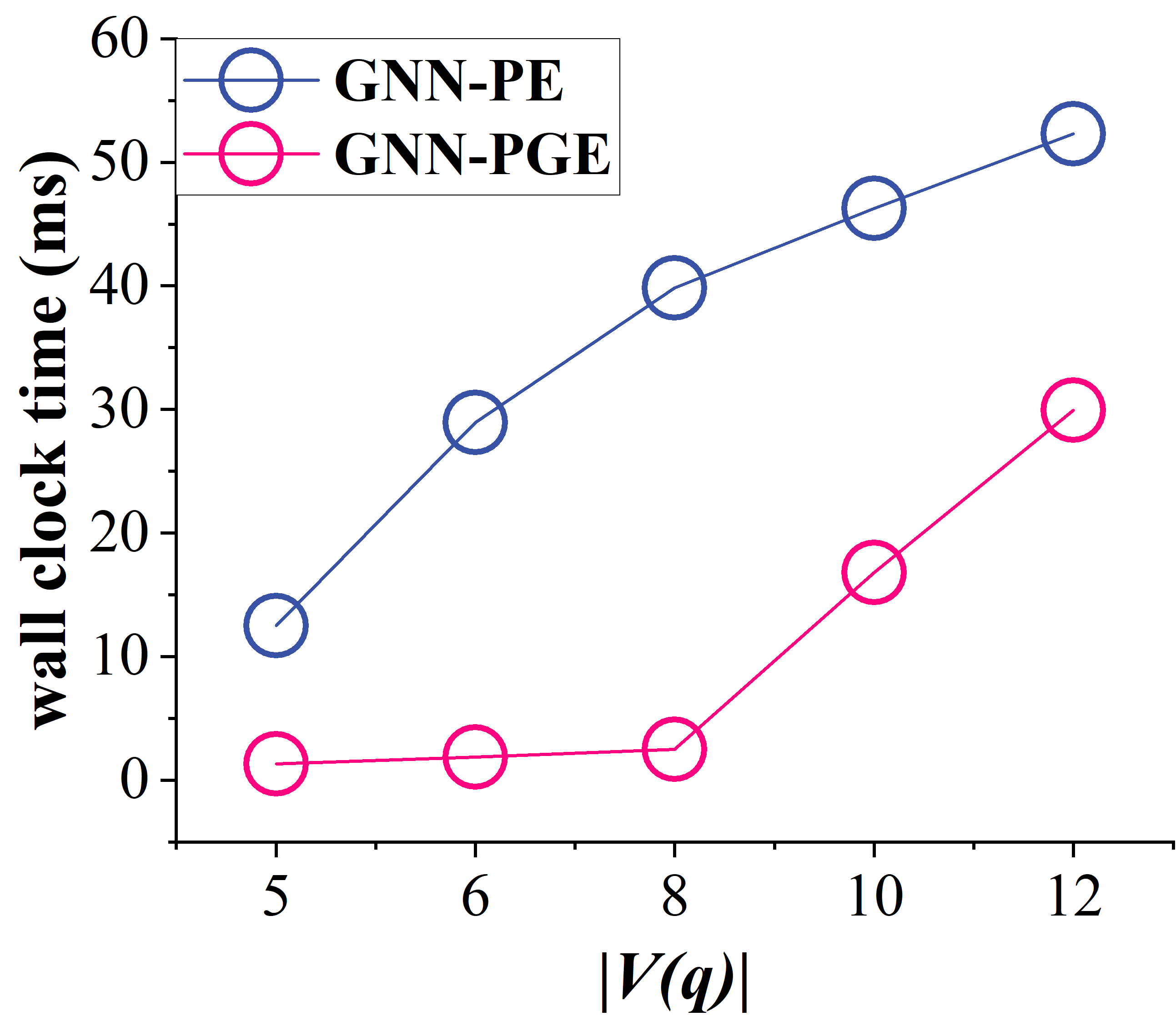}}\label{subfig:qsize_ye}}
\subfigure[][{\small Human}]{                    
\scalebox{0.14}[0.15]{\includegraphics{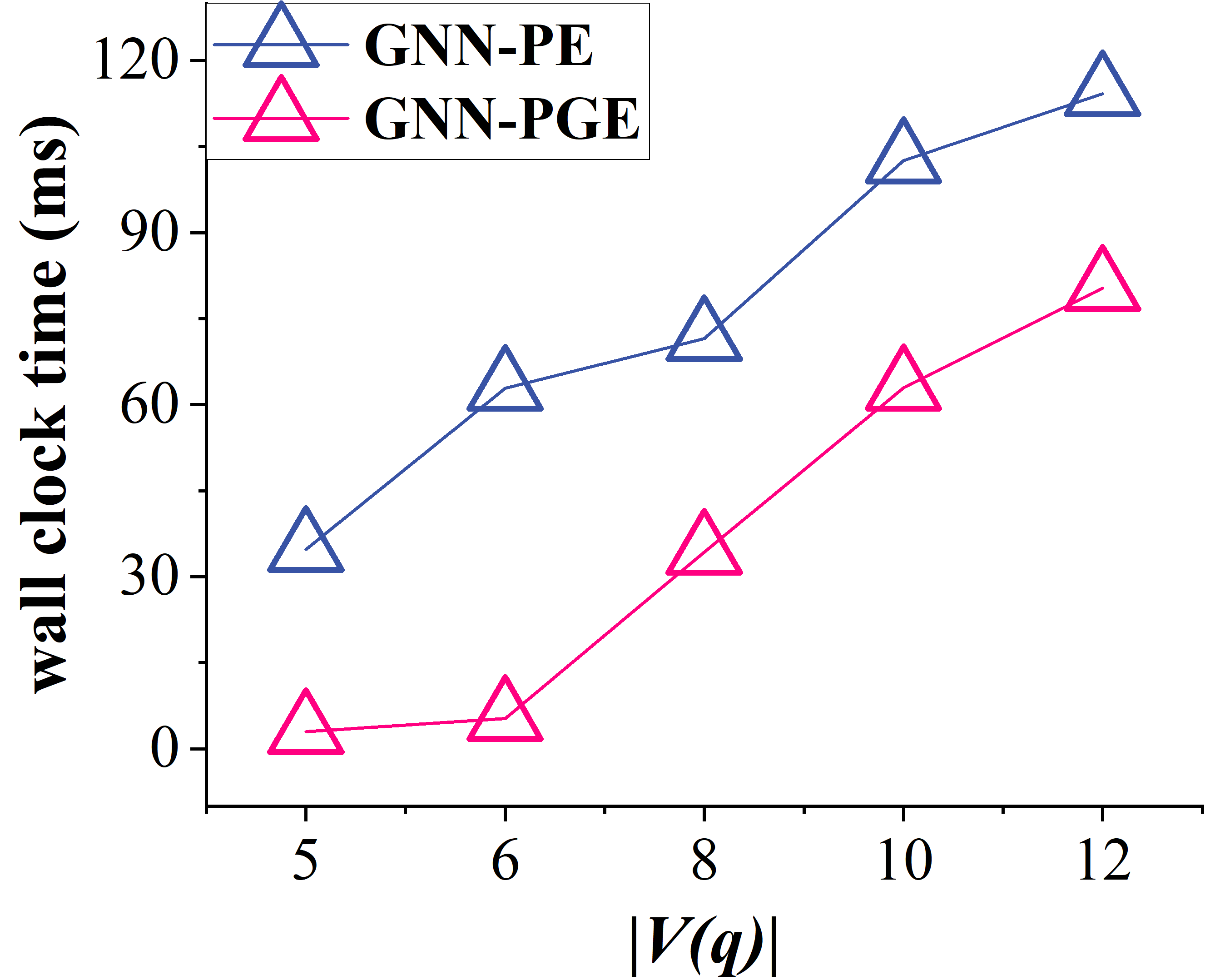}}\label{subfig:qsize_hu}}
\subfigure[][{\small HPRD}]{                    
\scalebox{0.14}[0.15]{\includegraphics{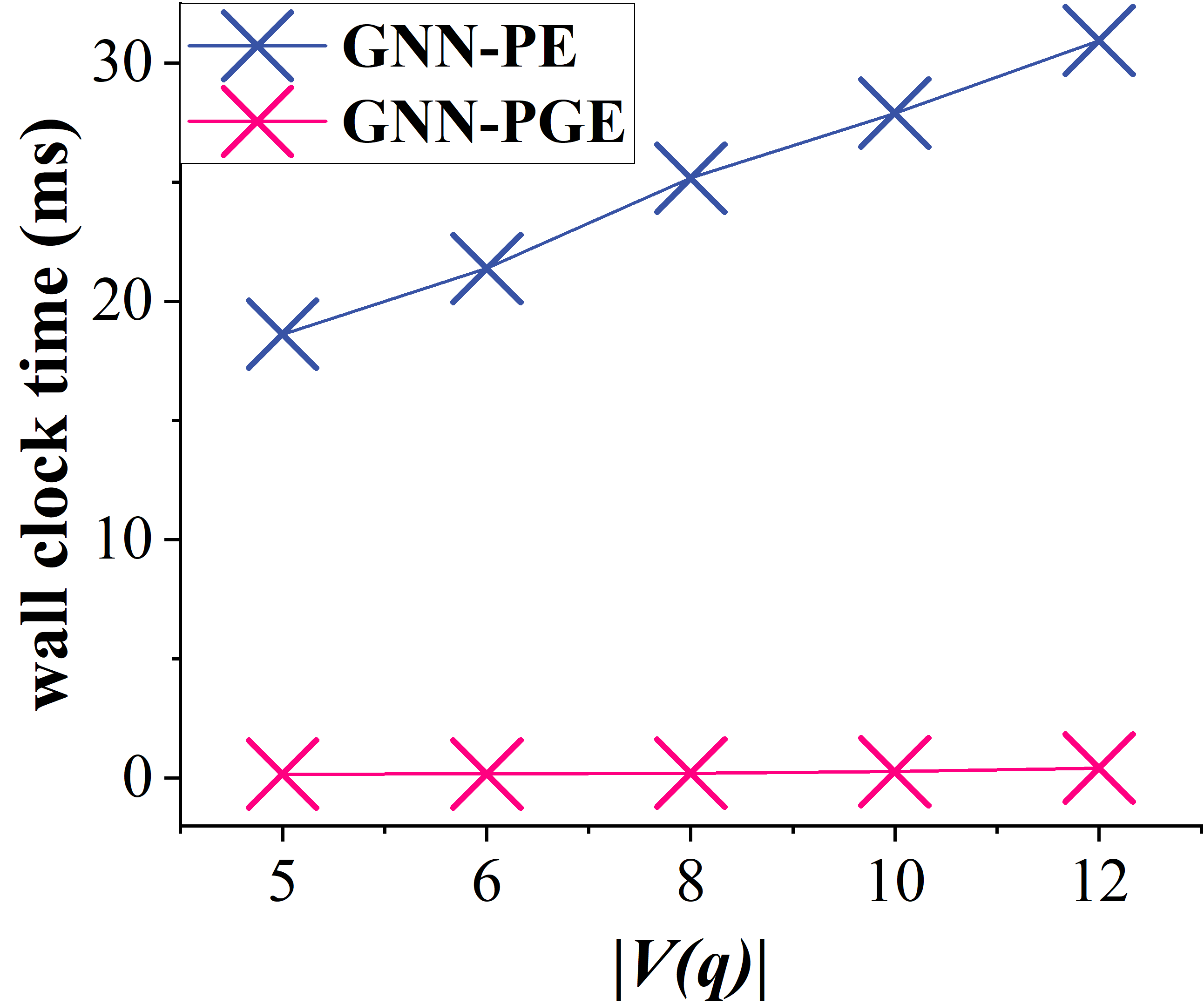}}\label{subfig:qsize_hp}}
\subfigure[][{\small WordNet}]{                    
\scalebox{0.14}[0.15]{\includegraphics{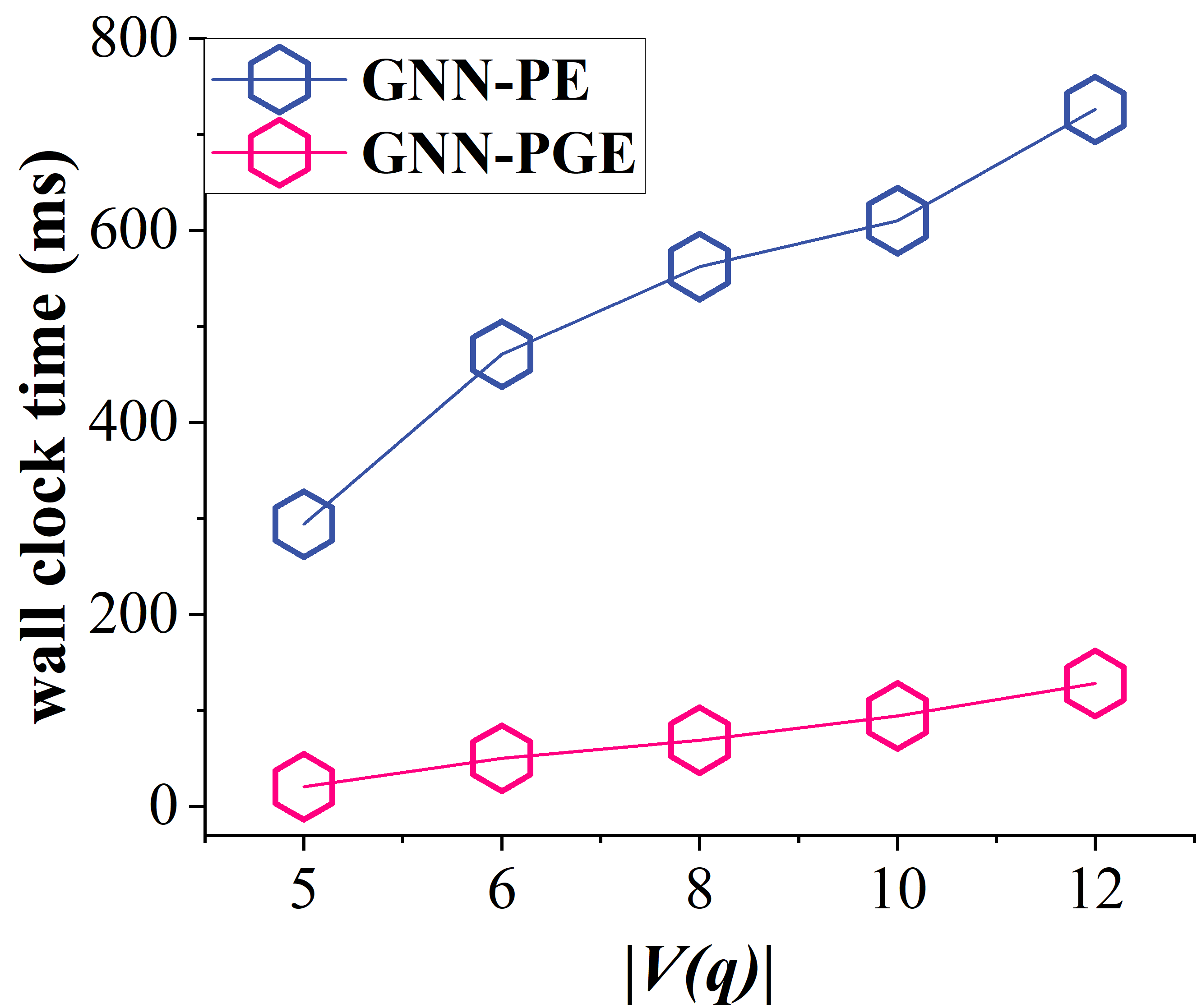}}\label{subfig:qsize_wn}}
\subfigure[][{\small DBLP}]{                    
\scalebox{0.14}[0.15]{\includegraphics{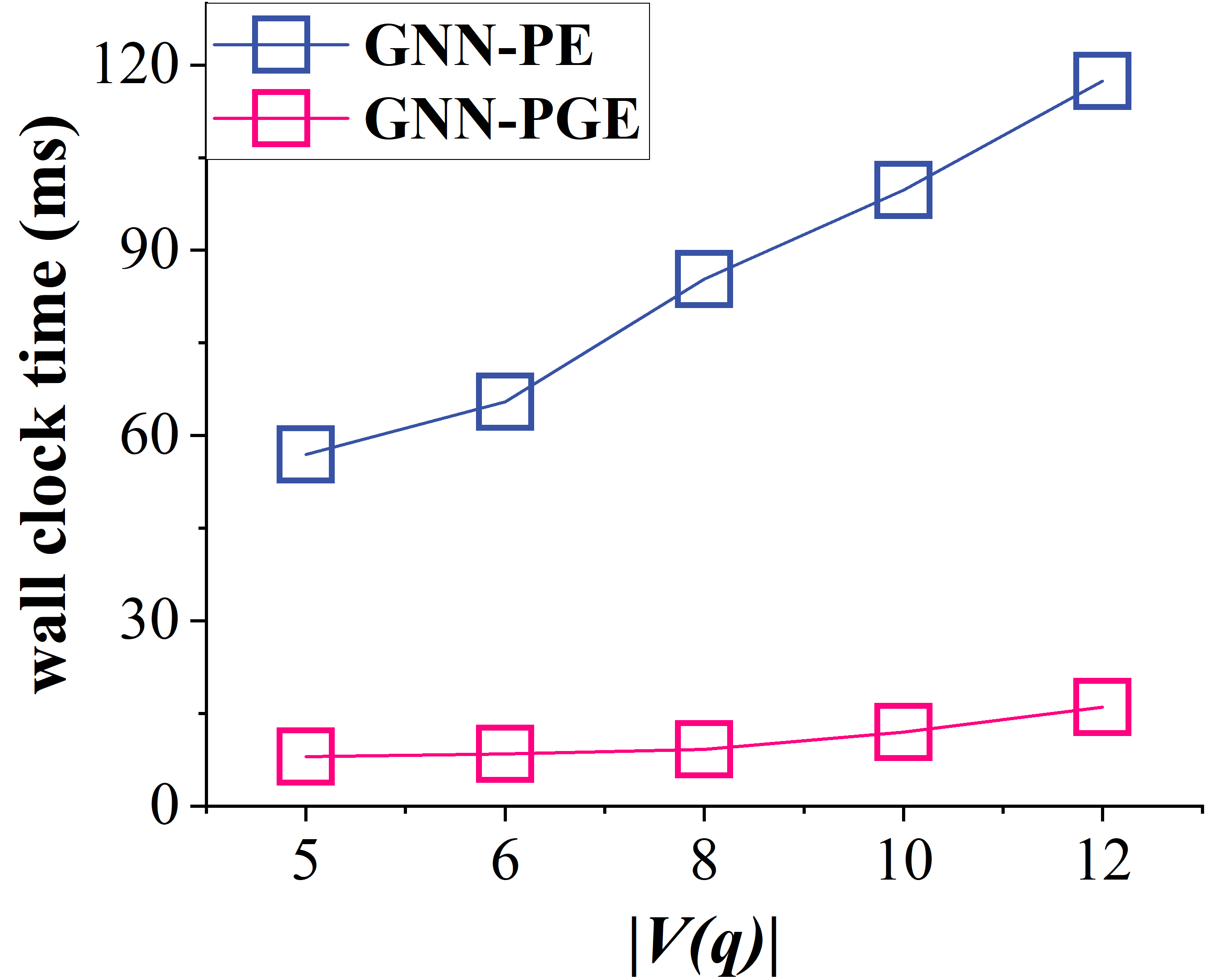}}\label{subfig:qsize_db}}
\\
\subfigure[][{\small YouTube}]{
\scalebox{0.14}[0.15]{\includegraphics{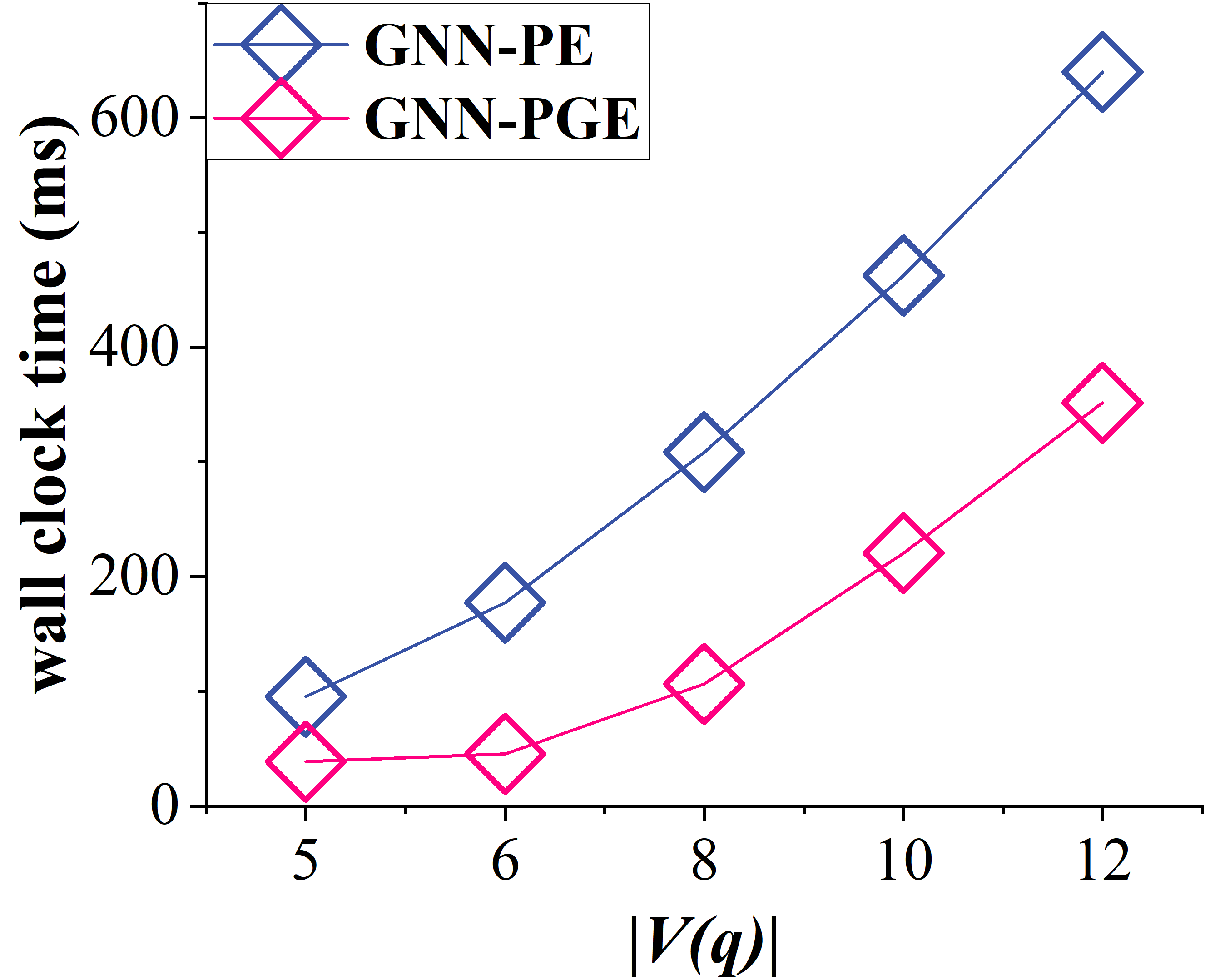}}\label{subfig:qsize_yt}}
\subfigure[][{\small USPatents}]{                    
\scalebox{0.14}[0.15]{\includegraphics{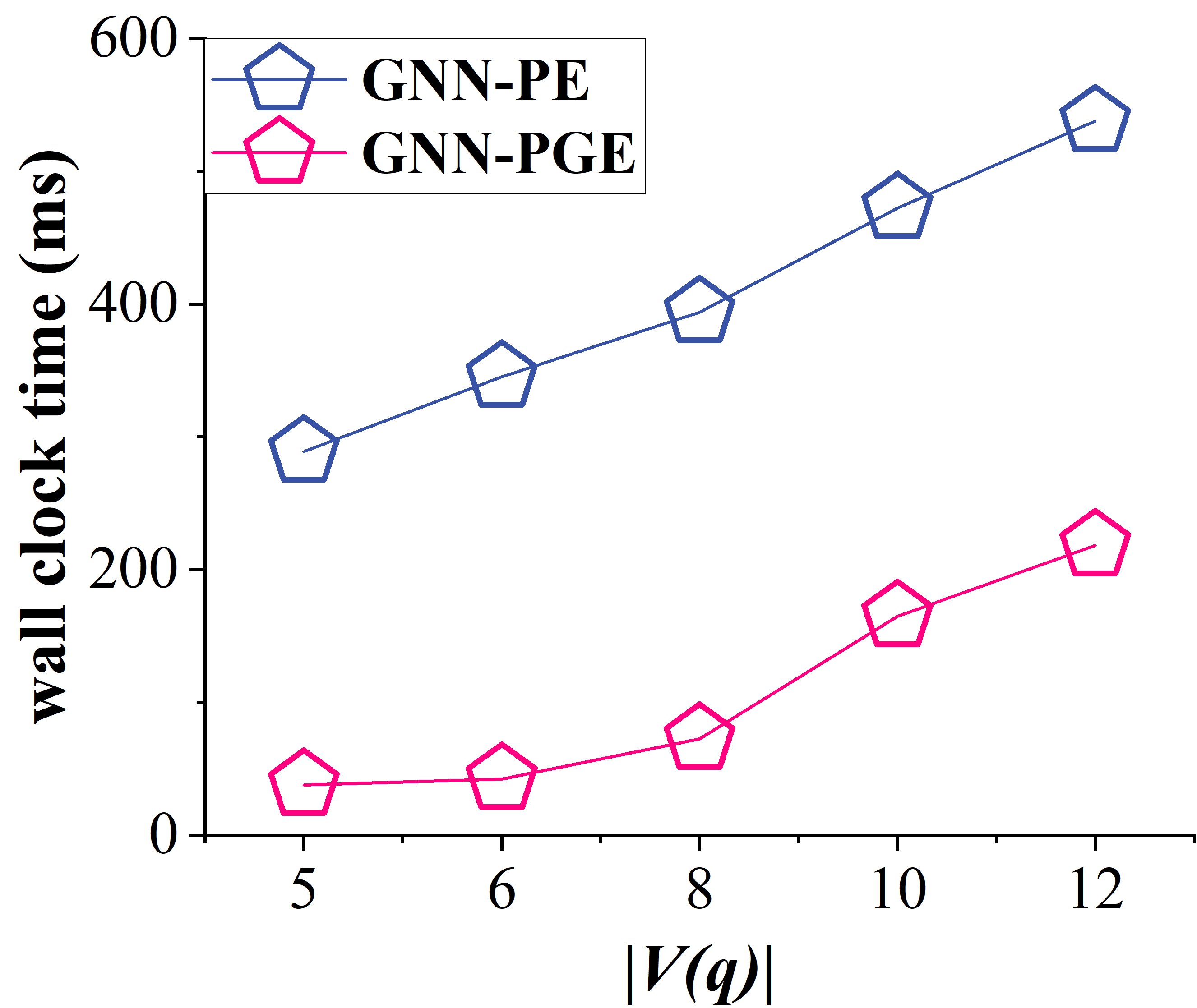}}\label{subfig:qsize_up}}
\subfigure[][{\small Syn-Uni}]{                    
\scalebox{0.14}[0.15]{\includegraphics{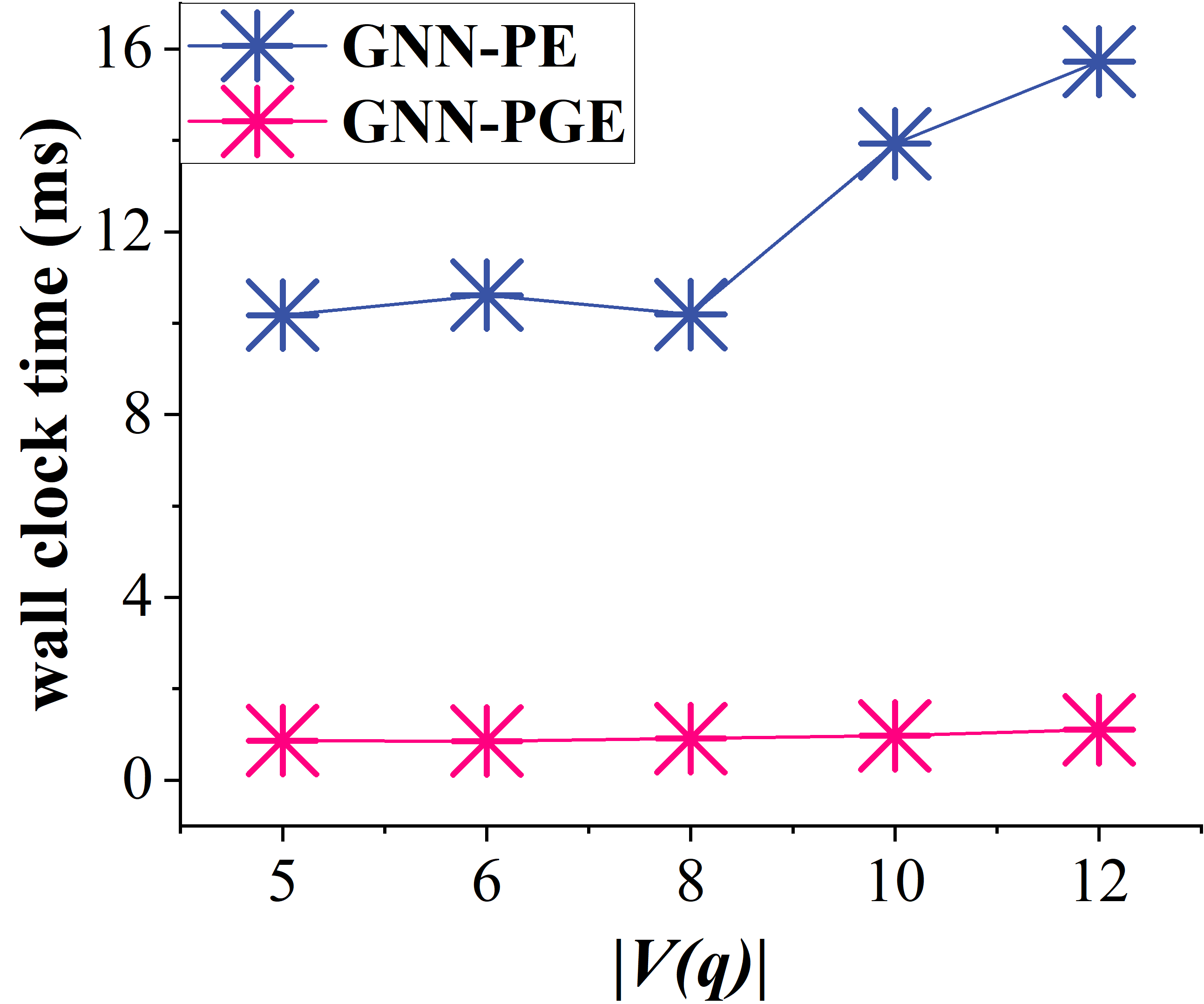}}\label{subfig:qsize_uni}}
\subfigure[][{\small Syn-Gau}]{                    
\scalebox{0.14}[0.15]{\includegraphics{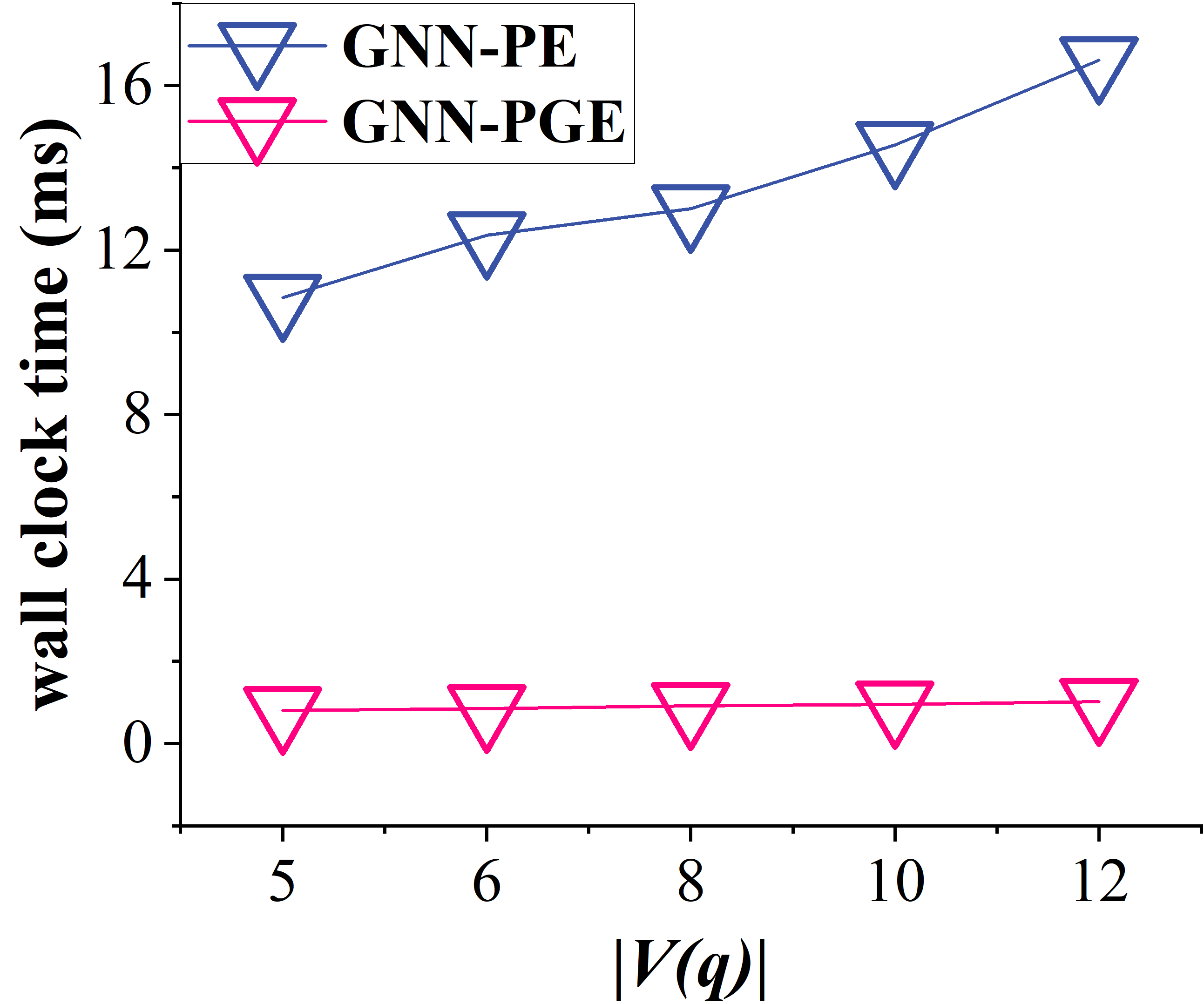}}\label{subfig:qsize_gau}}
\subfigure[][{\small Syn-Zipf}]{                    
\scalebox{0.14}[0.15]{\includegraphics{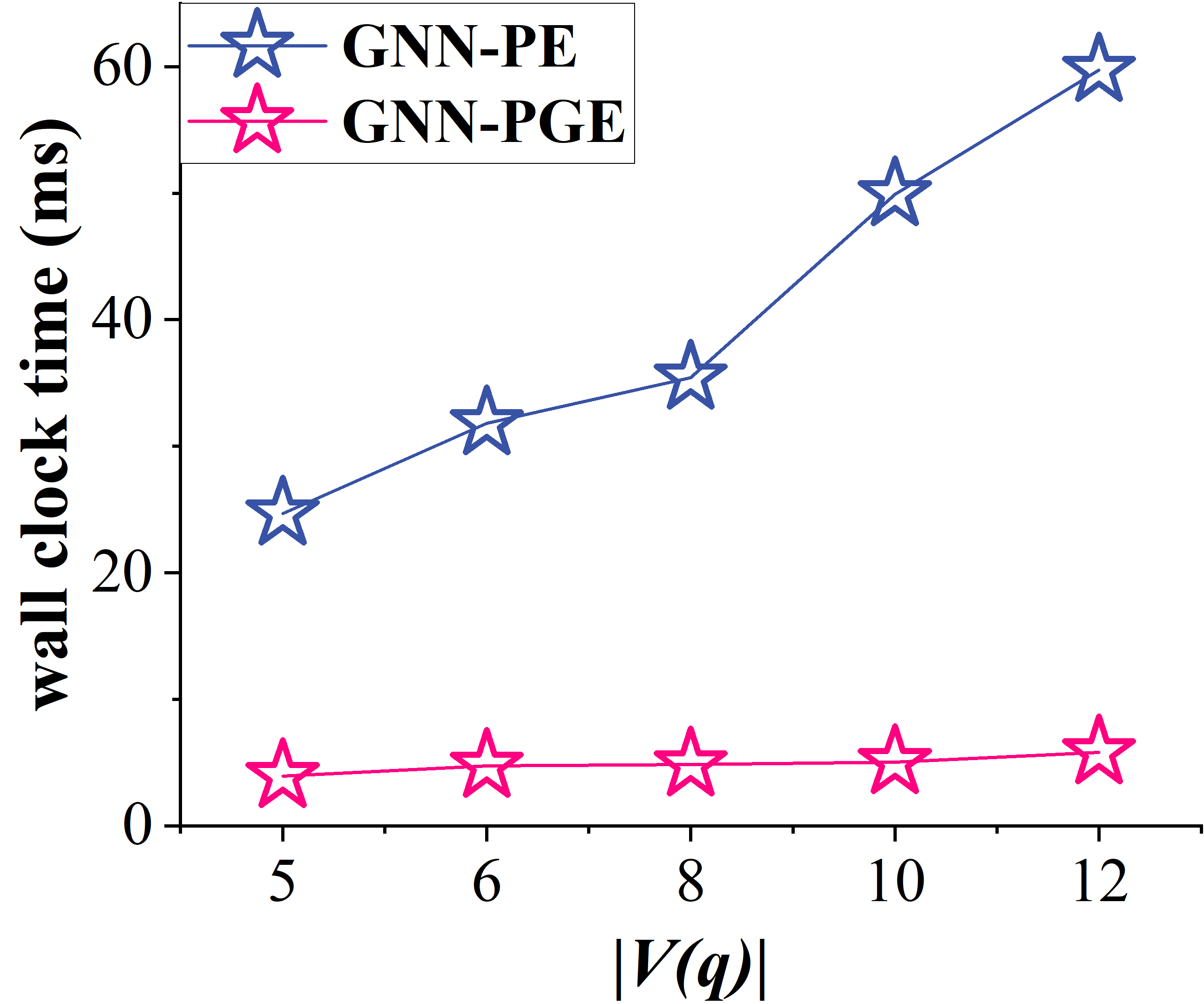}}\label{subfig:qsize_zip}}
\caption{Efficiency evaluation w.r.t. query graph size $\bm{|V(q)|}$.}
\label{fig:query_graph_size}
\end{figure*}

\begin{figure}[t]
\centering
\subfigure[][{\small real-world graphs}]{                    
\scalebox{0.19}[0.19]{\includegraphics{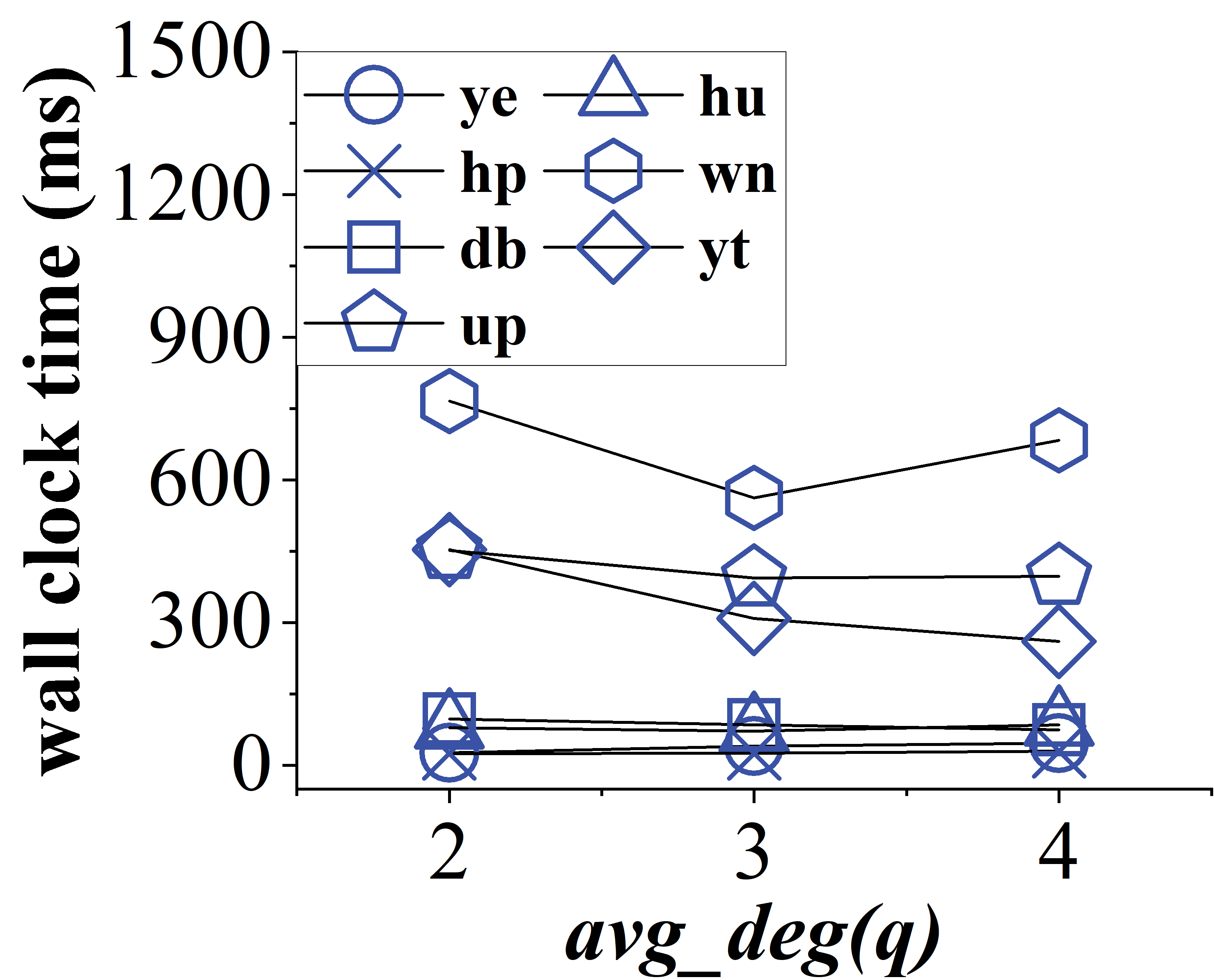}}\label{subfig:real_querydegree}}%
\subfigure[][{\small synthetic graphs}]{
\scalebox{0.19}[0.19]{\includegraphics{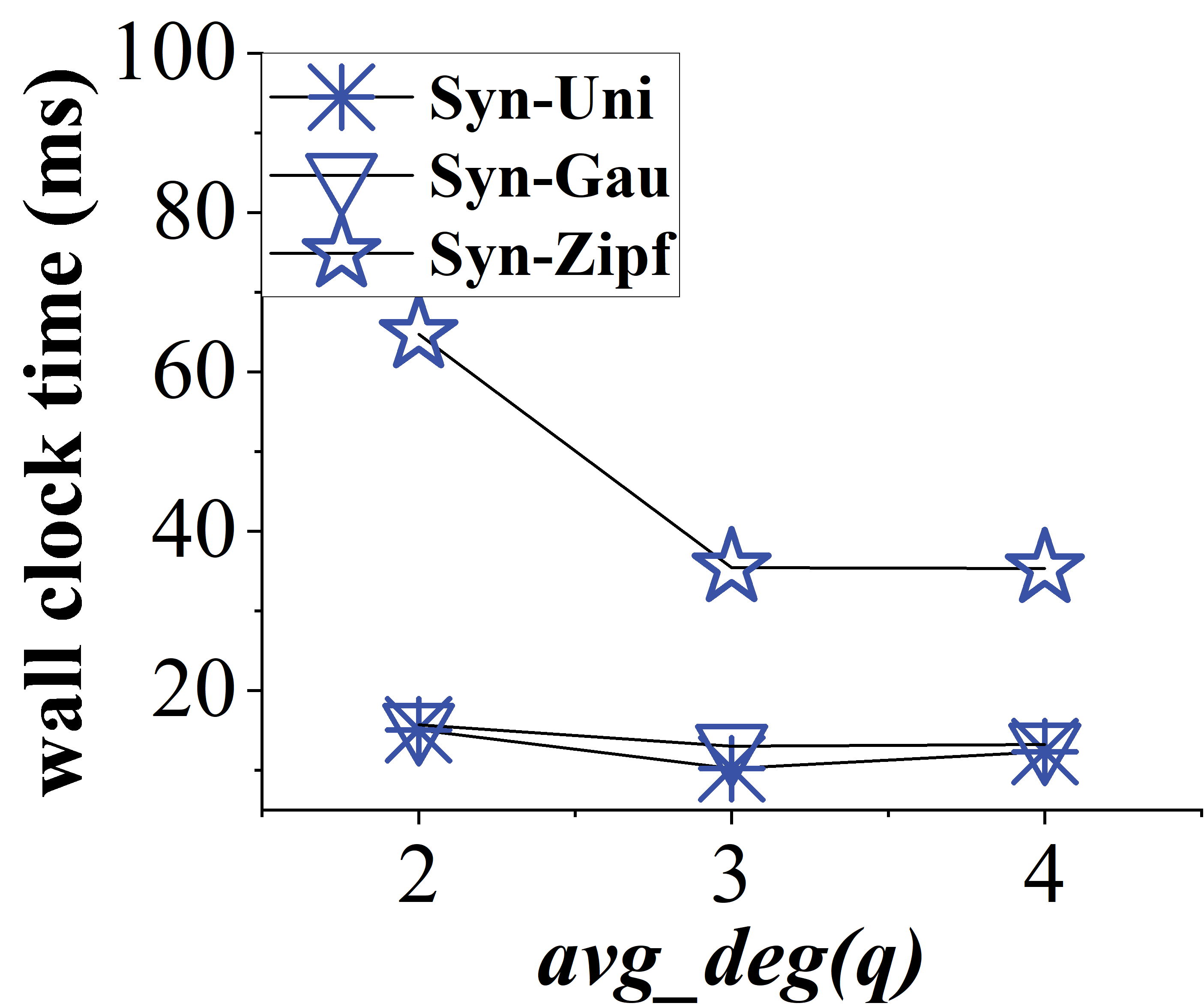}}\label{subfig:syn_querydegree}}
\caption{GNN-PE efficiency evaluation w.r.t. different average degrees, $\bm{avg\_deg(q)}$, of the query graph $\bm{q}$.}
\label{fig:querydegree}
\end{figure}

\begin{figure}
\centering
\subfigure[][{\small partition size, $|V(G)|/m$}]{
\scalebox{0.17}[0.17]{\includegraphics{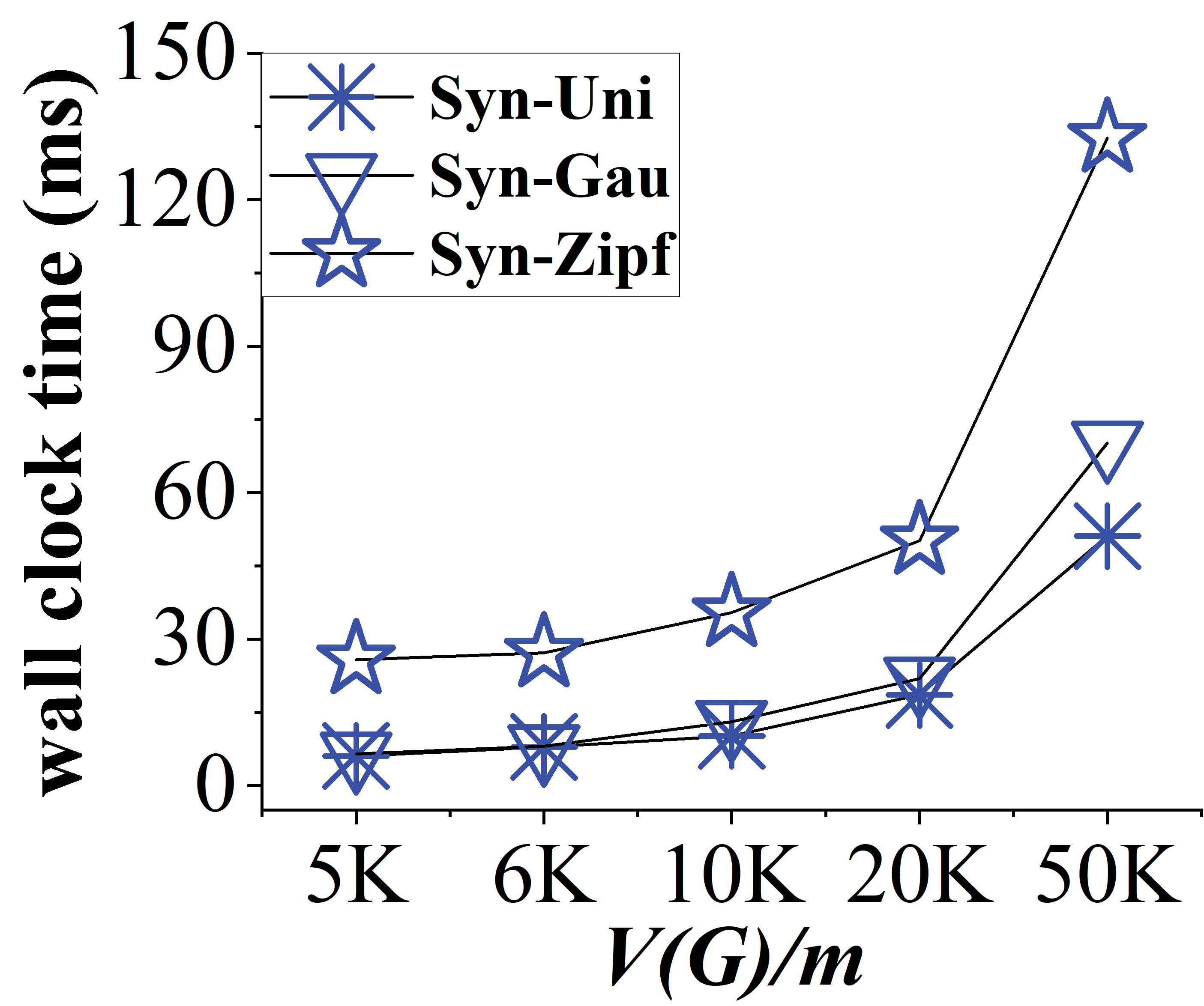}}\label{subfig:syn_par_size}}%
\subfigure[][{\small \# of distinct vertex labels, $|\sum|$}]{                    
\scalebox{0.17}[0.17]{\includegraphics{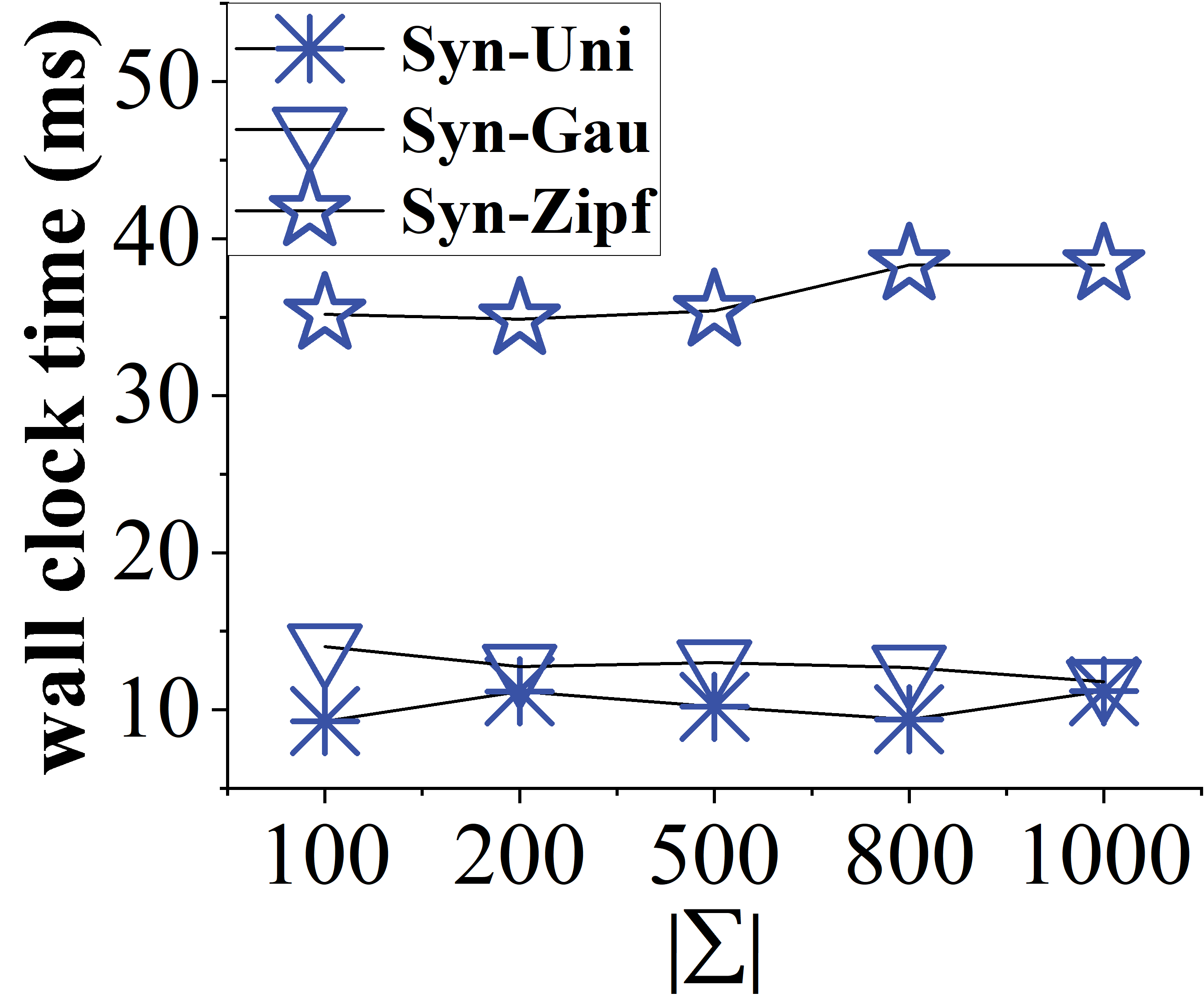}}\label{subfig:labelnumber}}\\
\subfigure[][{\small average degree of partitions}]{     
\scalebox{0.17}[0.17]{\includegraphics{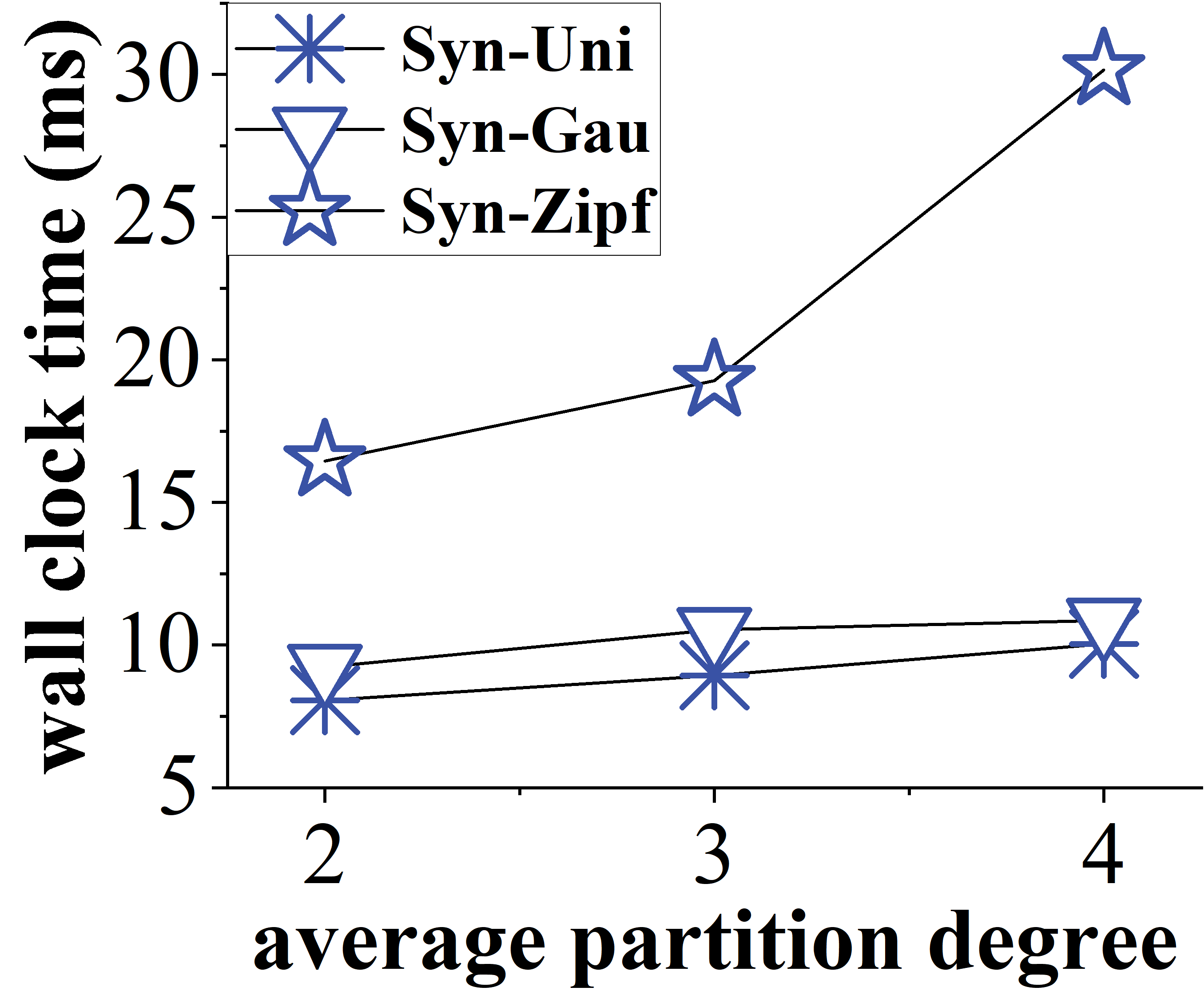}}\label{subfig:part_degree}}
\subfigure[][{\small $\#$ of edge cuts between partitions}]{                    
\scalebox{0.17}[0.17]{\includegraphics{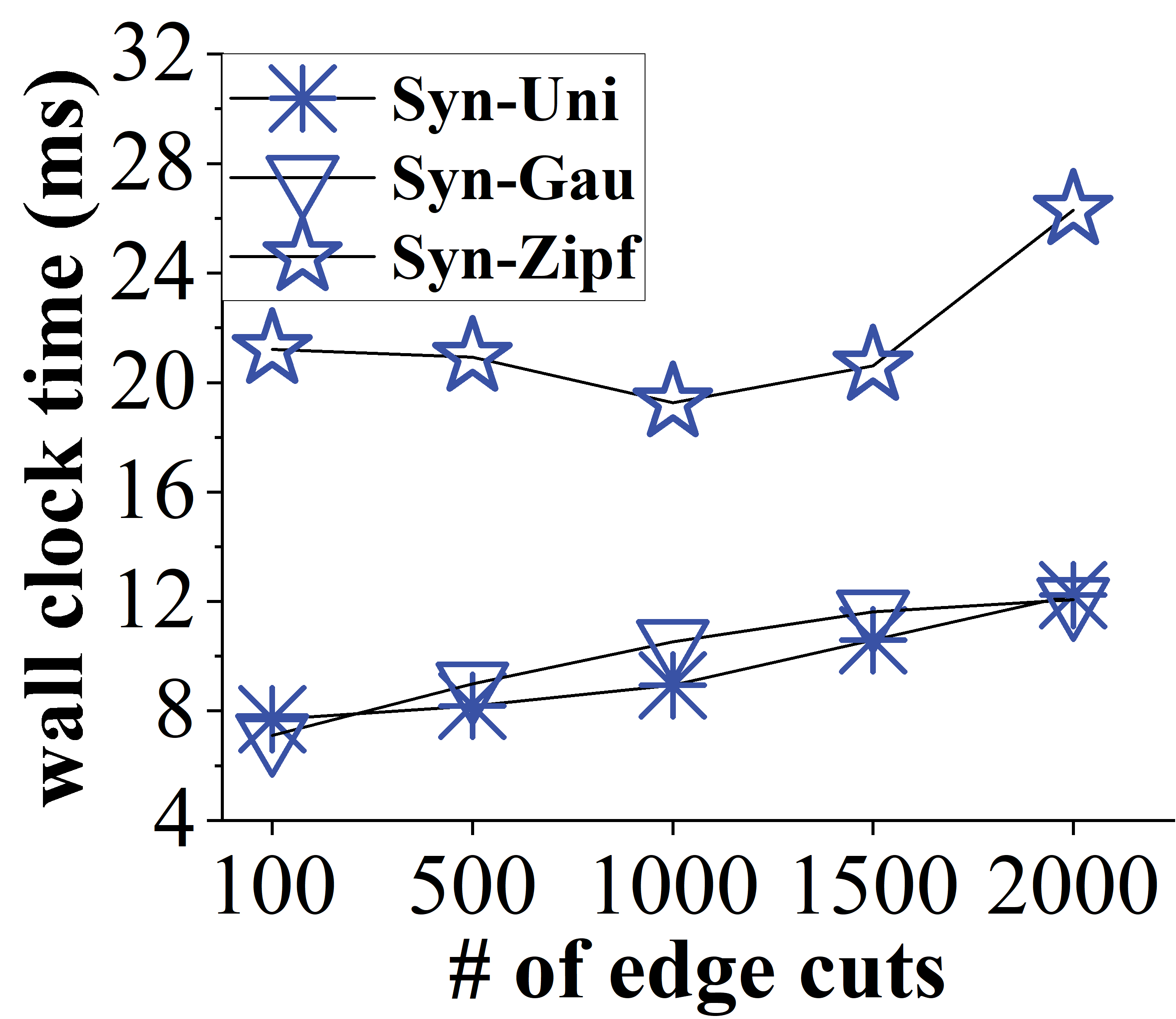}}\label{subfig:part_cuts}}
\caption{GNN-PE efficiency evaluation over synthetic graphs.}
\label{fig:efficiency}
\end{figure}

\noindent {\bf The GNN-PE Efficiency Evaluation w.r.t. Average Degree, $\bm{avg\_deg(q)}$, of the Query Graph $\bm{q}$.}
Figure \ref{fig:querydegree} examines the GNN-PE performance by varying the average degree, $avg\_deg(q)$, of the query graph $q$ from $2$ to $4$, where other parameters are set to default values. Higher degree $avg\_deg(q)$ of $q$ may produce more query paths, but meanwhile incur higher pruning power of query paths. Therefore, for most real/synthetic graphs, when $avg\_deg(q)$ increases, the wall clock time decreases. For real graph $wn$, due to more query paths from $q$, the time cost increases for $avg\_deg(q)=4$. With different $avg\_deg(q)$ values, the GNN-PE query costs remain low (i.e., 0.01 $sec$ $\sim$ 0.77 $sec$).

\noindent {\bf The GNN-PE Efficiency Evaluation w.r.t. Subgraph Partition Size, $\bm{|V(G)|/m}$.}
Figure \ref{subfig:syn_par_size} reports the GNN-PE efficiency for different subgraph partition sizes $|V(G)|/m$ from 5K to 50K, where we use default values for other parameters. With the increase of $|V(G)|/m$, the index traversal on partitions of a larger size incurs higher time. Nonetheless, for different $|V(G)|/m$ values, the time costs remain low (i.e., 0.006 $sec$ $\sim$ 0.133 $sec$) on all synthetic graphs.

\noindent {\bf The GNN-PE Efficiency Evaluation w.r.t. Number, $\bm{|\sum|}$, of Distinct Vertex Labels.}
Figure~\ref{subfig:labelnumber} shows the wall clock time of our GNN-PE approach, where $|\sum|$ varies from $100$ to $1,000$ and other parameters are set to their default values. From the figure, we can see that the GNN-PE performance is not very sensitive to $|\sum|$. Nonetheless, the query cost remains low (i.e., 0.009 $sec$ $\sim$ 0.038 $sec$) with different $|\sum|$ values, which indicates the efficiency of our GNN-PE approach.

\begin{figure}[t]
\centering
\subfigure[][{\small Syn-Uni}]{                    
\scalebox{0.12}[0.12]{\includegraphics{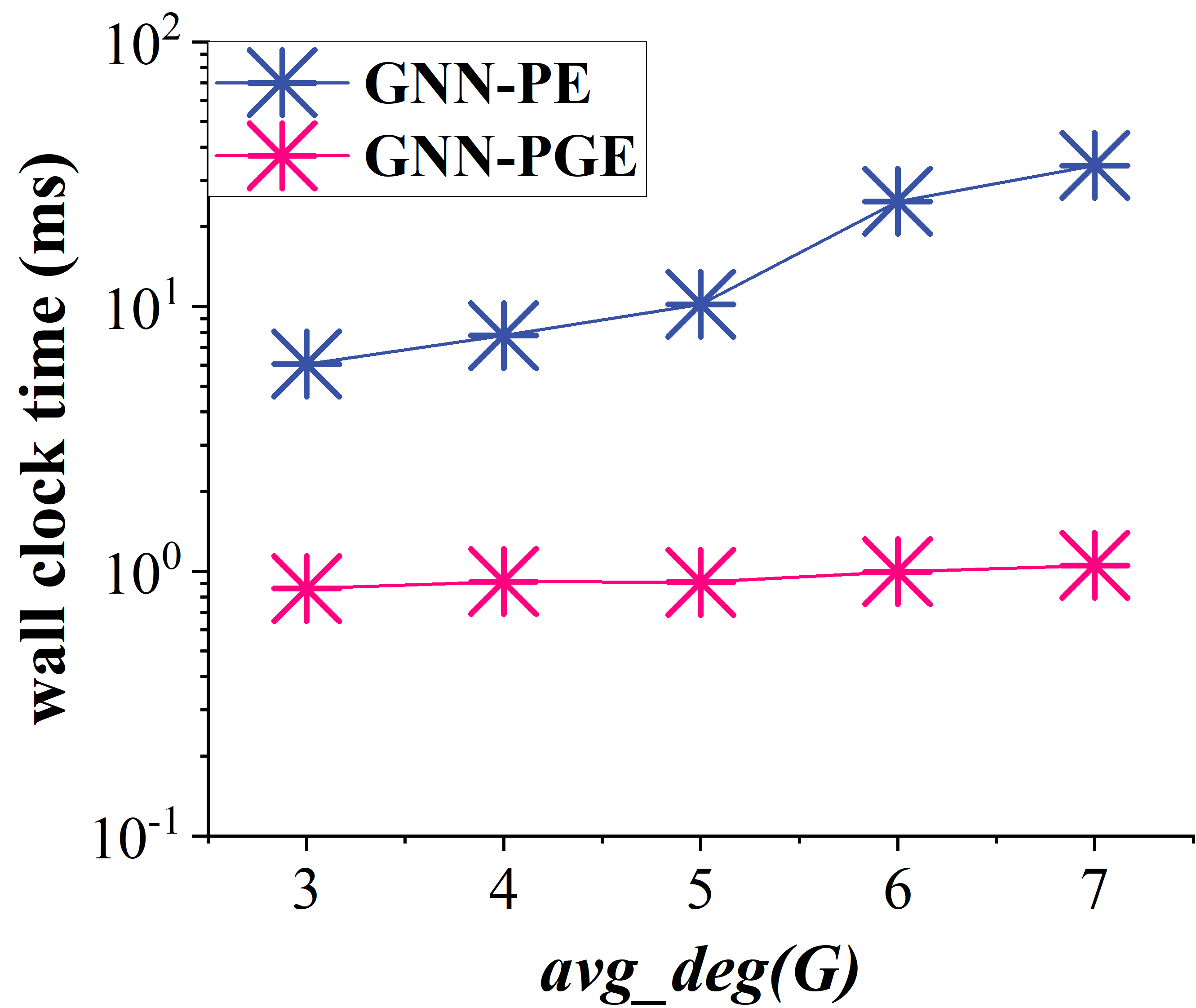}}\label{subfig:l}}
\subfigure[][{\small Syn-Gau}]{
\scalebox{0.12}[0.12]{\includegraphics{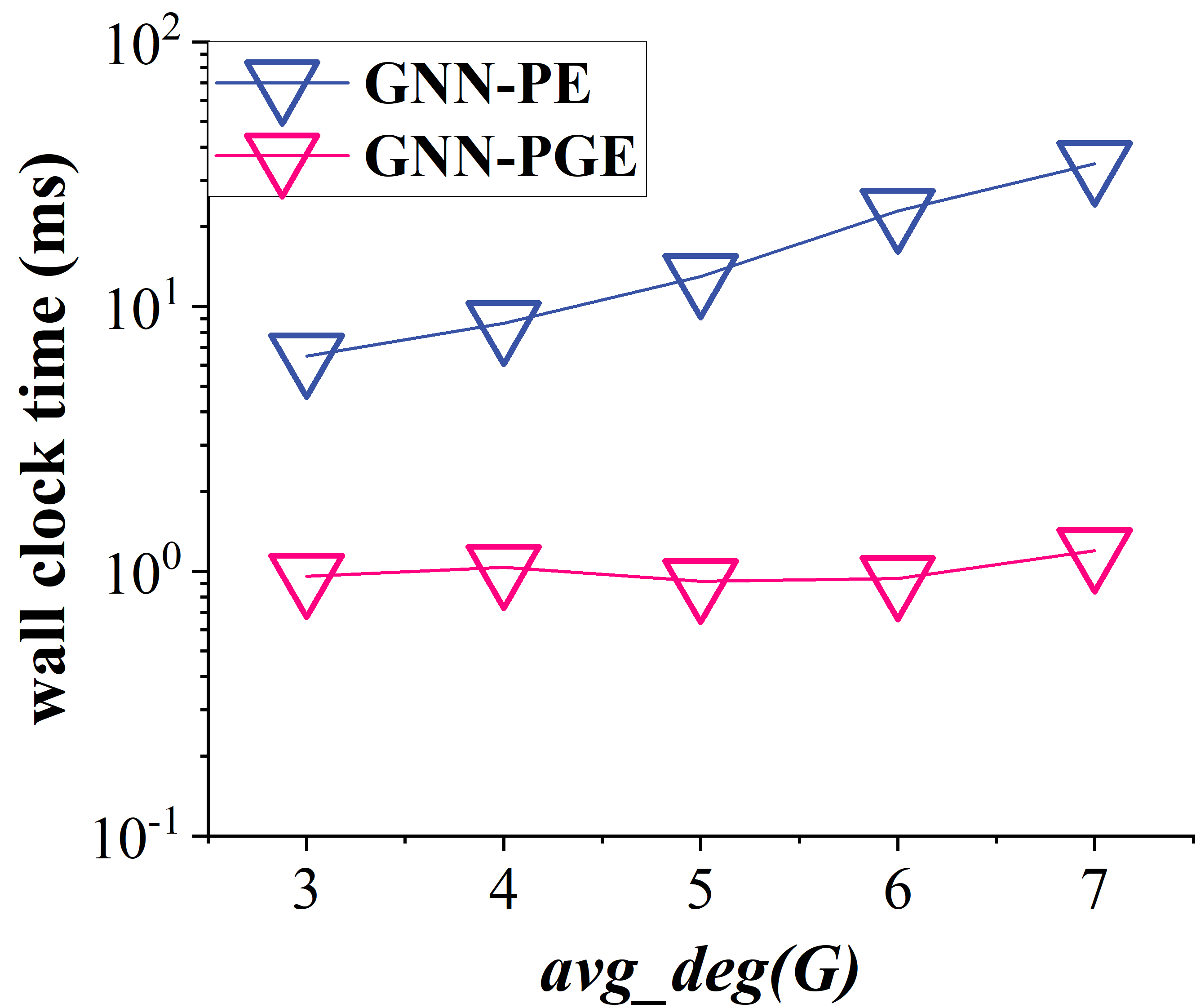}}\label{subfig:b}}
\subfigure[][{\small Syn-Zipf}]{
\scalebox{0.12}[0.12]{\includegraphics{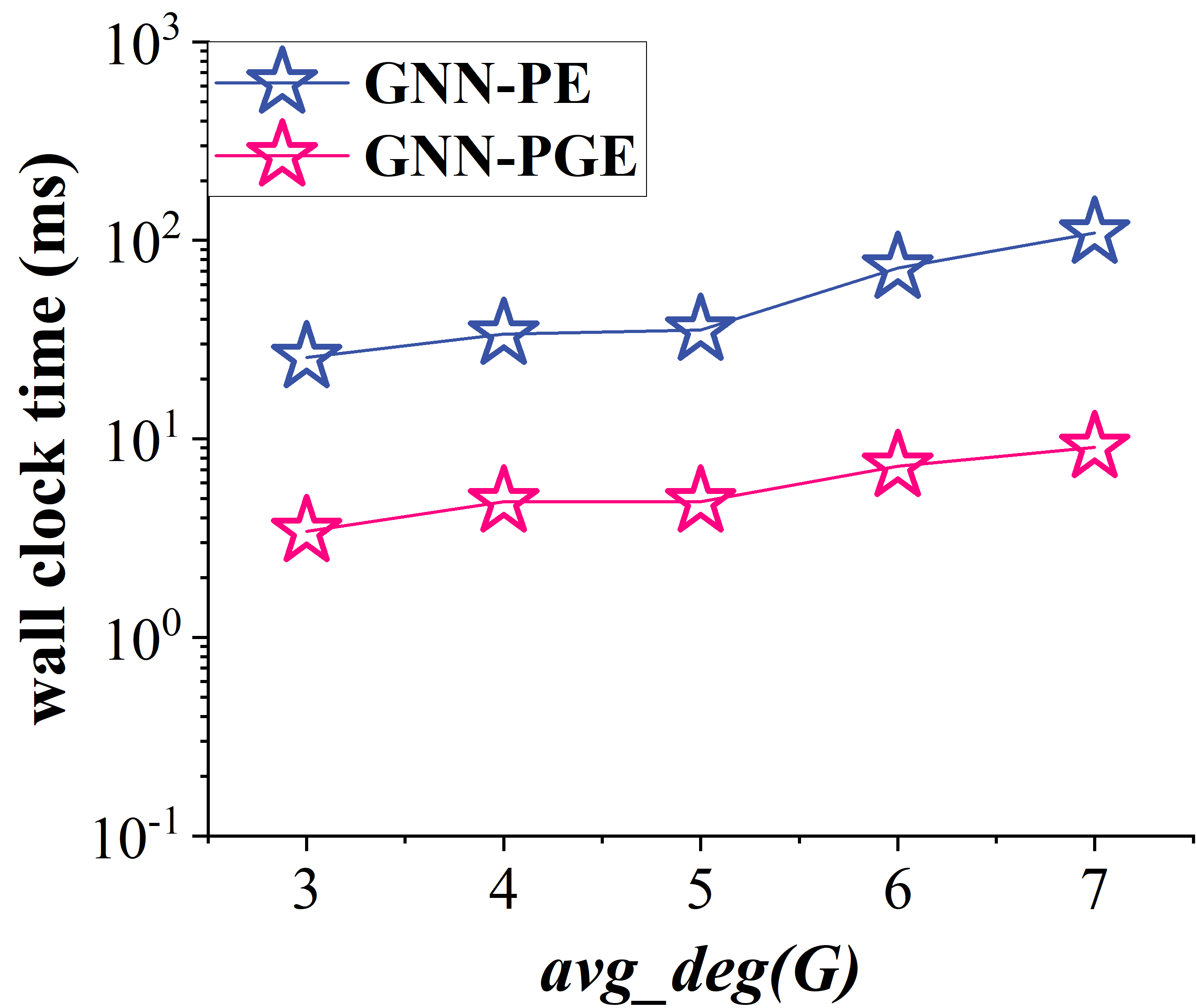}}\label{subfig:strategies}}
\caption{Efficiency evaluation w.r.t. average degree, $\bm{avg\_deg(G)}$, of the data graph $\bm{G}$}
\label{fig:data_avg_deg}
\end{figure}

\begin{figure}[t]
\centering
\subfigure[][{\small Syn-Uni}]{                    
\scalebox{0.12}[0.12]{\includegraphics{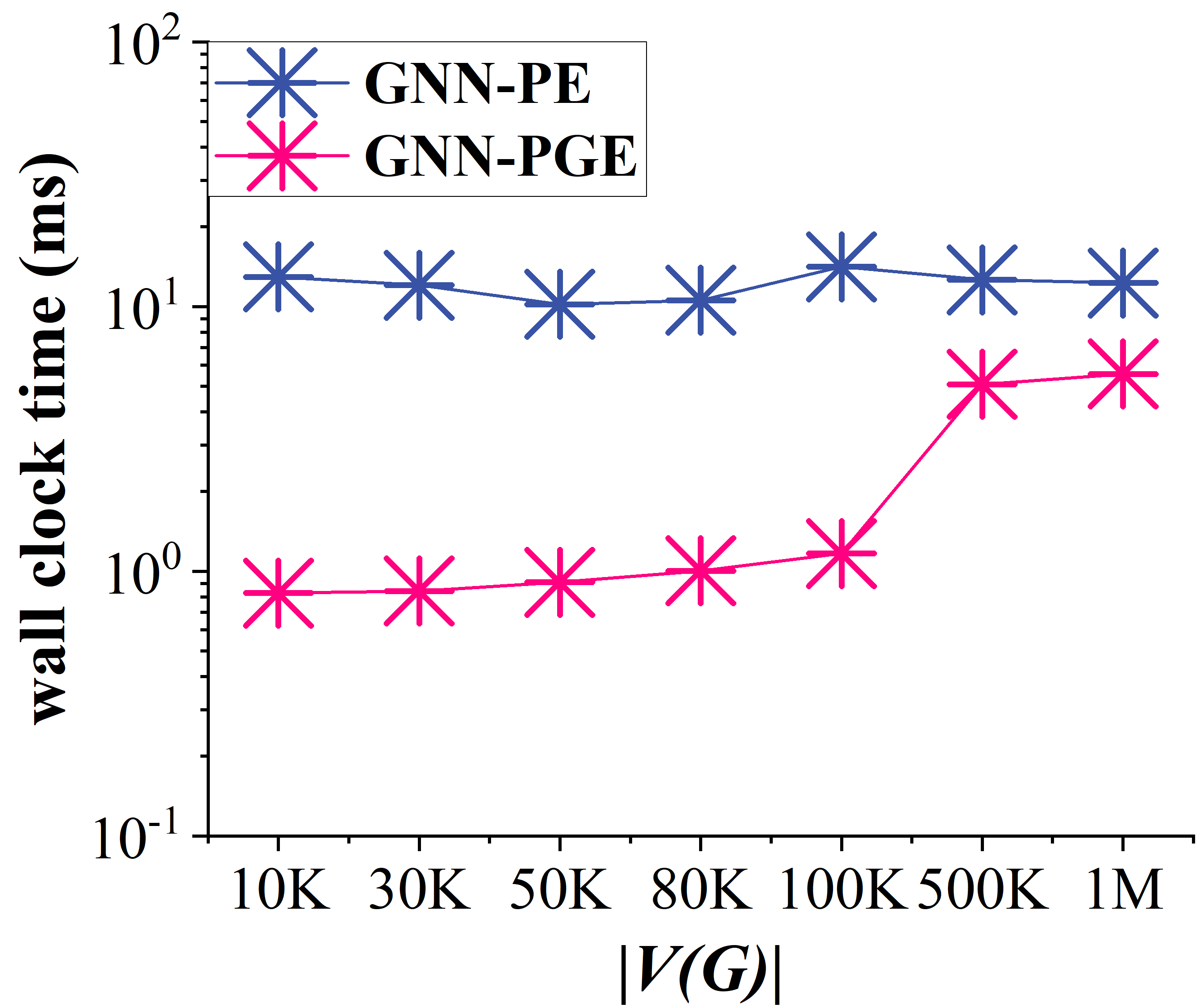}}\label{subfig:l}}
\subfigure[][{\small Syn-Gau}]{
\scalebox{0.12}[0.12]{\includegraphics{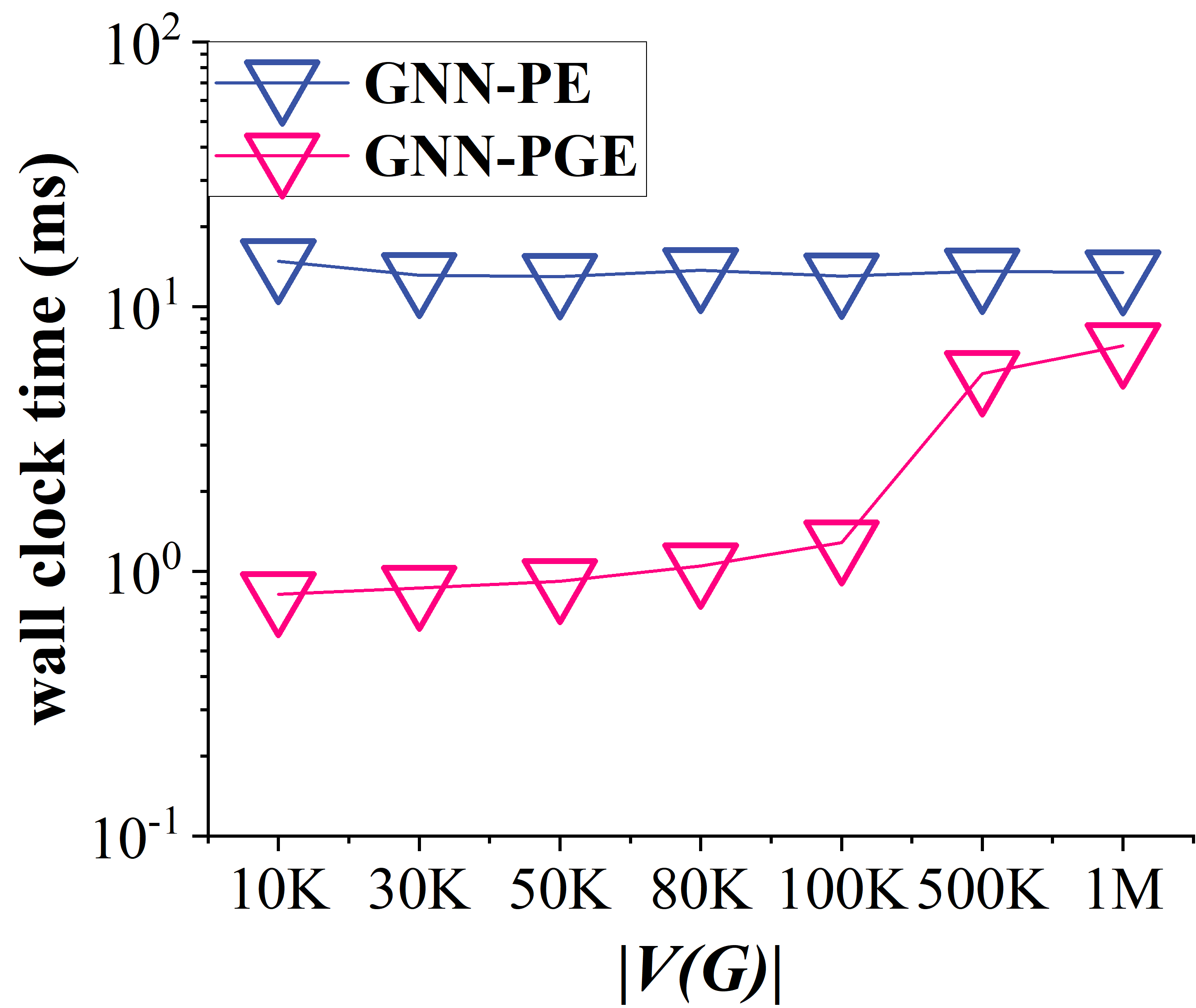}}\label{subfig:b}}
\subfigure[][{\small Syn-Zipf}]{
\scalebox{0.12}[0.12]{\includegraphics{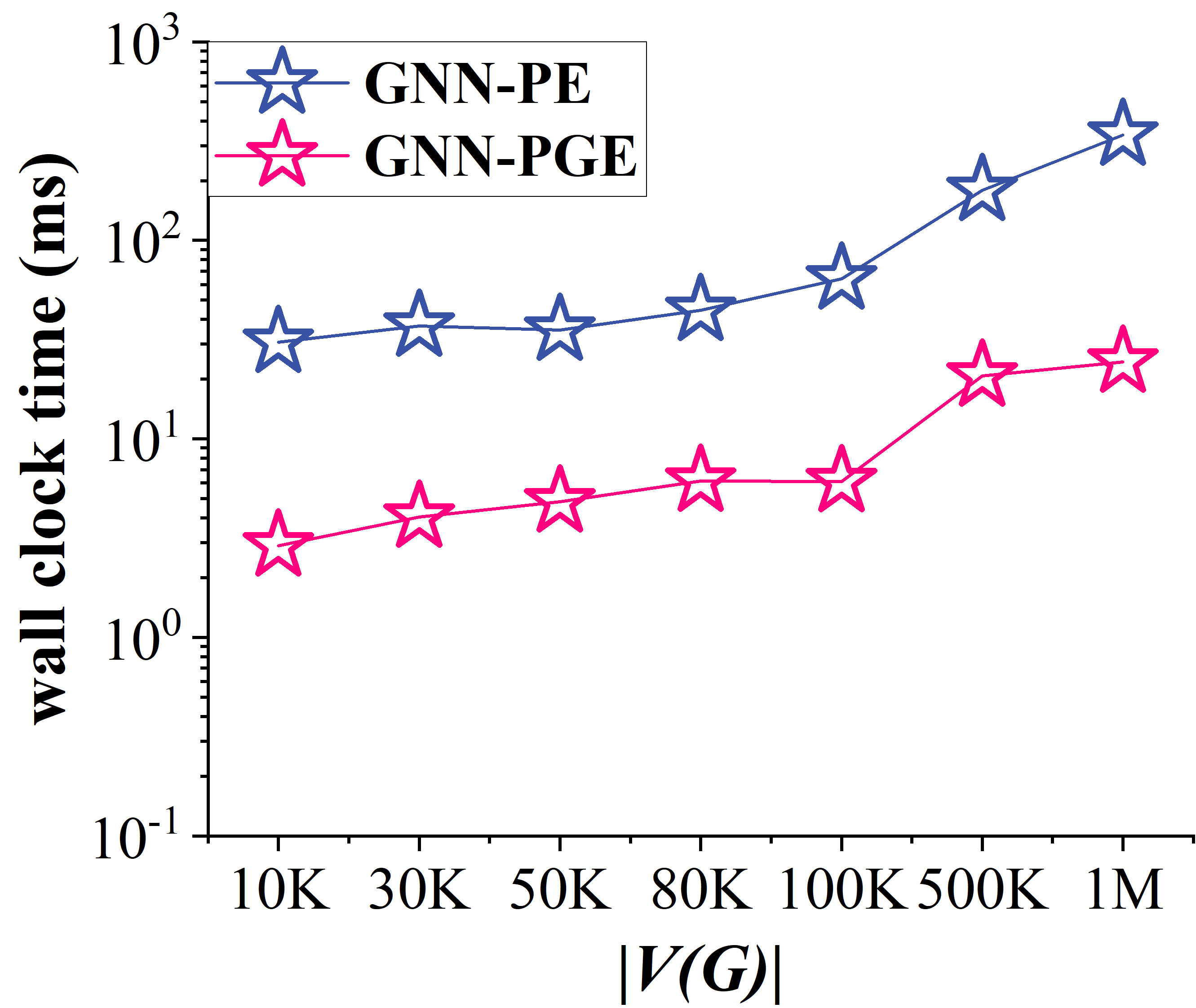}}\label{subfig:strategies}}
\caption{Scalability test w.r.t. data graph size $\bm{|V(G)|}$}
\label{fig:datasize}
\end{figure}

\noindent {\bf The GNN-PE/GNN-PGE Efficiency w.r.t. Average Degree, $\bm{avg\_deg(G)}$, of the Data Graph $\bm{G}$.}
Figure~\ref{fig:data_avg_deg} presents the performance of our GNN-PE and GNN-PGE approaches with different average degrees, $avg\_deg(G)$, of the data graph $G$, where $avg\_deg(G) = 3$, $4$, ..., and $7$, and default values are used for other parameters. Intuitively, higher degree $avg\_deg(G)$ in data graph $G$ incurs lower pruning power and more candidate paths. Thus, when $avg\_deg(G)$ becomes higher, the wall clock time also increases. Nevertheless, the wall clock time remains small (i.e., for GNN-PE, less than 0.035 $sec$ on $Syn\text{-}Uni$ and $Syn\text{-}Gau$, and 0.109 $sec$ on $Syn\text{-}Zipf$; for GNN-PGE, less than 0.001 $sec$ on $Syn\text{-}Uni$ and $Syn\text{-}Gau$, and 0.009 $sec$ on $Syn\text{-}Zipf$) for different $avg\_deg(G)$ values.

\noindent {\bf The GNN-PE/GNN-PGE Scalability Test w.r.t. Data Graph Size $\bm{|V(G)|}$.}
Figure~\ref{fig:datasize} tests the scalability of our GNN-PE and GNN-PGE approaches with different data graph sizes, $|V(G)|$, from 10K to 1M, where default values are assigned to other parameters. Since graphs are divided into partitions of similar sizes and processed with multiple threads in parallel, the GNN-PE performance over $Syn\text{-}Uni$ or $Syn\text{-}Gau$ is not very sensitive to $|V(G)|$. Moreover, for $Syn\text{-}Zipf$, due to the skewed keyword distributions, the refinement step generates more intermediate results, which are more costly to join for larger $|V(G)|$. 
On the other hand, since GNN-PGE utilizes the vertex as the query unit, its query cost increases as the graph size increases.
Nevertheless, for graph sizes from 10K to 1M, the wall clock time remains low (i.e., for GNN-PE, 0.010 $sec$ $\sim$ 0.015 $sec$ on $Syn\text{-}Uni$ and $Syn\text{-}Gau$, and 0.031 $sec$ $\sim$ 0.34 $sec$ on $Syn\text{-}Zipf$; for GNN-PGE, 0.0008 $sec$ $\sim$ 0.0071 $sec$ on $Syn\text{-}Uni$ and $Syn\text{-}Gau$, and 0.0029 $sec$ $\sim$ 0.0244 $sec$ on $Syn\text{-}Zipf$), which confirms the scalability of our GNN-PE approach for large graph sizes.

\noindent{\bf The GNN-PE Efficiency Evaluation w.r.t. Average Degree of Partitions.} Figure~\ref{subfig:part_degree} reports the performance of our GNN-PE approach by varying the average degree of partitions from 2 to 4, where $\#$ of edge cuts between any two partitions is 1,000 on average, and default values are used for other parameters. 
Higher degrees of partitions lead to more edge cuts and cross-partition paths, which incurs higher filtering/refinement costs. Thus, the time cost increases for a higher degree of partitions but remains small (i.e., $<$0.03 $sec$).

\noindent{\bf The GNN-PE Efficiency Evaluation w.r.t. $\#$ of Edge Cuts Between Partitions.} Figure~\ref{subfig:part_cuts} shows the GNN-PE performance by varying the average number of edge cuts between any two partitions from $100$ to $2,000$, where the partition degree is 3, and default parameter values are used. In general, more edge cuts between partitions produce more candidate paths, which results in higher retrieval/refinement costs. There are some exceptions over $Syn\text{-}Zipf$ (i.e., high query costs for fewer cross-partition edge cuts), which is due to the skewed distribution of vertex labels in $Syn\text{-}Zipf$.
Nevertheless, the query time remains small (i.e., $<$0.03 $sec$) for different numbers of cross-partition edge cuts.

\begin{figure}[t]
\centering
\subfigure[][{\small real-world graphs}]{                    
\scalebox{0.17}[0.17]{\includegraphics{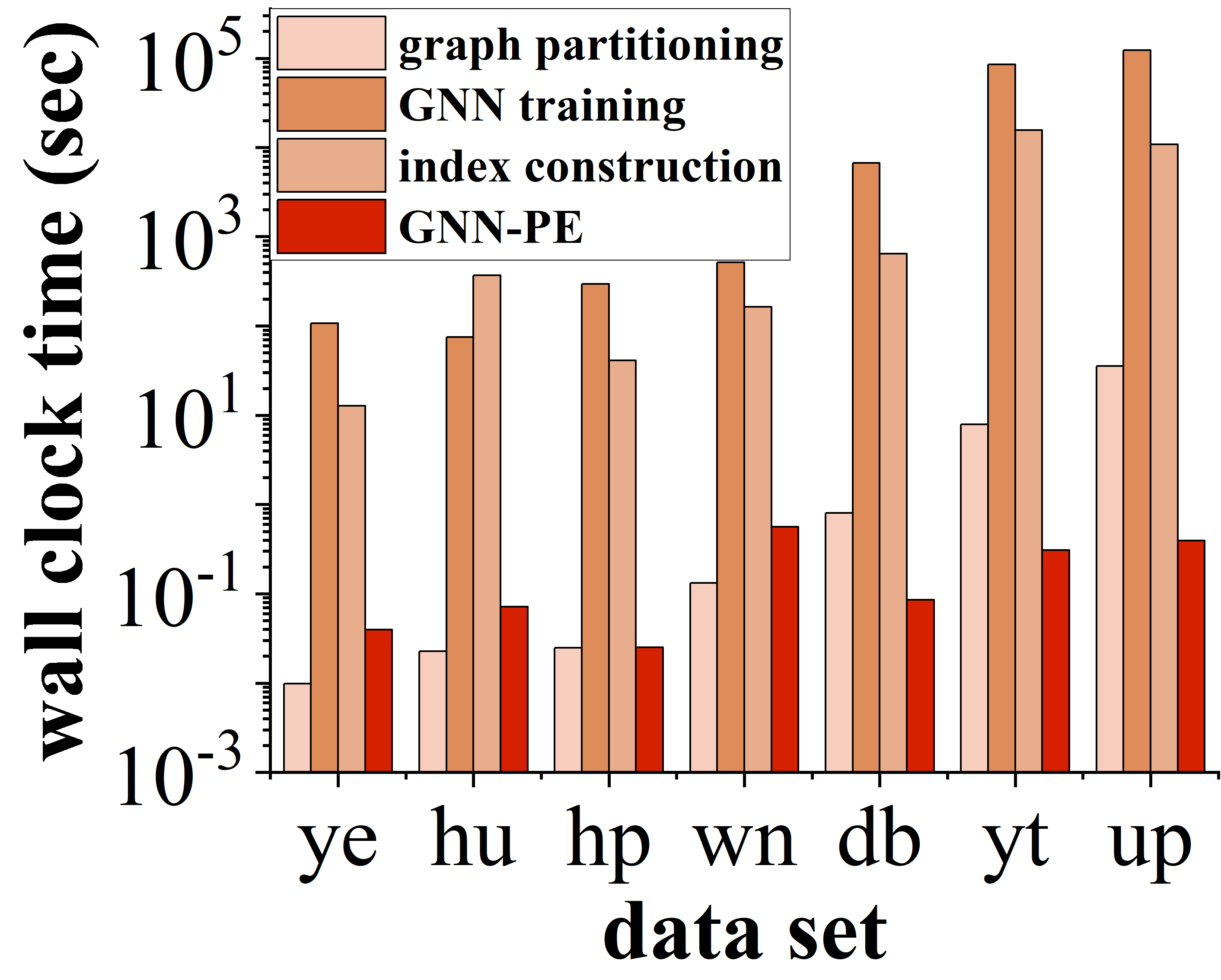}}\label{subfig:real_offline}}%
\subfigure[][{\small synthetic graphs}]{
\scalebox{0.17}[0.17]{\includegraphics{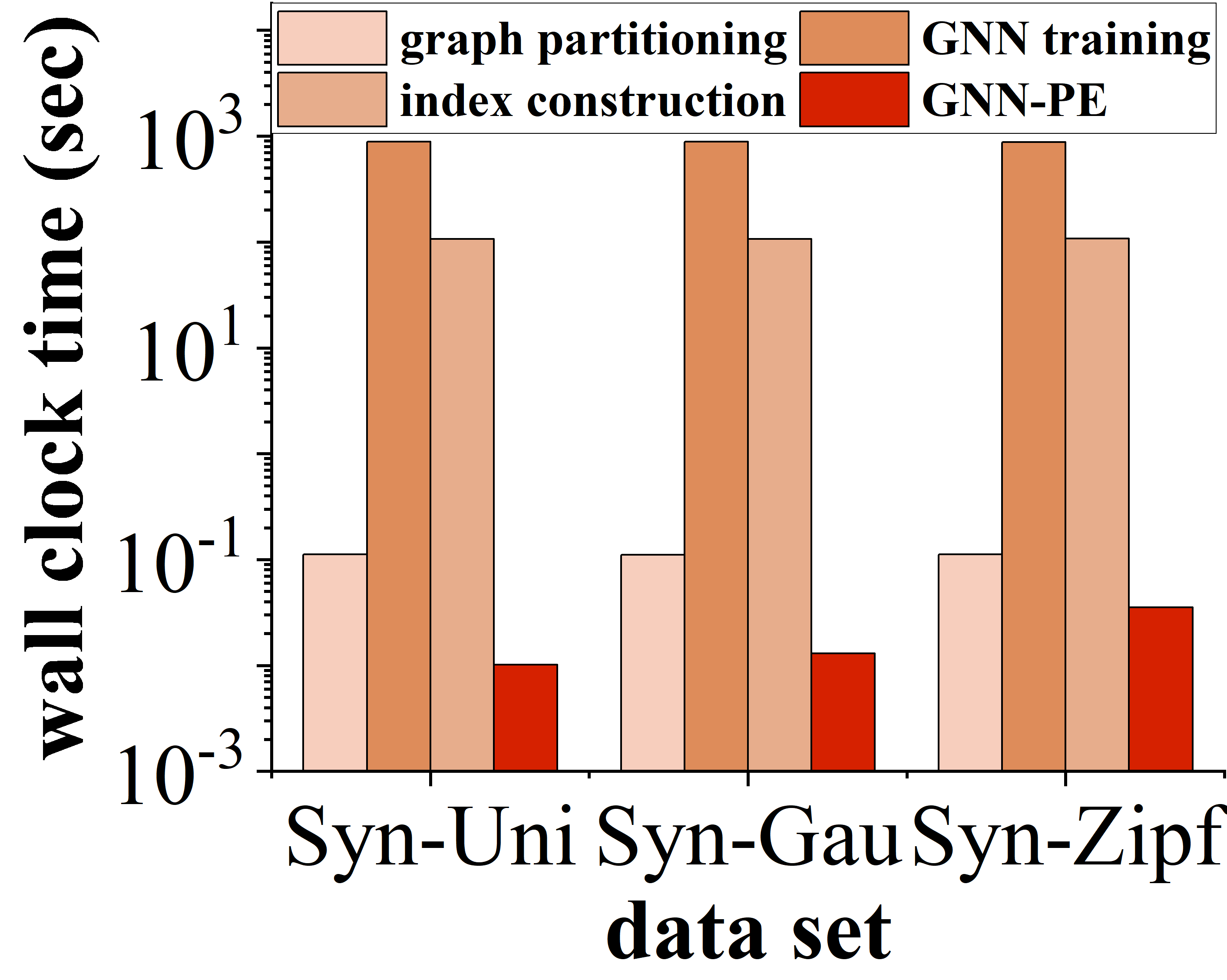}}\label{subfig:syn_offline}}
\caption{{The comparison analysis of the GNN-PE offline pre-computation and online query costs on real/synthetic graphs.}}
\label{fig:offlinecost}
\end{figure}

\begin{figure}[t]
\subfigcapskip=-0.cm
\centering
\subfigure[][{\small real-world graphs}]{
\scalebox{0.16}[0.16]{\includegraphics{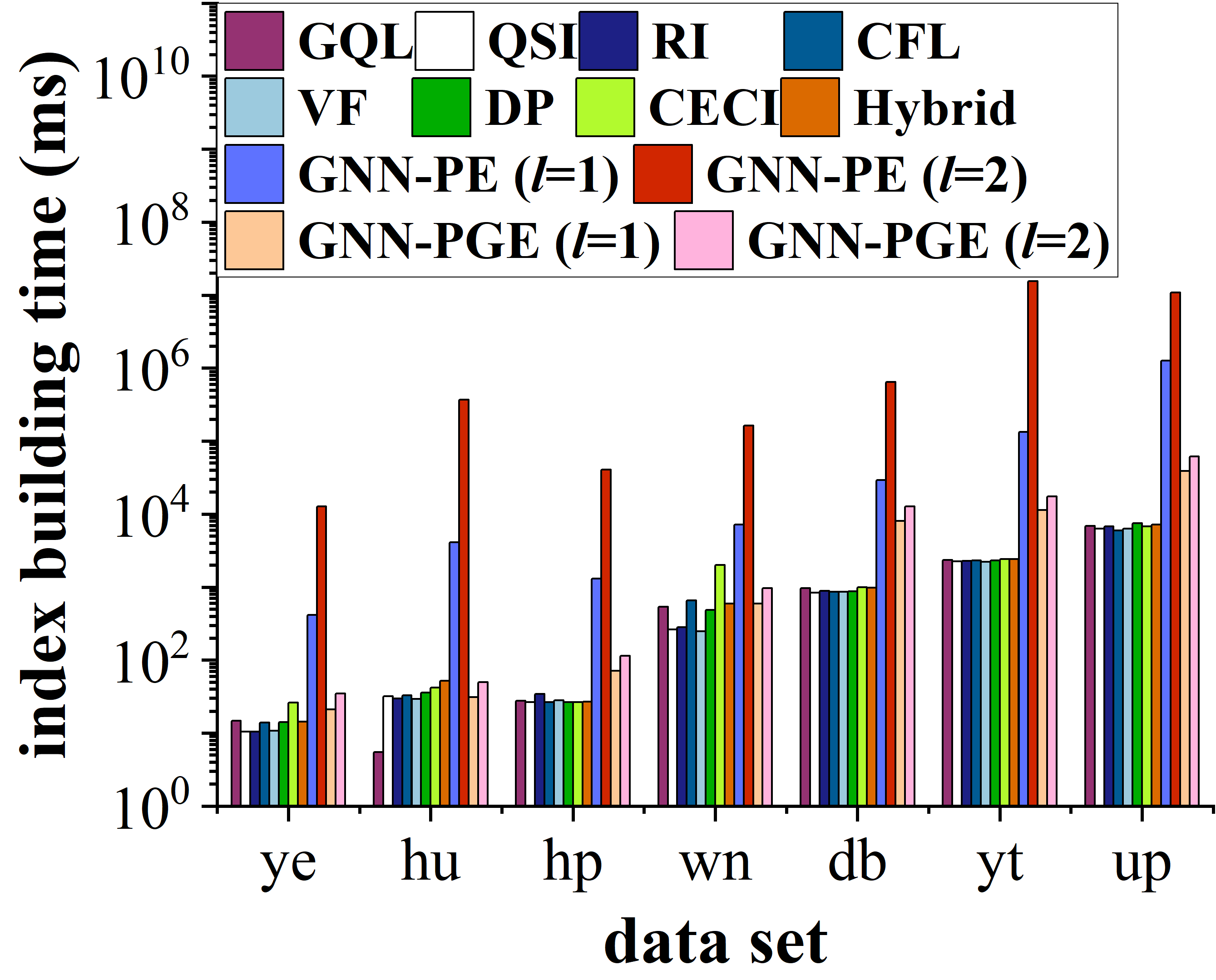}}\label{subfig:time_real}}
\quad
\subfigure[][{\small synthetic graphs}]{
\scalebox{0.16}[0.16]{\includegraphics{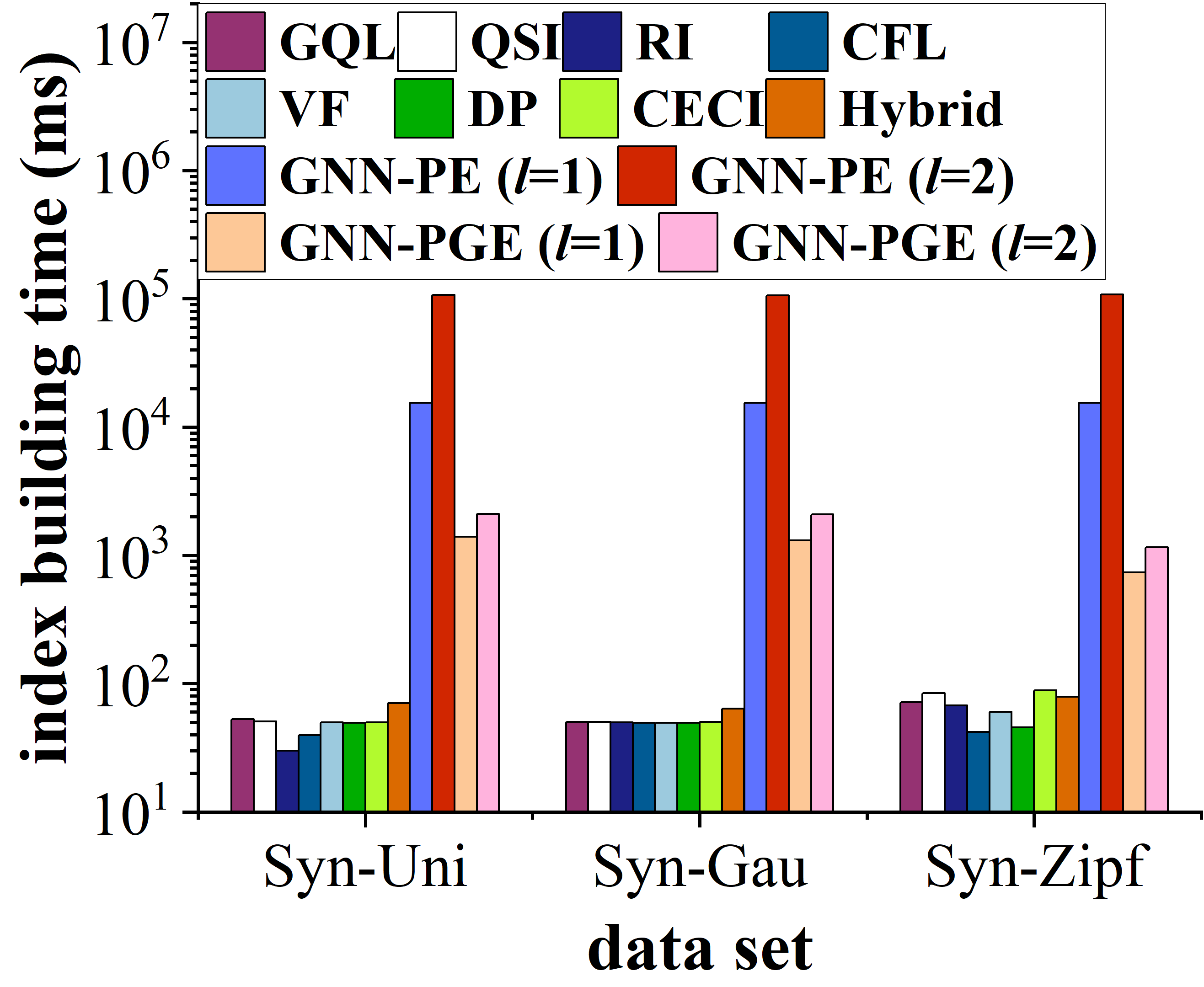}}\label{subfig:time_syn}}
\caption{Index building time on real/synthetic graphs, compared with baseline methods.}
\label{fig:index_time}
\end{figure}

\begin{figure}[t]
\subfigcapskip=-0.cm
\centering
\subfigure[][{\small real-world graphs}]{
\scalebox{0.16}[0.16]{\includegraphics{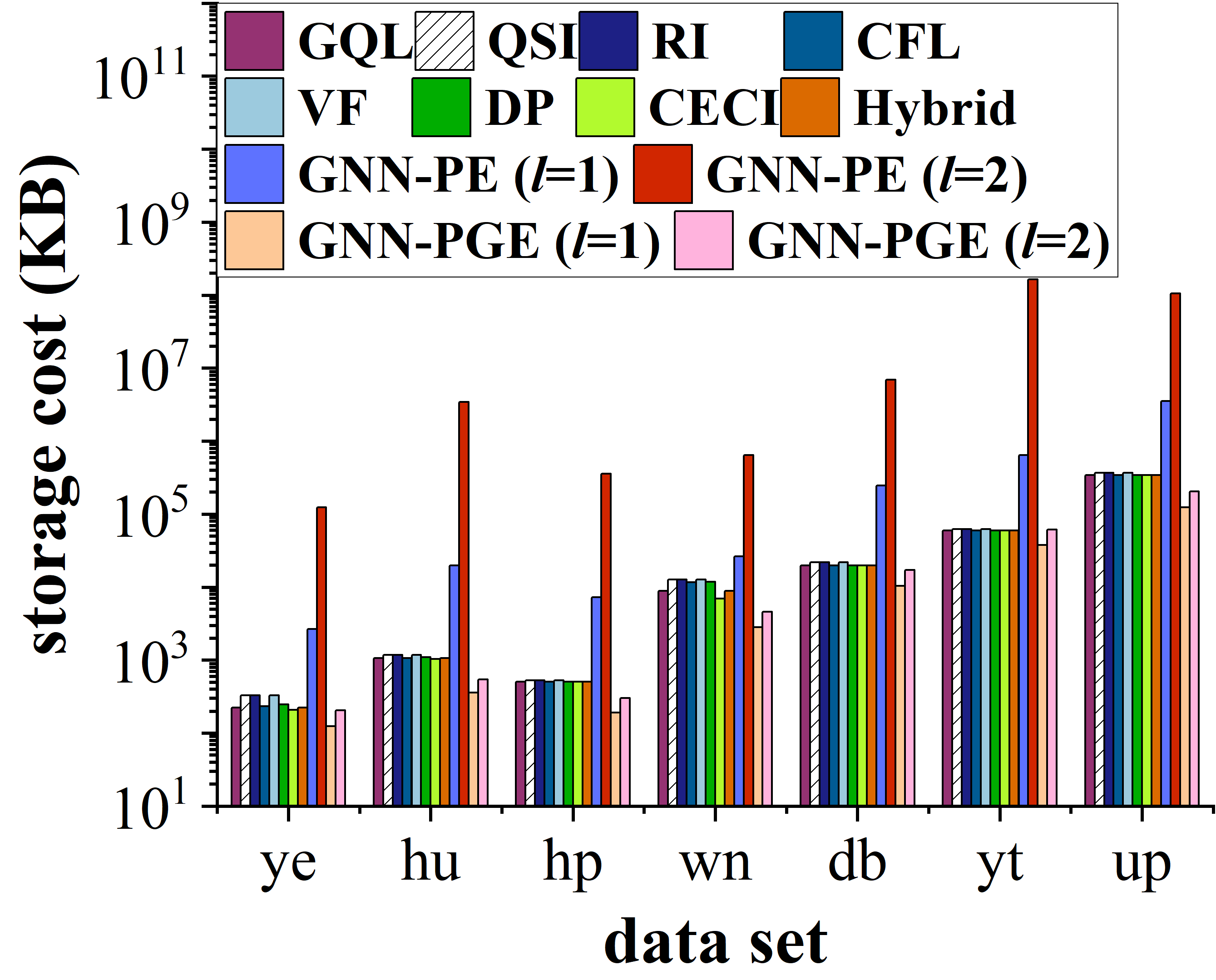}}\label{subfig:storage_real}}
\quad
\subfigure[][{\small synthetic graphs}]{
\scalebox{0.16}[0.16]{\includegraphics{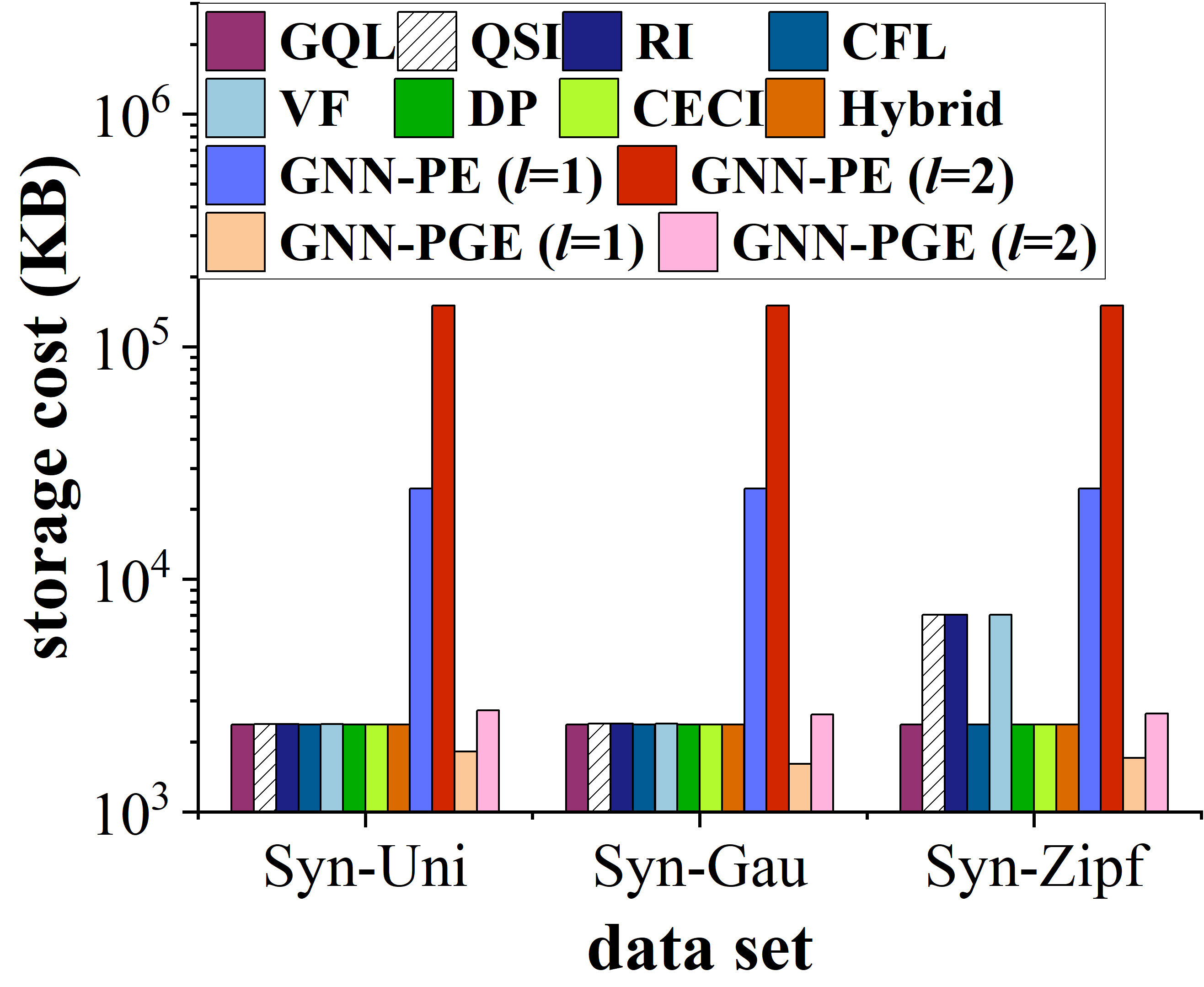}}\label{subfig:storage_syn}}
\caption{Index storage cost on real/synthetic graphs, compared with baseline methods.}
\label{fig:index_storage}
\end{figure}

\begin{figure*}[h!]
\centering\hspace{-2ex}
\subfigure[][{\small $Syn\text{-}Uni$}]{                    
\scalebox{0.2}[0.2]{\includegraphics{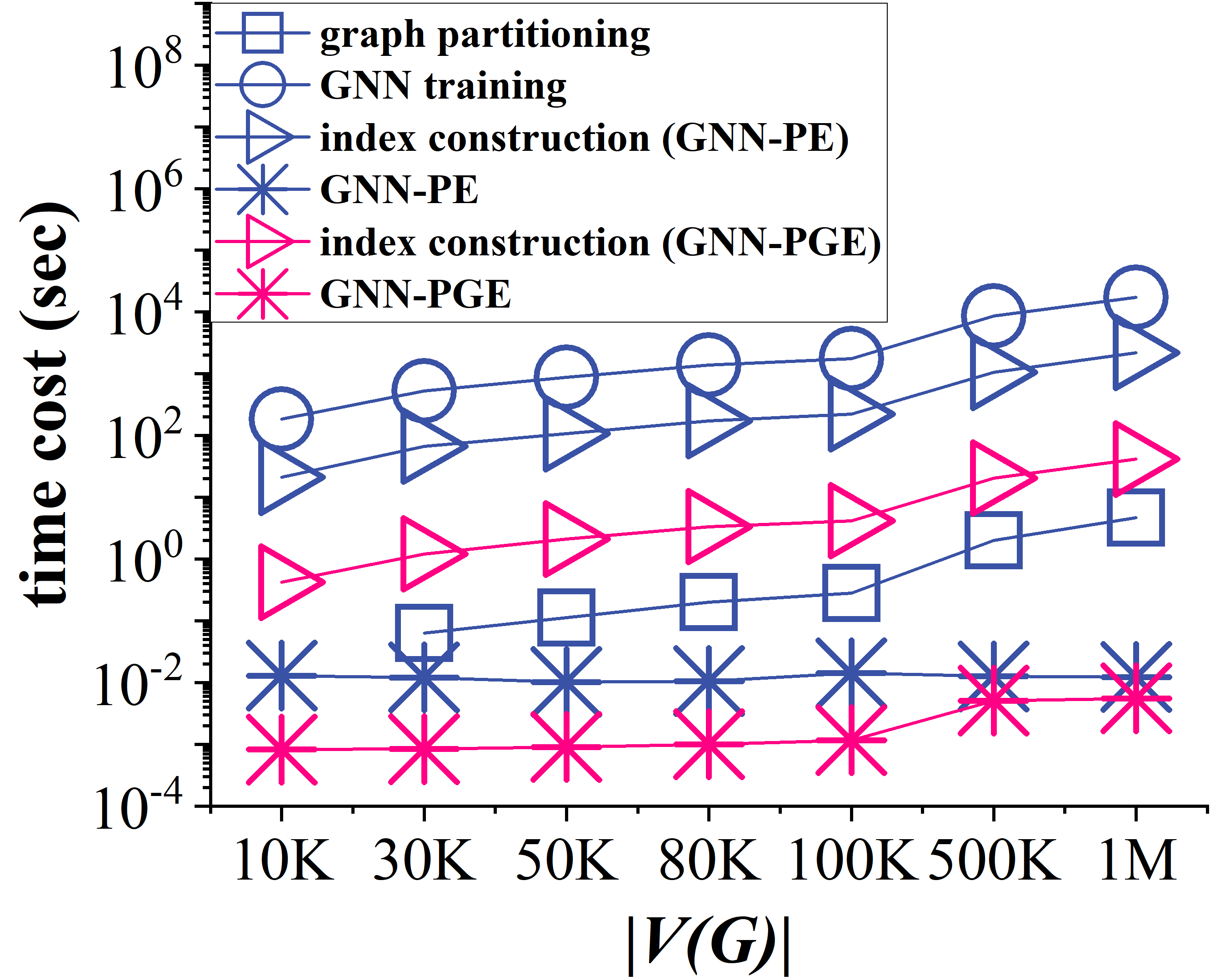}}\label{subfig:precost_uni}}\qquad\qquad
\subfigure[][{\small $Syn\text{-}Gau$}]{
\scalebox{0.2}[0.2]{\includegraphics{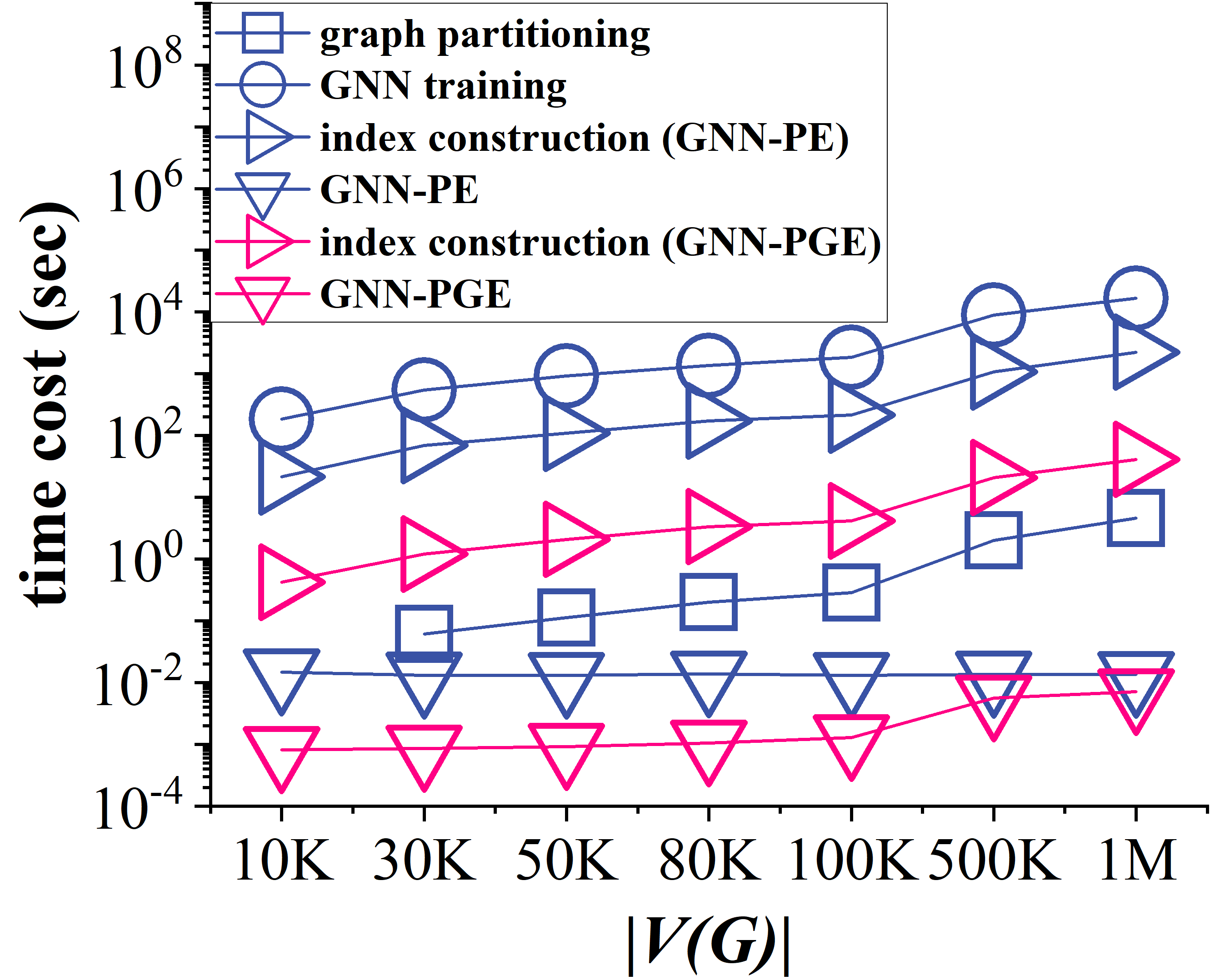}}\label{subfig:precost_gau}}\qquad\qquad
\subfigure[][{\small $Syn\text{-}Zipf$}]{
\scalebox{0.2}[0.2]{\includegraphics{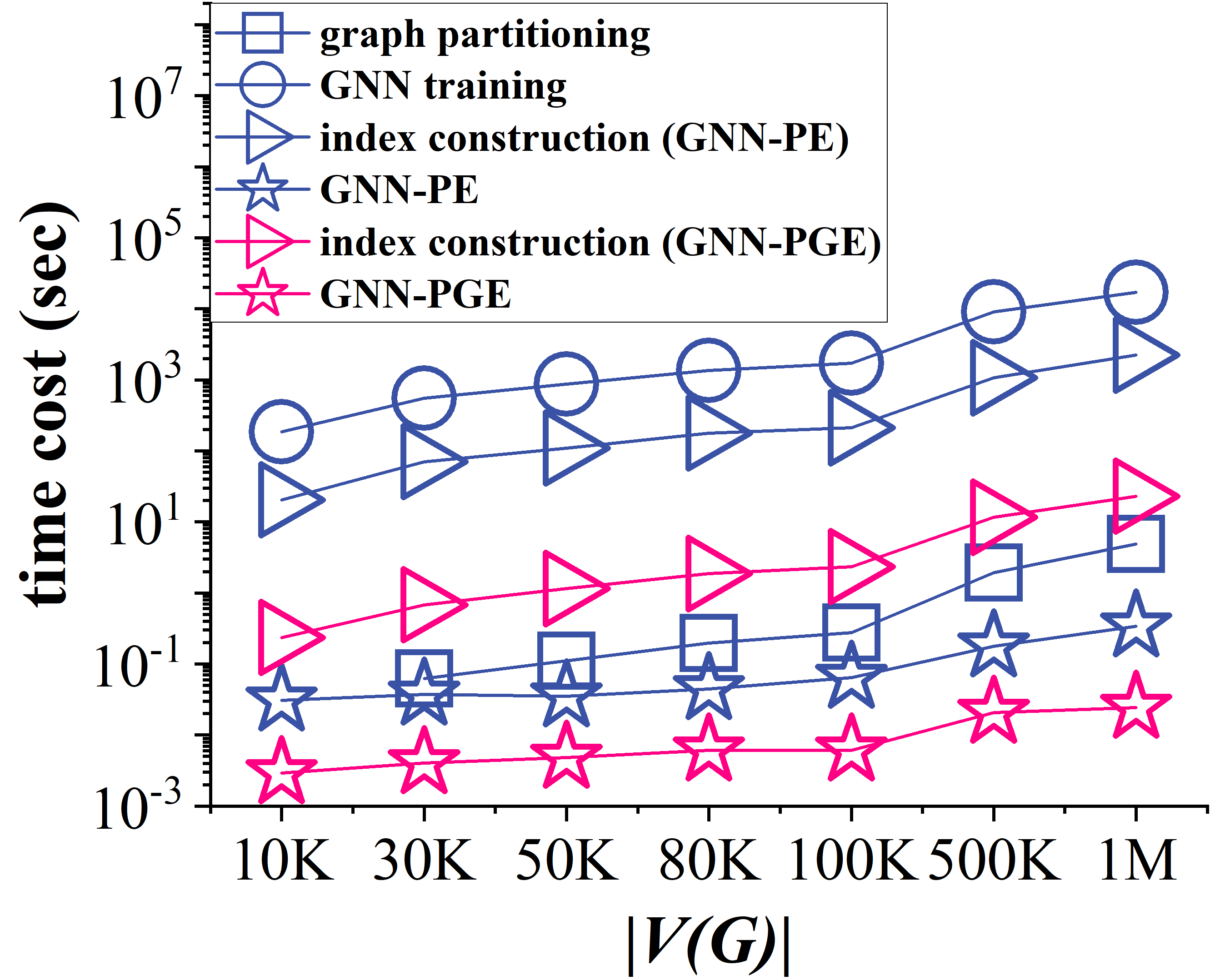}}\label{subfig:precost_zipf}}
\caption{The GNN-PE pre-computation and query costs w.r.t data graph size $\bm{|V(G)|}$.}
\label{fig:precost_size}
\end{figure*}

\begin{figure}[h!]
\centering
\subfigure[][{\small real-world graphs}]{                    
\scalebox{0.17}[0.17]{\includegraphics{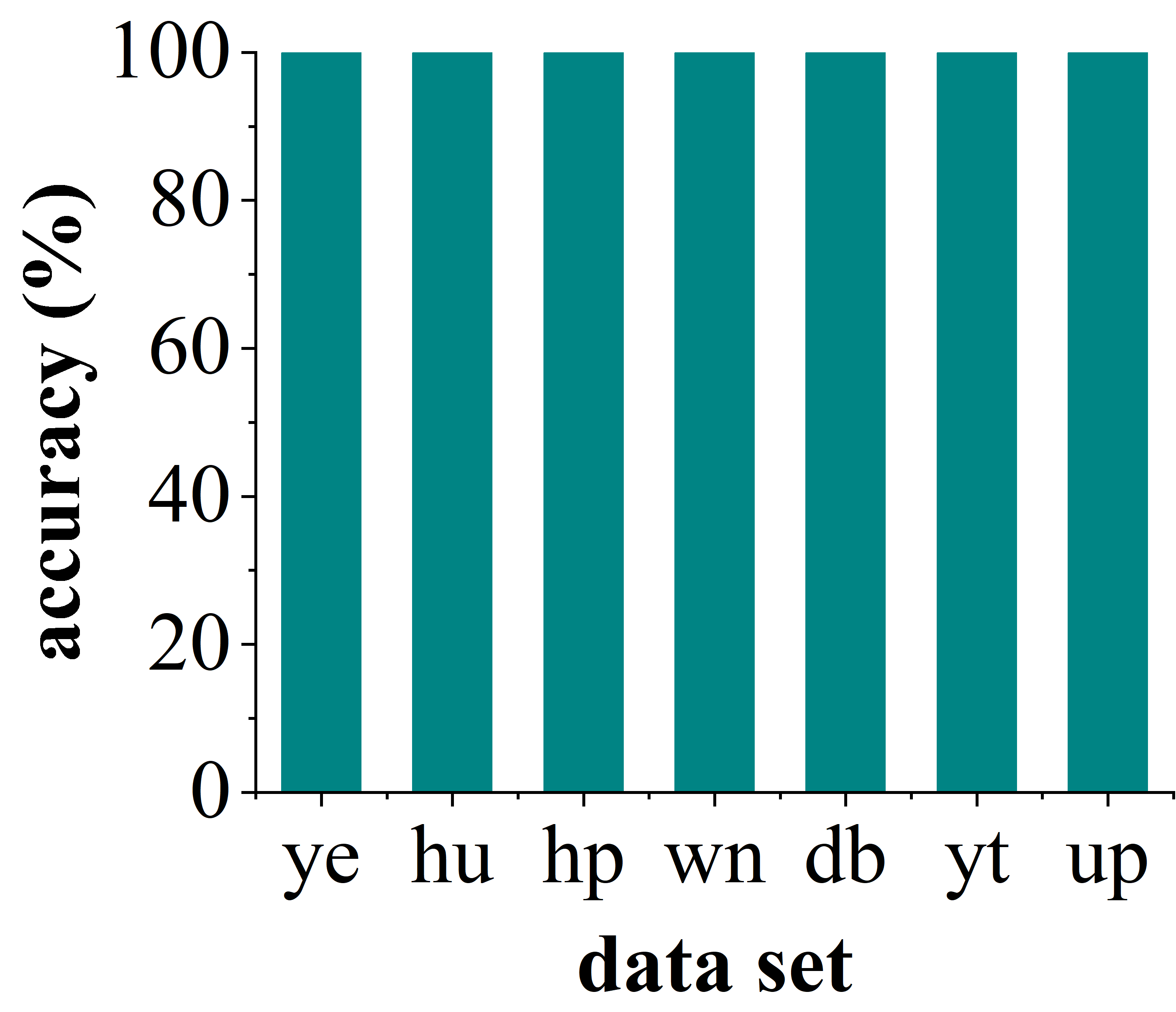}}\label{subfig:corr_real}}
\subfigure[][{\small synthetic graphs}]{
\scalebox{0.17}[0.17]{\includegraphics{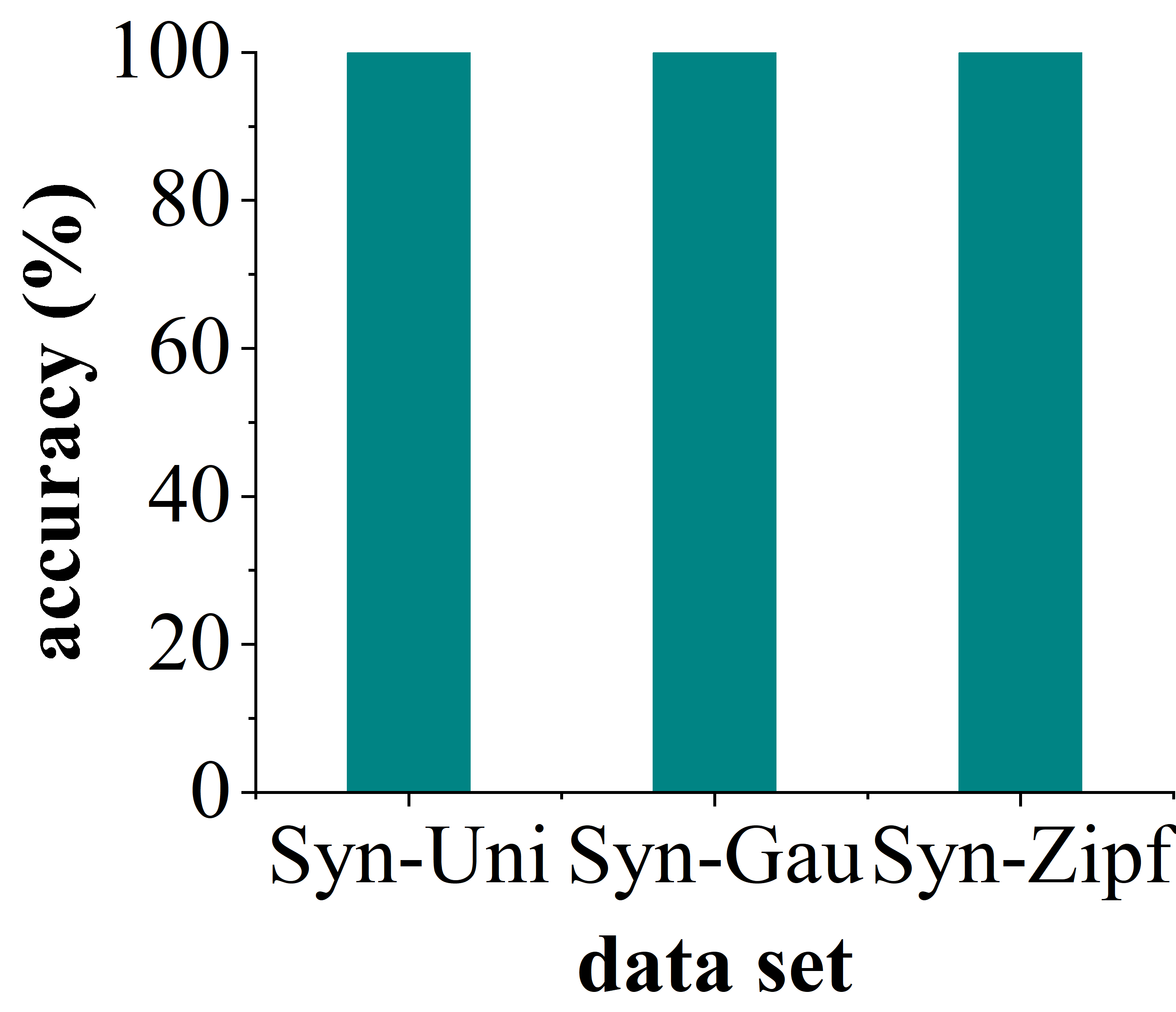}}\label{subfig:corr_syn}}
\caption{GNN-PE accuracy analysis.}
\label{fig:correctness}
\end{figure}

\subsection{Offline Pre-Computation Performance}

\noindent{\bf The GNN-PE Offline Pre-Computation Cost on Real/Synthetic Graphs.}
We compare the offline pre-computation time of our GNN-PE approach (including time costs of the graph partitioning, GNN training, and index construction on path embeddings) with online query time over real/synthetic graphs on a \textit{single} machine, where parameters are set to default values. In Figure~\ref{fig:offlinecost}, for graph size from 3K to 3.77M, the overall pre-computation time varies from 2 $min$ to 37.22 $hours$ and the online query time varies from 0.01 $sec$ to 0.56 $sec$. Specifically, the time costs of the graph partitioning, GNN training, and index construction are 0.01$\sim$35.36 $sec$, 1.24 $min$$\sim$34.21 $hours$, 12.77 $sec$$\sim$2.96 $hours$, respectively.

\noindent{\bf The GNN-PE/GNN-PGE Index Construction Time/Space Costs on Synthetic Graphs.} Figures \ref{fig:index_time} and \ref{fig:index_storage} compare the index construction time and storage cost, respectively, between our proposed GNN-PE and GNN-PGE approaches (path length $l=1$ or $2$) with 8 baselines, where other parameters are set to default values. 
From these figures, compared with baselines, in general, our GNN-PE approach needs higher time costs for offline index construction and more space costs for storing/indexing pre-computed GNN-based path embeddings. 
By grouping the paths that start from the same vertices, our GNN-PGE approach can significantly reduce both time and storage costs and achieve comparable performance with baselines on some datasets, such as $hu$ and $wn$.
Note that, unlike baseline methods that construct an index for each query graph during online subgraph matching, our index construction is \textit{offline and one-time only}, and the index construction time and storage cost are the summations over all partitions, which can be further optimized by using \textit{multiple} servers in a distributed parallel environment (as our future work). Thus, our offline constructed index can be used to accelerate \textit{numerous} online subgraph matching requests from users simultaneously with \textit{high throughput}. 

\noindent{\bf Offline Pre-Computation Cost w.r.t. Data Graph Size $|V(G)|$.} 
Figure~\ref{fig:precost_size} evaluates the offline pre-computation time of our GNN-PE and GNN-PGE approaches, including time costs of the graph partitioning, GNN training, and index construction on path embeddings (and path group embeddings), compared with online GNN-PE and GNN-PGE query time, over synthetic graphs on a \textit{single} machine, where we vary the graph size $|V(G)|$ from $10K$ to $1M$ and other parameters are set to default values. Specifically, for graph sizes from $10K$ to $1M$, the time costs of the graph partitioning and GNN training are 0.06$\sim$4.7 $sec$ and 3.07 $min$$\sim$4.79 $hours$, respectively. For index construction, GNN-PE needs 21.07 $sec$$\sim$36.63 $min$ to insert all paths in the data graph. Due to the path grouping, the cost of GNN-PGE is reduced to 0.42 $sec$ $\sim$ 41.5 $sec$.
Compared with the offline pre-computation time, the subgraph matching query cost is much smaller (i.e., 0.01$\sim$0.34 $sec$ for GNN-PE and $<$ 0.02 $sec$ for GNN-PGE).

In Figures \ref{subfig:precost_gau} and \ref{subfig:precost_zipf}, we can see similar experimental results over $Syn\text{-}Gau$ and $Syn\text{-}Zipf$ graphs, respectively, where the subgraph matching query cost is much smaller than the overall offline pre-computation time.

\subsection{The Accuracy of GNN-PE Query Answers}
\label{subsec:correctness}

\noindent{\bf The Correctness of GNN-PE Query Answers on Real/Synthetic Graphs.}
We report the correctness of query results returned by our GNN-PE approach over real/synthetic graphs, where parameters are set to default values. 
In Figure~\ref{fig:correctness}, we compare query answers of our GNN-PE approach with Hybrid \cite{sun2020memory} (a baseline that provides exact subgraph matching results). From the figure, we can see that the matching accuracy of our GNN-PE approach over all data sets is 100\%, which confirms that GNN-PE introduces no false dismissals.

The experiments with other parameters (e.g., $K$, $F$, $F'$, etc.) over real/synthetic graphs have similar results, which are thus omitted here.

\balance

\section{Related Work}
\label{sec:related_work}
The subgraph matching has been proved to be an NP-complete problem \cite{lewis1983michael}, which is not tractable. Many prior works \cite{shang2008taming,he2008graphs,bi2016efficient,bhattarai2019ceci,han2019efficient,bonnici2013subgraph,carletti2015vf2,lou2020neural,duong2021efficient} designed various database techniques (e.g., pruning \cite{shang2008taming,he2008graphs,bi2016efficient,bhattarai2019ceci,han2019efficient,bonnici2013subgraph,carletti2015vf2} or approximation \cite{lou2020neural, duong2021efficient}) to reduce the computational overhead to optimize the computational complexity in the average case. However, in this work, we consider a novel and effective perspective, that is, an AI-based approach via GNNs, to enable efficient processing of exact subgraph matching queries without any false dismissals.

\noindent{\bf Exact Subgraph Matching.} Prior works on exact subgraph matching fall into two major categories, i.e.,  join-based \cite{lai2015scalable, lai2016scalable, aberger2017emptyheaded, ammar2018distributed, mhedhbi2019optimizing, lai2019distributed, sun2020rapidmatch} and backtracking-search-based algorithms \cite{shang2008taming, carletti2015vf2, bonnici2013subgraph, sun2012efficient, he2008graphs, bi2016efficient, bhattarai2019ceci, han2019efficient, carletti2017challenging}. 
Specifically, the join-based algorithms split the query graph into different subgraphs (e.g., edges \cite{aberger2017emptyheaded, ammar2018distributed}, paths \cite{mhedhbi2019optimizing}, and triangles \cite{sun2020rapidmatch}) and join candidates of each subgraph to obtain final answers.
Moreover, the backtracking search-based algorithms search candidates of each vertex in the query graph, and then apply the backtracking search in the data graph to enumerate subgraph matching answers.
In contrast, our proposed GNN-PE approach transforms graph paths to GNN-based embeddings in the embedding space (instead of directly over graph structures) and utilizes spatial indexes to conduct efficient searches at low query costs.
Different from GNN-PE in the conference version \cite{ye2024efficient}, in this long version, we propose a more efficient GNN-PGE approach, which groups embedding vectors of paths with the same starting vertices as the basic search unit to further optimize the subgraph matching efficiency and reduce the (index) space cost.

Several existing works \cite{shasha2002algorithmics,james2004daylight,yan2004graph} also considered subgraph matching by using feature vectors (e.g., frequency vectors based on path \cite{shasha2002algorithmics,james2004daylight} or non-path \cite{yan2004graph} patterns) to represent graphs. However, they retrieve small graphs in a graph database (e.g., chemical molecules) that contain the query graph, which differs from our problem that finds subgraphs in a single large data graph.

\noindent{\bf Approximate Subgraph Matching.} An alternative problem is to develop methods to quickly return approximate subgraphs similar to a given query graph $q$. 
Existing works on approximate subgraph matching usually search for top-$k$ most similar subgraphs from the data graph by using different measures \cite{zhu2011structure,du2017first,dutta2017neighbor,li2018efficient}. With AI techniques, recent works  \cite{li2019graph, bai2019simgnn, xu2019cross, mcfee2009partial, vendrov2015order, lou2020neural} proposed to use GNNs/DNNs to improve the efficiency of approximate subgraph matching. Although these methods can quickly check the subgraph isomorphism by avoiding the comparison of graph structures, they cannot guarantee the accuracy of query answers and have limited task scenarios (e.g., only applicable to approximately comparing two graphs, instead of finding subgraph locations in a data graph).  

\noindent {\bf Graph Neural Networks.} The GNN representation for graph data has exhibited great success in various graph-related applications \cite{wang2020gognn,hao2021ks,wang2021binarized,wang2022powerful,huang2022able}. 
Typically, GNNs can learn the node representations with iterative aggregation of neighborhood information for each node, e.g., GCN \cite{kipf2017semi}, GAT \cite{velivckovic2018graph}, GraphSAGE \cite{hamilton2017inductive}, and GIN \cite{xu2019powerful}. 
With variants of GNNs, recent works \cite{li2019graph,bai2019simgnn,duong2021efficient} have used embedding techniques for graph matching. 
However, these embedding methods do not strictly impose structural relationships between the data graph and query graph (i.e., subgraph relationships). 
In contrast, our GNN-based path dominance embedding approach transforms vertices/paths in the data graph into embedding vectors, where subgraph relationships of the corresponding vertices/paths can be guaranteed without introducing false dismissals. 
To our best knowledge, no prior works used the GNN model to process exact subgraph matching over a large-scale data graph with 100\% accuracy.

\noindent {\bf Learning-Based Graph Data Analytics.} Recent works on learning-based graph data analytics utilized GNNs to generate graph embeddings for different tasks, such as subgraph counting \cite{chen2020can,yu2023learning}, graph distance prediction \cite{ranjan2022greed}, and subgraph matching \cite{lou2020neural,roy2022interpretable}). These works, however, either considered a different problem (e.g., counting or graph distance estimation) or only reported approximate subgraph matching answers (instead of exact subgraph matching studied in our problem). Therefore, we cannot apply their GNN embedding approaches to solve our problem.

\nop{
\textcolor{blue}{
Due to the prosperity of GNN, various AI-based methods \cite{yu2023learning,chen2020can,ranjan2022greed,roy2022interpretable,lou2020neural} have been studied in graph data management. These methods use GNN to generate graph embeddings for different downstream tasks (i.e., subgraph counting \cite{chen2020can,yu2023learning}, graph distance prediction \cite{ranjan2022greed}, and subgraph matching \cite{lou2020neural,roy2022interpretable}). To obtain the final result in a learning-based way, these AI-based methods apply a similar framework and workflow. Specifically, the architecture usually contains two main components: a GNN model for embedding generation and a DNN model for result prediction. Thus, the workflow includes two corresponding steps: first generating the embeddings of two given graphs, and then predicting the results of different tasks over their embeddings. However, these works either considered a different problem (e.g., counting or graph distance estimation) or only reported approximate subgraph matching answers (instead of exact subgraph matching studied in our problem). Therefore, we cannot apply these AI-based methods to solve our problem.
}

}

\section{Conclusions}
\label{sec:conclusions}

In this paper, we propose a novel \textit{GNN-based path embedding} (GNN-PE) framework for efficient processing of exact subgraph matching queries over a large-scale data graph. We carefully design GNN models to encode paths (and their surrounding 1-hop neighbors) in the data graph into embedding vectors, where subgraph relationships are strictly reflected by vector dominance constraints in the embedding space. 
The resulting embedding vectors can be used for efficient exact subgraph matching without false dismissals.
To further optimize GNN-PE, we design a GNN-based path group embedding (GNN-PGE) approach, which performs subgraph matching over the grouped path embedding vectors. Due to the path grouping with the same starting vertices, we can significantly reduce both the search space and the storage cost.
Extensive experiments have been conducted to show the efficiency and effectiveness of our proposed GNN-PE and GNN-PGE approaches over both real and synthetic graph data sets.

\nop{

\begin{acks}
 This work was supported by the [...] Research Fund of [...] (Number [...]). Additional funding was provided by [...] and [...]. We also thank [...] for contributing [...].
\end{acks}

}

\nop{

Potential extensions:

multiple keywords per vertex
label/keyword grouping 
high dimensions - high dimensional index like X-tree?

}

\balance

\clearpage
\bibliographystyle{ACM-Reference-Format}
\bibliography{sample}


\begin{thebibliography}{98}


\ifx \showCODEN    \undefined \def \showCODEN     #1{\unskip}     \fi
\ifx \showDOI      \undefined \def \showDOI       #1{#1}\fi
\ifx \showISBNx    \undefined \def \showISBNx     #1{\unskip}     \fi
\ifx \showISBNxiii \undefined \def \showISBNxiii  #1{\unskip}     \fi
\ifx \showISSN     \undefined \def \showISSN      #1{\unskip}     \fi
\ifx \showLCCN     \undefined \def \showLCCN      #1{\unskip}     \fi
\ifx \shownote     \undefined \def \shownote      #1{#1}          \fi
\ifx \showarticletitle \undefined \def \showarticletitle #1{#1}   \fi
\ifx \showURL      \undefined \def \showURL       {\relax}        \fi
\providecommand\bibfield[2]{#2}
\providecommand\bibinfo[2]{#2}
\providecommand\natexlab[1]{#1}
\providecommand\showeprint[2][]{arXiv:#2}

\bibitem[\protect\citeauthoryear{Aberger, Lamb, Tu, N{\"o}tzli, Olukotun, and R{\'e}}{Aberger et~al\mbox{.}}{2017}]%
        {aberger2017emptyheaded}
\bibfield{author}{\bibinfo{person}{Christopher~R Aberger}, \bibinfo{person}{Andrew Lamb}, \bibinfo{person}{Susan Tu}, \bibinfo{person}{Andres N{\"o}tzli}, \bibinfo{person}{Kunle Olukotun}, {and} \bibinfo{person}{Christopher R{\'e}}.} \bibinfo{year}{2017}\natexlab{}.
\newblock \showarticletitle{Emptyheaded: A relational engine for graph processing}.
\newblock \bibinfo{journal}{\emph{ACM Transactions on Database Systems}} \bibinfo{volume}{42}, \bibinfo{number}{4} (\bibinfo{year}{2017}), \bibinfo{pages}{1--44}.
\newblock


\bibitem[\protect\citeauthoryear{Albert and Barab{\'a}si}{Albert and Barab{\'a}si}{2002}]%
        {albert2002statistical}
\bibfield{author}{\bibinfo{person}{R{\'e}ka Albert} {and} \bibinfo{person}{Albert-L{\'a}szl{\'o} Barab{\'a}si}.} \bibinfo{year}{2002}\natexlab{}.
\newblock \showarticletitle{Statistical mechanics of complex networks}.
\newblock \bibinfo{journal}{\emph{Reviews of Modern Physics}} \bibinfo{volume}{74}, \bibinfo{number}{1} (\bibinfo{year}{2002}), \bibinfo{pages}{47}.
\newblock


\bibitem[\protect\citeauthoryear{Alon}{Alon}{2007}]%
        {alon2007network}
\bibfield{author}{\bibinfo{person}{Uri Alon}.} \bibinfo{year}{2007}\natexlab{}.
\newblock \showarticletitle{Network motifs: theory and experimental approaches}.
\newblock \bibinfo{journal}{\emph{Nature Reviews Genetics}} \bibinfo{volume}{8}, \bibinfo{number}{6} (\bibinfo{year}{2007}), \bibinfo{pages}{450--461}.
\newblock


\bibitem[\protect\citeauthoryear{Ammar, McSherry, Salihoglu, and Joglekar}{Ammar et~al\mbox{.}}{2018}]%
        {ammar2018distributed}
\bibfield{author}{\bibinfo{person}{Khaled Ammar}, \bibinfo{person}{Frank McSherry}, \bibinfo{person}{Semih Salihoglu}, {and} \bibinfo{person}{Manas Joglekar}.} \bibinfo{year}{2018}\natexlab{}.
\newblock \showarticletitle{Distributed evaluation of subgraph queries using worstcase optimal lowmemory dataflows}. In \bibinfo{booktitle}{\emph{Proceedings of the International Conference on Very Large Data Bases (PVLDB)}}. \bibinfo{pages}{691–--704}.
\newblock


\bibitem[\protect\citeauthoryear{Anagnostopoulos, Becchetti, Castillo, Gionis, and Leonardi}{Anagnostopoulos et~al\mbox{.}}{2012}]%
        {AnagnostopoulosBCGL12}
\bibfield{author}{\bibinfo{person}{Aris Anagnostopoulos}, \bibinfo{person}{Luca Becchetti}, \bibinfo{person}{Carlos Castillo}, \bibinfo{person}{Aristides Gionis}, {and} \bibinfo{person}{Stefano Leonardi}.} \bibinfo{year}{2012}\natexlab{}.
\newblock \showarticletitle{Online team formation in social networks}. In \bibinfo{booktitle}{\emph{Proceedings of the Web Conference (WWW)}}. \bibinfo{pages}{839--848}.
\newblock


\bibitem[\protect\citeauthoryear{Archibald, Dunlop, Hoffmann, McCreesh, Prosser, and Trimble}{Archibald et~al\mbox{.}}{2019}]%
        {archibald2019sequential}
\bibfield{author}{\bibinfo{person}{Blair Archibald}, \bibinfo{person}{Fraser Dunlop}, \bibinfo{person}{Ruth Hoffmann}, \bibinfo{person}{Ciaran McCreesh}, \bibinfo{person}{Patrick Prosser}, {and} \bibinfo{person}{James Trimble}.} \bibinfo{year}{2019}\natexlab{}.
\newblock \showarticletitle{Sequential and parallel solution-biased search for subgraph algorithms}. In \bibinfo{booktitle}{\emph{Proceedings of the Integration of Constraint Programming, Artificial Intelligence, and Operations Research (CPAIOR)}}. \bibinfo{pages}{20--38}.
\newblock


\bibitem[\protect\citeauthoryear{Babai}{Babai}{2018}]%
        {babai2018group}
\bibfield{author}{\bibinfo{person}{L{\'a}szl{\'o} Babai}.} \bibinfo{year}{2018}\natexlab{}.
\newblock \showarticletitle{Group, graphs, algorithms: the graph isomorphism problem}. In \bibinfo{booktitle}{\emph{Proceedings of the International Congress of Mathematicians: Rio de Janeiro 2018}}. World Scientific, \bibinfo{pages}{3319--3336}.
\newblock


\bibitem[\protect\citeauthoryear{Bai, Ding, Bian, Chen, Sun, and Wang}{Bai et~al\mbox{.}}{2019}]%
        {bai2019simgnn}
\bibfield{author}{\bibinfo{person}{Yunsheng Bai}, \bibinfo{person}{Hao Ding}, \bibinfo{person}{Song Bian}, \bibinfo{person}{Ting Chen}, \bibinfo{person}{Yizhou Sun}, {and} \bibinfo{person}{Wei Wang}.} \bibinfo{year}{2019}\natexlab{}.
\newblock \showarticletitle{Simgnn: A neural network approach to fast graph similarity computation}. In \bibinfo{booktitle}{\emph{Proceedings of the International Conference on Web Search and Data Mining (WSDM)}}. \bibinfo{pages}{384--392}.
\newblock


\bibitem[\protect\citeauthoryear{Barab{\'a}si and Albert}{Barab{\'a}si and Albert}{1999}]%
        {barabasi1999emergence}
\bibfield{author}{\bibinfo{person}{Albert-L{\'a}szl{\'o} Barab{\'a}si} {and} \bibinfo{person}{R{\'e}ka Albert}.} \bibinfo{year}{1999}\natexlab{}.
\newblock \showarticletitle{Emergence of scaling in random networks}.
\newblock \bibinfo{journal}{\emph{science}} \bibinfo{volume}{286}, \bibinfo{number}{5439} (\bibinfo{year}{1999}), \bibinfo{pages}{509--512}.
\newblock


\bibitem[\protect\citeauthoryear{Barab{\'a}si and Bonabeau}{Barab{\'a}si and Bonabeau}{2003}]%
        {barabasi2003scale}
\bibfield{author}{\bibinfo{person}{Albert-L{\'a}szl{\'o} Barab{\'a}si} {and} \bibinfo{person}{Eric Bonabeau}.} \bibinfo{year}{2003}\natexlab{}.
\newblock \showarticletitle{Scale-free networks}.
\newblock \bibinfo{journal}{\emph{Scientific American}} \bibinfo{volume}{288}, \bibinfo{number}{5} (\bibinfo{year}{2003}), \bibinfo{pages}{60--69}.
\newblock


\bibitem[\protect\citeauthoryear{Beckmann, Kriegel, Schneider, and Seeger}{Beckmann et~al\mbox{.}}{1990}]%
        {beckmann1990r}
\bibfield{author}{\bibinfo{person}{Norbert Beckmann}, \bibinfo{person}{Hans-Peter Kriegel}, \bibinfo{person}{Ralf Schneider}, {and} \bibinfo{person}{Bernhard Seeger}.} \bibinfo{year}{1990}\natexlab{}.
\newblock \showarticletitle{The R*-tree: An efficient and robust access method for points and rectangles}. In \bibinfo{booktitle}{\emph{Proceedings of the International Conference on Management of Data (SIGMOD)}}. \bibinfo{pages}{322--331}.
\newblock


\bibitem[\protect\citeauthoryear{Berchtold, Keim, and Kriegel}{Berchtold et~al\mbox{.}}{1996}]%
        {BerchtoldKK96}
\bibfield{author}{\bibinfo{person}{Stefan Berchtold}, \bibinfo{person}{Daniel~A. Keim}, {and} \bibinfo{person}{Hans{-}Peter Kriegel}.} \bibinfo{year}{1996}\natexlab{}.
\newblock \showarticletitle{The X-tree : An Index Structure for High-Dimensional Data}. In \bibinfo{booktitle}{\emph{Proceedings of the International Conference on Very Large Data Bases (PVLDB)}}. \bibinfo{pages}{28--39}.
\newblock


\bibitem[\protect\citeauthoryear{Bhattarai, Liu, and Huang}{Bhattarai et~al\mbox{.}}{2019}]%
        {bhattarai2019ceci}
\bibfield{author}{\bibinfo{person}{Bibek Bhattarai}, \bibinfo{person}{Hang Liu}, {and} \bibinfo{person}{H~Howie Huang}.} \bibinfo{year}{2019}\natexlab{}.
\newblock \showarticletitle{Ceci: Compact embedding cluster index for scalable subgraph matching}. In \bibinfo{booktitle}{\emph{Proceedings of the International Conference on Management of Data (SIGMOD)}}. \bibinfo{pages}{1447--1462}.
\newblock


\bibitem[\protect\citeauthoryear{Bi, Chang, Lin, Qin, and Zhang}{Bi et~al\mbox{.}}{2016}]%
        {bi2016efficient}
\bibfield{author}{\bibinfo{person}{Fei Bi}, \bibinfo{person}{Lijun Chang}, \bibinfo{person}{Xuemin Lin}, \bibinfo{person}{Lu Qin}, {and} \bibinfo{person}{Wenjie Zhang}.} \bibinfo{year}{2016}\natexlab{}.
\newblock \showarticletitle{Efficient subgraph matching by postponing cartesian products}. In \bibinfo{booktitle}{\emph{Proceedings of the International Conference on Management of Data (SIGMOD)}}. \bibinfo{pages}{1199--1214}.
\newblock


\bibitem[\protect\citeauthoryear{Bisong and Bisong}{Bisong and Bisong}{2019}]%
        {bisong2019introduction}
\bibfield{author}{\bibinfo{person}{Ekaba Bisong} {and} \bibinfo{person}{Ekaba Bisong}.} \bibinfo{year}{2019}\natexlab{}.
\newblock \showarticletitle{Introduction to Scikit-learn}.
\newblock \bibinfo{journal}{\emph{Building Machine Learning and Deep Learning Models on Google Cloud Platform: A Comprehensive Guide for Beginners}} (\bibinfo{year}{2019}), \bibinfo{pages}{215--229}.
\newblock


\bibitem[\protect\citeauthoryear{Bonnici, Giugno, Pulvirenti, Shasha, and Ferro}{Bonnici et~al\mbox{.}}{2013}]%
        {bonnici2013subgraph}
\bibfield{author}{\bibinfo{person}{Vincenzo Bonnici}, \bibinfo{person}{Rosalba Giugno}, \bibinfo{person}{Alfredo Pulvirenti}, \bibinfo{person}{Dennis Shasha}, {and} \bibinfo{person}{Alfredo Ferro}.} \bibinfo{year}{2013}\natexlab{}.
\newblock \showarticletitle{A subgraph isomorphism algorithm and its application to biochemical data}.
\newblock \bibinfo{journal}{\emph{BMC bioinformatics}} \bibinfo{volume}{14}, \bibinfo{number}{7} (\bibinfo{year}{2013}), \bibinfo{pages}{1--13}.
\newblock


\bibitem[\protect\citeauthoryear{Borzsony, Kossmann, and Stocker}{Borzsony et~al\mbox{.}}{2001}]%
        {Borzsonyi01}
\bibfield{author}{\bibinfo{person}{S. Borzsony}, \bibinfo{person}{D. Kossmann}, {and} \bibinfo{person}{K. Stocker}.} \bibinfo{year}{2001}\natexlab{}.
\newblock \showarticletitle{The skyline operator}. In \bibinfo{booktitle}{\emph{Proceedings of the International Conference on Data Engineering (ICDE)}}. \bibinfo{pages}{421--430}.
\newblock


\bibitem[\protect\citeauthoryear{Carletti, Foggia, Saggese, and Vento}{Carletti et~al\mbox{.}}{2017}]%
        {carletti2017challenging}
\bibfield{author}{\bibinfo{person}{Vincenzo Carletti}, \bibinfo{person}{Pasquale Foggia}, \bibinfo{person}{Alessia Saggese}, {and} \bibinfo{person}{Mario Vento}.} \bibinfo{year}{2017}\natexlab{}.
\newblock \showarticletitle{Challenging the time complexity of exact subgraph isomorphism for huge and dense graphs with VF3}.
\newblock \bibinfo{journal}{\emph{IEEE Transactions on Pattern Analysis and Machine Intelligence}} \bibinfo{volume}{40}, \bibinfo{number}{4} (\bibinfo{year}{2017}), \bibinfo{pages}{804--818}.
\newblock


\bibitem[\protect\citeauthoryear{Carletti, Foggia, and Vento}{Carletti et~al\mbox{.}}{2015}]%
        {carletti2015vf2}
\bibfield{author}{\bibinfo{person}{Vincenzo Carletti}, \bibinfo{person}{Pasquale Foggia}, {and} \bibinfo{person}{Mario Vento}.} \bibinfo{year}{2015}\natexlab{}.
\newblock \showarticletitle{VF2 Plus: An improved version of VF2 for biological graphs}. In \bibinfo{booktitle}{\emph{International Workshop on Graph-Based Representations in Pattern Recognition (GbRPR)}}. \bibinfo{pages}{168--177}.
\newblock


\bibitem[\protect\citeauthoryear{Chen, Chen, Villar, and Bruna}{Chen et~al\mbox{.}}{2020}]%
        {chen2020can}
\bibfield{author}{\bibinfo{person}{Zhengdao Chen}, \bibinfo{person}{Lei Chen}, \bibinfo{person}{Soledad Villar}, {and} \bibinfo{person}{Joan Bruna}.} \bibinfo{year}{2020}\natexlab{}.
\newblock \showarticletitle{Can graph neural networks count substructures?}. In \bibinfo{booktitle}{\emph{Proceedings of the Advances in Neural Information Processing Systems (NeurIPS)}}. \bibinfo{pages}{10383--10395}.
\newblock


\bibitem[\protect\citeauthoryear{Chen, Shen, Zhou, and Yu}{Chen et~al\mbox{.}}{2009}]%
        {chen2009monitoring}
\bibfield{author}{\bibinfo{person}{Zaiben Chen}, \bibinfo{person}{Heng~Tao Shen}, \bibinfo{person}{Xiaofang Zhou}, {and} \bibinfo{person}{Jeffrey~Xu Yu}.} \bibinfo{year}{2009}\natexlab{}.
\newblock \showarticletitle{Monitoring path nearest neighbor in road networks}. In \bibinfo{booktitle}{\emph{Proceedings of the International Conference on Management of Data (SIGMOD)}}. \bibinfo{pages}{591--602}.
\newblock


\bibitem[\protect\citeauthoryear{Cordella, Foggia, Sansone, and Vento}{Cordella et~al\mbox{.}}{2004}]%
        {cordella2004sub}
\bibfield{author}{\bibinfo{person}{Luigi~P Cordella}, \bibinfo{person}{Pasquale Foggia}, \bibinfo{person}{Carlo Sansone}, {and} \bibinfo{person}{Mario Vento}.} \bibinfo{year}{2004}\natexlab{}.
\newblock \showarticletitle{A (sub) graph isomorphism algorithm for matching large graphs}.
\newblock \bibinfo{journal}{\emph{IEEE Transactions on Pattern Analysis and Machine Intelligence}} \bibinfo{volume}{26}, \bibinfo{number}{10} (\bibinfo{year}{2004}), \bibinfo{pages}{1367--1372}.
\newblock


\bibitem[\protect\citeauthoryear{Dorogovtsev and Mendes}{Dorogovtsev and Mendes}{2003}]%
        {dorogovtsev2003evolution}
\bibfield{author}{\bibinfo{person}{Sergei~N Dorogovtsev} {and} \bibinfo{person}{Jos{\'e}~FF Mendes}.} \bibinfo{year}{2003}\natexlab{}.
\newblock \bibinfo{booktitle}{\emph{Evolution of networks: From biological nets to the Internet and WWW}}.
\newblock \bibinfo{publisher}{Oxford university press}.
\newblock


\bibitem[\protect\citeauthoryear{Du, Zhang, Cao, and Tong}{Du et~al\mbox{.}}{2017}]%
        {du2017first}
\bibfield{author}{\bibinfo{person}{Boxin Du}, \bibinfo{person}{Si Zhang}, \bibinfo{person}{Nan Cao}, {and} \bibinfo{person}{Hanghang Tong}.} \bibinfo{year}{2017}\natexlab{}.
\newblock \showarticletitle{First: Fast interactive attributed subgraph matching}. In \bibinfo{booktitle}{\emph{Proceedings of the International Conference on Knowledge Discovery and Data Mining (SIGKDD)}}. \bibinfo{pages}{1447--1456}.
\newblock


\bibitem[\protect\citeauthoryear{Du, Zhai, Poczos, and Singh}{Du et~al\mbox{.}}{2019}]%
        {du2019gradient}
\bibfield{author}{\bibinfo{person}{Simon~S Du}, \bibinfo{person}{Xiyu Zhai}, \bibinfo{person}{Barnabas Poczos}, {and} \bibinfo{person}{Aarti Singh}.} \bibinfo{year}{2019}\natexlab{}.
\newblock \showarticletitle{Gradient descent provably optimizes over-parameterized neural networks}. In \bibinfo{booktitle}{\emph{Proceedings of the International Conference on Learning Representations (ICLR)}}. \bibinfo{pages}{1--19}.
\newblock


\bibitem[\protect\citeauthoryear{Duong, Hoang, Yin, Weidlich, Nguyen, and Aberer}{Duong et~al\mbox{.}}{2021}]%
        {duong2021efficient}
\bibfield{author}{\bibinfo{person}{Chi~Thang Duong}, \bibinfo{person}{Trung~Dung Hoang}, \bibinfo{person}{Hongzhi Yin}, \bibinfo{person}{Matthias Weidlich}, \bibinfo{person}{Quoc Viet~Hung Nguyen}, {and} \bibinfo{person}{Karl Aberer}.} \bibinfo{year}{2021}\natexlab{}.
\newblock \showarticletitle{Efficient streaming subgraph isomorphism with graph neural networks}. In \bibinfo{booktitle}{\emph{Proceedings of the International Conference on Very Large Data Bases (PVLDB)}}. \bibinfo{pages}{730--742}.
\newblock


\bibitem[\protect\citeauthoryear{Dutta, Nayek, and Bhattacharya}{Dutta et~al\mbox{.}}{2017}]%
        {dutta2017neighbor}
\bibfield{author}{\bibinfo{person}{Sourav Dutta}, \bibinfo{person}{Pratik Nayek}, {and} \bibinfo{person}{Arnab Bhattacharya}.} \bibinfo{year}{2017}\natexlab{}.
\newblock \showarticletitle{Neighbor-aware search for approximate labeled graph matching using the chi-square statistics}. In \bibinfo{booktitle}{\emph{Proceedings of the Web Conference (WWW)}}. \bibinfo{pages}{1281--1290}.
\newblock


\bibitem[\protect\citeauthoryear{Garey and Johnson}{Garey and Johnson}{1983}]%
        {lewis1983michael}
\bibfield{author}{\bibinfo{person}{Michael~R. Garey} {and} \bibinfo{person}{David~S. Johnson}.} \bibinfo{year}{1983}\natexlab{}.
\newblock \showarticletitle{Computers and intractability: A guide to the theory of NP-completeness}.
\newblock \bibinfo{journal}{\emph{The Journal of Symbolic Logic}} \bibinfo{volume}{48}, \bibinfo{number}{2} (\bibinfo{year}{1983}), \bibinfo{pages}{498--500}.
\newblock


\bibitem[\protect\citeauthoryear{Goodfellow, Bengio, and Courville}{Goodfellow et~al\mbox{.}}{2016}]%
        {goodfellow2016deep}
\bibfield{author}{\bibinfo{person}{Ian Goodfellow}, \bibinfo{person}{Yoshua Bengio}, {and} \bibinfo{person}{Aaron Courville}.} \bibinfo{year}{2016}\natexlab{}.
\newblock \bibinfo{booktitle}{\emph{Deep learning}}.
\newblock \bibinfo{publisher}{MIT press}.
\newblock


\bibitem[\protect\citeauthoryear{Grohe and Schweitzer}{Grohe and Schweitzer}{2020}]%
        {grohe2020graph}
\bibfield{author}{\bibinfo{person}{Martin Grohe} {and} \bibinfo{person}{Pascal Schweitzer}.} \bibinfo{year}{2020}\natexlab{}.
\newblock \showarticletitle{The graph isomorphism problem}.
\newblock \bibinfo{journal}{\emph{Commun. ACM}} \bibinfo{volume}{63}, \bibinfo{number}{11} (\bibinfo{year}{2020}), \bibinfo{pages}{128--134}.
\newblock


\bibitem[\protect\citeauthoryear{Hagberg and Conway}{Hagberg and Conway}{2020}]%
        {hagberg2020networkx}
\bibfield{author}{\bibinfo{person}{Aric Hagberg} {and} \bibinfo{person}{Drew Conway}.} \bibinfo{year}{2020}\natexlab{}.
\newblock \showarticletitle{Networkx: Network analysis with python}.
\newblock \bibinfo{journal}{\emph{URL: https://networkx. github. io}} (\bibinfo{year}{2020}).
\newblock


\bibitem[\protect\citeauthoryear{Hamilton, Ying, and Leskovec}{Hamilton et~al\mbox{.}}{2017}]%
        {hamilton2017inductive}
\bibfield{author}{\bibinfo{person}{Will Hamilton}, \bibinfo{person}{Zhitao Ying}, {and} \bibinfo{person}{Jure Leskovec}.} \bibinfo{year}{2017}\natexlab{}.
\newblock \showarticletitle{Inductive representation learning on large graphs}. In \bibinfo{booktitle}{\emph{Proceedings of the Advances in Neural Information Processing Systems (NeurIPS)}}. \bibinfo{pages}{1--11}.
\newblock


\bibitem[\protect\citeauthoryear{Han and Moraga}{Han and Moraga}{1995}]%
        {han1995influence}
\bibfield{author}{\bibinfo{person}{Jun Han} {and} \bibinfo{person}{Claudio Moraga}.} \bibinfo{year}{1995}\natexlab{}.
\newblock \showarticletitle{The influence of the sigmoid function parameters on the speed of backpropagation learning}. In \bibinfo{booktitle}{\emph{Proceedings of the International Workshop on Artificial Neural Networks (IWANN)}}. \bibinfo{pages}{195--201}.
\newblock


\bibitem[\protect\citeauthoryear{Han, Kim, Gu, Park, and Han}{Han et~al\mbox{.}}{2019}]%
        {han2019efficient}
\bibfield{author}{\bibinfo{person}{Myoungji Han}, \bibinfo{person}{Hyunjoon Kim}, \bibinfo{person}{Geonmo Gu}, \bibinfo{person}{Kunsoo Park}, {and} \bibinfo{person}{Wook-Shin Han}.} \bibinfo{year}{2019}\natexlab{}.
\newblock \showarticletitle{Efficient subgraph matching: Harmonizing dynamic programming, adaptive matching order, and failing set together}. In \bibinfo{booktitle}{\emph{Proceedings of the International Conference on Management of Data (SIGMOD)}}. \bibinfo{pages}{1429--1446}.
\newblock


\bibitem[\protect\citeauthoryear{Han, Lee, and Lee}{Han et~al\mbox{.}}{2013}]%
        {han2013turboiso}
\bibfield{author}{\bibinfo{person}{Wook-Shin Han}, \bibinfo{person}{Jinsoo Lee}, {and} \bibinfo{person}{Jeong-Hoon Lee}.} \bibinfo{year}{2013}\natexlab{}.
\newblock \showarticletitle{Turboiso: towards ultrafast and robust subgraph isomorphism search in large graph databases}. In \bibinfo{booktitle}{\emph{Proceedings of the International Conference on Management of Data (SIGMOD)}}. \bibinfo{pages}{337--348}.
\newblock


\bibitem[\protect\citeauthoryear{Hao, Cao, Sheng, Fang, and Wang}{Hao et~al\mbox{.}}{2021}]%
        {hao2021ks}
\bibfield{author}{\bibinfo{person}{Yu Hao}, \bibinfo{person}{Xin Cao}, \bibinfo{person}{Yufan Sheng}, \bibinfo{person}{Yixiang Fang}, {and} \bibinfo{person}{Wei Wang}.} \bibinfo{year}{2021}\natexlab{}.
\newblock \showarticletitle{Ks-gnn: Keywords search over incomplete graphs via graphs neural network}. In \bibinfo{booktitle}{\emph{Proceedings of the Advances in Neural Information Processing Systems (NeurIPS)}}. \bibinfo{pages}{1700--1712}.
\newblock


\bibitem[\protect\citeauthoryear{He and Singh}{He and Singh}{2008}]%
        {he2008graphs}
\bibfield{author}{\bibinfo{person}{Huahai He} {and} \bibinfo{person}{Ambuj~K Singh}.} \bibinfo{year}{2008}\natexlab{}.
\newblock \showarticletitle{Graphs-at-a-time: query language and access methods for graph databases}. In \bibinfo{booktitle}{\emph{Proceedings of the International Conference on Management of Data (SIGMOD)}}. \bibinfo{pages}{405--418}.
\newblock


\bibitem[\protect\citeauthoryear{Huang, Fang, Lin, Cao, and Zhang}{Huang et~al\mbox{.}}{2022}]%
        {huang2022able}
\bibfield{author}{\bibinfo{person}{Chenji Huang}, \bibinfo{person}{Yixiang Fang}, \bibinfo{person}{Xuemin Lin}, \bibinfo{person}{Xin Cao}, {and} \bibinfo{person}{Wenjie Zhang}.} \bibinfo{year}{2022}\natexlab{}.
\newblock \showarticletitle{Able: Meta-path prediction in heterogeneous information networks}.
\newblock \bibinfo{journal}{\emph{ACM Transactions on Knowledge Discovery from Data}} \bibinfo{volume}{16}, \bibinfo{number}{4} (\bibinfo{year}{2022}), \bibinfo{pages}{1--21}.
\newblock


\bibitem[\protect\citeauthoryear{James}{James}{2004}]%
        {james2004daylight}
\bibfield{author}{\bibinfo{person}{Craig~A James}.} \bibinfo{year}{2004}\natexlab{}.
\newblock \showarticletitle{Daylight theory manual}.
\newblock \bibinfo{journal}{\emph{http://www. daylight. com/dayhtml/doc/theory/theory. toc. html}} (\bibinfo{year}{2004}).
\newblock


\bibitem[\protect\citeauthoryear{Jin, Yang, Lin, Yang, Qin, and Peng}{Jin et~al\mbox{.}}{2021}]%
        {jin2021fast}
\bibfield{author}{\bibinfo{person}{Xin Jin}, \bibinfo{person}{Zhengyi Yang}, \bibinfo{person}{Xuemin Lin}, \bibinfo{person}{Shiyu Yang}, \bibinfo{person}{Lu Qin}, {and} \bibinfo{person}{You Peng}.} \bibinfo{year}{2021}\natexlab{}.
\newblock \showarticletitle{Fast: Fpga-based subgraph matching on massive graphs}. In \bibinfo{booktitle}{\emph{Proceedings of the International Conference on Data Engineering (ICDE)}}. \bibinfo{pages}{1452--1463}.
\newblock


\bibitem[\protect\citeauthoryear{J{\"u}ttner and Madarasi}{J{\"u}ttner and Madarasi}{2018}]%
        {juttner2018vf2++}
\bibfield{author}{\bibinfo{person}{Alp{\'a}r J{\"u}ttner} {and} \bibinfo{person}{P{\'e}ter Madarasi}.} \bibinfo{year}{2018}\natexlab{}.
\newblock \showarticletitle{VF2++—An improved subgraph isomorphism algorithm}.
\newblock \bibinfo{journal}{\emph{Discrete Applied Mathematics}}  \bibinfo{volume}{242} (\bibinfo{year}{2018}), \bibinfo{pages}{69--81}.
\newblock


\bibitem[\protect\citeauthoryear{Kankanamge, Sahu, Mhedbhi, Chen, and Salihoglu}{Kankanamge et~al\mbox{.}}{2017}]%
        {kankanamge2017graphflow}
\bibfield{author}{\bibinfo{person}{Chathura Kankanamge}, \bibinfo{person}{Siddhartha Sahu}, \bibinfo{person}{Amine Mhedbhi}, \bibinfo{person}{Jeremy Chen}, {and} \bibinfo{person}{Semih Salihoglu}.} \bibinfo{year}{2017}\natexlab{}.
\newblock \showarticletitle{Graphflow: An active graph database}. In \bibinfo{booktitle}{\emph{Proceedings of the International Conference on Management of Data (SIGMOD)}}. \bibinfo{pages}{1695--1698}.
\newblock


\bibitem[\protect\citeauthoryear{Karlebach and Shamir}{Karlebach and Shamir}{2008}]%
        {karlebach2008modelling}
\bibfield{author}{\bibinfo{person}{Guy Karlebach} {and} \bibinfo{person}{Ron Shamir}.} \bibinfo{year}{2008}\natexlab{}.
\newblock \showarticletitle{Modelling and analysis of gene regulatory networks}.
\newblock \bibinfo{journal}{\emph{Nature Reviews Molecular Cell Biology}} \bibinfo{volume}{9}, \bibinfo{number}{10} (\bibinfo{year}{2008}), \bibinfo{pages}{770--780}.
\newblock


\bibitem[\protect\citeauthoryear{Karypis and Kumar}{Karypis and Kumar}{1998}]%
        {karypis1998fast}
\bibfield{author}{\bibinfo{person}{George Karypis} {and} \bibinfo{person}{Vipin Kumar}.} \bibinfo{year}{1998}\natexlab{}.
\newblock \showarticletitle{A fast and high quality multilevel scheme for partitioning irregular graphs}.
\newblock \bibinfo{journal}{\emph{SIAM Journal on Scientific Computing}} \bibinfo{volume}{20}, \bibinfo{number}{1} (\bibinfo{year}{1998}), \bibinfo{pages}{359--392}.
\newblock


\bibitem[\protect\citeauthoryear{Katsarou, Ntarmos, and Triantafillou}{Katsarou et~al\mbox{.}}{2017}]%
        {katsarou2017subgraph}
\bibfield{author}{\bibinfo{person}{Foteini Katsarou}, \bibinfo{person}{Nikos Ntarmos}, {and} \bibinfo{person}{Peter Triantafillou}.} \bibinfo{year}{2017}\natexlab{}.
\newblock \showarticletitle{Subgraph querying with parallel use of query rewritings and alternative algorithms}. In \bibinfo{booktitle}{\emph{Proceedings of the International Conference on Extending Database Technology (EDBT)}}. \bibinfo{pages}{25--36}.
\newblock


\bibitem[\protect\citeauthoryear{Kim, Choi, Park, Lin, Hong, and Han}{Kim et~al\mbox{.}}{2021}]%
        {kim2021versatile}
\bibfield{author}{\bibinfo{person}{Hyunjoon Kim}, \bibinfo{person}{Yunyoung Choi}, \bibinfo{person}{Kunsoo Park}, \bibinfo{person}{Xuemin Lin}, \bibinfo{person}{Seok-Hee Hong}, {and} \bibinfo{person}{Wook-Shin Han}.} \bibinfo{year}{2021}\natexlab{}.
\newblock \showarticletitle{Versatile equivalences: Speeding up subgraph query processing and subgraph matching}. In \bibinfo{booktitle}{\emph{Proceedings of the International Conference on Management of Data (SIGMOD)}}. \bibinfo{pages}{925--937}.
\newblock


\bibitem[\protect\citeauthoryear{Kingma and Ba}{Kingma and Ba}{2015}]%
        {kingma2015adam}
\bibfield{author}{\bibinfo{person}{Diederik~P Kingma} {and} \bibinfo{person}{Jimmy Ba}.} \bibinfo{year}{2015}\natexlab{}.
\newblock \showarticletitle{Adam: A method for stochastic optimization}. In \bibinfo{booktitle}{\emph{Proceedings of the International Conference on Learning Representations (ICLR)}}. \bibinfo{pages}{1--15}.
\newblock


\bibitem[\protect\citeauthoryear{Kipf and Welling}{Kipf and Welling}{2017}]%
        {kipf2017semi}
\bibfield{author}{\bibinfo{person}{Thomas~N Kipf} {and} \bibinfo{person}{Max Welling}.} \bibinfo{year}{2017}\natexlab{}.
\newblock \showarticletitle{Semi-supervised classification with graph convolutional networks}. In \bibinfo{booktitle}{\emph{Proceedings of the International Conference on Learning Representations (ICLR)}}. \bibinfo{pages}{1--14}.
\newblock


\bibitem[\protect\citeauthoryear{Lai, Qin, Lin, and Chang}{Lai et~al\mbox{.}}{2015}]%
        {lai2015scalable}
\bibfield{author}{\bibinfo{person}{Longbin Lai}, \bibinfo{person}{Lu Qin}, \bibinfo{person}{Xuemin Lin}, {and} \bibinfo{person}{Lijun Chang}.} \bibinfo{year}{2015}\natexlab{}.
\newblock \showarticletitle{Scalable subgraph enumeration in mapreduce}. In \bibinfo{booktitle}{\emph{Proceedings of the International Conference on Very Large Data Bases (PVLDB)}}. \bibinfo{pages}{974--985}.
\newblock


\bibitem[\protect\citeauthoryear{Lai, Qin, Lin, Zhang, Chang, and Yang}{Lai et~al\mbox{.}}{2016}]%
        {lai2016scalable}
\bibfield{author}{\bibinfo{person}{Longbin Lai}, \bibinfo{person}{Lu Qin}, \bibinfo{person}{Xuemin Lin}, \bibinfo{person}{Ying Zhang}, \bibinfo{person}{Lijun Chang}, {and} \bibinfo{person}{Shiyu Yang}.} \bibinfo{year}{2016}\natexlab{}.
\newblock \showarticletitle{Scalable distributed subgraph enumeration}. In \bibinfo{booktitle}{\emph{Proceedings of the International Conference on Very Large Data Bases (PVLDB)}}. \bibinfo{pages}{217--228}.
\newblock


\bibitem[\protect\citeauthoryear{Lai, Qing, Yang, Jin, Lai, Wang, Hao, Lin, Qin, Zhang, et~al\mbox{.}}{Lai et~al\mbox{.}}{2019}]%
        {lai2019distributed}
\bibfield{author}{\bibinfo{person}{Longbin Lai}, \bibinfo{person}{Zhu Qing}, \bibinfo{person}{Zhengyi Yang}, \bibinfo{person}{Xin Jin}, \bibinfo{person}{Zhengmin Lai}, \bibinfo{person}{Ran Wang}, \bibinfo{person}{Kongzhang Hao}, \bibinfo{person}{Xuemin Lin}, \bibinfo{person}{Lu Qin}, \bibinfo{person}{Wenjie Zhang}, {et~al\mbox{.}}} \bibinfo{year}{2019}\natexlab{}.
\newblock \showarticletitle{Distributed subgraph matching on timely dataflow}. In \bibinfo{booktitle}{\emph{Proceedings of the International Conference on Very Large Data Bases (PVLDB)}}. \bibinfo{pages}{1099--1112}.
\newblock


\bibitem[\protect\citeauthoryear{Lazaridis and Mehrotra}{Lazaridis and Mehrotra}{2001}]%
        {Lazaridis01}
\bibfield{author}{\bibinfo{person}{Iosif Lazaridis} {and} \bibinfo{person}{Sharad Mehrotra}.} \bibinfo{year}{2001}\natexlab{}.
\newblock \showarticletitle{Progressive Approximate Aggregate Queries with a Multi-Resolution Tree Structure}. In \bibinfo{booktitle}{\emph{Proceedings of the International Conference on Management of Data (SIGMOD)}}. \bibinfo{pages}{401--412}.
\newblock


\bibitem[\protect\citeauthoryear{Lee, Han, Kasperovics, and Lee}{Lee et~al\mbox{.}}{2012}]%
        {lee2012depth}
\bibfield{author}{\bibinfo{person}{Jinsoo Lee}, \bibinfo{person}{Wook-Shin Han}, \bibinfo{person}{Romans Kasperovics}, {and} \bibinfo{person}{Jeong-Hoon Lee}.} \bibinfo{year}{2012}\natexlab{}.
\newblock \showarticletitle{An in-depth comparison of subgraph isomorphism algorithms in graph databases}.
\newblock  (\bibinfo{year}{2012}), \bibinfo{pages}{133--144}.
\newblock


\bibitem[\protect\citeauthoryear{Li, Gu, Dullien, Vinyals, and Kohli}{Li et~al\mbox{.}}{2019}]%
        {li2019graph}
\bibfield{author}{\bibinfo{person}{Yujia Li}, \bibinfo{person}{Chenjie Gu}, \bibinfo{person}{Thomas Dullien}, \bibinfo{person}{Oriol Vinyals}, {and} \bibinfo{person}{Pushmeet Kohli}.} \bibinfo{year}{2019}\natexlab{}.
\newblock \showarticletitle{Graph matching networks for learning the similarity of graph structured objects}. In \bibinfo{booktitle}{\emph{Proceedings of the International Conference on Machine Learning (ICML)}}. \bibinfo{pages}{3835--3845}.
\newblock


\bibitem[\protect\citeauthoryear{Li, Jian, Lian, and Chen}{Li et~al\mbox{.}}{2018}]%
        {li2018efficient}
\bibfield{author}{\bibinfo{person}{Zijian Li}, \bibinfo{person}{Xun Jian}, \bibinfo{person}{Xiang Lian}, {and} \bibinfo{person}{Lei Chen}.} \bibinfo{year}{2018}\natexlab{}.
\newblock \showarticletitle{An efficient probabilistic approach for graph similarity search}. In \bibinfo{booktitle}{\emph{Proceedings of the International Conference on Data Engineering (ICDE)}}. IEEE, \bibinfo{pages}{533--544}.
\newblock


\bibitem[\protect\citeauthoryear{Lian and Chen}{Lian and Chen}{2011}]%
        {lian2011efficient}
\bibfield{author}{\bibinfo{person}{Xiang Lian} {and} \bibinfo{person}{Lei Chen}.} \bibinfo{year}{2011}\natexlab{}.
\newblock \showarticletitle{Efficient query answering in probabilistic RDF graphs}. In \bibinfo{booktitle}{\emph{Proceedings of the International Conference on Management of Data (SIGMOD)}}. \bibinfo{pages}{157--168}.
\newblock


\bibitem[\protect\citeauthoryear{Lou, You, Wen, Canedo, Leskovec, et~al\mbox{.}}{Lou et~al\mbox{.}}{2020}]%
        {lou2020neural}
\bibfield{author}{\bibinfo{person}{Zhaoyu Lou}, \bibinfo{person}{Jiaxuan You}, \bibinfo{person}{Chengtao Wen}, \bibinfo{person}{Arquimedes Canedo}, \bibinfo{person}{Jure Leskovec}, {et~al\mbox{.}}} \bibinfo{year}{2020}\natexlab{}.
\newblock \showarticletitle{Neural subgraph matching}.
\newblock \bibinfo{journal}{\emph{arXiv preprint arXiv:2007.03092}} (\bibinfo{year}{2020}).
\newblock


\bibitem[\protect\citeauthoryear{Loukas}{Loukas}{2020}]%
        {loukas2020graph}
\bibfield{author}{\bibinfo{person}{Andreas Loukas}.} \bibinfo{year}{2020}\natexlab{}.
\newblock \showarticletitle{What graph neural networks cannot learn: depth vs width}. In \bibinfo{booktitle}{\emph{Proceedings of the International Conference on Learning Representations (ICLR)}}. \bibinfo{pages}{1--17}.
\newblock


\bibitem[\protect\citeauthoryear{McFee and Lanckriet}{McFee and Lanckriet}{2009}]%
        {mcfee2009partial}
\bibfield{author}{\bibinfo{person}{Brian McFee} {and} \bibinfo{person}{Gert Lanckriet}.} \bibinfo{year}{2009}\natexlab{}.
\newblock \showarticletitle{Partial order embedding with multiple kernels}. In \bibinfo{booktitle}{\emph{Proceedings of the International Conference on Machine Learning (ICML)}}. \bibinfo{pages}{721--728}.
\newblock


\bibitem[\protect\citeauthoryear{Mhedhbi and Salihoglu}{Mhedhbi and Salihoglu}{2019}]%
        {mhedhbi2019optimizing}
\bibfield{author}{\bibinfo{person}{Amine Mhedhbi} {and} \bibinfo{person}{Semih Salihoglu}.} \bibinfo{year}{2019}\natexlab{}.
\newblock \showarticletitle{Optimizing subgraph queries by combining binary and worst-case optimal joins}. In \bibinfo{booktitle}{\emph{Proceedings of the International Conference on Very Large Data Bases (PVLDB)}}. \bibinfo{pages}{1692--1704}.
\newblock


\bibitem[\protect\citeauthoryear{Nair and Hinton}{Nair and Hinton}{2010}]%
        {nair2010rectified}
\bibfield{author}{\bibinfo{person}{Vinod Nair} {and} \bibinfo{person}{Geoffrey~E Hinton}.} \bibinfo{year}{2010}\natexlab{}.
\newblock \showarticletitle{Rectified linear units improve restricted Boltzmann machines}. In \bibinfo{booktitle}{\emph{Proceedings of the International Conference on Machine Learning (ICML)}}. \bibinfo{pages}{807--814}.
\newblock


\bibitem[\protect\citeauthoryear{Newman}{Newman}{2005}]%
        {newman2005power}
\bibfield{author}{\bibinfo{person}{Mark~EJ Newman}.} \bibinfo{year}{2005}\natexlab{}.
\newblock \showarticletitle{Power laws, Pareto distributions and Zipf's law}.
\newblock \bibinfo{journal}{\emph{Contemporary Physics}} \bibinfo{volume}{46}, \bibinfo{number}{5} (\bibinfo{year}{2005}), \bibinfo{pages}{323--351}.
\newblock


\bibitem[\protect\citeauthoryear{Orogat and El-Roby}{Orogat and El-Roby}{2022}]%
        {orogat2022smartbench}
\bibfield{author}{\bibinfo{person}{Abdelghny Orogat} {and} \bibinfo{person}{Ahmed El-Roby}.} \bibinfo{year}{2022}\natexlab{}.
\newblock \showarticletitle{SmartBench: demonstrating automatic generation of comprehensive benchmarks for question answering over knowledge graphs}. In \bibinfo{booktitle}{\emph{Proceedings of the International Conference on Very Large Data Bases (PVLDB)}}. \bibinfo{pages}{3662--3665}.
\newblock


\bibitem[\protect\citeauthoryear{Qiao, Zhang, and Cheng}{Qiao et~al\mbox{.}}{2017}]%
        {qiao2017subgraph}
\bibfield{author}{\bibinfo{person}{Miao Qiao}, \bibinfo{person}{Hao Zhang}, {and} \bibinfo{person}{Hong Cheng}.} \bibinfo{year}{2017}\natexlab{}.
\newblock \showarticletitle{Subgraph matching: on compression and computation}. In \bibinfo{booktitle}{\emph{Proceedings of the International Conference on Very Large Data Bases (PVLDB)}}. \bibinfo{pages}{176--188}.
\newblock


\bibitem[\protect\citeauthoryear{Qiu, Cen, Qian, Peng, Zhang, Lin, and Zhou}{Qiu et~al\mbox{.}}{2018}]%
        {qiu2018real}
\bibfield{author}{\bibinfo{person}{Xiafei Qiu}, \bibinfo{person}{Wubin Cen}, \bibinfo{person}{Zhengping Qian}, \bibinfo{person}{You Peng}, \bibinfo{person}{Ying Zhang}, \bibinfo{person}{Xuemin Lin}, {and} \bibinfo{person}{Jingren Zhou}.} \bibinfo{year}{2018}\natexlab{}.
\newblock \showarticletitle{Real-time constrained cycle detection in large dynamic graphs}. In \bibinfo{booktitle}{\emph{Proceedings of the International Conference on Very Large Data Bases (PVLDB)}}. \bibinfo{pages}{1876--1888}.
\newblock


\bibitem[\protect\citeauthoryear{Ranjan, Grover, Medya, Chakaravarthy, Sabharwal, and Ranu}{Ranjan et~al\mbox{.}}{2022}]%
        {ranjan2022greed}
\bibfield{author}{\bibinfo{person}{Rishabh Ranjan}, \bibinfo{person}{Siddharth Grover}, \bibinfo{person}{Sourav Medya}, \bibinfo{person}{Venkatesan Chakaravarthy}, \bibinfo{person}{Yogish Sabharwal}, {and} \bibinfo{person}{Sayan Ranu}.} \bibinfo{year}{2022}\natexlab{}.
\newblock \showarticletitle{Greed: A neural framework for learning graph distance functions}. In \bibinfo{booktitle}{\emph{Proceedings of the Advances in Neural Information Processing Systems (NeurIPS)}}. \bibinfo{pages}{22518--22530}.
\newblock


\bibitem[\protect\citeauthoryear{Ren and Wang}{Ren and Wang}{2015}]%
        {ren2015exploiting}
\bibfield{author}{\bibinfo{person}{Xuguang Ren} {and} \bibinfo{person}{Junhu Wang}.} \bibinfo{year}{2015}\natexlab{}.
\newblock \showarticletitle{Exploiting vertex relationships in speeding up subgraph isomorphism over large graphs}. In \bibinfo{booktitle}{\emph{Proceedings of the International Conference on Very Large Data Bases (PVLDB)}}. \bibinfo{pages}{617--628}.
\newblock


\bibitem[\protect\citeauthoryear{Roy, Velugoti, Chakrabarti, and De}{Roy et~al\mbox{.}}{2022}]%
        {roy2022interpretable}
\bibfield{author}{\bibinfo{person}{Indradyumna Roy}, \bibinfo{person}{Venkata Sai Baba~Reddy Velugoti}, \bibinfo{person}{Soumen Chakrabarti}, {and} \bibinfo{person}{Abir De}.} \bibinfo{year}{2022}\natexlab{}.
\newblock \showarticletitle{Interpretable neural subgraph matching for graph retrieval}. In \bibinfo{booktitle}{\emph{Proceedings of the AAAI Conference on Artificial Intelligence (AAAI)}}. \bibinfo{pages}{8115--8123}.
\newblock


\bibitem[\protect\citeauthoryear{Sahu, Mhedhbi, Salihoglu, Lin, and {\"O}zsu}{Sahu et~al\mbox{.}}{2017}]%
        {sahu2017ubiquity}
\bibfield{author}{\bibinfo{person}{Siddhartha Sahu}, \bibinfo{person}{Amine Mhedhbi}, \bibinfo{person}{Semih Salihoglu}, \bibinfo{person}{Jimmy Lin}, {and} \bibinfo{person}{M~Tamer {\"O}zsu}.} \bibinfo{year}{2017}\natexlab{}.
\newblock \showarticletitle{The ubiquity of large graphs and surprising challenges of graph processing}. In \bibinfo{booktitle}{\emph{Proceedings of the International Conference on Very Large Data Bases (PVLDB)}}. \bibinfo{pages}{420--431}.
\newblock


\bibitem[\protect\citeauthoryear{Shang, Zhang, Lin, and Yu}{Shang et~al\mbox{.}}{2008}]%
        {shang2008taming}
\bibfield{author}{\bibinfo{person}{Haichuan Shang}, \bibinfo{person}{Ying Zhang}, \bibinfo{person}{Xuemin Lin}, {and} \bibinfo{person}{Jeffrey~Xu Yu}.} \bibinfo{year}{2008}\natexlab{}.
\newblock \showarticletitle{Taming verification hardness: an efficient algorithm for testing subgraph isomorphism}. In \bibinfo{booktitle}{\emph{Proceedings of the International Conference on Very Large Data Bases (PVLDB)}}. \bibinfo{pages}{364--375}.
\newblock


\bibitem[\protect\citeauthoryear{Shasha, Wang, and Giugno}{Shasha et~al\mbox{.}}{2002}]%
        {shasha2002algorithmics}
\bibfield{author}{\bibinfo{person}{Dennis Shasha}, \bibinfo{person}{Jason~TL Wang}, {and} \bibinfo{person}{Rosalba Giugno}.} \bibinfo{year}{2002}\natexlab{}.
\newblock \showarticletitle{Algorithmics and applications of tree and graph searching}. In \bibinfo{booktitle}{\emph{Proceedings of the Principles of Database Systems (PODS)}}. \bibinfo{pages}{39--52}.
\newblock


\bibitem[\protect\citeauthoryear{Sun and Luo}{Sun and Luo}{2020}]%
        {sun2020memory}
\bibfield{author}{\bibinfo{person}{Shixuan Sun} {and} \bibinfo{person}{Qiong Luo}.} \bibinfo{year}{2020}\natexlab{}.
\newblock \showarticletitle{In-memory subgraph matching: An in-depth study}. In \bibinfo{booktitle}{\emph{Proceedings of the International Conference on Management of Data (SIGMOD)}}. \bibinfo{pages}{1083--1098}.
\newblock


\bibitem[\protect\citeauthoryear{Sun, Sun, Che, Luo, and He}{Sun et~al\mbox{.}}{2020}]%
        {sun2020rapidmatch}
\bibfield{author}{\bibinfo{person}{Shixuan Sun}, \bibinfo{person}{Xibo Sun}, \bibinfo{person}{Yulin Che}, \bibinfo{person}{Qiong Luo}, {and} \bibinfo{person}{Bingsheng He}.} \bibinfo{year}{2020}\natexlab{}.
\newblock \showarticletitle{Rapidmatch: A holistic approach to subgraph query processing}. In \bibinfo{booktitle}{\emph{Proceedings of the International Conference on Very Large Data Bases (PVLDB)}}. \bibinfo{pages}{176--188}.
\newblock


\bibitem[\protect\citeauthoryear{Sun, Wang, Wang, Shao, and Li}{Sun et~al\mbox{.}}{2012}]%
        {sun2012efficient}
\bibfield{author}{\bibinfo{person}{Zhao Sun}, \bibinfo{person}{Hongzhi Wang}, \bibinfo{person}{Haixun Wang}, \bibinfo{person}{Bin Shao}, {and} \bibinfo{person}{Jianzhong Li}.} \bibinfo{year}{2012}\natexlab{}.
\newblock \showarticletitle{Efficient Subgraph Matching on Billion Node Graphs}. In \bibinfo{booktitle}{\emph{Proceedings of the International Conference on Very Large Data Bases (PVLDB)}}. \bibinfo{pages}{788--799}.
\newblock


\bibitem[\protect\citeauthoryear{Szklarczyk, Franceschini, Wyder, Forslund, Heller, Huerta-Cepas, Simonovic, Roth, Santos, Tsafou, et~al\mbox{.}}{Szklarczyk et~al\mbox{.}}{2015}]%
        {szklarczyk2015string}
\bibfield{author}{\bibinfo{person}{Damian Szklarczyk}, \bibinfo{person}{Andrea Franceschini}, \bibinfo{person}{Stefan Wyder}, \bibinfo{person}{Kristoffer Forslund}, \bibinfo{person}{Davide Heller}, \bibinfo{person}{Jaime Huerta-Cepas}, \bibinfo{person}{Milan Simonovic}, \bibinfo{person}{Alexander Roth}, \bibinfo{person}{Alberto Santos}, \bibinfo{person}{Kalliopi~P Tsafou}, {et~al\mbox{.}}} \bibinfo{year}{2015}\natexlab{}.
\newblock \showarticletitle{STRING v10: protein--protein interaction networks, integrated over the tree of life}.
\newblock \bibinfo{journal}{\emph{Nucleic Acids Research}} \bibinfo{volume}{43}, \bibinfo{number}{D1} (\bibinfo{year}{2015}), \bibinfo{pages}{D447--D452}.
\newblock


\bibitem[\protect\citeauthoryear{Tang, Zhang, Yao, Li, Zhang, and Su}{Tang et~al\mbox{.}}{2008}]%
        {tang2008arnetminer}
\bibfield{author}{\bibinfo{person}{Jie Tang}, \bibinfo{person}{Jing Zhang}, \bibinfo{person}{Limin Yao}, \bibinfo{person}{Juanzi Li}, \bibinfo{person}{Li Zhang}, {and} \bibinfo{person}{Zhong Su}.} \bibinfo{year}{2008}\natexlab{}.
\newblock \showarticletitle{Arnetminer: extraction and mining of academic social networks}. In \bibinfo{booktitle}{\emph{Proceedings of the International Conference on Knowledge Discovery and Data Mining (SIGKDD)}}. \bibinfo{pages}{990--998}.
\newblock


\bibitem[\protect\citeauthoryear{Vaswani, Shazeer, Parmar, Uszkoreit, Jones, Gomez, Kaiser, and Polosukhin}{Vaswani et~al\mbox{.}}{2017}]%
        {vaswani2017attention}
\bibfield{author}{\bibinfo{person}{Ashish Vaswani}, \bibinfo{person}{Noam Shazeer}, \bibinfo{person}{Niki Parmar}, \bibinfo{person}{Jakob Uszkoreit}, \bibinfo{person}{Llion Jones}, \bibinfo{person}{Aidan~N Gomez}, \bibinfo{person}{{\L}ukasz Kaiser}, {and} \bibinfo{person}{Illia Polosukhin}.} \bibinfo{year}{2017}\natexlab{}.
\newblock \showarticletitle{Attention is all you need}. In \bibinfo{booktitle}{\emph{Proceedings of the Advances in Neural Information Processing Systems (NeurIPS)}}. \bibinfo{pages}{1--11}.
\newblock


\bibitem[\protect\citeauthoryear{Veli{\v{c}}kovi{\'c}, Cucurull, Casanova, Romero, Lio, and Bengio}{Veli{\v{c}}kovi{\'c} et~al\mbox{.}}{2018}]%
        {velivckovic2018graph}
\bibfield{author}{\bibinfo{person}{Petar Veli{\v{c}}kovi{\'c}}, \bibinfo{person}{Guillem Cucurull}, \bibinfo{person}{Arantxa Casanova}, \bibinfo{person}{Adriana Romero}, \bibinfo{person}{Pietro Lio}, {and} \bibinfo{person}{Yoshua Bengio}.} \bibinfo{year}{2018}\natexlab{}.
\newblock \showarticletitle{Graph attention networks}. In \bibinfo{booktitle}{\emph{Proceedings of the International Conference on Learning Representations (ICLR)}}. \bibinfo{pages}{1--12}.
\newblock


\bibitem[\protect\citeauthoryear{Vendrov, Kiros, Fidler, and Urtasun}{Vendrov et~al\mbox{.}}{2016}]%
        {vendrov2015order}
\bibfield{author}{\bibinfo{person}{Ivan Vendrov}, \bibinfo{person}{Ryan Kiros}, \bibinfo{person}{Sanja Fidler}, {and} \bibinfo{person}{Raquel Urtasun}.} \bibinfo{year}{2016}\natexlab{}.
\newblock \showarticletitle{Order-embeddings of images and language}. In \bibinfo{booktitle}{\emph{Proceedings of the International Conference on Learning Representations (ICLR)}}. \bibinfo{pages}{1--12}.
\newblock


\bibitem[\protect\citeauthoryear{Wang, Lian, Liu, Wen, Chen, and Wang}{Wang et~al\mbox{.}}{2022a}]%
        {wang2022powerful}
\bibfield{author}{\bibinfo{person}{Hanchen Wang}, \bibinfo{person}{Defu Lian}, \bibinfo{person}{Wanqi Liu}, \bibinfo{person}{Dong Wen}, \bibinfo{person}{Chen Chen}, {and} \bibinfo{person}{Xiaoyang Wang}.} \bibinfo{year}{2022}\natexlab{a}.
\newblock \showarticletitle{Powerful graph of graphs neural network for structured entity analysis}. In \bibinfo{booktitle}{\emph{Proceedings of the Web Conference (WWW)}}. \bibinfo{pages}{609--629}.
\newblock


\bibitem[\protect\citeauthoryear{Wang, Lian, Zhang, Qin, He, Lin, and Lin}{Wang et~al\mbox{.}}{2021}]%
        {wang2021binarized}
\bibfield{author}{\bibinfo{person}{Hanchen Wang}, \bibinfo{person}{Defu Lian}, \bibinfo{person}{Ying Zhang}, \bibinfo{person}{Lu Qin}, \bibinfo{person}{Xiangjian He}, \bibinfo{person}{Yiguang Lin}, {and} \bibinfo{person}{Xuemin Lin}.} \bibinfo{year}{2021}\natexlab{}.
\newblock \showarticletitle{Binarized graph neural network}. In \bibinfo{booktitle}{\emph{Proceedings of the Web Conference (WWW)}}. \bibinfo{pages}{825--848}.
\newblock


\bibitem[\protect\citeauthoryear{Wang, Lian, Zhang, Qin, and Lin}{Wang et~al\mbox{.}}{2020}]%
        {wang2020gognn}
\bibfield{author}{\bibinfo{person}{Hanchen Wang}, \bibinfo{person}{Defu Lian}, \bibinfo{person}{Ying Zhang}, \bibinfo{person}{Lu Qin}, {and} \bibinfo{person}{Xuemin Lin}.} \bibinfo{year}{2020}\natexlab{}.
\newblock \showarticletitle{Gognn: Graph of graphs neural network for predicting structured entity interactions}. In \bibinfo{booktitle}{\emph{Proceedings of the International Joint Conference on Artificial Intelligence (IJCAI)}}. \bibinfo{pages}{1317--1323}.
\newblock


\bibitem[\protect\citeauthoryear{Wang, Zhang, Qin, Wang, Zhang, and Lin}{Wang et~al\mbox{.}}{2022b}]%
        {wang2022reinforcement}
\bibfield{author}{\bibinfo{person}{Hanchen Wang}, \bibinfo{person}{Ying Zhang}, \bibinfo{person}{Lu Qin}, \bibinfo{person}{Wei Wang}, \bibinfo{person}{Wenjie Zhang}, {and} \bibinfo{person}{Xuemin Lin}.} \bibinfo{year}{2022}\natexlab{b}.
\newblock \showarticletitle{Reinforcement Learning Based Query Vertex Ordering Model for Subgraph Matching}. In \bibinfo{booktitle}{\emph{Proceedings of the International Conference on Data Engineering (ICDE)}}. \bibinfo{pages}{245--258}.
\newblock


\bibitem[\protect\citeauthoryear{Wasserman and Faust}{Wasserman and Faust}{1994}]%
        {wasserman1994social}
\bibfield{author}{\bibinfo{person}{Stanley Wasserman} {and} \bibinfo{person}{Katherine Faust}.} \bibinfo{year}{1994}\natexlab{}.
\newblock \showarticletitle{Social network analysis: Methods and applications}.
\newblock  (\bibinfo{year}{1994}).
\newblock


\bibitem[\protect\citeauthoryear{Watts and Strogatz}{Watts and Strogatz}{1998}]%
        {watts1998collective}
\bibfield{author}{\bibinfo{person}{Duncan~J Watts} {and} \bibinfo{person}{Steven~H Strogatz}.} \bibinfo{year}{1998}\natexlab{}.
\newblock \showarticletitle{Collective dynamics of ‘small-world’networks}.
\newblock \bibinfo{journal}{\emph{Nature}} \bibinfo{volume}{393}, \bibinfo{number}{6684} (\bibinfo{year}{1998}), \bibinfo{pages}{440--442}.
\newblock


\bibitem[\protect\citeauthoryear{Wu, Pan, Chen, Long, Zhang, and Philip}{Wu et~al\mbox{.}}{2020}]%
        {wu2020comprehensive}
\bibfield{author}{\bibinfo{person}{Zonghan Wu}, \bibinfo{person}{Shirui Pan}, \bibinfo{person}{Fengwen Chen}, \bibinfo{person}{Guodong Long}, \bibinfo{person}{Chengqi Zhang}, {and} \bibinfo{person}{S~Yu Philip}.} \bibinfo{year}{2020}\natexlab{}.
\newblock \showarticletitle{A comprehensive survey on graph neural networks}.
\newblock \bibinfo{journal}{\emph{IEEE Transactions on Neural Networks and Learning Systems}} \bibinfo{volume}{32}, \bibinfo{number}{1} (\bibinfo{year}{2020}), \bibinfo{pages}{4--24}.
\newblock


\bibitem[\protect\citeauthoryear{Xu, Hu, Leskovec, and Jegelka}{Xu et~al\mbox{.}}{2019a}]%
        {xu2019powerful}
\bibfield{author}{\bibinfo{person}{Keyulu Xu}, \bibinfo{person}{Weihua Hu}, \bibinfo{person}{Jure Leskovec}, {and} \bibinfo{person}{Stefanie Jegelka}.} \bibinfo{year}{2019}\natexlab{a}.
\newblock \showarticletitle{How powerful are graph neural networks?}. In \bibinfo{booktitle}{\emph{Proceedings of the International Conference on Learning Representations (ICLR)}}. \bibinfo{pages}{1--17}.
\newblock


\bibitem[\protect\citeauthoryear{Xu, Wang, Yu, Feng, Song, Wang, and Yu}{Xu et~al\mbox{.}}{2019b}]%
        {xu2019cross}
\bibfield{author}{\bibinfo{person}{Kun Xu}, \bibinfo{person}{Liwei Wang}, \bibinfo{person}{Mo Yu}, \bibinfo{person}{Yansong Feng}, \bibinfo{person}{Yan Song}, \bibinfo{person}{Zhiguo Wang}, {and} \bibinfo{person}{Dong Yu}.} \bibinfo{year}{2019}\natexlab{b}.
\newblock \showarticletitle{Cross-lingual knowledge graph alignment via graph matching neural network}.
\newblock \bibinfo{journal}{\emph{arXiv preprint arXiv:1905.11605}} (\bibinfo{year}{2019}).
\newblock


\bibitem[\protect\citeauthoryear{Yan, Yu, and Han}{Yan et~al\mbox{.}}{2004}]%
        {yan2004graph}
\bibfield{author}{\bibinfo{person}{Xifeng Yan}, \bibinfo{person}{Philip~S Yu}, {and} \bibinfo{person}{Jiawei Han}.} \bibinfo{year}{2004}\natexlab{}.
\newblock \showarticletitle{Graph indexing: a frequent structure-based approach}. In \bibinfo{booktitle}{\emph{Proceedings of the International Conference on Management of Data (SIGMOD)}}. \bibinfo{pages}{335--346}.
\newblock


\bibitem[\protect\citeauthoryear{Ye, Lian, and Chen}{Ye et~al\mbox{.}}{2024}]%
        {ye2024efficient}
\bibfield{author}{\bibinfo{person}{Yutong Ye}, \bibinfo{person}{Xiang Lian}, {and} \bibinfo{person}{Mingsong Chen}.} \bibinfo{year}{2024}\natexlab{}.
\newblock \showarticletitle{Efficient exact subgraph matching via gnn-based path dominance embedding}.
\newblock \bibinfo{journal}{\emph{Proceedings of the International Conference on Very Large Data Bases (PVLDB)}} \bibinfo{volume}{17}, \bibinfo{number}{7} (\bibinfo{year}{2024}), \bibinfo{pages}{1628--1641}.
\newblock


\bibitem[\protect\citeauthoryear{Ying, You, Morris, Ren, Hamilton, and Leskovec}{Ying et~al\mbox{.}}{2018}]%
        {ying2018hierarchical}
\bibfield{author}{\bibinfo{person}{Zhitao Ying}, \bibinfo{person}{Jiaxuan You}, \bibinfo{person}{Christopher Morris}, \bibinfo{person}{Xiang Ren}, \bibinfo{person}{Will Hamilton}, {and} \bibinfo{person}{Jure Leskovec}.} \bibinfo{year}{2018}\natexlab{}.
\newblock \showarticletitle{Hierarchical graph representation learning with differentiable pooling}. In \bibinfo{booktitle}{\emph{Proceedings of the Advances in Neural Information Processing Systems (NeurIPS)}}. \bibinfo{pages}{4805--4815}.
\newblock


\bibitem[\protect\citeauthoryear{Yu, Liu, Fang, and Zhang}{Yu et~al\mbox{.}}{2023}]%
        {yu2023learning}
\bibfield{author}{\bibinfo{person}{Xingtong Yu}, \bibinfo{person}{Zemin Liu}, \bibinfo{person}{Yuan Fang}, {and} \bibinfo{person}{Xinming Zhang}.} \bibinfo{year}{2023}\natexlab{}.
\newblock \showarticletitle{Learning to count isomorphisms with graph neural networks}. In \bibinfo{booktitle}{\emph{Proceedings of the AAAI Conference on Artificial Intelligence (AAAI)}}. \bibinfo{pages}{4845--4853}.
\newblock


\bibitem[\protect\citeauthoryear{Yuan, Ma, Wen, Zhang, and Wang}{Yuan et~al\mbox{.}}{2021}]%
        {yuan2021subgraph}
\bibfield{author}{\bibinfo{person}{Ye Yuan}, \bibinfo{person}{Delong Ma}, \bibinfo{person}{Zhenyu Wen}, \bibinfo{person}{Zhiwei Zhang}, {and} \bibinfo{person}{Guoren Wang}.} \bibinfo{year}{2021}\natexlab{}.
\newblock \showarticletitle{Subgraph matching over graph federation}. In \bibinfo{booktitle}{\emph{Proceedings of the International Conference on Very Large Data Bases (PVLDB)}}. \bibinfo{pages}{437--450}.
\newblock


\bibitem[\protect\citeauthoryear{Zhang, Cui, Neumann, and Chen}{Zhang et~al\mbox{.}}{2018}]%
        {zhang2018end}
\bibfield{author}{\bibinfo{person}{Muhan Zhang}, \bibinfo{person}{Zhicheng Cui}, \bibinfo{person}{Marion Neumann}, {and} \bibinfo{person}{Yixin Chen}.} \bibinfo{year}{2018}\natexlab{}.
\newblock \showarticletitle{An end-to-end deep learning architecture for graph classification}. In \bibinfo{booktitle}{\emph{Proceedings of the AAAI Conference on Artificial Intelligence (AAAI)}}. \bibinfo{pages}{4438--4445}.
\newblock


\bibitem[\protect\citeauthoryear{Zhang and Yu}{Zhang and Yu}{2022}]%
        {zhang2022relative}
\bibfield{author}{\bibinfo{person}{Yikai Zhang} {and} \bibinfo{person}{Jeffrey~Xu Yu}.} \bibinfo{year}{2022}\natexlab{}.
\newblock \showarticletitle{Relative Subboundedness of Contraction Hierarchy and Hierarchical 2-Hop Index in Dynamic Road Networks}. In \bibinfo{booktitle}{\emph{Proceedings of the International Conference on Management of Data (SIGMOD)}}. \bibinfo{pages}{1992--2005}.
\newblock


\bibitem[\protect\citeauthoryear{Zhang, Zheng, Zhang, Peng, and Zhang}{Zhang et~al\mbox{.}}{2022}]%
        {zhang2022hybrid}
\bibfield{author}{\bibinfo{person}{Yuejia Zhang}, \bibinfo{person}{Weiguo Zheng}, \bibinfo{person}{Zhijie Zhang}, \bibinfo{person}{Peng Peng}, {and} \bibinfo{person}{Xuecang Zhang}.} \bibinfo{year}{2022}\natexlab{}.
\newblock \showarticletitle{Hybrid Subgraph Matching Framework Powered by Sketch Tree for Distributed Systems}. In \bibinfo{booktitle}{\emph{Proceedings of the International Conference on Data Engineering (ICDE)}}. \bibinfo{pages}{1031--1043}.
\newblock


\bibitem[\protect\citeauthoryear{Zhao and Han}{Zhao and Han}{2010}]%
        {zhao2010graph}
\bibfield{author}{\bibinfo{person}{Peixiang Zhao} {and} \bibinfo{person}{Jiawei Han}.} \bibinfo{year}{2010}\natexlab{}.
\newblock \showarticletitle{On graph query optimization in large networks}. In \bibinfo{booktitle}{\emph{Proceedings of the International Conference on Very Large Data Bases (PVLDB)}}. \bibinfo{pages}{340--351}.
\newblock


\bibitem[\protect\citeauthoryear{Zhu, Ng, and Cheng}{Zhu et~al\mbox{.}}{2011}]%
        {zhu2011structure}
\bibfield{author}{\bibinfo{person}{Linhong Zhu}, \bibinfo{person}{Wee~Keong Ng}, {and} \bibinfo{person}{James Cheng}.} \bibinfo{year}{2011}\natexlab{}.
\newblock \showarticletitle{Structure and attribute index for approximate graph matching in large graphs}.
\newblock \bibinfo{journal}{\emph{Information Systems}} \bibinfo{volume}{36}, \bibinfo{number}{6} (\bibinfo{year}{2011}), \bibinfo{pages}{958--972}.
\newblock


\end{thebibliography}

\end{document}